\documentclass[11pt]{article}
\usepackage{amsmath,amsfonts,amssymb,amsthm}
\usepackage{fullpage}
\usepackage{algorithm}
\usepackage{algorithmic}
\usepackage{subfigure}
\usepackage{epstopdf}
\usepackage{graphicx}
\usepackage{cite}
\usepackage{color}
\usepackage[sans]{dsfont}
\usepackage{mathtools}
\usepackage{hyperref}
\usepackage{cleveref}
\usepackage{tcolorbox}
\usepackage{multirow}
\usepackage{xspace}
\tcbuselibrary{skins,breakable}
\tcbset{enhanced jigsaw}

\newcommand{\myparskip}{3pt}
\parskip \myparskip
\newtheorem{theorem}{Theorem}
\newtheorem{heuristic algorithm}{Heuristic Algorithm}
\newtheorem{lemma}{Lemma}[section]
\newtheorem{definition}{Definition}[section]
\newtheorem{corollary}{Corollary}[section]
\newtheorem{claim}{Claim}[section]

\newtheorem{proposition}{Proposition}[section]
\newtheorem{problem}{Problem}
\newtheorem{observation}{Observation}[section]

\newtheorem{fact}{Fact}

\makeatletter\@addtoreset{section}{part}\makeatother%

\begin{document}
\bibliographystyle{alpha}

\newcommand{\algline}{
	% \nointerlineskip \vspace{\baselineskip}%
	\rule{0.5\linewidth}{.1pt}\hspace{\fill}%
	\par\nointerlineskip \vspace{.1pt}
}
% Table: 
\newenvironment{tbox}{\begin{tcolorbox}[
		enlarge top by=5pt,
		enlarge bottom by=5pt,
		breakable,
		boxsep=0pt,
		left=4pt,
		right=4pt,
		top=10pt,
		boxrule=1pt,toprule=1pt,
		colback=white,
		arc=-1pt,
		%          drop shadow={black,opacity=1},
		% arc=0pt
		]%%
	}
	{\end{tcolorbox}}

%--------------------------------------------------------------
%--------------------------------------------------------------
%Comments
%--------------------------------------------------------------
%--------------------------------------------------------------

\newenvironment{proofof}[1]{\noindent{\bf Proof of #1.}}
{\hspace*{\fill}\stopproof}

\newenvironment{properties}[2][0]
{\renewcommand{\theenumi}{#2\arabic{enumi}}
	\begin{enumerate} \setcounter{enumi}{#1}}{\end{enumerate}\renewcommand{\theenumi}{\arabic{enumi}}}

\newif\ifnocomments
%\nocommentstrue

% \nocommentsfalse %Uncommenting this  gives the comments version
%%%

\ifnocomments

\newcommand{\znote}[1]{}

\else
\newcommand{\znote}[1]{\textcolor{red}{\sc{[ZT: #1]}}}

\fi

\ifnocomments

\newcommand{\snote}[1]{}

\else
\newcommand{\snote}[1]{\textcolor{red}{\sc{[SK: #1]}}}

\fi

%--------------------------------------------------------------
%--------------------------------------------------------------
%Special Notations
%--------------------------------------------------------------
%--------------------------------------------------------------

\newcommand{\tG}{\textbf{G}}
\newcommand{\tH}{\textbf{H}}
\newcommand{\tE}{\textbf{E}'}
\newcommand{\tC}{\textbf{C}}
\newcommand{\tphi}{\bm{\phi}}
\newcommand{\tpsi}{\bm{\psi}}
\newcommand{\tSigma}{\bm{\Sigma}}
\newcommand{\tB}{\tilde B}
\newcommand{\dout}{D_{\mbox{\tiny{out}}}}
\newcommand{\notF}{\overline{F}}

%--------------------------------------------------------------
%--------------------------------------------------------------
%Complexity Classes
%--------------------------------------------------------------
%--------------------------------------------------------------
\renewcommand{\P}{\mbox{\sf P}}
\newcommand{\NP}{\mbox{\sf NP}}
\newcommand{\PCP}{\mbox{\sf PCP}}
\newcommand{\ZPP}{\mbox{\sf ZPP}}
\newcommand{\DTIME}{\mbox{\sf DTIME}}
\newcommand{\opt}{\mathsf{OPT}}
\newcommand{\optcro}{\mathsf{OPT}_{\mathsf{cr}}}
\newcommand{\optcrors}{\mathsf{OPT}_{\mathsf{cnwrs}}}
%--------------------------------------------------------------
%--------------------------------------------------------------
%Sets
%--------------------------------------------------------------
%--------------------------------------------------------------
\newcommand{\set}[1]{\left\{ #1 \right\}}
\newcommand{\sse}{\subseteq}
\newcommand{\B}{{\mathcal{B}}}
\newcommand{\tset}{{\mathcal T}}
\newcommand{\uset}{{\mathcal U}}
\newcommand{\iset}{{\mathcal{I}}}
\newcommand{\pset}{{\mathcal{P}}}
\newcommand{\nset}{{\mathcal{N}}}
\newcommand{\dset}{{\mathcal{D}}}
\newcommand{\tpset}{\tilde{\mathcal{P}}}
\newcommand{\qset}{{\mathcal{Q}}}
\newcommand{\tqset}{\tilde{\mathcal{Q}}}
\newcommand{\lset}{{\mathcal{L}}}
\newcommand{\bset}{{\mathcal{B}}}
\newcommand{\tbset}{\tilde{\mathcal{B}}}
\newcommand{\aset}{{\mathcal{A}}}
\newcommand{\cset}{{\mathcal{C}}}
\newcommand{\fset}{{\mathcal{F}}}
\newcommand{\mset}{{\mathcal M}}
\newcommand{\jset}{{\mathcal{J}}}
\newcommand{\xset}{{\mathcal{X}}}
\newcommand{\wset}{{\mathcal{W}}}
\newcommand{\gset}{{\mathcal{G}}}
\newcommand{\oset}{{\mathcal{O}}}
\newcommand{\yset}{{\mathcal{Y}}}
\newcommand{\rset}{{\mathcal{R}}}
\newcommand{\I}{{\mathcal I}}
\newcommand{\hset}{{\mathcal{H}}}
\newcommand{\sset}{{\mathcal{S}}}
\newcommand{\zset}{{\mathcal{Z}}}
\newcommand{\notu}{\overline U}
\newcommand{\vol}{\operatorname{vol}}
\newcommand{\nots}{\overline S}
\newcommand{\eint}{E^{\tiny\mbox{int}}}
\newcommand{\event}{{\cal{E}}}
\newcommand{\floor}[1]{\ensuremath{\left\lfloor#1\right\rfloor}}
\newcommand{\ceil}[1]{\ensuremath{\left\lceil#1\right\rceil}}
%--------------------------------------------------------------
%--------------------------------------------------------------

\newcommand{\marcon}{{\mathsf{MC}}}
\newcommand{\cov}{{\mathsf{cov}}}
\newcommand{\mst}{{\mathsf{MST}}}
\newcommand{\card}[1]{|#1|}
\newcommand{\coi}{{\mathsf{COI}}}

%--------------------------------------------------------------
%--------------------------------------------------------------
%For Metric TSP
%--------------------------------------------------------------
%--------------------------------------------------------------
\newcommand{\cover}{\textsf{cover}}
\newcommand{\eps}{\varepsilon}
\newcommand{\bfs}{\textnormal{\textsf{BFS}}}
\newcommand{\pbfs}{\textnormal{\textsf{BFS}}}
\newcommand{\lv}{\textsf{lv}}
\newcommand{\tsp}{\mathsf{TSP}}
\newcommand{\gtsp}{\textsf{GTSP}}
\newcommand{\ebt}{\tset}
\newcommand{\eb}{\textsf{EB}}
\newcommand{\optmst}{\textsf{MST}}
\newcommand{\defi}{\textsf{def}}
\newcommand{\ord}{\textsf{ord}}
\newcommand{\rc}{\textnormal{\textsf{rc}}}
\newcommand{\dist}{\textnormal{\textsf{dist}}}
\newcommand{\cost}{\textnormal{\textsf{cost}}}
\newcommand{\bw}{\textsf{bw}}
\newcommand{\local}{\textsf{Local}}
\newcommand{\pseudo}{\textsf{Pseudo-IP}}
\newcommand{\vin}{v^{\textnormal{\textsf{in}}}}
\newcommand{\vout}{v^{\textnormal{\textsf{out}}}}
\newcommand{\diam}{\textsf{diam}}
\newcommand{\expect}{\mathbb{E}}
\newcommand{\proover}{\pi_{\textsf{Overwrite}}}
\newcommand{\promst}{\pi_{\textsf{MST}}}
\newcommand{\protsp}{\pi_{\textsf{TSP}}}
\newcommand{\mstest}{\textsf{MST}_{\textsf{apx}}}
\newcommand{\tspest}{\textsf{TSP}_{\textsf{apx}}}
\newcommand{\proind}{\pi_{\textsf{Index}}}
\newcommand{\ind}{\textsf{Index}}
\newcommand{\distIND}{\mathcal{D}_{\textsf{Index}}}
\newcommand{\distMST}{\mathcal{D}_{\textsf{MST}}}
\newcommand{\ic}{\textnormal{\textsf{IC}}}
\newcommand{\cc}{\textnormal{\textsf{CC}}}
\newcommand{\tvd}[2]{\ensuremath{\Delta_{\textnormal{\texttt{TV}}}(#1,#2)}}
\newcommand{\dkl}[2]{\ensuremath{D_{\textnormal{\textsf{KL}}}(#1 \| #2)}}

\newcommand{\hel}{h}
\newcommand{\II}{I}
\newcommand{\HH}{H}

\newcommand{\RV}[1]{\mathbf{#1}}
\newcommand{\prot}{\ensuremath{\Pi}}
\newcommand{\Prot}{\ensuremath{\Pi}}
\newcommand{\findmiss}{\sf{FindBit}}
\newcommand{\overwrite}{\sf{Overwrite}}
\newcommand{\distfind}{\mathcal{D}_{\textsf{FindBit}}}
\newcommand{\distover}{\mathcal{D}_{\textsf{Overwrite}}}
\newcommand{\temp}{\textsf{temp}}
\newcommand{\IA}{\textsf{IA}}
\newcommand{\IB}{\textsf{IB}}

\newcommand{\row}{\textsf{Row}}
\newcommand{\col}{\textsf{Col}}
\newcommand{\alg}{\ensuremath{\mathsf{Alg}}\xspace}

\newcommand{\opttsp}{\textnormal{\textsf{TSP}}}

\newcommand{\sep}{\sf{sep}}
\newcommand{\core}{\sf{core}}
\newcommand{\scut}{\sf{Shortcut}}
\newcommand{\adv}{\mathsf{adv}}
\newcommand{\lig}{\sf{light}}
\newcommand{\maxmat}{\mathsf{MM}}
\newcommand{\midd}{\mathsf{mid}}
\newcommand{\bottom}{\mathsf{bot}}
\newcommand{\topp}{\mathsf{top}}
\newcommand{\snfl}{tree\xspace}
\newcommand{\snfls}{trees\xspace}
\newcommand{\inn}{\sf in}
\newcommand{\wD}{w_{\downarrow}}
\newcommand{\wU}{w_{\uparrow}}
\newcommand{\walkcost}{\mathsf{MWC}}

\begin{titlepage}
	
	\title{Sublinear Algorithms and Lower Bounds for Estimating MST and TSP Cost in General Metrics}

%\iffalse	
\author{Yu Chen\thanks{University of Pennsylvania, Philadelphia, PA, USA. Email: {\tt chenyu2@cis.upenn.edu}.}   \and Sanjeev Khanna\thanks{University of Pennsylvania, Philadelphia, PA, USA. Email: {\tt  sanjeev@cis.upenn.edu}.} \and Zihan Tan\thanks{University of Chicago, Chicago, IL, USA. Email: {\tt zihantan@uchicago.edu}.}}
%\fi
	\maketitle

	\thispagestyle{empty}
	\begin{abstract}

We consider the design of sublinear space and query complexity algorithms for estimating the cost of a minimum spanning tree (MST) and the cost of a minimum traveling salesman (TSP) tour in a metric on $n$ points. We start by exploring this estimation task in the regime of $o(n)$ space, when the input is presented as a stream of all $\binom{n}{2}$ entries of the metric in an arbitrary order (a metric stream). For any $\alpha \ge 2$, we show that both MST and TSP cost can be $\alpha$-approximated using $\tilde{O}(n/\alpha)$ space, and moreover, $\Omega(n/\alpha^2)$ space is necessary for this task. We further show that even if the streaming algorithm is allowed $p$ passes over a metric stream, it still requires $\tilde{\Omega}(\sqrt{n/\alpha p^2})$ space. %So while the space complexity significantly improves as one transitions from graph streams to metric streams, we rule out the possibility of an $\tilde O(1)$-approximation to MST/TSP cost unless polynomially many passes are allowed.

We next consider the well-studied semi-streaming regime. In this regime, it is straightforward to compute MST cost exactly even in the case where the input stream only contains the edges of a weighted graph that induce the underlying metric (a graph stream), and the main challenging problem is to estimate TSP cost to within a factor that is strictly better than $2$. We show that in graph streams, for any $\eps > 0$, any one-pass $(2-\eps)$-approximation of TSP cost requires $\Omega(\eps^2 n^2)$ space. On the other hand, we show that there is an $\tilde{O}(n)$ space two-pass algorithm that approximates the TSP cost to within a factor of 1.96.

Finally, we consider the query complexity of estimating metric TSP cost to within a factor that is strictly better than $2$ when the algorithm is given access to an $n \times n$ matrix that specifies pairwise distances between $n$ points. The problem of MST cost estimation in this model is well-understood and a $(1+\eps)$-approximation is achievable by $\tilde{O}(n/\varepsilon^{O(1)})$ queries. However, for estimating TSP cost, it is known that obtaining a $(1+\eps)$-approximation requires $\Omega(n^2)$ queries even for $(1,2)$-TSP, and for general metrics, no algorithm that achieves a better than $2$-approximation with $o(n^2)$ queries is known. We make progress on this task by designing an algorithm that performs $\tilde{O}(n^{1.5})$ distance queries and achieves a strictly better than $2$-approximation when either the metric is known to contain a spanning tree supported on weight-$1$ edges or the algorithm is given access to a minimum spanning  tree of the graph. 
%\znote{do we want this ``near" here?} \snote{We can remove it from the abstract. But I think it will be good to make a remark somewhere that all we need is a $(1+eps)$-approximate MST. This looks a bit more realistic.} 
Prior to our work, such results were only known for the special cases of graphic TSP and $(1,2)$-TSP.

In terms of techniques, our algorithms for metric TSP cost estimation in both streaming and query settings rely on estimating the {\em cover advantage} which intuitively measures the cost needed to turn an MST into an Eulerian graph. One of our main algorithmic contributions is to show that this quantity can be meaningfully estimated by a sublinear number of queries in the query model. On one hand, the fact that a metric stream reveals pairwise distances for all pairs of vertices provably helps algorithmically. On the other hand, it also seems to render useless techniques for proving space lower bounds via reductions from well-known hard communication problems.
Our main technical contribution in lower bounds is to identify and characterize the communication complexity of new problems that can serve as canonical starting point for proving metric stream lower bounds. 				
		
	\end{abstract}
\end{titlepage}

\pagenumbering{gobble}
\setcounter{tocdepth}{2}
\tableofcontents
\newpage 
\pagenumbering{arabic}

\section{Introduction}

The minimum spanning tree (MST) problem and the metric traveling salesman (TSP) problem are among the most well-studied combinatorial optimization problems with a long and rich history (see, e.g., \cite{christofides1976worst,karpinski2015new,CQ17,CQ18,Gao18,MnichM18,karlin2021slightly}\!). 
The two problems are intimately connected to one another, as many approximation algorithms for metric TSP use a minimum spanning tree as a starting point for efficiently constructing an approximate solution. In particular, any algorithm for estimating the MST cost to within a factor of $\alpha$ immediately implies an algorithm for estimating the metric TSP cost to within a factor of $2\alpha$. 
In this work, we consider the design of sublinear space and query complexity algorithms for estimating the cost of a minimum spanning tree (MST) and the cost of a minimum metric traveling salesman (TSP) tour in an $n$-vertex weighted undirected graph $G$. An equivalent view of both problems is that we are given an $n \times n$ matrix $w$ specifying pairwise distances between them, where the entry $w[u,v]$ corresponds to the weight of the shortest path from $u$ to $v$ in $G$. It is clear that any algorithm that works with a weighted graph as input also works when the input is presented as the complete metric. However, the converse is not true. For instance, no single-pass streaming algorithm can obtain a finite approximation to the diameter (or even determine the connectivity) of a graph in $o(n)$ space when the graph is presented as a sequence of edges (a {\em graph stream}). But if instead we are presented a stream of $n^2$ entries of the metric matrix $w$ (a {\em metric stream}), there is a trivial $\tilde{O}(1)$ space algorithm for this problem -- simply track the largest entry seen. 

\subsection{Our Results}

In the first part of this work, we explore the power and limitations of graph and metric streams for MST and TSP cost estimation. We start by exploring this estimation task in the regime of $o(n)$ space in the streaming model. It is easy to show that no finite approximation to MST/TSP cost is achievable in this regime when the input stream simply contains the edges of a weighted graph that induce the underlying metric (a graph stream). However, we show that this state of affairs changes completely if the input is instead presented as all entries of the shortest-path-distance metric induced by the input graph (a metric stream). 

\begin{theorem}
\label{thm: 1 pass alpha MST upper}
For any $\alpha>1$, there is a randomized one-pass $\alpha$-approximation streaming algorithm for MST cost estimation in  metric streams using $\tilde O(n/\alpha)$ space.
\end{theorem}

Note that this also immediately gives a one-pass $\tilde O(n/\alpha)$-space algorithm for TSP cost estimation for any $\alpha \ge 2$ by simply doubling the MST cost estimate. The result above is in sharp contrast to what is achievable in graph streams. Using a simple reduction from the $\mathsf{Disjointness}$ problem, we can show the following lower bound for graph streams.

\begin{theorem}
\label{thm: p pass alpha graph MST lower}
For any $\alpha>1$, any randomized $p$-pass $\alpha$-approximation streaming algorithm for MST cost estimation in graph streams requires $\tilde\Omega (n/p)$ space.
\end{theorem}

We next show that there are limits to the power of metric streams, and in particular, any non-trivial approximation of MST cost still requires polynomial space even if we allow multiple passes over the stream.

\begin{theorem}
\label{thm: 1 pass alpha MST lower}
For any $\alpha>1$, any randomized one-pass $\alpha$-approximation streaming algorithm for MST cost estimation in metric streams requires $\Omega (n/\alpha^2)$ space.
\end{theorem}
	
\begin{theorem}
\label{thm: p pass alpha MST lower}
For any $\alpha>1$, any randomized $p$-pass $\alpha$-approximation streaming algorithm for MST cost estimation requires $\tilde\Omega (\sqrt{n/\alpha p^2})$ space.
\end{theorem}

Table~\ref{tab: MST_results} summarizes our results for MST (and TSP) cost estimation in the regime of $o(n)$ space.

\setlength{\tabcolsep}{0.5em} % for the horizontal padding
{\renewcommand{\arraystretch}{1.5}% for the vertical padding

\begin{table}[h]
	\centering
	\begin{tabular}{|c|c|c|c|l|l|}
		\hline
		\multirow{2}{*}{Stream Type}  & \multicolumn{5}{c|}{MST estimation}                                                      \\ \cline{2-6} 
		& \# of passes & Approximation ratio & \multicolumn{3}{c|}{Upper or Lower bounds}       \\ \hline
		\multirow{3}{*}{Metric Stream} & 1                & $1$                 & \multicolumn{3}{c|}{$\tilde O(n)$ (trivial)}              \\ \cline{2-6} 
		& 1                & $\alpha$                 & \multicolumn{3}{c|}{$\tilde O(n/\alpha)$ (\Cref{thm: 1 pass alpha MST upper}), $\tilde \Omega(n/\alpha^2)$ (\Cref{thm: 1 pass alpha MST lower})}              \\ \cline{2-6} 
		& $p$                & $\alpha$             & \multicolumn{3}{c|}{$\tilde \Omega(\sqrt{n/\alpha p^2})$ (\Cref{thm: p pass alpha MST lower})}              \\ \hline
		\multirow{2}{*}{Graph Stream} & 1                & $1$                 & \multicolumn{3}{c|}{$\tilde \Theta(n)$ (trivial)}              \\ \cline{2-6} 
		& $p$                & any              & \multicolumn{3}{c|}{$\Omega(n/p)$ (\Cref{thm: p pass alpha graph MST lower})}              \\ \hline
	\end{tabular}
	\caption{Summary of results for MST-cost estimation streaming algorithms.}\label{tab: MST_results}
\end{table}

We next consider the well-studied semi-streaming regime when the streaming algorithm is allowed to use $\tilde{O}(n)$ space. In this regime, it is straightforward to design a deterministic one-pass streaming algorithm to compute MST cost exactly even in graph streams, and this in turn, immediately gives an $\tilde{O}(n)$ space algorithm to estimate TSP cost to within a factor of $2$. Thus in the semi-streaming regime, the key challenging problem is to estimate TSP cost to within a factor that is strictly better than $2$. 
A special case of this problem, {\em graphic TSP cost} estimation, where the input metric corresponds to the shortest-path distances induced by an unweighted undirected graph, was recently studied in~\cite{chen2020sublinear}, and the authors gave an $\tilde{O}(n)$ space randomized one-pass streaming algorithm that achieves an $(11/6)$-approximation even in the setting of graph streams. 
This ratio was recently improved\footnote{In their paper, they give an $\tilde O(n)$-time $1.83$-approximation algorithm, which can be easily turned into a one-pass streaming algorithm with space $\tilde{O}(n)$ with the same approximation ratio.} by Behnezhad, Roghani, Rubinstein, and Saberi to $1.83$ \cite{behnezhad2023sublinear}.
However, no analogous result is known for general TSP. We show that there is in fact a good reason for this state of affairs:

\begin{theorem}
\label{thm: 1 pass TSP lower}
For any $0<\epsilon<1$, any randomized one-pass $(2\!-\!\eps)$-approximation streaming algorithm for TSP cost estimation in graph streams requires $\Omega (\epsilon^2n^2)$ space.
\end{theorem}

However, we show that the situation changes considerably once we allow two passes and indeed there is now a deterministic $\tilde{O}(n)$ space algorithm that achieves better than a $2$-approximation to TSP cost.

\begin{theorem}
\label{thm: 2 pass TSP upper}
There is a deterministic two-pass $1.96$-approximation algorithm for TSP cost estimation in graph streams using $\tilde O(n)$ space.
\end{theorem}

We note that an interesting remaining question here is if a similar result is achievable using one pass when the input is a metric stream. As a step towards understanding the power of metric streams in semi-streaming regime, we show that any one-pass algorithm that computes TSP cost exactly requires $\Omega(n^2)$ space.  
Table~\ref{tab: TSP_results} summarizes our results for TSP cost estimation in the regime of semi-streaming space.

\setlength{\tabcolsep}{0.5em} % for the horizontal padding
{\renewcommand{\arraystretch}{1.5}% for the vertical padding

\begin{table}[h]
	\centering
	\begin{tabular}{|c|c|c|c|l|l|}
		\hline
		\multirow{2}{*}{Stream Type}  & \multicolumn{5}{c|}{TSP estimation}                                                      \\ \cline{2-6} 
		& \# of passes & Approximation ratio & \multicolumn{3}{c|}{Upper or Lower bounds}                              
		\\ \hline
		\multirow{2}{*}{Metric Stream} & 1                & $1$                 & \multicolumn{3}{c|}{$\Omega(n^2)$*}              \\ \cline{2-6} 
		& 1                & $2-\eps$                 & \multicolumn{3}{c|}{ Open }              
		\\             \hline
		\multirow{3}{*}{Graph Stream} & 1                & $2$                 & \multicolumn{3}{c|}{$\tilde \Theta(n)$ (trivial)}              \\ \cline{2-6} 
		& 1                & $2-\eps$                 & \multicolumn{3}{c|}{$\tilde \Omega(\eps^2n^2)$ (\Cref{thm: 1 pass TSP lower})}              \\ \cline{2-6} 
		& 2                & $1.96$              & \multicolumn{3}{c|}{$\tilde O(n)$ (\Cref{thm: 2 pass TSP upper})}              \\ \hline
	\end{tabular}
	\caption{Summary of results for TSP-cost estimation streaming algorithms. The statements and proofs for entries marked by $(*)$ is deferred to the full version.}\label{tab: TSP_results}
\end{table}

The second part of our paper focuses on the design of sublinear query complexity algorithms for TSP cost estimation. The related problem of estimating the MST cost using sublinear queries was first studied in the graph adjacency-list model by Chazelle, Rubinfeld, and Trevisan~\cite{ChazelleRT05}. The authors gave an $\tilde{O}(d W /\eps^2)$-time algorithm to estimate the MST cost to within a factor of $(1+\eps)$ in a graph where the average degree is $d$, and all edge costs are integers in $\set{1,\ldots,W}$. For certain parameter regimes this gives a sublinear time algorithm for estimating the MST cost, but in general, this run-time need not be sublinear. In fact, it is not difficult to show that in general, even checking if a graph is connected requires $\Omega(n^2)$ queries in the graph adjacency-list model, and hence no finite approximation to MST cost can be achieved in $o(n^2)$ queries. However, the situation changes if one restricts attention to the {\em metric} MST problem where the edge weights satisfy the triangle inequality, and the algorithm is given access to an $n \times n$ matrix $w$ specifying pairwise distances between vertices. Czumaj and Sohler~\cite{czumaj2009estimating} showed that for any $\eps > 0$, there exists an $\tilde{O}(n/\eps^{O(1)})$ query algorithm that returns a $(1 + \eps)$-approximate estimate of the metric MST cost. This result immediately implies an $\tilde{O}(n/\eps^{O(1)})$ time algorithm to estimate the TSP cost to within a factor of $(2 + \eps)$ for any $\eps > 0$. In sharp contrast to this result, so far no $o(n^2)$ query algorithms are known to approximate metric TSP cost to a factor that is strictly better than $2$. In this work, we consider sublinear query algorithms for TSP cost when the algorithm is given query access to the $n \times n$ distance matrix $w$. {\em We will assume throughout the paper that all entries of $w$ are positive integers.}

For the special case of graphic TSP, where the metric corresponds to shortest path distances of some underlying connected unweighted graph, the algorithm of Chen, Kannan, and Khanna~\cite{chen2020sublinear} combined with the recent result of Behnezhad~\cite{behnezhad2021time} (which builds on the work of Yoshida et al.~\cite{yoshida2012improved} and Onak et al.~\cite{onak2012near}),
gives an $\tilde O(n)$-query $(27/14)$-approximation algorithm for estimating graphic TSP cost. The authors in~\cite{chen2020sublinear} also show that there exists an $\eps_0 > 0$, such that any algorithm that estimates the cost of graphic TSP (or even $(1,2)$-TSP) to within a $(1 + \eps_0)$-factor, necessarily requires $\Omega(n^2)$ queries. 
Later on, Behnezhad, Roghani, Rubinstein, and Saberi \cite{behnezhad2023sublinear} improved the graphic TSP result by giving an $\tilde O(n)$-query $1.83$-approximation, and they also gave an $\tilde O(n)$-query $(1.5+\eps)$-approximation algorithm for $(1,2)$-TSP.
This leaves open the following question: Is there an $o(n^2)$ query algorithm to estimate TSP cost to a factor strictly better than $2$ when the metric is {\em arbitrary}?

We make progress on this question by designing an $\tilde{O}(n^{1.5})$-query algorithm that achieves a strictly better than $2$-approximation when either the metric is known to contain a spanning tree supported on weight-$1$ edges or the algorithm is given access to a minimum spanning  tree of the graph. 
Prior to our work, such results were only known for the special cases of graphic TSP and $(1,2)$-TSP.

\begin{theorem}
\label{thm: main G_1 connected}
There is a randomized algorithm, that, given access to an $n$-point metric $w$ with the promise that $w$ contains a minimum spanning tree supported only on weight-$1$ edges, estimates with high probability the metric TSP cost to within a factor of $(2-\eps_0)$ for some universal constant $\eps_0 > 0$, by performing $\tilde O(n^{1.5})$ queries to $w$.
\end{theorem}

We note that the setting of Theorem~\ref{thm: main G_1 connected} captures as a special case graphic TSP but is considerably more general, and hence difficult.

\begin{theorem}
\label{thm: main with MST}
There is a randomized algorithm, that, given access to an $n$-point metric $w$ and an arbitrary minimum spanning tree of the complete graph with edge weights given by $w$, estimates with high probability the metric TSP cost to within a factor of $(2-\eps_0)$ for some universal constant $\eps_0>0$, by performing $\tilde O(n^{1.5})$ queries to $w$.
\end{theorem}

In what follows, we give an overview of the techniques underlying our results.

\subsection{Technical Overview}
			
\subsubsection{Overview of Algorithmic Techniques}

Our streaming algorithm for MST estimation (\Cref{thm: 1 pass alpha MST upper}) utilizes a rather natural idea. We sample $O(n/\alpha)$ vertices and maintain a MST $T'$ over them. For the remaining vertices, we maintain an estimate of the cost of connecting them to the nearest vertex in $T'$. We show that these estimates can be suitably combined to obtain an $\alpha$-approximation of MST cost.
%utilizes a standard randomized approach based on sampling a random subset of vertices and record the distances from other vertices to them. 
In this subsection we focus on providing a high-level overview of the algorithms for TSP estimation.

It is well-known that $\mst\le \tsp\le 2\cdot \mst$ holds for any graph/metric, since we can construct a TSP-tour by doubling all edges of a MST (and then shortcut the obtained walk into a tour).
Since the MST cost of a graph/metric can be exactly computed by a one-pass $\tilde O(n)$ space algorithm (the greedy algorithm) in the streaming model, and can be approximated to within a factor of $(1+\eps)$ by performing $\tilde O(n)$ queries in the query model \cite{czumaj2009estimating}, to obtain a factor $(2-\eps)$ approximation for $\tsp$, it suffices to establish either $\tsp\ge (1+\eps)\cdot\mst$ or $\tsp\le (2-\eps)\cdot\mst$ holds.
From the approach due to \cite{christofides1976worst}, the minimum weight of a perfect matching on the set of all odd-degree vertices in an MST can immediately gives us the answer. However, obtaining a good approximation to the minimum weight of such a perfect matching appears hard to do, both for semi-streaming algorithms and for a query algorithm that performs $o(n^2)$ queries, even if we are given an MST at the start. To get around this issue, we consider an alternative measure, called the \emph{cover advantage}, that turns out to be more tractable in both models.

\paragraph{Cover Advantage.} Let $T$ be a MST of the input graph/metric. For an edge $f\in E(T)$ and an edge $e\notin E(T)$, we say that $f$ is \emph{covered} by $e$, iff $f$ belongs to the unique tree-path in $T$ connecting the endpoints of $e$. 
For a set $E'$ of edges, we denote by $\cov(E', T)$ the set of all edges in $E(T)$ that are covered by at least one edge in $E'$. The \emph{cover advantage} of $E'$, denoted by $\adv(E')$, is defined to be the total weight of all edges in $\cov(E', T)$ minus the total weight of all edges in $E'$. 
Intuitively, if a single-edge set $\set{e}$ where $e=(u,v)$ has cover advantage $c$, then we can construct a tour by starting from some Euler-tour of $T$ and replacing the segment corresponding to the tree path of $T$ connecting $u$ to $v$ by the single edge $e$, and thereby ``saving a cost of $c$'' from $2\cdot \mst$, the cost of the Euler-tour obtained by doubling MST edges.
Generalizing this idea, we show that if there exists a set $E'$ with cover advantage bounded away from $0$ (at least $\eps\cdot \mst$), then $\tsp\le (2-\eps/2)\cdot \mst$. Conversely, if there does not exist any set $E'$ with cover advantage close to $\mst/2$ (say at least $(1/2-\eps/2)\cdot \mst$), then $\tsp\ge (1+\eps)\cdot \mst$.
%if the maximum cover advantage of any subset $E'$ of edges is bounded away from $0$ (at least $\eps\cdot \mst$), then $\tsp\le (2-\eps/2)\cdot \mst$, and if it is bounded away from $\mst/2$ (at most $(1/2-\eps/2)\cdot \mst$), then $\tsp\ge (1+\eps)\cdot \mst$.
In fact, we show that the same hold for a more restricted notion called \emph{special cover advantage}, which is defined to be the maximum cover advantage of any subset $E'$ of edges that have at least one endpoint being a special vertex in $T$ (a vertex $v$ is called a \emph{special vertex} of $T$ iff $\deg_T(v)\ne 2$).
Therefore, to obtain a better-than-factor-2 approximation for $\tsp$, it suffices to obtain a constant-factor approximation for the maximum cover advantage or the maximum special cover advantage.

\paragraph{Estimating maximum cover advantage in the streaming setting.}
We construct a one-pass streaming algorithm $O(1)$-estimating the maximum cover advantage, which leads to a two-pass algorithm in \Cref{thm: 2 pass TSP upper} where in the first pass we only compute an MST of the input graph. We store edges with substantial cover advantage with respect to the MST in a greedy manner. Since all edges appear in the stream, it can be shown that, if we end up not discovering a large cover advantage, then the real maximum cover advantage is indeed small (bounded away from $\mst/2$).

\paragraph{Estimating maximum cover advantage in the query model.} The task of obtaining a constant-factor approximation to maximum cover advantage turns out to be distinctly more challenging in the query model, even if we are given explicit access to an MST of the metric. The design of sublinear query algorithms for estimating cover advantage is indeed our central algorithmic contribution.
We design an $\tilde O(n^{1.5})$-query algorithm for estimating the maximum cover advantage when either an MST is explicitly given or we can assume that the MST is supported on weight-$1$ edges. Note that the latter case generalizes graphic TSP studied in \cite{chen2020sublinear}.

\iffalse
%
%In particular, we are not even given an MST in the first place. 
Although the weight of an MST can be accurately approximated with $\tilde O(n)$ queries, computing an MST (or even a spanning tree with a cost at most $(2-\eps)\cdot\mst$) requires $\Omega(n^2)$ queries, as shown in \cite{chen2020sublinear}.
Therefore, to be able to exploit the power of the cover advantage technique, we need some additional MST-related properties of the input metric. Specifically, we need additional information of the MST in the following two aspects:
\begin{enumerate}
\item What structure does the MST have?
\item What are the weights of the MST edges?
\end{enumerate}

%We design an $\tilde O(n^{1.5})$-query algorithm for estimating the maximum cover advantage when either an MST is explicitly given or we can assume that the MST is supported on weight-$1$ edges. Note that the latter case generalizes graphic TSP studied in \cite{chen2020sublinear}.
%Even if we are given access to an MST, this is our main algorithmic contribution.

The two special cases we explored in this paper reflect the two aspects mentioned above. In the first special case, we assume that the MST consists of only weight-$1$ edges, but make no assumptions on its structure.
In the second special case, we assume that we are given the structure of the MST, but make no assumptions on its edge weights.
\fi

The algorithms for these two cases share several similarities. To illustrate the ideas behind them, it might be instructive to consider the following two examples. In the first example, we are given an MST $T$ on $V$ that has at most $O(\sqrt{n})$ leaves. 
We can simply query the distances between all pairs $u,v\in V$ where $u$ is a special vertex of $T$, and then use the obtained information to compute the maximum special cover advantage, which takes $\tilde O(n^{1.5})$ queries since there can be at most $O(\sqrt n)$ special vertices (or in fact we can even query the distances between all pairs of special vertices in $T$ and compute the minimum weight perfect matching on them).
In the second example, we are given an MST $T$ on $V$ which is a star graph centered at a vertex $r\in V$, and all edges have weight $1$. Note that, since all edge weights are integers, in this case the distances between every pair of vertices in $V\setminus \set{r}$ is either $1$ or $2$, and it is not hard to see that the maximum cover advantage is exactly the size of a maximum weight-$1$ matching on $V\setminus \set{r}$.
Therefore, we can adapt the algorithm from \cite{behnezhad2021time} to obtain an $O(1)$-approximation of the maximum weight-$1$ matching size, using $\tilde{O}(n)$ queries.
Note that, in this case we obtain an estimate of the maximum cover advantage without computing a set of edges that achieves it.

Taking a step back, we observe that, in the first example where the number of special vertices is small, the cover advantage can be computed in a local and exhaustive manner, while in the second example where the number of special vertices is large, the cover advantage has to be estimated in a global and ``superficial'' manner. Intuitively, our query algorithms interpolate between these two approaches in an organic manner.

We now provide more details of our query algorithms in the two special cases. We first consider the special case where we are given the structure of an MST.

\paragraph{When MST is given:} We root the given MST $T$ at an arbitrary vertex. For each vertex $v\in V$, we say that it is \emph{light} iff the subtree of $T$ rooted at $v$, denoted by $T_v$, contains at most $\sqrt n$ vertices, and we call $T_v$ a \emph{light subtree} of $T$. On the one hand, the cover advantages that are local at some light subtree (achieved by edges with both endpoints in the same light subtree) can be efficiently estimated in an exhaustive manner. On the other hand, if we peel off all light subtrees from $T$, then the remaining subtree, that we denote by $T'$, contains at most $\sqrt{n}$ leaves, and therefore the special cover advantage achieved by any set of edges with at least one endpoint being a special vertex of $T'$ can also be computed in an exhaustive manner. The only type of cover advantages that is not yet computed are the one achieved by edges with endpoints in different light subtrees. We then observe that the light subtrees hanged at $T$ are similar to the edges of a star graph hanged at its root, and eventually manage to adapt the algorithm from \cite{behnezhad2021time} in a delicate way to estimate the cover advantage by edges of this type in a global manner.

\paragraph{When MST consists of only weight-$1$ edges:} %We now consider the special case where the MST consists of only weight-$1$ edges. 
This special case appears trickier since we do not know the structure of an MST at the start, and there may not even be a unique MST.
To circumvent this, we need to utilize the following technical result of \cite{chen2020sublinear}:
Let $G_1$ be the graph on $V$ induced by all weight-$1$ edges in the given metric, then if $G_1$ contains a size-$s$ matching consisting of only edges in $2$-edge-connected components of $G_1$, then $\tsp\le 2n-\Omega(s)$.
This result allows us to construct a local procedure that explores some neighborhood of the unknown graph $G_1$ up to a certain size, such that in the end we either reconstruct a size-$\sqrt n$ subgraph of the (locally) unique MST, or certify that a set of $\Omega(\sqrt n)$ vertices belong to some $2$-edge-connected components of $G_1$,
which will be later collected to estimate the maximum weight-$1$ matching size.

We then use this local procedure on a set of vertices randomly sampled from $V$. Let $T$ be an MST. Intuitively, if the total size of light subtrees of $T$ is non-negligible, then with high probability some of the sampled vertices will lie in light subtrees of $T$, and we can obtain an estimate of the local cover advantage within subtrees. If the total size of light subtrees is negligible, then $T'$, the subtree obtained from $T$ by peeling off all light subtrees, has roughly the same size as $T$, which means that $T$ is close to the first instructive example mentioned before -- a tree with only $O(\sqrt{n})$ special vertices. Then we can apply the local procedure to $\Omega(\sqrt{n})$ sampled vertices, to almost reconstruct the whole tree $T$, and the rest of the algorithm is similar to the algorithm in the first special case.

\subsubsection{Overview of Lower Bound Techniques}

%As it turns out, proving space lower bound for streaming algorithms in metric streams is essentially different from (and also harder than) the same task in graph streams.

As our algorithmic results illustrate that metric streams are more powerful than graph streams, it is perhaps not surprising that proving space lower bounds for metrics streams turns out to be a more challenging task that requires new tools.
To illustrate this point, it might be instructive to consider the following simplified versions of metric and graph streams. Let $G$ be a graph.
\begin{itemize}
\item Unweighted Graph Stream: a sequence that contains all edges of $E(G)$, and the same edge may appear more than once in the stream;
\item Unweighted Metric Stream: a sequence that contains, for each pair $u,v$ of $V(G)$, a symbol $f(u,v)$ indicating whether or not the edge $(u,v)$ belongs to $E(G)$.
%, and the same tuple $(u,v,f(u,v))$ may appear more than once in the stream.
\end{itemize}
Note that in unweighted metric streams, the non-edge information between pairs of vertices is also explicitly given (as the edge information), as opposed to being given implicitly in the unweighted graph stream. This seemingly unimportant distinction, unexpectedly, makes proving lower bounds for several problems much harder in unweighted metric streams than unweighted graph streams. 

For example, consider the problem of deciding whether the input graph is a clique.
On the one hand, to prove a space lower bound for streaming algorithms in unweighted graph streams, we consider the following two-player one-way communication game: Alice is given a graph $G_A$ and Bob is given a graph $G_B$ on a common vertex set $V$, and Alice and Bob want to decide if $G_A\cup G_B$ is the complete graph on $V$. It is easy to show that this communication game has back-and-forth communication complexity $\Omega(n^2)$. In fact, Alice's input graph $G_A$ can be viewed as a vector $x^A\in \set{0,1}^{\binom{V}{2}}$ and Bob's input graph $G_B$ can be viewed as a vector $x^B\in \set{0,1}^{\binom{V}{2}}$, where the coordinate $x^A_{(u,u')}$ indexed by the pair $u,u'$ of vertices in $V$ indicates whether or not the edge $(u, u')$ appears in graph $G_A$, and similarly $x^B_{(u,u')}$ indicates whether or not the edge $(u, u')$ appears graph $G_B$. It is then easy to see that the two players need to detect whether or not the bitwise-OR of vectors $x^A$ and $x^B$ is the all-one vector, which requires $\Omega(n^2)$-bits information exchange even in the back-and-forth communication model.
On the other hand, in the corresponding two-player one-way communication game for unweighted metric streams, Alice and Bob are each given a set of edge/non-edge information, with the promise that the edge/non-edge information between each pair of vertices appears in at least one of the player's input. There is a one-bit protocol: Alice simply sends to Bob a signal indicating whether or not in her input there is non-edge information between any pair of vertices, and Bob outputs ``Not a Clique'' iff either he sees Alice's ``non-edge" signal or he sees a non-edge information in his input.

The distinction that all non-edge information is explicitly given in the unweighted stream seems to fail all reductions from standard problems (like Disjointness and Index) to prove lower bounds. Therefore, in the lower bound proofs of \Cref{thm: 1 pass alpha MST lower} and \Cref{thm: p pass alpha MST lower}, we identify new ``primitive'' graph-theoretic problems, prove communication lower bounds for them, and then reduce them to MST-estimation problems. Here we briefly provide some ideas for the proof of \Cref{thm: p pass alpha MST lower}.

We consider the special type of metrics, in which the distance between every pair of vertices is either $1$ or a large enough real number. Intuitively, the problem of estimating the MST cost is equivalent to the problem of estimating the number of connected components of the graph induced by all weight-$1$ edges, which is essentially a graph-theoretic problem in unweighted metric streams.

As a first step, we consider the following problem: given an unweighted metric stream, decide whether the underlying graph is a perfect matching or a perfect matching minus one edge.
Unlike the previous clique-identification problem, we show that the corresponding two-player communication game for this problem has  communication complexity $\Omega(n)$ in the back and forth communication model, even if the complete edge/non-edge information is split between Alice and Bob. The proof is by analyzing the information complexity of any protocol for the problem,
We construct several similar input combinations for Alice and Bob, among which some cross-combination lead to different answers, and then lower bound the mutual information between the protocol transcript and the players' inputs.

However, this perfect matching vs perfect matching minus one edge problem is not sufficient for our purpose, since a perfect matching graph on $n$ vertices has $n/2$ connected components, while a perfect matching minus one edge graph on $n$ vertices has $n/2-1$ connected components, and the ratio between $n/2$ and $n/2-1$ are too small to provide a space lower bound for $\alpha$-approximation of the number of connected components.
To fix this issue, we next consider a generalization of this problem, called the \emph{Clique or Independent Set} problem ($\coi_{a,b}$) parametrized by two integers $a,b$. In this problem, we are required to decide whether the input graph is the disjoint union of $b$ cliques of size $a$ each (Yes case) or it is a disjoint union of $(b-1)$ cliques of size $a$ each and an independent set of size $a$ (No case). Note that if $a=2$ and $b=n/2$ then this problem is exactly the perfect matching vs perfect matching minus one edge problem. Now if we let $a\gg b$, then the ratio between numbers of connected components in Yes case and in No case is $(a+b-1)/b=\Omega(a/b)$, which is enough for giving a space lower bound for $o(a/b)$-approximation streaming algorithms for MST estimation.
For the proof of the communication lower bound of problem $\coi_{a,b}$, we first consider the special case $\coi_{a,2}$ and show that the communication complexity is $\tilde\Omega(1)$ via a Hellinger distance analysis on transcript distributions on certain input combinations, and then use a direct sum type argument to show that the communication complexity of $\coi_{a,b}$ is $\tilde\Omega(b)$.
Both steps use techniques similar to the ones used in the proof of communication lower bound for the perfect matching vs perfect matching minus one edge problem.
Now for a given approximation ratio $\alpha>1$, setting $a=\Theta(\sqrt{\alpha n})$ and $b=\Theta(\sqrt{n/\alpha})$ yields the desired communication lower bound, which then implies the space lower bounds for streaming algorithms.

\subsection{Organization}

We first introduce the notion of cover advantage in \Cref{subsec: cover adv}. The remainder of the paper is divided into three parts. In the first part, we present our results on streaming algorithms and lower bounds for estimating MST and TSP cost (the proofs of \Cref{thm: 1 pass alpha MST upper} to \ref{thm: 2 pass TSP upper}).
In the second part, we present the algorithm in the query model where we assume that MST is supported on weight-$1$ edges (the proof of \Cref{thm: main G_1 connected}). 
In the third part, we present the algorithm in the query model where we are additionally given access to an MST (the proof of \Cref{thm: main with MST}).

\section{Cover Advantage}
\label{subsec: cover adv}

In this section, we introduce the notion of cover advantage, which is a key notion that 
captures the gap between the MST cost and the TSP cost. Our TSP cost estimation algorithms in both streaming and query settings will crucially utilize this notion.

At a high-level, the TSP estimation algorithms in this paper are based on converting an MST $T$ of the input graph/metric into a spanning Eulerian subgraph. A trivial approach is to simply double all edges in $T$ obtaining a $2$-approximation. A more clever approach due to Christofides \cite{christofides1976worst} instead makes $T$ Eulerian by adding a minimum weight perfect matching on odd-degree vertices in $T$, obtaining a $3/2$-appproximation. However, computing a good approximation to the minimum weight perfect matching on a set of vertices appears hard to do either in the semi-streaming setting or with sublinear number of queries. We instead identify the more tractable notion, called the cover advantage, that can be efficiently implemented in the semi-streaming and query model.

We say that an edge $f$ of tree $T$ is \emph{covered} by an edge $e$ that may or may not belong to $T$, iff $f\in E(P^T_{e})$; and we say that $f$ is covered by a set $E'$ of edges, iff it is covered by some edge of $E'$. We denote by $\cov(e)$ the set of all edges of $T$ that are covered by $e$, and define $\cov(E')=\bigcup_{e\in E'}\cov(e)$.
%Additionally, we define $\cov_1(E')$ the beset of edges that are covered by an odd number of edges in $E'$.

Let $T'$ be a subtree of $T$. For each edge $e\notin E(T)$, we define $\cov(e,T')=\cov(e)\cap E(T')$. Similarly, for a set $E'$ of edges, we define $\cov(E',T')=\bigcup_{e\in E'}\cov(e,T')$.
%and $\cov_1(E',T_v)$ to be the set of edges in $E(T_v)$ that are covered by an even number of edges in $E'$. 
Clearly, $\cov(E',T')=\cov(E')\cap E(T')$.
% and $\cov_1(E',T_v)=\cov_1(E')\cap E(T_v)$.

We define the \emph{cover advantage} of a set $E'$ of edges on a subtree $T'$ of $T$, denoted by $\adv(E',T')$, to be $\adv(E',T')=w(\cov(E',T'))-w(E')$. 
The \emph{optimal cover advantage} of a subtree $T'$, denoted by $\adv(T')$, is defined to be the maximum cover advantage of any set $E'$ of edges that have at least one endpoint lying in $V(T')$ on $T'$.
The \emph{optimal special cover advantage} of a subtree $T'$, denoted by $\adv^{*}(T')$, is defined to be the maximum cover advantage of any set $E'$ of edges that have at least one endpoint being a special vertex of $T'$ (a vertex $v$ is a special vertex of $T'$ iff $\deg_{T'}(v)\ne 2$).
Clearly, by definition, $\adv(T')\ge \adv^*(T')\ge 0$.
%The \emph{optimal inner cover advantage} of a subtree $T'$, denoted by $\adv^{\inn}(T')$, is defined to be the maximum cover advantage of any set $E'$ of edges that have both endpoints in $V(T')$ on $T'$.

%We use the following lemmas.

The next two lemmas show that the optimal cover advantage and the optimal special cover advantage of any subtree can be computed using a small number of queries.

\iffalse
\begin{observation}
	For any tree $T$ and any subtree $T'$ of $T$, $\adv(T')\ge 0$.
\end{observation}
\begin{proof}
	By definition, for each edge $e\in E(T')$, $\cov(e,T')=\set{e}$, Therefore, $\cov(E(T'),T')=E(T')$, and so $\adv(E(T'),T')=0$. Then by the definition of the optimal cover advantage, $\adv(T')\ge 0$.
\end{proof}
\fi

\begin{lemma}
\label{lem: max advantage}
There is an algorithm, that given a subtree $T'$ of $T$, computes the optimal cover advantage of $T'$ as well as a set $E'$ of edges achieving the optimal cover advantage of $T'$, by performing at most $O(n\cdot |V(T')|)$ queries.
\end{lemma}
\begin{proof}
We first query the distances between all pairs $(u,u')$ of vertices with $u\in V(T')$ and $u'\in V$. We then enumerate all eligible edge sets (sets of edges with at least one endpoint in $V(T')$) and compute the cover advantage of each of them on $T'$, from the distance information we acquired. Finally, we return the eligible set $E'$ with the maximum cover advantage on $T'$, and the cover advantage on $T'$ it achieves. Clearly, we have performed $O(n\cdot |V(T')|)$ queries.
\end{proof}

Similarly, we can prove the following lemma.

\begin{lemma}
\label{lem: max special cover advantage}
There is an algorithm, that given a subtree $T'$ of $T$, computes the optimal special cover advantage of $T'$ as well as a set $E'$ of edges achieving the optimal special cover advantage of $T'$, by performing $O(n\cdot k_{T'})$ queries, where $k_{T'}$ is the number of special vertices in $T'$.
\end{lemma}

We now show that the cover advantage of edge-disjoint subtrees of an MST translates into a TSP tour whose cost is much better than $2$ times the MST cost.

\begin{lemma}
\label{lem: cover_advantage}
Let $T$ be an MST on a set $V$ of vertices, and
let $\tset$ be a set of edge-disjoint subtrees of $T$. Then $\tsp\le 2\cdot\mst-\frac{1}{2}\cdot \sum_{T'\in \tset}\adv(T')$.
\end{lemma}

\begin{proof}
We introduce some definitions before providing the proof.	

Let $E'$ be a set of edges that do not belong to $E(T)$. We define the multi-graph $H_{T,E'}$ as follows. Its vertex set is $V(H_{T,E'})=V$. Its edge set is the union of (i) the set $E'$; and (ii) the set $E_{[T,E']}$ that contains, for each edge $f\in E(T)$, $2$ copies of $f$ iff $f$ is covered by an even number of edges in $E'$, $1$ copy of $f$ iff $f$ is covered by an odd number of edges in $E'$. 
Equivalently, graph $H_{T,E'}$ can be obtained from the following iterative algorithm. Throughout, we maintain a graph $\hat H$ on the vertex set $V$, that initially contains two copies of each edge of $E(T)$. We will maintain the invariant that, over the course of the algorithm, for each edge $f$ of $E(T)$, graph $\hat H$ contains either one copy or two copies of $f$.
We then process edges of $E'$ one-by-one (in arbitrary order) as follows. Consider now an edge $e\in E'$ and the tree-path $P^T_e$. We add one copy of edge $e$ to $\hat {H}$. Then for each edge $f\in E(P^T_e)$, if currently the graph $\hat H$ contains $2$ copies of $f$, then we remove one copy of it from $\hat H$; if currently the graph $\hat H$ contains $1$ copy of $f$, then we add one copy of it into $\hat H$. Clearly after each iteration of processing some edge of $E'$, the invariant still holds. It is also easy to see that the resulting graph we obtain after processing all edges of $E'$ is exactly the graph $H_{T,E'}$ defined above.

We prove the following observation.
\begin{observation}
	\label{obs: displace_graph Eulerian}
	For any set $E'$, graph $H_{T,E'}$ is Eulerian.
\end{observation}
\begin{proof}
	Consider the algorithm that produces the graph $H_{T,E'}$. Initially, graph $\hat H$ contains $2$ copies of each edge of $T$, and is therefore Eulerian. It is easy to see that, in the iteration of processing the edge $e\in E'$, we only modify the degrees of vertices in the cycle $e\cup P^T_e$. Specifically, for each vertex in the cycle $e\cup P^T_e$, either its degree is increased by $2$ (if a copy is added to both of its incident edges in the cycle), or its degree is decreased by $2$ (if a copy is removed from both of its incident edges in the cycle), or its degree remains unchanged (if a copy is removed from one of its incident edges, and a copy is added to the other incident edge). Therefore, the graph $\hat H$ remains Eulerian after this iteration, and it follows that the resulting graph $H_{T,E'}$ is Eulerian.
\end{proof}

We now provide the proof of \Cref{lem: cover_advantage}. 
Denote $\tset=\set{T_1,\ldots,T_k}$.
For each index $1\le i\le k$, let $E^*_i$ be the set of edges that achieves the maximum cover advantage on $T_i$. Denote $E^*=\bigcup_{1\le i\le k}E^*_i$, and then we let $E'$ be the random subset of $E^*$ that includes each edge of $E^*$ independently with probability $1/2$. 
We will show that the expected total weight of all edges in $E(H_{T,E'})$ is at most $2\cdot\mst(G)-\frac{1}{2}\cdot \sum_{1\le i\le k}\adv(T_i)$. Note that this implies that there exists a subset $E^{**}$ of $E^*$, such that the weight of graph $H_{T,E^{**}}$ is at most $2\cdot\mst(G)-\frac{1}{2}\cdot \sum_{1\le i\le k}\adv(T_i)$. Combined with \Cref{obs: displace_graph Eulerian} and \Cref{lem: TSP bounded by Eulerian Multigraph}, this implies $\tsp\le 2\cdot\mst(G)-\frac{1}{2}\cdot \sum_{1\le i\le k}\adv(T_i)$, completing the proof of \Cref{lem: cover_advantage}.

We now show that $\expect[w(H_{T,E'})]\le 2\cdot\mst(G)-\frac{1}{2}\cdot \sum_{1\le i\le k}\adv(T_i)$. From the definition of graph $H_{T,E'}$, $E(H_{T,E'})=E'\cup E_{[T,E']}$. On one hand, from the construction of set $E'$, $\expect[w(E')]= w(E^*)/2$.
On the other hand, for each edge $f\in \cov(E^*)$, with probability $1/2$ graph $H_{T,E'}$ contains $1$ copy of it, and with probability $1/2$ graph $H_{T,E'}$ contains $2$ copies of it. Therefore, $\expect[w(E_{[T,E']})]= 2\cdot w(E(T))-w(\cov(E^*))/2$.
Note that subtrees $\set{T_i}_{1\le i\le k}$ are edge-disjoint, so the edge sets $\set{\cov(E^*,T_i)}_{1\le i\le k}$ are mutually disjoint.
Altogether, 
\begin{equation*}
\begin{split}
\expect[w(H_{T,E'})] & =2\cdot \mst-\frac{w(\cov(E^*))-w(E^*)}{2}\\
& \le 2\cdot \mst-\frac{1}{2}\cdot\sum_{1\le i\le k}\bigg(w(\cov(E^*,T_i))-w(E^*_i)\bigg)\\
& \le 2\cdot \mst-\frac{1}{2}\cdot\sum_{1\le i\le k}\bigg(w(\cov(E^*_i,T_i))-w(E^*_i)\bigg)\\
& = 2\cdot \mst-\frac{1}{2}\cdot\sum_{1\le i\le k}\adv(T_i).
\end{split}
\end{equation*}
\end{proof}

We show in the next lemma that, although the special cover advantage of any tree is upper bounded by its cover advantage by definition, but it is greater than the cover advantage of any single edge.

\begin{lemma}
\label{lem: single_edge_adv}
Let $T'$ be a subtree of $T$. Let $e=(u,v)$ be an edge with $u\in V(T')$ and $v\notin V(T')$, then $\adv(e,T')\le \adv^*(T')$.
\end{lemma}
\begin{proof}
If $u$ is a special vertex of $T'$, then the lemma follows immediately from the definition of $\adv^*(T')$. Assume now that $u$ is not a special vertex of $T'$, so $\deg_{T'}(u)=2$.
Since $u\in V(T')$ and $v\notin V(T')$, $\cov(e,T')$ is a path $P$ of $T'$ with $u$ being one of its endpoint. Clearly. we can find another leaf $u'$ of $T'$ such that the tree-path $P^{T'}_{u,u'}$ is edge-disjoint from $P$. Consider now the edge $e'=(u',v)$. On the one hand, since $u'$ is a special vertex of $T'$, $\adv(e',T')\le \adv^*(T')$. On the other hand, clearly $\cov(e',T')=P^{T'}_{u,u'}\cup P$, and from triangle inequality, $w(u',v)\le w(u',u)+w(u,v)\le w(P^{T'}_{u,u'})+w(u,v)$. Therefore, $\adv(e',T')\ge \adv(e,T')$. Altogether, we get that $\adv(e,T')\le \adv^*(T')$.
\end{proof}

We show in the next technical lemma that the edge set of any TSP tour can be partitioned into two subsets that both cover the whole tree.

\begin{lemma}
\label{lem: tour edge into two cover}
Let $T'$ be a subtree of $T$ and let $\pi$ be a tour that traverses all vertices of $T'$. Then $E(\pi)$ can be partitioned into two subsets $E(\pi)=E_0\cup E_1$, such that $E(T')\subseteq \cov(E_0)$ and $E(T')\subseteq \cov(E_1)$.
\end{lemma}
\begin{proof}
Let $V'$ be the set of all odd-degree vertices in $T'$, so $|V'|$ is even. %Let $\pi'$ be obtained from $\pi$ by deleting all vertices of $V\setminus V'$. Clearly, $|V'|$ is even and $\pi'$ is a tour that traverses all vertices of $V'$. 
Denote $V'=\set{v_1,v_2,\ldots,v_{2k}}$, where the vertices are index according to the order in which they appear in $\pi$. 
For each $1\le i\le 2k$, we define edge $e_i=(v_i,v_{i+1})$ and define $E^{i}_{\pi}$ to be the set of all edges traversed by $\pi$ between vertices $v_i$ and $v_{i+1}$.
Clearly, $\cov(e_i)\subseteq \cov (E^{i}_{\pi})$.
We define $E_0=\bigcup_{0\le i\le k-1}E^{2i}_{\pi}$ and $E_1=\bigcup_{0\le i\le k-1}E^{2i+1}_{\pi}$.

Consider now the tour $\pi'$ induced by edges of $e_1,\ldots,e_{2k}$. Clearly, $\pi'$ is a tour that visits all vertices of $V'$. We define sets
$F_0=\set{e_{2i}\mid 0\le i\le k-1}$ and $F_1=\set{e_{2i+1}\mid 0\le i\le k-1}$, so  $E(\pi')=F_0\cup F_1$. 
We now show that $E(T')\subseteq \cov(F_0)$ and $E(T')\subseteq \cov(F_1)$.
Note that this implies that $E(T')\subseteq \cov(E_0)$ and $E(T')\subseteq \cov(E_1)$, since 
\[\cov(F_0)=
\bigg(\bigcup_{0\le i\le k-1}\cov(e_{2i})\bigg)
\subseteq 
\bigg(\bigcup_{0\le i\le k-1}\cov(E^{2i}_{\pi})\bigg) \subseteq \cov \bigg(\bigcup_{0\le i\le k-1}E^{2i}_{\pi}\bigg)=\cov(E_0),\]
and similarly $\cov(F_1)\subseteq \cov(E_1)$.

In fact, note that $F_0$ is a perfect matching on $V'$. Since $V'$ is the set of odd-degree vertices of $T'$, the graph on $V(T')$ induced by edges of $E(T')\cup F_0$ is Eulerian. Therefore, every edge of $T'$ appears in at least two sets of $\set{\cov(e')\mid e'\in E(T')\cup M_o}$. Note that for each $e'\in E(T')$, $\cov(e')=\set{e'}$. Therefore, every edge $e\in E(T')$ appears in at least one set of $\set{\cov(e')\mid e'\in F_0}$, i.e., $E(T')\subseteq \cov(F_0)$. Similarly, we get that $E(T')\subseteq \cov(F_1)$. 
This completes the proof of \Cref{lem: tour edge into two cover}.
\end{proof}

Using \Cref{lem: tour edge into two cover}, we show in the next lemma that, if the special cover advantage is small, then the TSP cost is close to $2$ times the MST cost.

\begin{lemma}
\label{lem: cover_advantage lower bound}
Let $T'$ be any subtree of an MST $T$. Then $\tsp\ge 2\cdot w(T')-2\cdot \adv^*(T')$.
\end{lemma}
\begin{proof}
Let $\pi^*$ be the optimal TSP-tour that visits all vertices of $V$, so $\tsp=w(\pi^*)$.
Let $\pi$ be the tour obtained from $\pi^*$ by deleting all vertices of $T\setminus T'$, so $\pi$ is a tour that visits all vertices of $V(T')$, and, from triangle inequality, $w(\pi^*)\ge w(\pi)$. We define set $V'$ and edge sets $F_0, F_1, E_0, E_1$ in the same way as in the proof of \Cref{lem: tour edge into two cover}. From triangle inequality, $w(F_0)+w(F_1)\le w(E_0)+w(E_1)= w(\pi)$. Note that edges of $F_0$ and $F_1$ have both endpoint in $V'$, and moreover, from the definition of $V'$, all vertices of $V'$ are special vertices of $T'$.
Since $E(T')\subseteq \cov(F_0)$, from the definition of $\adv^*(T')$, we get that
$\adv^*(T')\ge w(\cov(F_0,T'))-w(F_0)=w(T')-w(F_0)$, and similarly $\adv^*(T')\ge w(T')-w(F_1)$. Therefore,
$
\tsp=w(\pi^*)\ge w(\pi)=w(E_0)+w(E_1)\ge w(F_0)+w(F_1)\ge 2\cdot (w(T')-\adv^*(T')).
$
\end{proof}

\iffalse
\begin{lemma}
	Let $T'_1,T'_2$ be two subtrees of $T$. Then $\adv(T_1)-\adv(T_2)\le w(T_1\setminus T_2)$.
\end{lemma}
\begin{proof}
	\znote{To Complete.}
\end{proof}
\fi

The last two lemmas in this section show that the total cover advantage and the total special cover advantage of a set of small edge-disjoint paths of an MST $T$ can be efficiently and accurately estimated.

\begin{lemma}
\label{lem: estimate adv of a set of segments}
There is an algorithm, that, given a constant $0<\eps<1$ and a set $\tset$ of edge-disjoint subtrees of $T$, with high probability, either correctly reports that $\sum_{T'\in \tset}\adv(T')\ge \eps\cdot \mst$, or correctly reports that $\sum_{T'\in \tset}\adv(T')\le 2\eps\cdot \mst$, by performing $\tilde O((n/\eps^2)\cdot\max\set{|V(T')|}_{T'\in \tset})$ queries.
\end{lemma}

\begin{proof}
	Denote $s=\max\set{|V(T')|}_{T'\in \tset}$.
	We define the random variable $\gamma$ as follows. Denote $w(\tset)=\sum_{T'\in \tset}w(T')$. For each subtree $T'\in \tset$, we let $\gamma$ takes value $\adv(T')/w(T')$ with probability $w(T')/w(\tset)$. Therefore, $\gamma$ is supported on $[0,1]$, and
	$$\mathbb{E}[\gamma]=\sum_{T'\in \tset}\frac{\adv(T')}{w(T')}\cdot \frac{w(T')}{w(\tset)}
	=\frac{\sum_{T'\in \tset}\adv(T')}{w(\tset)}=\frac{\sum_{T'\in \tset}\adv(T')}{\mst}\cdot \bigg(\frac{\mst}{w(\tset)}\bigg).$$
	From Chernoff Bound, for every $0<\eps<1$, the value of $\mathbb{E}[\gamma]$ can be estimated to within an additive error of $\eps/2$, by $\tilde O(1/\eps^2)$ independent samples of $\gamma$. 
	Since $w(\tset)\le \mst$, it follows that the value of $\sum_{T'\in \tset}\adv(T')/w(\tset)$ can be also estimated to within an additive error of $\eps/2$ by $\tilde O(1/\eps^2)$ independent samples of $\gamma$. 
	It suffices to show that we can obtain one independent sample of $\gamma$ using $O(ns)$ queries. In fact, we can sample a subtree from the distribution on $\tset$, in which each $T'\in \tset$ has probability $w(T')/w(\tset)$, and then we use the algorithm from \Cref{lem: max advantage} to compute the value $\adv(T')$ using at most $O(ns)$ queries, and finally return the value $\adv(T')/w(T')$. It is easy to verify that this value is a sample of $\gamma$. This completes the proof of \Cref{lem: estimate adv of a set of segments}.
\end{proof}

Similarly, using the algorithm from \Cref{lem: max special cover advantage}, we have the following lemma.

\begin{lemma}
\label{lem: estimate special adv of a set of segments}
There is an algorithm, that, given a constant $0<\eps<1$ and a set $\tset$ of edge-disjoint subtrees of $T$, with high probability, either correctly reports that $\sum_{T'\in \tset}\adv^*(T')\ge \eps\cdot \mst$, or correctly reports that $\sum_{T'\in \tset}\adv^*(T')\le 2\eps\cdot \mst$, by performing $\tilde O((n/\eps^2)\cdot\max\set{k_{T'}\mid T'\in \tset})$ queries, where $k_{T'}$ is the number of special vertices in $T'$.
\end{lemma}

\newpage

\part{Streaming Algorithms and Lower Bounds}

%streaming part

In this part we present our results on streaming algorithms and lower bounds for MST and TSP cost estimation. We start with some preliminaries in this part, and then provide the proofs of \Cref{thm: 1 pass alpha MST upper,thm: 1 pass alpha MST lower,thm: p pass alpha MST lower,thm: p pass alpha graph MST lower,thm: 2 pass TSP upper,thm: 1 pass TSP lower} in the subsequent sections.

\section{Preliminaries}

By default, all logarithms in this paper are to the base of $2$. All graphs are undirected and simple.

\subsection{Graph-Theoretic Notation}

We follow standard graph-theoretic notation. 
Let $G=(V,E)$ be an undirected unweighted graph.
For a vertex $v\in V$, we denote by $\deg_G(v)$ the degree of vertex $v$ in $G$.
For a pair $A,B$ of disjoint subsets of vertices of $G$, we denote by $E_G(A,B)$ the set of edges in $G$ with one endpoint in $A$ and the other endpoint in $B$.
For a subset $A$ of vertices of $G$, we denote by $\delta_G(A)$ the set of edges with exactly one endpoint in $A$.
We may sometimes omit the subscript $G$ in the notation above if the graph $G$ is clear from the context.

%Similarly, for a subset $V'\subseteq V$ of vertices, we denote by $\delta(V')$ the subset of edges in $G$ with exactly one endpoint in $V'$.
%For a pair $v,v'$ of vertices in $G$, we denote by $\dist_G(v,v')$ the length (the number of edges) of the shortest path in $G$ connecting $v$ to $v'$.
%We denote by $N_G(v,r)$ the set of vertices $v'$ in $G$, such that $\dist_G(v,v')\le r$. 

Let $V$ be a set of vertices and let $w: V\times V \to \mathbb{R}^{\ge 0}$ be a metric on $V$. For an edge $e=(u,v)$ where $u,v\in V$, we sometimes write $w(e)=w(u,v)$. We denote by $\mst(w)$ the MST cost of the weighted complete graph on $V$ with edge weights given by the metric $w$, and we denote by $\tsp(w)$ minimum metric TSP cost of same graph.

Let $T$ be a tree rooted at some vertex of $V(T)$. For a vertex $v$ of $T$, we denote by $T_v$ the subtree of $T$ rooted at vertex $v$. 
We say that a vertex $v$ in $T$ is a \emph{special vertex}, iff $\deg_T(v)\ne 2$.
For a pair $v,v'$ of vertices in $T$, we call the unique path in $T$ connecting $v$ to $v'$ the \emph{tree path} connecting $v$ to $v'$, which we denote by $P^T_{v,v'}$. 
%We say that $v,v'$ are at tree-distance $d$ in $T$, iff $|E(P^T_{v,v'})|=d$.
If vertices $u,v$ are connected by an edge $e$ that does not belong to $T$, then we also write $P^T_{e}=P^T_{u,v}$.
Let $T$ be a tree rooted at some vertex $r\in V(T)$. Let $\hat T$ be a subtree of $T$. We call the vertex of $\hat T$ that is closest to $r$ in $T$ the root of subtree $\hat T$.

\paragraph{Euler Tour of a tree.} Let $T$ be a rooted tree. The \emph{euler tour} of $T$, or equivalently the \emph{depth-first traversal} of $T$, is a vertex traversal of $T$ that begins and ends at the root of $T$, traversing each edge of $T$ exactly twice -- once to enter the subtree and once to exit it.

\paragraph{Graph Streams and Metric Streams.} In the first part of the paper, we consider two types of streams: \emph{graph streams} and \emph{metric streams}. Let $w$ be a metric on a set $V$ of $n$ vertices. A metric stream of $w$ contains the distances $w(v,v')$ between all pairs $v,v'$ of vertices in $V$. A graph stream of $w$ contains all edges of an $n$-vertex weighted graph that induces the metric $w$. Specifically, let $G$ be a minimal (inclusion-wise) weighted graph on $V$ such that the shortest-path distance metric in $G$ on $V$ is identical to $w$, then a graph stream of $w$ contains all edges (and their weights) of $G$.

\subsection{Tools from Information Theory}

Let $\RV{X},\RV{Y}$ be random variables. We denote by $\HH(\RV{X})$ the \emph{Shannon Entropy} of $\RV{X}$, and denote by $\II(\RV{X};\RV{Y})$ the \emph{mutual information} of $\RV{X}$ and $\RV{Y}$. By definition, 
$\II(\RV{X};\RV{Y})=\HH(\RV{X})-\HH(\RV{X}|\RV{Y})=\HH(\RV{Y})-\HH(\RV{Y}|\RV{X})$.
We will sometimes use $I_{\dset}(\RV{X};\RV{Y})$ instead of $I(\RV{X};\RV{Y})$ if the distribution $\dset$ of the random variables $\RV{X},\RV{Y}$ is unclear from the context.
We use the following basic properties of entropy and mutual information.

\begin{claim}
\label{clm: basic properties}
Let $\RV{X},\RV{Y},\RV{Z},\RV{W}$ be random variables. Then
\begin{enumerate}
\item $\II(\RV{X};\RV{Y},\RV{Z})=\II(\RV{X};\RV{Z}|\RV{Y})+\II(\RV{X};\RV{Y})$ (the chain rule for mutual information);
\label{fact:chain-rule}
\item If $\RV{X}$ and $\RV{W}$ are independent conditioned on $\RV{Z}$, then 
$\II(\RV{X};\RV{Y} | \RV{Z},\RV{W}) \ge \II(\RV{X};\RV{Y}|\RV{Z})$.
\label{fac:mul-ind}
\end{enumerate}
\end{claim}

\iffalse
\begin{fact} [chain rule] \label{fact:chain-rule}
	$$\HH(\RV{X}|\RV{Y})+\HH(\RV{Y}) = \HH(\RV{X},\RV{Y})$$
	$$\II(\RV{X};\RV{Z}|\RV{Y}) + \II(\RV{X};\RV{Y}) = \II(\RV{X};\RV{Y},\RV{Z})$$
\end{fact}

\begin{fact} \label{fac:mul-ind}
	If $\RV{X}$ and $\RV{Y}$ are independent given $\RV{W}$, then 
	$$\II(\RV{X};\RV{Z} | \RV{Y},\RV{W}) \ge \II(\RV{X};\RV{Z}|\RV{W}).$$
\end{fact}
\fi

\paragraph{Total variation distance, KL-divergence, and Hellinger distance.}
Let $P,Q$ be distributions on a discrete set $\Omega$. The \emph{total variation distance} between $P$ and $Q$, is defined as
$$\tvd{P}{Q} = \frac{1}{2} \cdot \sum_{\omega \in \Omega} \card{P(\omega)-Q(\omega)}.$$
The \emph{KL-divergence} between $P$ and $Q$ is defined as 
$$\dkl{P}{Q} = \sum_{\omega \in \Omega} P(\omega) \cdot\log \left( \frac{P(\omega)}{Q(\omega)} \right).$$
The \emph{Hellinger distance} between $P$ and $Q$ is defined as 
$$\hel(P,Q) = \sqrt{\frac{1}{2} \sum_{\omega \in \Omega} \bigg(\sqrt{P(\omega)} - \sqrt{Q(\omega)}\bigg)^2} = \sqrt{1-\sum_{\omega \in \Omega} \sqrt{P(\omega)Q(\omega)}}.$$

We use properties of the above measures between distributions.

\begin{claim} \label{clm:triangle}
    Let $P,Q,R$ be distributions on $\Omega$. Then $\hel(P,Q) \le \hel(P,R)+\hel(Q,R)$, and $\tvd{P}{Q} \le \tvd{P}{R}+\tvd{Q}{R}$.
\end{claim}

\begin{claim}[Pinsker's Inequality] \label{clm:pinsker}
    Let $P,Q$ be distributions on $\Omega$. Then $\tvd{P}{Q} \le \sqrt{2\dkl{P}{Q}}$.
\end{claim}

\begin{claim}[\!\!\!\cite{prahladhlecture}\!]
\label{fac:hel-tvd}
Let $P,Q$ be distributions on $\Omega$. Then $\hel^2(P,Q) \le \tvd{P}{Q} \le \sqrt{2} \cdot \hel(P,Q)$.
\end{claim}

\begin{claim}[\!\!\!\cite{Lin91}\!]
\label{fac:hel-dkl}
Let $P,Q$ be distributions on $\Omega$. Then $$\hel^2(P,Q) \le \frac{1}{2}\cdot D_{\textnormal{\textsf{KL}}}\bigg(P\text{ } \bigg|\bigg|\text{ }\frac{P+Q}{2}\bigg)+ \frac{1}{2}\cdot D_{\textnormal{\textsf{KL}}}\bigg(Q\text{ } \bigg|\bigg|\text{ }\frac{P+Q}{2}\bigg).$$
\end{claim}

Let $\RV{X},\RV{Y}$ be any random variables.
We denote by $P_{\RV{X}}$ the distribution of $\RV{X}$, by $P_{\RV{Y}}$ the distribution of $\RV{X}$, and by $P_{\RV{X}\mid \RV{Y}}$ the conditional distribution of $\RV{X}$ given $\RV{Y}$. We use the following claim.

\begin{claim}
\label{clm: II to KL}
$\II(\RV{X};\RV{Y}) = \mathbb{E}_{\RV{Y}}
\bigg[D_{\textnormal{\textsf{KL}}}\big(P_{\RV{X} \mid \RV{Y}}\mid\mid P_{\RV{X}}\big)\bigg]$.
\end{claim}

\subsection{Communication Complexity and Information Complexity}

In this paper we will work with two standard communication models: the \emph{2-player one-way communication} model and the
\emph{multi-party blackboard communication} model.
In this subsection we provide the definitions of communication complexity and information complexity in these two models.

\paragraph{Two-player one-way communication model.}
Let $\mathcal{X},\mathcal{Y},\mathcal{Z}$ be discrete sets and let $P\subseteq \mathcal{X} \times \mathcal{Y} \times \mathcal{Z}$ be a relation (we will also call $P$ a \emph{problem}). 
Let $\dset$ be a distribution over the product set $\mathcal{X}\times\mathcal{Y}$, and let $(X,Y)$ (where $X\in \xset$ and $Y\in \yset$) be a pair sampled from $\dset$.
Alice receives $X$ as her input and Bob receives $Y$ as his input.
%In addition to private randomness, the players also have access to a public tape of random bits, that we denote by $R$. 
In the 2-player one-way communication model, Alice is allowed to send to Bob a single message based on her input $X$, and Bob, upon receiving this message, needs to output an answer $Z$. We say that the answer is correct iff $(X,Y,Z)\in P$.

We always allow a one-way protocol of Alice and Bob to use both public and private randomness, even against a prior distribution $\dset$ on the input pairs $(X,Y)\in \xset\times\yset$.
Let $\pi$ be an one-way protocol. We say that $\pi$ is an $\delta$-error protocol (for some $0<\delta<1$) for $P$ over the distribution $\dset$, iff the probability that, upon receiving a random pair $(X,Y)\in \xset\times\yset$ from distribution $\dset$ as the input, the protocol computes an answer $Z$ with $(X,Y,Z)\notin P$, is at most $\delta$.

\begin{definition}
The cost of a one-way protocol $\pi$ for a problem $P$ with respect to an input distribution $\dset$, denoted by $\cc_{\dset}^{\textnormal{1-way}}(\pi)$, is defined to be the expected bit-length of the transcript communicated from Alice to Bob in $\pi$, when the input pair is sampled from distribution $\dset$. 
%The $\delta$-error \emph{one-way communication complexity} of a problem $P$ with respect to an input distribution $\dset$, denoted by $\cc^{\delta, \textnormal{1-way}}_{\dset}(P)$, is the minimum cost of a $\delta$-error one-way protocol for $P$ over $\dset$.
The communication complexity of an one-way protocol $\pi$, denoted by $\cc^{\textnormal{1-way}}(\pi)$, is defined to be the maximum cost of $\pi$ with respect to any distribution, namely  $\cc^{\textnormal{1-way}}(\pi)=\max_{\dset}\set{\cc_{\dset}^{\textnormal{1-way}}(\pi)}$.
\end{definition}

If the input pair is $(X,Y)$, we denote by $\Pi_{X,Y}$ the transcript of protocol $\pi$ when executed on the input pair $(X,Y)$. %When the inputs are random variables $(\RV{X},\RV{Y})$, 
When the procotol $\pi$ is randomized, we denote by $\RV{\Pi}_{X,Y}$ the random variable representing the transcript of $\pi$ when the input pair is $(X,Y)$. When the inputs are random variables $(\RV{X},\RV{Y})$, we denote $\RV{\Pi}=\RV{\Pi}_{\RV{X},\RV{Y}}$.

\begin{definition}
Let $\pi$ be a protocol for a problem $P$ with respect to input distribution $\dset$. The \emph{information cost} of $\pi$ is defined to be $\ic_{\dset}(\pi) = \II_{\dset}(\RV{\Pi} ; \RV{X} | \RV{Y}) + \II_{\dset}(\RV{\Pi} ; \RV{Y} | \RV{X})$.
%The $\delta$-error \emph{one-way information cost} of problem $P$ with respect to the input distribution $\dset$, denoted by $\ic^{\delta, \textnormal{1-way}}_{\dset}(P)$, is defined to be the minimum information cost $\ic_{\dset}(\pi)$ over all $\delta$-error one-way protocols $\pi$ for $P$ over $\dset$.
\end{definition}

We use the following well-known propositions.

\begin{proposition}[\!\!\!\cite{BravermanR14}\!] \label{fac:cc-ic}
Let $P$ be a problem and let $\dset$ be an input distribution of $P$. Then for any one-way protocol $\pi$ for $P$ over the distribution $\dset$, $\ic_{\dset}(\pi) \le \cc^{\textnormal{1-way}}_{\dset}(\pi)$.
\end{proposition}

\begin{proposition}[\!\!\!\cite{Bar-YossefJKS04}\!] \label{fac:rect}
Let $\pi$ be any randomized protocol. Let $X,X'$ be any pair of elements in $\mathcal{X}$ and let $Y,Y'$ be any pair of elements in $\mathcal{Y}$. Then, $\hel(\RV{\Prot}_{X,Y},\RV{\Prot}_{X',Y'}) = \hel(\RV{\Prot}_{X,Y'},\RV{\Prot}_{X',Y})$. 
\end{proposition}

\paragraph{Multi-party blackboard communication model.} Let $P\subseteq \mathcal{X}_1 \times\dots \times \mathcal{X}_k \times \mathcal{Y}$ be a relation.
In the multi-party communication model, there are $k$ players, that are denoted by $P_1,\ldots,P_k$.
For each $1\le i\le k$, the player $P_i$ receives an element $X_i\in \xset_i$ as her own input, where the tuple $(X_1,\ldots,X_k)$ is sampled from some input distribution $\dset$ on the product set $\mathcal{X}_1 \times\dots \times \mathcal{X}_k$.
The players $P_1,\ldots,P_k$ may exchange several messages publicly by writing them on a shared blackboard for all other players to see.
Eventually, some players will announce an output $Y$ of the function. We say that the answer is correct iff $(X_1,\ldots,X_k,Y)\in P$.
For convenience, we will also call it the blackboard model.

Similar to the two-player one-way communication model, we allow a protocol to use both public and private randomness, and we say that $\pi$ is an $\delta$-error protocol (for some $0<\delta<1$) for $P$ over the distribution $\dset$, iff the probability that, upon receiving a random tuple $(X_1,\ldots,X_k)\in \mathcal{X}_1 \times\dots \times \mathcal{X}_k$ from distribution $\dset$ as the input, the protocol computes an answer $Y$ with $(X_1,\ldots,X_k,Y)\notin P$, is at most $\delta$.
The cost of a protocol $\pi$ for a problem $P$ with respect to an input distribution $\dset$, denoted by $\cc_{\dset}(\pi)$, is defined to be the expected bit-length of the transcript communicated between the players $P_1,\ldots,P_k$ in $\pi$, when their input is sampled from $\dset$. 
%For a real number $0<\delta<1$, we define the $\delta$-error \emph{communication complexity} of problem $P$ over distribution $\dset$, denoted by $\cc^{\delta,\textnormal{pub}}_{\dset}(P)$, as the minimum cost of any $\delta$-error protocol for $P$ over $\dset$ in the blackboard model.
The communication complexity of a protocol $\pi$, denoted by $\cc(\pi)$, is defined to be the maximum cost of $\pi$ with respect to any distribution.

We denote $X = (X_1,\ldots,X_k)$, and for each $1\le i\le k$, we denote $X_{-i} = (X_1,\ldots,X_{i-1},X_{i+1},\ldots,X_k)$. Let $\pi$ be a protocol for $P$ over $\dset$. We denote by $\Pi$ the transcript of $\pi$ and by $\RV{\Pi}$ the random variable representing the transcript of $\pi$, when $\pi$ is randomized.
%We define the \emph{(internal) information cost} of $\pi$ to be $\ic_{\dset}(\pi)=\sum_{1\le i \le k}\II_{\dset}(\RV{\Pi};\RV{X}_{-i}\mid \RV{X}_{i})$.
We use the following propositions.

\begin{proposition} \label{fac:multi-cc}
Let $P$ be a problem and let $\dset$ be an input distribution.
Let $\pi$ be any $\delta$-error protocol for $P$ over $\dset$ in the blackboard model. Then for each $1\le i\le k$, $\II_{\dset}(\RV{\Pi};\RV{X}_{-i}|\RV{X}_i) \le \cc_{\dset}(\pi)$.
\end{proposition}
\iffalse
\begin{proof}
	$\II_{\mu}(\Pi;\RV{X}_{-i} | \RV{X}_i) \le \HH_{\mu}(\Pi) \le \expect_{\mu}(\card{\Pi}) = \cc_{\mu}(\pi)$.
\end{proof}
\fi

\begin{proposition}[\!\!\!\cite{Bar-YossefJKS04}\!] \label{fac:multi-rect}
Let $P$ be a problem and let $\pi$ be a randomized protocol for $P$. Let $X,X'$ be a pair of inputs to $P$. Let $A\subseteq\set{1,\ldots,k}$ be an index set. 
If we define $X^{a}$ as the vector such that for each $i \in A$, $X^{a}_i=X_i$, and for each $i \notin A$, $X^{a}_i=X'_i$; and we define $X^{b}$ as the vector such that for each $i\in A$, $X^{b}_i=X'_i$, and for each $i \notin A$, $X^{b}_i=X_i$, then $\hel(\RV{\Pi}_X,\RV{\Pi}_{X'}) = \hel(\RV{\Pi}_{X^a},\RV{\Pi}_{X^b})$.
\end{proposition}

Lastly, we mention the connection between
the communication complexity defined in this section and the space complexity of streaming algorithms. It is well-known that any streaming algorithm can be immediately converted into a one-way communication protocol, and therefore any communication complexity lower bound in the one-way communication model implies the same lower bound on the space complexity of one-pass streaming algorithms. 
Similarly, any communication complexity lower bound in the blackboard communication model also implies the same lower bound on the space complexity of multi-pass streaming algorithms.

\section{One-Pass MST Estimation in Metric Streams}

In this section, we provide results on one-pass streaming algorithms for MST cost estimation, when the algorithm is given as input a metric stream. We will show an upper bound on the space complexity in \Cref{subsec: 1 pass MST upper} and a lower bound in \Cref{subsec: 1 pass MST lower proof}.
Our proof for the lower bound in \Cref{subsec: 1 pass MST lower proof} is via analyzing a two-player communication game, and we show in \Cref{Optimality of 1-pass n/S^2 lower bound} that our analysis in \Cref{subsec: 1 pass MST lower proof} is asymptotically optimal by giving a near-optimal protocol for the game.

\subsection{A One-Pass Streaming Algorithm for MST Estimation}
\label{subsec: 1 pass MST upper}

In this subsection, we will construct a one-pass streaming algorithm that outputs an $\alpha$-approximation of the minimum spanning tree weight, using $\tilde O(n/\alpha)$ space, thus establishing \Cref{thm: 1 pass alpha MST upper}.
At a high level, our algorithm computes the estimate by sampling a subset $V'$ of vertices, computing the weight of an MST on vertices of $V'$, and then adding to it an estimate of the distances of all remaining vertices to $V'$.
Clearly, the result computed by the algorithm is the total weight of some spanning tree, and is, therefore, lower bounded by $\mst$. %The fact that it is indeed an $\alpha$-approximation requires some proof.

\paragraph{Algorithm.}
We perform three tasks in parallel over the stream.
In the first task, we simply record the largest weight over all edges in the stream, namely we record $\diam(w)=\max\set{w(u,v)\mid u,v\in V}$, the diameter of metric $w$, using $\tilde O(1)$ space.
In the second task, we sample a subset $V'$ of $k=\ceil{n/\alpha}$ distinct vertices in $V$ uniformly at random at the start, and then compute a minimum spanning tree $T'$ on vertices of $V'$ and its cost $\mst'$, using $\tilde O(n/\alpha)$ space over the stream.
In the third task, we sample a subset $V''$ of $k=\ceil{n/\alpha}$ vertices in set $V\setminus V'$ uniformly at random at the start, and for each vertex $v\in V''$, we record $\dist_w(v,V')=\min\set{w(v,v')\mid v'\in V'}$, using a total of $\tilde O(n/\alpha)$ space.
Finally, we output $\mst'+100\cdot\alpha\log n\cdot(\diam(w)+\sum_{v\in V''}\dist_w(v,V'))$ as the estimate of the cost of a minimum spanning tree on $V$.

\paragraph{Approximation Ratio Analysis.}
It is clear that the algorithm described above uses $\tilde O(n/\alpha)$ space. We show that with high probability, 
\begin{equation}
\label{eqn: main}
\mst\le \mst'+100\cdot\alpha\log n\cdot\bigg(\diam(w)+\sum_{v\in V''}\dist_w(v,V')\bigg)\le \tilde O(\alpha)\cdot\mst.
\end{equation}

Let $T$ be a minimum spanning tree on $V$. Let $\pi=(v_1,v_2,\ldots,v_{2n-2})$ be an Euler-tour of $T$, and for each $1\le t\le 2n-2$, we denote by $R_{\pi,t}=\set{v_i\mid t\le i\le t+100\alpha\log n}$.
We now define two bad events such that if neither of them happens, then \ref{eqn: main} holds. We then show that with high probability, neither event happens.

\paragraph{Bad event $\xi_1$.}
Let $\xi_1$ be the event that there exists some $t$, such that $R_{\pi,t}\cap V'=\emptyset$.
We now show that $\Pr[\xi_1]=O(n^{-49})$.
Since each edge of $T$ appears at most twice in the set $\set{(v_i,v_{i+1})\mid v_i\in R_{\pi,t}}$, $R_{\pi,t}$ contains at least $50\alpha\log n$ distinct vertices. Therefore, the probability that a random subset of $\ceil{n/\alpha}$ vertices in $V$ does not intersect with $R_{\pi,t}$ is at most $(1-(50\alpha\log n/n))^{n/\alpha}\le n^{-50}$. Taking the union bound over all $1\le t\le 2n-2$, we get that $\Pr[\xi_1]=O(n^{-49})$. 
Note that, if the event $\xi_1$ does not happen, then every consecutive window of $\pi$ of length $100\alpha\log n$ contains at least one element of $V'$.

Let $v_{i_1},v_{i_2},\ldots,v_{i_{k'}}$ be the vertices of $\pi$ that belongs to $V'$ ($k'$ may be larger than $k$ since we keep all copies of the same vertex), then for each $1\le j\le k'$, $|i_j-i_{j+1}|\le 100\alpha\log n$.

\begin{claim}
\label{claim: mst simple lower bound}
If the event $\xi_1$ does not happen, then $\mst\le \mst'+\sum_{v\in V\setminus V'}\emph{\dist}_w(v,V')\le \tilde O(\alpha)\cdot\mst$. 
\end{claim}
\begin{proof}
On the one hand, if we denote, for each $v\in V\setminus V'$, by $v^*$ the vertex of $V'$ that is closest to $v$ (under the metric $w$), then it is clear that the edges of $E(T')\cup \set{(v,v^*)\mid v\in V\setminus V'}$ form a spanning tree on $V$, and therefore 
$\mst'+\sum_{v\in V\setminus V'}\dist_w(v,V')\ge \mst$.

On the other hand, consider now a vertex $v\in V\setminus V'$. Assume that the first appearance of vertex $v$ in the tour $\pi$ appears between $v_{i_j}$ and $v_{i_{j+1}}$. We then define $\dist_{\pi}(v,V')=\min\set{w(v,v_{i_j}),w(v,v_{i_{j+1}})}$, and it is easy to see that, by triangle inequality $\dist_{\pi}(v,V')\le \sum_{i_j\le t\le i_{j+1}-1}w(v_t,v_{t+1})$.
Since the event $\xi_1$ does not happen, we get that $\sum_{v\in V\setminus V'}\dist_{\pi}(v,V')\le 100\alpha\log n\cdot\mst$.
Finally, since $\mst'\le \mst$, we conclude that $\sum_{v\in V\setminus V'}\dist_w(v,V')\le \tilde O(\alpha)\cdot\mst$. 
\end{proof}

Define $W=\sum_{v\in V\setminus V'}\dist_w(v,V')$ and $W''=\sum_{v\in V''}\dist_w(v,V')$. We prove the following claim.
\begin{claim}
\label{claim: accurate estimate}
If $\emph{\diam}(w)\le W/(100\alpha\log n)$, then with probability $1-O(n^{-50})$, $W/2\le \alpha\cdot W''\le 2W$. 
\end{claim}
\begin{proof}
Let $U_{V\setminus V'}$ be the uniform distribution on all vertices of $V\setminus V'$.
We define the random variable $X: V\setminus V'\to \mathbb{R}^{\ge 0}$ as $X(v)=\dist_{w}(v,V')$ for each vertex $v\in V\setminus V'$. Therefore,  $X$ is supported on the interval $[0,\diam(w)]\subseteq [0, W/(100\alpha\log n)]$, and $\expect[X]=W/(n-k)$. Observe that $W''$ is in fact the sum of $k$ i.i.d samples of $X$. From multiplicative Chernoff bound,
\[\Pr\bigg[\frac{W k}{(3/2)(n-k)}\le W''\le  \frac{(3/2)W k}{n-k}\bigg] \ge 1-\exp\bigg(-k\cdot \frac{W/(n-k)}{W/(100\alpha\log n)}\bigg)\ge 1-n^{-50}.\]
Since $\alpha\cdot \frac{(3/2)W k}{n-k}\le 2W$ and $\alpha\cdot \frac{W k}{(3/2)(n-k)}\ge W/2$, with probability $1-O(n^{-50})$, $W/2\le \alpha\cdot W''\le W$.
\end{proof}

\paragraph{Bad event $\xi_2$.} Let $\xi_2$ be the event that $\diam(w)\le W/(100\alpha\log n)$ holds, and either $\alpha W''> 2W$ or $\alpha W''<W/2$ holds.
From \Cref{claim: accurate estimate}, $\Pr[\xi_2]=O(n^{-50})$.

We are now ready to show that, if both events $\xi_1$ and $\xi_2$ do not happen, then \eqref{eqn: main} holds.
On the one hand, if $\diam(w)\ge W/(100\alpha\log n)$, then $100\alpha\log n\cdot\diam(w)\ge \mst$; if $\diam(w)\le W/(100\alpha\log n)$, then since $\xi_2$ does not happen, $\alpha\cdot W''\ge W/2$, and then since $\xi_1$ does not happen and from \Cref{claim: mst simple lower bound}, $\mst'+100\alpha\log n\cdot W''\ge \mst'+W\ge \mst$.
On the other hand, note that $\mst'\le \mst$ and $\diam(w)\le \mst$. Moreover, since $\xi_2$ does not happen, $\alpha\cdot W''\le 2W$. Therefore,
$$\mst+100\alpha\log n\cdot\diam(w)+100\alpha\log n\cdot W''\le 
\mst+ \tilde O(\alpha)\cdot\mst + \tilde O(\alpha)\cdot\mst
= \tilde O(\alpha)\cdot\mst.$$

\subsection{Space Lower Bound for One-Pass Streaming Algorithms for MST Estimation}
\label{subsec: 1 pass MST lower proof}

In this subsection, we will show that 
%any one-pass streaming algorithm that uses $S$-space can only obtain an $\Omega (\sqrt{n/S})$-approximation. In other words, 
any one-pass streaming algorithm that outputs an $\alpha$-approximation of the cost of a minimum spanning tree needs to use space $\Omega(n/\alpha^2)$, thus establishing \Cref{thm: 1 pass alpha MST lower}.
At a high level, our proof consider a 2-player communication game, in which Alice and Bob are each given a partial metric, and are asked to estimate the MST cost of the complete metric obtained by combining their partial metrics to within a factor of $\alpha$. It is well-known that the space complexity of any one-pass $\alpha$-approximation streaming algorithm is lower bounded by the one-way communication complexity of the game. We then show that the one-way communication complexity of the game is $\Omega(n/\alpha^2)$ by constructing a hard distribution and analyzing the information complexity of any randomized protocol for this game.

In order to describe the communication game, we first introduce the following notion on metrics and partial metrics.

\paragraph{Metrics and partial metrics.} 
We say that a function $\bar w:V\times V\to \mathbb{R}^{\ge 0}\cup \set{*}$ is a \emph{partial metric} on $V$ iff for all sequences $x_1,\ldots,x_t$ of vertices of $V$ such that $\bar w(x_i,x_{i+1})\ne *$ for all $1\le i\le t$ (where we use the convention that $t+1=1$), the inequality $\sum_{1\le i\le t-1}\bar w(x_i,x_{i+1})\ge \bar w(x_1,x_t)$ holds.
Equivalently, $\bar w$ is a partial metric iff it can be obtained from some complete metric $w$ by masking some of its entries by $*$.
Given a subset $S\subseteq V\times V$, we define the \emph{restriction of $w$ on set $S$}, denoted by $\bar w_{\mid S}$, to be the partial metric where for each pair $(x,y)\in V\times V$, if $(x,y)\in S$, then $\bar w_{\mid S}(x,y)=w(x,y)$; otherwise $\bar w_{\mid S}(x,y)=*$.
Given a partial metric $\bar w$, we say that a complete metric $w$ is a \emph{completion} of $\bar w$, iff for all pairs $x,y$ of vertices in $V$ such that $\bar w(x,y)\ne *$, $w(x,y)=\bar w(x,y)$. In other words, $w$ is a completion of $\bar w$ iff $\bar w$ is the restriction of $w$ on some subset.

Let $\bar w, \bar w'$ be two partial metrics on $V$, we say that partial metrics $\bar w, \bar w'$ are \emph{complementary}, iff there is a complete metric $w$ on $V$ and two subsets $S,S'\subseteq V\times V$ where $S\cup S'= V\times V$, such that $\bar w$ is the restriction of $w$ on $S$, and $\bar w'$ is the restriction of $w$ on $\bar S'$.
In this case, we also write $w=\bar w\cup \bar w'$.

\iffalse
\begin{tbox}
	\textbf{Distribution} $\distIND$: An distribution on input pairs $(x,i)$ (where $x\in \set{0,1}^n$ and $i\in [n]$): 
	
	\begin{enumerate}
		\item For each $j \in [n]$, $x_j$ is chosen uniformly at random from $\set{0,1}$. 
		\item The index $i$ is chosen uniformly at random from $[n]$.
	\end{enumerate}
\end{tbox}
\begin{theorem}[\znote{cite?}]
\label{thm: index}
Let $\delta$ be any constant such that $0<\delta<1/2$. Then any $\delta$-error protocol for $\emph{\ind}$ problem over the input distribution $\dset_{\sf Index}$ must satisfy that $\emph{\cc}(\pi)=\Omega(n)$.
\end{theorem}
\fi

%\begin{problem}
	
\paragraph{Two-player metric-MST-cost approximation problem ($\emph{\sf MST}_{\sf apx}$):} This is a one-way communication problem, in which Alice is given a partial metric $\bar{w}_A$ and Bob is given a partial metric $\bar{w}_B$. The goal is to compute an estimate $Y$ of $\mst(\bar{w}_A\cup \bar{w}_B)$, when partial metrics $\bar{w}_A,\bar{w}_B$ are complementary. If $\bar{w}_A,\bar{w}_B$ are not complementary, then any answer returned by the protocol will be viewed as a correct answer.
%\end{problem}

Let $\dset$ be a distribution on the input pairs $(\bar{w}_A,\bar{w}_B)$ of partial metrics.
We say that a protocol $\pi$ of Alice and Bob $(\alpha,\delta)$-approximates the problem $\emph{\sf MST}_{\sf apx}$ over the distribution $\dset$ for some real numbers $\alpha>1$ and $0<\delta<1$, iff, whenever $\bar{w}_A,\bar{w}_B$ are complementary, with probability at least $1-\delta$, the estimate $Y$ returned by protocol $\pi$ satisfies that $Y\le \mst(w)\le \alpha\cdot Y$.

The main result of this section is the following theorem, which immediately implies \Cref{thm: 1 pass alpha MST lower}.

\begin{theorem}
\label{thm: main_1-pass n/S^2 lower bound}
For any real numbers $\alpha>1$ and $0<\delta< 1/10$, there exists a distribution $\dset$ on input pairs of the problem $\emph{\sf MST}_{\sf apx}$, such that any one-way protocol $\pi$ that $(\alpha,\delta)$-approximates the problem $\emph{\sf MST}_{\sf apx}$ over $\dset$ has $\emph{\cc}^{\textnormal{1-way}}_{\dset}(\pi)=\Omega(n/\alpha^2)$.
\end{theorem}

The remainder of this section is dedicated to the proof of \Cref{thm: main_1-pass n/S^2 lower bound}. We first define a new problem called the \emph{Overwrite} problem, and then prove \Cref{thm: main_1-pass n/S^2 lower bound} by reducing this new problem to the $\mstest$ problem.

\subsubsection{The Problem $\emph{\overwrite}$ and its Communication Complexity}

{\bf Overwrite problem ($\emph{\overwrite}$)}: This is a $2$-player one-way communication problem, in which Alice is given a vector $X^A\in \set{0,1,*}^{n}$, and Bob is given a vector $X^B\in \set{0,1}^{n}$ and an index $i^*\in [n]$.
They are promised that, for each $i\ne i^*$, either $X^A_i=X^B_i$ or $X^A_i =*$ holds.
Alice is allowed to send a message to Bob, and upon receiving the message, the goal of Bob is to figure out the value of $X^A_{i^*}$, the $i^*$-th coordinate of Alice's input vector.  If $X^A_{i^*}\ne *$, then the correct output  of Bob is $X^A_{i^*}$; otherwise both $0$ and $1$ are accepted as correct outputs.

We consider the following distribution of input pairs to the problem $\overwrite$.

\begin{tbox}
	\textbf{Distribution} $\distover$: A distribution on input pairs $(X^A, X^B)$ to $\overwrite$: 
	\begin{enumerate}
		\item Choose an index $i^*$ uniformly at random from $[n]$.
		\item Choose $X^B$ uniformly at random from the set $\set{0,1}^{n}$. 
		\item Set $X^A_{i^*}$ to be $*$ with probability $1/2$, $0$ with probability $1/4$, and $1$ with probability $1/4$. For each $i\in [n]$ such that $i\ne i^*$, set $X^A_{i}$ to be $*$ with probability $1/2$ and $X^B_i$ with probability $1/2$.
	\end{enumerate}
\end{tbox}

The main result of this subsection is the following lemma.
\begin{lemma}
\label{lem: hardness for overwrite}
Let $\delta$ be any constant such that $0<\delta<1/10$. Then any $\delta$-error one-way communication protocol $\pi$ for $\emph{\overwrite}$ problem over the input distribution $\dset_{\sf Overwrite}$ has $\emph{\cc}^{\textnormal{1-way}}_{\dset_{\textnormal{\sf Overwrite} }}(\pi)=\Omega(n)$.
\end{lemma}
\begin{proof}
Consider the following one-bit version of the $\overwrite$ problem. Alice is given a bit $\hat X^A$, which is $*$ with probability $1/2$, $0$ with probability $1/4$, and $1$ with probability $1/4$. Bob is also given a bit $\hat X^B$, which is chosen uniformly at random from set $\set{0,1}$.
We denote this input distribution by $\dset^*$.
Alice is allowed to send a message to Bob, such that upon receiving this message, Bob needs to figure out the value of $\hat X^A$. Note that if $\hat X^A=*$, then Bob can return either $0$ or $1$ as the correct answer. We first show in the following claim that any protocol that allows Alice and Bob to return a correct answer with probability at least $9/10$ has information complexity $\Omega(1)$. % \znote{To add a proof in Appendix}

\begin{claim}
\label{clm: 1 bit overwrite}
For any $\delta$-error one-way protocol $\pi$ for the one-bit $\overwrite$ problem, $\ic_{\dset^*}(\pi)=\Omega(1)$.
\end{claim}
\begin{proof}
Let $\pi$ be any $\delta$-error one-way protocol for the one-bit $\overwrite$ problem. Recall that $0<\delta<1/10$. We will show that $\ic_{\dset^*}(\pi)>10^{-5}$. Assume for contradiction that this does not hold. We denote by $\Pi$ the random variable representing the transcript of $\pi$. By definition, $\II_{\dset^*}(\Pi;\hat X^A|\hat X^B) \le \ic_{\dset^*}(\pi) \le 10^{-5}$. %Let $\Pi$ be the random message of $\pi$ sent by Alice when $\hat X^A$ is sampled by the distribution $\nu$, 
We denote by $\Pi_{*}$ the random variable representing the transcript of $\pi$ when $\hat X^A=*$, and we define $\Pi_{0}, \Pi_{1}$ similarly. Since random variables $\hat X^A$ and $\hat X^B$ are independent in $\dset^*$, from \ref{fac:mul-ind} of \Cref{clm: basic properties} and \Cref{clm: II to KL},
\[
10^{-5} \ge \II_{\dset^*}(\Pi;\hat X^A|\hat X^B) \ge \II_{\dset^*}(\Pi;\hat X^A) = \frac{\dkl{\Pi_*}{\Pi}}{2} + \frac{\dkl{\Pi_0}{\Pi}}{4} + \frac{\dkl{\Pi_1}{\Pi}}{4}.
\] 
From \Cref{clm:pinsker}, $\tvd{\Pi_0}{\Pi} \le \sqrt{2 \dkl{\Pi_0}{\Pi}} \le 10^{-2}$, and similarly $\tvd{\Pi_1}{\Pi}\le  \sqrt{2 \dkl{\Pi_1}{\Pi}} \le 10^{-2}$. Then from \Cref{clm:triangle}, $\tvd{\Pi_0}{\Pi_1} \le 1/50$. Recall that in distribution $\dset^*$, with probability $1/2$, $\hat X_A=*$; with probability $1/4$, $\hat X_A=0$; and with probability $1/4$, $\hat X_A=1$.
Since $\tvd{\Pi_0}{\Pi_1} \le 1/50$,  protocol $\pi$ is correct with probability at most 
$$\Pr[\hat X_A=*]+ \frac{1}{2}\cdot\left(\Pr[\hat X_A=0]+\Pr[\hat X_A=1]+2\cdot \tvd{\Pi_0}{\Pi_1}\right)\le \frac 1 2+\frac 1 4 +\frac{1}{50} < \frac{9}{10},$$ a contradiction to the fact that $\pi$ is an $\delta$-error protocol for $0<\delta<1/10$.
\end{proof}

Let $\pi_{\overwrite}$ be a $\delta$-error protocol for some $0<\delta<1/10$ for the $\overwrite$ problem. We now construct a $\delta$-error protocol for the one-bit $\overwrite$ problem using $\pi_{\overwrite}$ as follows.
%Let $\dset_1$ be the one-bit distribution that gets value $*$ with probability $1/2$, gets value $0$ with probability $1/4$ and  gets value $1$ with probability $1/4$.

\begin{tbox}
	\textbf{Protocol} $\pi$: A protocol for one-bit $\overwrite$ using a protocol $\pi_{\overwrite}$ for $\overwrite$. 
	
	\smallskip
	
	\textbf{Input:} An instance $(\hat X^A, \hat X^B)$ of the one-bit $\overwrite$ problem. \\
	\textbf{Output:} $\hat Z\in \set{0,1}$ as the answer to the one-bit $\overwrite$ problem.
	
	\algline
	
	\begin{enumerate}
		\item \textbf{Sampling the instance.} Alice and Bob create an instance $(X^A, X^B)$ of $\overwrite$ as follows.
		\begin{enumerate}
		\item Choose an index $i^*$ uniformly at random from $[n]$, using \underline{public coins}.
		\item For each $1\le i< i^*$, using \underline{public coins}, Alice and Bob jointly sample $X^A_i$ from distribution $\dset^*$. If $X^A_i=*$, then Bob chooses $X^B_i$ uniformly at random from $\set{0,1}$, using \underline{private coins}. If $X^A_i\ne *$, then Bob sets $X^B_i=X^A_i$.
		\item Alice sets $X^A_{i^*}=\hat X^A$ and Bob sets $X^B_{i^*}=\hat X^B$.
		\item For each $i^*<i\le n$, using \underline{public coins}, Alice and Bob jointly sample $X^B_i$ uniformly at random from $\set{0,1}$. Then Alice, using \underline{private coins}, sets $X^A_{i^*}=X^B_{i^*}$ with probability $1/2$, and sets $X^A_{i^*}=*$ with probability $1/2$.
		\end{enumerate}
		\item \textbf{Computing the answer.} Alice and Bob run the protocol $\pi_{\overwrite}$ on the instance $(X^A, X^B)$ of $\overwrite$ to compute an answer $Z$, and return $\hat Z=Z$. 
	\end{enumerate}
\end{tbox}

It is easy to see that the distribution on instances $(X^A,X^B)$ for $\overwrite$ created in the reduction by the choice of pairs $(\hat X^A, \hat X^B)$ is exactly the same as the distribution $\distover$. 
Moreover, the protocol $\pi$ returns a correct answer to the one-bit $\overwrite$ problem iff the protocol $\pi_{\sf Overwrite}$ returns a correct answer to the $\overwrite$ problem. Since $\pi_{\sf Overwrite}$ is a $\delta$-error one-way protocol to the $\overwrite$ problem for some $0<\delta<1/10$, $\pi$ is a $\delta$-error one-way protocol to the one-bit $\overwrite$ problem for the same $\delta$.
Therefore, from \Cref{clm: 1 bit overwrite}, if we denote by $\Pi$ the random variable representing the transcript of $\pi$, then $\II(\Pi; \hat X^A | \hat X^B)=\Omega(1)$.

We denote by $\Pi_{\overwrite}$ the random variable representing the transcript of $\pi_{\overwrite}$. Then,
%from the chain rule (\ref{fact:chain-rule} in \Cref{clm: basic properties}), we get that 
\begin{align*}
\Omega(1)=\II(\Pi; \hat X^A| \hat X^B) & \le \II(\Pi_{\overwrite};\hat X^A\mid i^*,X^A_1,X^A_2,\dots,X^A_{i^*-1},X^B_1,X^B_2,\dots,X^B_n) \\
&= \frac{1}{n}\cdot \sum_{i=1}^n \II(\Pi_{\overwrite}; X^A_i \mid X^A_1,X^A_2,\dots,X^A_{i-1}, X^B) \\
&= \frac{1}{n}\cdot \II(\Pi_{\overwrite}; X^A | X^B),
\end{align*}
where the first equality is from \ref{fac:mul-ind} in \Cref{clm: basic properties} the mutual independence between the random variables in $\set{X^A_t\mid 1\le t< i^*}$ and $\set{X^B_t\mid 1\le t\le n}$ (note that $\Pi=\Pi_{\overwrite}$, $\hat X_B=X^B_{i^*}$ and the variables $i^*$, $\set{X^A_t}_{1\le t< i^*}$ and $\set{X^B_t}_{1\le t\le n, t\ne i^*}$ play the role of $W$ in \ref{fac:mul-ind}), the second equality follows from the fact that $i^*$ is chosen uniformly at random from set $[n]$, and the last equality is from the chain rule (\ref{fact:chain-rule} in \Cref{clm: basic properties}).

From \Cref{fac:cc-ic}, we get that $\emph{\cc}^{\textnormal{1-way}}_{\dset_{\textnormal{\sf Overwrite} }}(\pi_{\overwrite})\ge \emph{\ic}_{\dset_{\textnormal{\sf Overwrite} }}(\pi_{\overwrite}) \ge  \Omega(n)$. This completes the proof of \Cref{lem: hardness for overwrite}.
\end{proof}
%\end{proof}

\subsubsection{Completing the Proof of \Cref{thm: main_1-pass n/S^2 lower bound}}

We now prove \Cref{thm: main_1-pass n/S^2 lower bound} by reducing the problem $\overwrite$ to the problem $\mstest$. First we define two metrics for constructing the hard distribution for the $\mstest$ problem.

\paragraph{Metrics $w_{\sf Y}$ and $w_{\sf N}$.}
In order to define the distribution $\distMST$ on pairs $(\bar{w}_A,\bar{w}_B)$ of partial metrics, which serves as input to the $\mstest$ problem, we need to first define metrics $w_{\sf Y}$ and $w_{\sf N}$ and certain groups of vertices.
We use the parameters $k,L,r$ to be fixed later.
Let $p=(n-kr)/(2rk)$.
The set of vertices in the metric spaces $w_{\sf Y}$ and $w_{\sf N}$ is $V' = U \cup S \cup T$, where 
\begin{itemize}
\item $U=\set{u_{i,j}\mid i\in [k],j\in [r]}$, so $|U|=kr$;
\item $S=\set{s_{i,j,\ell}\mid i\in [k],j\in [r],\ell \in [p]}$, so $|S|=prk=(n-kr)/2$; and 
\item $T=\set{t_{i,j,\ell}\mid i\in [k],j\in [r],\ell \in [p]}$, so $|T|=prk=(n-kr)/2$.
\end{itemize}
We now define the distances between pairs of vertices of $V'$ in metrics $w_{\sf Y}$ and $w_{\sf N}$. In fact, the metrics $w_{\sf Y}$ and $w_{\sf N}$ differ in only the distance between vertices of $U$.
Consider now a pair $(v,v')$ of vertices of $V'$, we distinguish between the following cases.
\begin{itemize}
\item If $v,v'\in S$, and $v=s_{i,j,\ell}$ and $v=s_{i',j',\ell'}$, then\\ $w_{\sf Y}(v,v')=w_{\sf N}(v,v')=\min\set{L, |(irp+\ell r+j)-(i'rp+\ell' r+j')|}$.
\item If $v,v'\in T$, and $v=t_{i,j,\ell}$ and $v=t_{i',j',\ell'}$, then\\ $w_{\sf Y}(v,v')=w_{\sf N}(v,v')=\min\set{L, |(irp+jp+\ell)-(i'rp+j'q+\ell')|}$.
\item If $v\in S$ and $v'\in T$, then $w_{\sf Y}(v,v')=w_{\sf N}(v,v')=L$.
\item If $v\in S\cup T$ and $v'\in U$, then $w_{\sf Y}(v,v')=w_{\sf N}(v,v')=L$.
\item Assume $v=u_{i,j}\in U$ and $v'=u_{i',j'}\in U$. 
First, $w_{\sf N}(v,v')=L$. 
Then if $i=i'$,  $w_{\sf Y}(v,v')=|j-j'|$; if $i\ne i'$, $w_{\sf Y}(v,v')=L$.
\end{itemize}
This completes the definitions of $w_{\sf Y}$ and $w_{\sf N}$. See \Cref{fig:Yes_No_metrics} for an illustration. 
It is easy to verify that $\mst(w_{\sf Y})=(n-1)+(L-1)(k+1)$ and $\mst(w_{\sf N})=(n-1)+(L-1)((k-1)r+1)$. Additionally, for each index $1\le i\le p$, we define the set 
$S_{\ell}=\set{s_{i,j,\ell}\mid i \in [k], j\in [r]}$ and the set 
$T_{\ell}=\set{t_{i,j,\ell}\mid i \in [k], j\in [r]}$.
Sets $\set{S_{\ell}}_{\ell \in [p]}$, $\set{T_{\ell}}_{\ell \in [p]}$ and $U$ are called \emph{groups} of $V'$.
%Lastly, we also define a partial metric $\bar w$ of $w_{\sf Y}$ as follows. Denote $R=(V'\times V')\setminus (U\times U)$ as a subset of $V'\times V'$. Partial metric $\bar w$ is then defined as the restriction of metric $w_{\sf Y}$  on the set $R$ of pairs. Clearly, $\bar w$ can be also viewed as the restriction of metric $w_{\sf N}$ on the set $R$ of pairs.

\begin{figure}[h]
	\centering
	\subfigure[An illustration of $w_{\sf Y}$. For $v,v'\in V'$, $w_{\sf Y}(v,v')=\min\set{L, \text{ the distance between }v,v'\text{ on the path}}$.]{\scalebox{0.16}{\includegraphics{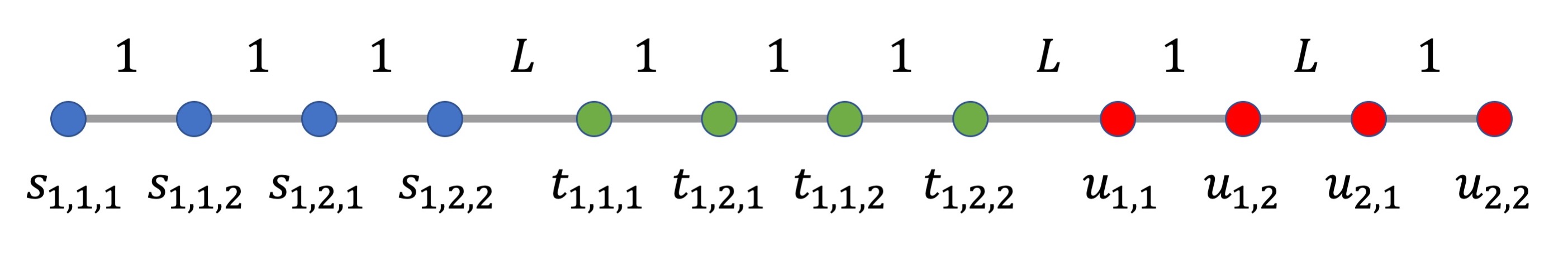}}
	}
	\hspace{0.5cm}
	\subfigure[An illustration of $w_{\sf N}$. For $v,v'\in V'$, $w_{\sf N}(v,v')=\min\set{L, \text{ the distance between }v,v'\text{ on the path}}$.]{\scalebox{0.16}{\includegraphics{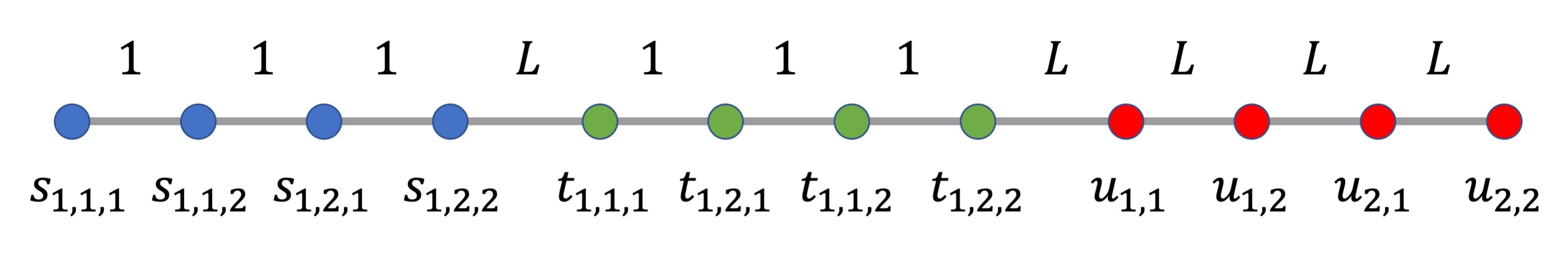}}}
	\caption{An illustration of metrics $w_{\sf Y}$ and $w_{\sf N}$ in the case $k=r=2$, $p=1$ and $n=12$.}\label{fig:Yes_No_metrics}
\end{figure}

Before we define the hard distribution on input partial metric pairs to the problem $\mstest$, we first describe the set $V$ of vertices on which the partial metrics will be defined. Denote $q=2p+1$.
Set $V$ contains $n$ vertices, and is partitioned into $q$ \emph{groups} $V=V_1\cup\ldots\cup V_q$, such that $|V_1|=\cdots=|V_q|=kr$.
We say that an one-to-one mapping $f: V\to V'$ is \emph{group-preserving}, iff $f$ maps an entire group of $V$ to an entire group of $V'$.

\begin{tbox}
	\textbf{Distribution} $\distMST(k,L,r)$: An distribution on pairs $(\bar{w}_A,\bar{w}_B)$ of partial metrics on $V$: 
	
	\begin{enumerate}
	\item Choose a random one-to-one group-preserving mapping $f: V\to V'$.
	\item Choose $Z\in \set{\sf Y,\sf N}$ uniformly at random, and then define metric $w$ as: for every pair $v,v'\in V$, $w(v,v')=w_{Z}(f(v),f(v'))$.
	\item Choose $i^*$ from $[p]$ uniformly at random. Then with probability $1/2$, set $R^{i^*}_A=\emptyset$ and with probability $1/2$, set $R^{i^*}_A=f^{-1}(U)\times f^{-1}(U)$.
	\item For each $1\le i\le p$, $i\ne i^*$, choose $z_i$ from $\set{\mathbb S,\mathbb T}$ uniformly at random. If $z_i=\mathbb S$, then define $R^i_A=f^{-1}(S_i)\times f^{-1}(S_i)$ as a subset of $V\times V$;  if $z_i=\mathbb T$, then define $R^i_A=f^{-1}(T_i)\times f^{-1}(T_i)$ as a subset of $V\times V$. 
	\item Define $\bar w_A$ as the restriction of metric $w$ on set $R_A=\bigcup_{1\le i\le p}R^i_A$ of pairs.
	\item Define $\bar w_B$ as the restriction of $w$ on the set $R_B=(V\times V)\setminus (f^{-1}(U)\times f^{-1}(U))$ of pairs.
	\end{enumerate}
\end{tbox}

Let $\alpha>1$ and $0<\delta<1/10$ be any constants.
Let $\promst$ be a protocol that $(\alpha,\delta)$-approximates the problem $\emph{\sf MST}_{\sf apx}$ over the distribution $\distMST(k,L,r)$.
We construct the following protocol $\proover$ for the problem $\overwrite$ using the protocol $\promst$. 

\begin{tbox}
	\textbf{Protocol} $\proover$: A protocol for $\overwrite$ using a protocol $\promst$ for $\mstest$. 
	
	\smallskip
	
	\textbf{Input:} An instance $(X^A, X^B)\sim \dset_{\sf Overwrite}$. \\
	\textbf{Output:} $Z\in \set{0,1}$ as the answer to $\overwrite$.
	
	\algline
	
	\begin{enumerate}
		\item \textbf{Sampling the instance.} Alice and Bob create an instance $(\bar w_A,\bar w_B)$ of $\mstest$ as follows.
		\begin{enumerate}
			\item Set parameters $k=r=\ceil{2\alpha}$ and $L=n/2\alpha$.
			\item Define the partial metric $\bar w_A$ as follows: for each $t\in [p]$, if $X^A_{t}=1$, then Alice sets the restricted metric on group $V_i$ to be the metric on $S_1$; if $X^A_{t}=0$, then Alice sets the restricted metric on group $V_t$ to be the metric on $T_1$; if $X^A_{t}=*$ then Alice does not modify the restricted metric on group $V_t$ (which remains unknown to Alice).
			\item Let $i^*$ be the index given to Bob. Bob first finds an one-to-one group-preserving mapping $f: V\to V'$, such that (i) $f(V_{i^*})=U$; and (ii) for each $t\in [q], t\ne i^*$, if $X^B_t=1$, then group $V_t$ is mapped to some group of $\set{S_i}_{i\in [p]}$; and if $X^B_t=0$, then group $V_t$ is mapped to some group of $\set{T_i}_{i\in [p]}$.
			\item Define the partial metric $\bar w_B$ as follows: for each pair $v,v'\in V$ such that at least one of $v,v'$ does not belong to group $V_{i^*}$, $\bar w_B(v,v')=w_{\sf Y}(f(v),f(v'))=w_{\sf N}(f(v),f(v'))$.
		\end{enumerate}
		\item \textbf{Computing the answer.} Alice and Bob run the protocol $\promst$ on $(\bar w_A,\bar w_B)$ to compute an estimate $Y$. If $Y\le \alpha \cdot \mst(w_{\sf Y})$, then they return $Z=1$; otherwise they return $Z=0$.  
	\end{enumerate}
\end{tbox}

\begin{claim}
\label{clm: protocol mst to index}
If the protocol $\pi_{\sf MST}$ $(\alpha,\delta)$-approximates the problem $\emph{\sf MST}_{\sf apx}$ over $\dset_{\textnormal{MST}}(k,L,r)$, then the protocol $\pi_{\overwrite}$ is a $\delta$-error protocol for the problem $\overwrite$ over $\dset_{\sf Overwrite}$.
\end{claim}
\begin{proof}
It is easy to see that the distribution of instances $(\bar w_A,\bar w_B)$ for $\mstest$ created in the reduction by the instance $(X^A,X^B)\sim \dset_{\sf Overwrite}$, is exactly the same as the distribution $\distMST(\ceil{2\alpha},n/2\alpha,\ceil{2\alpha})$. Moreover, if $X_{i^*}=1$, then $w=w_{\sf Y}$, otherwise $w=w_{\sf N}$.
Recall that $\mst(w_{\sf Y})=(n-1)+(L-1)(k+1)$ and $\mst(w_{\sf N})=(n-1)+(L-1)((k-1)r+1)$. Therefore, if we set $k=r=\alpha$ and $L=n/\alpha$, then
\[
\frac{\mst(w_{\sf N})}{\mst(w_{\sf Y})}
=
\frac{(n-1)+((n/2\alpha)-1)((\ceil{2\alpha}-1)\cdot\ceil{2\alpha}+1)}{(n-1)+((n/2\alpha)-1)(\ceil{2\alpha}+1)}> \alpha,
\]
and it follows that any $\alpha$-approximation of $\mst(w)$ can determine the value of $\mst(w)$ (to be either $\mst(w_{\sf N})$ or $\mst(w_{\sf Y})$).
Therefore, if the protocol $\pi_{\sf MST}$ $(\alpha,\delta)$-approximates the problem $\emph{\sf MST}_{\sf apx}$ over the distribution $\dset_{\textnormal{MST}}(k,L,r)$, then the protocol $\pi_{\overwrite}$ is an $\delta$-error protocol for the problem $\overwrite$ over the distribution $\dset_{\sf Overwrite}$.
\end{proof}

Combining \Cref{clm: protocol mst to index} and \Cref{lem: hardness for overwrite}, we derive that, for any $\alpha>1$ and any $0<\delta<1$, if a protocol $\pi$ $(\alpha,\delta)$-approximates the problem $\emph{\sf MST}_{\sf apx}$ over $\dset_{\textnormal{MST}}(k,L,r)$, then $$\cc_{\dset_{\textnormal{MST}}(k,L,r)}(\pi)=\Omega(p)=\Omega\left(\frac{n-\ceil{2\alpha}^2}{2\cdot\ceil{2\alpha}^2}\right)=\Omega(n/\alpha^2).$$ This completes the proof of \Cref{thm: main_1-pass n/S^2 lower bound}.

\subsection{Discussion of the Lower Bound in \Cref{thm: main_1-pass n/S^2 lower bound}}
\label{Optimality of 1-pass n/S^2 lower bound}

%We say that a function $\bar w:V\times V\to \mathbb{R}^{\ge 0}\cup \set{*}$ is a \emph{partial metric} on $V$ iff for all sequences $x_1,\ldots,x_t$ of vertices of $V$ such that $\bar w(x_i,x_{i+1})\ne *$ for all $1\le i\le t$ (where we use the convention that $t+1=1$), the inequality $\sum_{1\le i\le t-1}\bar w(x_i,x_{i+1})\ge \bar w(x_1,x_t)$ holds. Equivalently, $w$ is a partial metric iff it can be obtained from some complete metric by masking some of its entries by $*$. Given a subset $S\subseteq V\times V$, we define the \emph{restriction of $w$ on $S$} to be the partial metric $\bar w_{\mid S}$, where for each pair $(x,y)\in V\times V$, if $(x,y)\in S$, then $\bar w_{\mid S}(x,y)=w(x,y)$; otherwise $\bar w_{\mid S}(x,y)=*$. Given a partial metric $\bar w$, we say that a complete metric $w$ is a \emph{completion} of $\bar w$, iff for all pairs $x,y$ of vertices in $V$ such that $\bar w(x,y)\ne *$, $w(x,y)=\bar w(x,y)$. In other words, $w$ is a completion of $\bar w$ iff $\bar w$ is the restriction of $w$ on some subset.

In this subsection we show that 
%our lower bound proof of \Cref{thm: 1 pass alpha MST lower} in \Cref{subsec: 1 pass MST lower proof} is asymptotically optimal. In particular, we prove the following theorem, which shows that 
the exponent $2$ in the lower bound $\Omega(n/\alpha^2)$ in \Cref{thm: main_1-pass n/S^2 lower bound} is the best constant we can hope for.

\begin{theorem}
There is a deterministic one-way protocol that returns an $\tilde O(\sqrt{n})$-approximation to the problem $\emph{\sf MST}_{\sf apx}$ with communication complexity $\tilde O(1)$.
\end{theorem}

In particular, if we set $\alpha=\tilde O(\sqrt{n})$, then the lower bound in \Cref{thm: main_1-pass n/S^2 lower bound} shows that any protocol that returns an $\tilde O(\sqrt{n})$-approximation must have communication complexity $\tilde \Omega(1)=\tilde\Omega(n/\alpha^2)$.

%The main result in this section is that there is a simply one-way protocol of space $\tilde O(1)$ for approximating the MST-cost within a factor of $\tilde O(\sqrt{n})$.
%Specifically, Alice is given a partial metric $\bar w_A$ and Bob is given another partial metric $\bar w_B$, such that $\bar w_A,\bar w_B$ are restrictions of a complete metric on a pair of complementary subsets. In other words, there exists a complete metric $w$ and a subset $S\subseteq V\times V$, such that $\bar w_A=\bar w_{\mid S}$ and $\bar w_B=\bar w_{\mid \bar S}$. 
%We show that Alice can send a message of length $\tilde O(1)$, such that Bob, upon receiving this message, can output an estimate of $\mst(w)$ within a factor of $\tilde O(\sqrt{n})$.

Recall that in the communication game, Alice is given a partial metric $\bar w_A$ and Bob is given another partial metric $\bar w_B$, such that $\bar w_A,\bar w_B$ are complementary.
The protocol is simple and is decribed below.
%Intuitively, the message sent by Alice is the minimum MST cost of a metric completion of $\bar w_A$, the input partial metric of Alice.

\paragraph{Protocol.}
Alice computes $Z_A=\min\set{\mst(\hat w)\mid \hat w\text{ is a completion of }\bar w_A}$, and Bob computes $Z_B=\min\set{\mst(\hat w)\mid \hat w\text{ is a completion of }\bar w_B}$. Alice sends $Z_A$ to Bob and then Bob outputs $\max\set{Z_A,Z_B}$ as the estimate of $\mst(w)$.

It is clear that the protocol is deterministic and has communication complexity $\tilde O(1)$. We now show that it indeed returns an $\tilde O(\sqrt{n})$-approximation to the problem $\emph{\sf MST}_{\sf apx}$.

\paragraph{Proof of Correctness.}
On the one hand, by definition, $w$ is a completion of $\bar w_A$, so $\mst(w)\ge Z_A$, and similarly $\mst(w)\ge Z_B$, so $\mst(w)\ge \max\set{Z_A, Z_B}$. It now remains to show that $\mst(w)\le \tilde O(\sqrt{n})\cdot \max\set{Z_A, Z_B}$.

\begin{lemma}
\label{lem: uniform metric completion sqrt n bound}
If $w(x,y)=1$ for all pairs $x,y$ of vertices of $V$, then $\max\set{Z_A,Z_B}\ge \Omega(\sqrt n)$.
\end{lemma}
\begin{proof}
Let $w_A$ be the complete metric that, over all completion $\hat w$ of $\bar w_A$, minimizes $\mst(\hat w)$.
If $Z_A=\mst(w_A)\ge \sqrt{n}$, then $\max\set{Z_A,Z_B}\ge Z_A\ge \sqrt n/2$.
We will show that, if $Z_A=\mst(w_A)<\sqrt{n}$, then there is a set $V'\subseteq V$ of at least $\sqrt{n}/2$ vertices, such that for each pair $x,y\in V'$, $w_A(x,y)<1$. Note that this implies that, for each pair $x,y\in V'$, $\bar w_A(x,y)\ne 1$, and therefore $\bar w_B(x,y)= 1$ since the partial metrics $\bar w_A,\bar w_B$ are complementary.
It follows that $Z_B\ge |V'|/2=\Omega(\sqrt{n})$.

Let $T$ be a minimum spanning tree of graph $G_{w_A}$, the complete graph with weights $w_A$ on its edges, so $w(E(T))\le \sqrt n$. 
Let $T$ be rooted at an arbitrary vertex $r$.
We now iteratively partition vertices of $V$ into disjoint clusters as follows. Throughout, we maintain (i) a tree $\hat T$, that is initialized to be $T$; and (ii) a collection $\sset$ of clusters of $V$, that initially contains no clusters.
We will ensure that, over the course of the algorithm, $\hat T$ is either an empty graph or a connected subtree of $T$ that contains the root $r$.
We now describe an iteration. 
Let $v$ be the leaf of $\hat T$ that, over all leaves of $\hat T$, maximizes its tree-distance to root $r$ in $\hat T$. Let $e_v$ be the unique incident edge of $v$.
If $w_A(e_v)\ge 1/2$, then we simply (i) add a single-vertex cluster $\set{v}$ to $\sset$; (ii) delete vertex $v$ and edge $e_v$ from $\hat T$, and then continue to the next iteration.
If $w_A(e_v)< 1/2$, then we let $S_v=\set{v'\in V(\hat T)\mid w_A(v,v')<1/2}$, and then (i) add the cluster $S_v$ to $\sset$; (ii) delete vertices of $S_v$ and their incident edges from $\hat T$, and then continue to the next iteration.
It is easy to see that, after each iteration, the total weight (under $w_A$) of all edges of $\hat T$ is decreased by at least $1/2$, and the diameter of each cluster added into $\sset$ is less than $1$. Therefore, the algorithm terminates in at most $2\sqrt{n}$ iterations, so $\sset$ contains at most $2\sqrt{n}$ clusters whose union is $V$, and the diameter of each cluster is less than $1$. Let $V'$ be the cluster of $\sset$ with maximum cardinality. It is easy to verify that $V'$ satisfies all desired conditions.
\end{proof}

Let $G_w$ be the weighted complete graph on $V$, where for each pair $v,v'\in V$, the weight on edge $(v,v')$ is $w(v,v')$. For each $1\le i\le \log n$, we let $G^i_w$ be the graph obtained from $G_w$ by deleting all edges with weight greater than $2^i$, and we denote by $c_i$ the number of connected components in graph $G^i_w$. We use the following claims.
\begin{claim}
\label{clm: mst expressed in terms of c_i}
$\mst(w)\le \sum_{i} c_i\cdot 2^{i}$.
\end{claim}
\begin{proof}
%We first show that $\sum_{i} c_i\cdot 2^{i}\le \mst(w)$.
We define function $w^{+}$ as follows. For each pair $x,y\in V$, let $i$ be the unique integer such that $2^{i-1}< w(x,y)\le 2^i$, and we define $w^{+}(x,y)=2^i$. Although $w^+$ may no longer be a metric, $\mst(w^+)$ is well-defined and it is clear that $\mst(w^+)\ge \mst(w)$.
We define graph $G_{w^+}$ and graphs $\set{G^i_{w^+}}_{i}$ similarly. Clearly, for each $i$, graph $G^i_{w^+}$ contains the same set of edges as $G^i_w$. Therefore, the number of connected components in $G^i_{w^+}$ is $c_i$, and it follows that any minimum spanning tree of $G_{w^+}$ must contain, for each $i$, at least $c_i$ edges of weight at least $2^{i+1}$. Therefore, \[\mst(w^+)=\mst(G_{w^+})\ge \sum_{i}c_i\cdot (2^{i+1}-2^i)=\sum_{i}c_i\cdot 2^i.\]
This completes the proof of \Cref{clm: mst expressed in terms of c_i}.
\end{proof}
\begin{claim}
\label{clm: single c_i bound for mst}
For each $1\le i\le \log n$, $\max\set{Z_A, Z_B}\ge \sqrt{c_i}\cdot2^{i}$.
\end{claim}
\begin{proof}
For a partial metric $\bar{w}$, we define its \emph{truncated metric at level-$i$} to be the metric $\bar{w}^{(i)}$, where for every pair $x,y\in V$, if $\bar{w}(x,y)\ne *$, then $\bar{w}^{(i)}(x,y)= \min\set{2^i, \bar{w}(x,y)}$; otherwise $\bar{w}^{(i)}(x,y)=*$. The truncated metric at level-$i$ of a complete metric is defined similarly.
It is easy to see that a truncated metric of a complete metric is still a valid complete metric (that is, satisfying all triangle inequalities).

Fix now an integer $1\le i\le \log n$
and consider the truncated metrics $w^{(i)}$, $\bar w_A^{(i)}$ and $\bar w_B^{(i)}$.
%Clearly, $w^{(i)}$ is a completion of $\bar w_A^{(i)}$, so $\mst(w)\ge Z^{(i)}_A=\min\set{\hat w\mid \hat w\text{ is a completion of }\bar w^{(i)}_A}$, and similarly $\mst(w)\ge Z^{(i)}_B=\min\set{\hat w\mid \hat w\text{ is a completion of }\bar w^{(i)}_B}$. Therefore, $\mst(w)\ge \max\set{Z^{(i)}_A,Z^{(i)}_B}$.
We denote $Z^{(i)}_A=\min\set{\hat w\mid \hat w\text{ is a completion of }\bar w^{(i)}_A}$.
Let $\tilde w_A$ be the completion of partial metric $\bar w_A$ that minimizes $\mst(\tilde w_A)$. 
It is easy to see that the truncated metric $\tilde w^{(i)}_A$ of $\tilde w_A$ is a completion of $\bar w^{(i)}_A$, and $\mst(\tilde w^{(i)}_A)\le \mst(\tilde w_A)$. Altogether, we get that $Z^{(i)}_A\le Z_A$. Similarly, if we denote by $Z^{(i)}_B=\min\set{\hat w\mid \hat w\text{ is a completion of }\bar w^{(i)}_B}$, then $Z^{(i)}_B\le Z_B$.
So $\max\set{Z_A,Z_B}\ge \max\set{Z^{(i)}_A,Z^{(i)}_B}$.

We now show that $\max\set{Z^{(i)}_A,Z^{(i)}_B}\ge 2^i\cdot \sqrt{c_i}$, thus completing the proof of \Cref{clm: single c_i bound for mst}. Recall that the graph $G^i_{w}$ is obtained from graph $G_w$ by deleting all edges of weight greater than $2^i$. Recall that $G^i_{w}$ contains $c^{i}$ connected component. Let $V'=\set{v_1,\ldots,v_{c_i}}$ be a set of $c_i$ vertices from different connected components of $G^i_{w}$. We define a partial metric $\bar w'$ as follows: for each pair $v_p,v_q\in V'$, $\bar w'(v_p,v_q)=2^i$; for any other pair $(x,y)$, $\bar w'(x,y)=*$.
We define two other partial metrics $\bar w'_A, \bar w'_B$ as follows. For a pair $v_p,v_q\in V'$, if $\bar w_A(v_p,v_q)\ne *$, then $\bar w'_A(v_p,v_q)= 2^i$, otherwise $\bar w'_A(v_p,v_q)= *$. The metric $\bar w'_B$ is defined similarly. Clearly, partial metrics $\bar w'_A, \bar w'_B$ are restrictions of $\bar w'$. Moreover, $\bar w'_A$ is a restriction of $\bar w^{(i)}_A$, and $\bar w'_B$ is a restriction of $\bar w^{(i)}_B$.
Therefore, if we denote $Z'_A=\min\set{\hat w\mid \hat w\text{ is a completion of }\bar w'_A}$ (and define $Z'_B$ similarly), then 
$Z'_A\le Z^{(i)}_A$ (and similarly $Z'_B\le Z^{(i)}_B$).
Note that, from \Cref{lem: uniform metric completion sqrt n bound}, $\max\set{Z'_A,Z'_B}\ge \Omega(\sqrt{c_i}\cdot2^{i})$, so $\max\set{Z^{(i)}_A,Z^{(i)}_B}\ge \Omega(\sqrt{c_i}\cdot2^{i})$.
\end{proof}
From \Cref{clm: mst expressed in terms of c_i} and \Cref{clm: single c_i bound for mst}, we get that
\[
\mst(w)\le \sum_{1\le i\le \log n} c_i\cdot 2^{i}
\le \sum_{1\le i\le \log n} \sqrt{n\cdot c_i}\cdot 2^{i}
= \sqrt{n}\cdot\bigg(\sum_{1\le i\le \log n} \sqrt{c_i}\cdot 2^{i}\bigg)
\le O(\sqrt{n}\log n)\cdot \max\set{Z_A, Z_B}.
\]

\section{Space Lower Bound for Multi-Pass MST Estimation in Metric Streams}

In this section, we prove a space lower bound for multi-pass streaming algorithms for MST estimation in metric streams. We will show that, for any $\alpha>1$ and $p\ge 1$, any randomized $p$-pass $\alpha$-approximation streaming algorithm for metric MST estimation requires $\tilde\Omega (\sqrt{n/\alpha p^2})$ space, thus establishing \Cref{thm: p pass alpha MST lower}.
At a high level, our proof first analyzes a multi-player communication game for MST estimation in the blackboard model, and then shows that the lower bound of the communication complexity of this problem implies the space lower bound on multi-pass streaming algorithms.

%In the coordinator model, each player gets a set of pairs of vertices and gets the distance between each pair. Moreover, for any pair of vertices, at least one player gets the distance between them.

\paragraph{The $k$-player metric-MST-cost approximation problem ($\emph{\sf MST}^k_{\sf apx}$):} This is a communication problem between players $P_1,\ldots,P_k$, in which each player $P_i$ is given a partial metric $\bar{w}_i$ on some common vertex set $V$ known to all players. The goal of the players is to compute an estimate $Y$ of $\mst(w)$, if (i) there exists some complete metric $w$ on $V$, such that each $\bar w_i$ is the restriction of $w$ onto some subset $R_i\subseteq V\times V$ of pairs; and (ii) $\bigcup_{1\le i\le k}R_i=V\times V$.
If either (i) or (ii) does not hold, then any output by the algorithm will be viewed as a correct output.

In order to analyze the communication complexity of any randomized protocol that outputs an $\alpha$-approximation of the problem, we first define metrics $w_{\sf Y}$ and $w_{\sf N}$ as follows. The vertex set $V$ contains $n$ vertices, and is partitioned into $m$ groups of size $N=\floor{\sqrt{3\alpha n}}$ each, so $m=\Omega(\sqrt{n/3\alpha})$. 
We denote one of the groups by $V^*$, and call it the \emph{special group}. All other groups are called \emph{regular groups}.
We let $M$ be any parameter such that $M \gg n$. 
In both metrics $w_{\sf Y}$ and $w_{\sf N}$, the distance between every pair of vertices that come from different groups is $M$. In $w_{\sf Y}$, the distance between every pair of vertices that come from the same group is $1$.  In $w_{\sf N}$, the distance between every pair of vertices that come from the same regular group is $1$, while the distance between every pair of vertices from the special group is $M$. See \Cref{fig: metrics} for an illustration. It is easy to verify that $\mst(w_{\sf Y})\le 2mM$ and $\mst(w_{\sf N})\ge MN$, so $\mst(w_{\sf N})> \alpha \cdot \mst(w_{\sf Y})$. This means any protocol that returns an $\alpha$-approximation of the MST cost can indeed distinguish between the two metrics $w_{\sf Y}$ and $w_{\sf N}$.

\begin{figure}[h]
	\centering
	\subfigure[Metric $w_{\sf Y}$ (or graph $G_{\sf Y}$).]{\scalebox{0.24}{\includegraphics{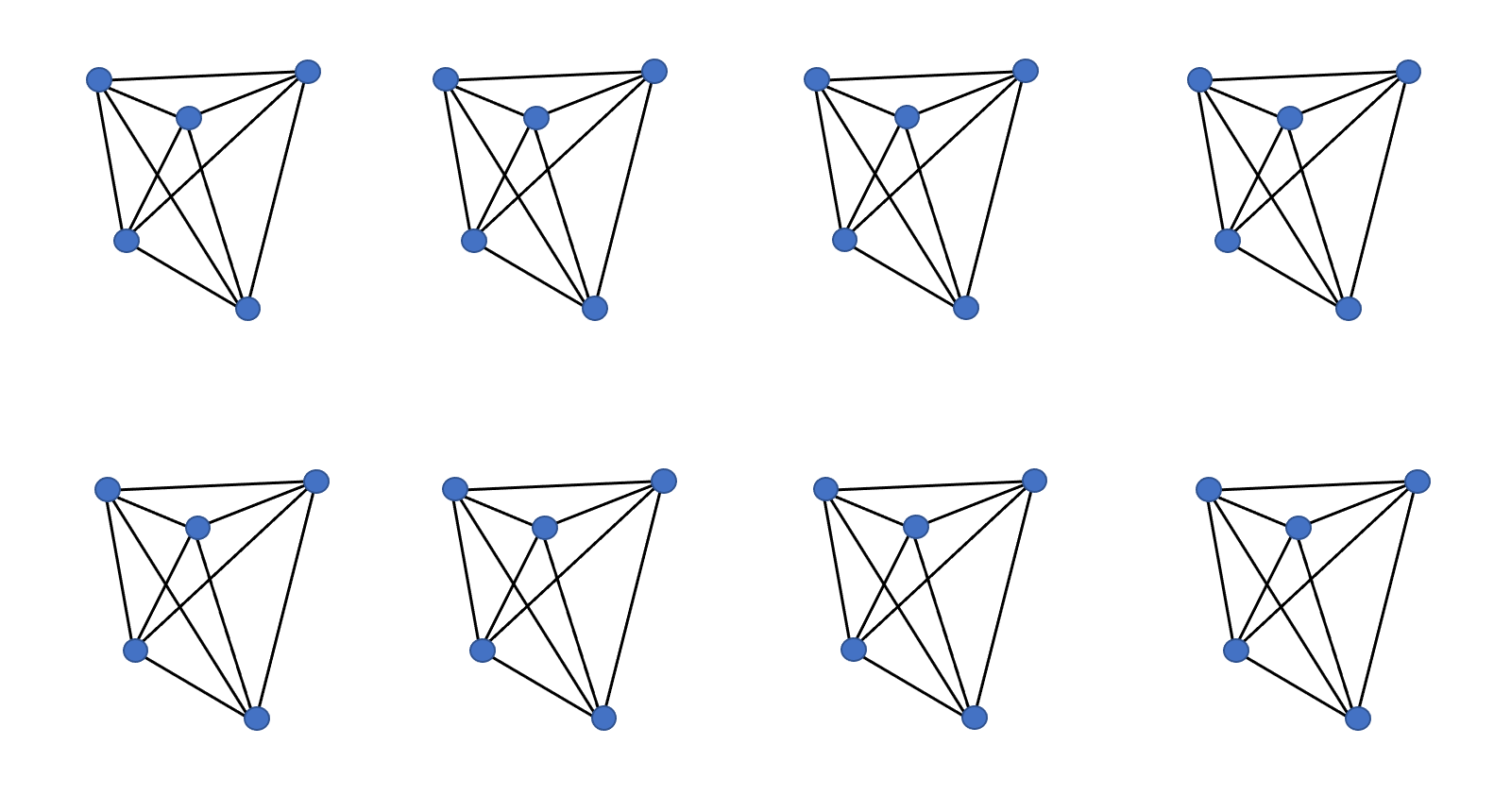}}}
	\hspace{0.5cm}
	\subfigure[Metric $w_{\sf N}$ (or graph $G_{\sf N}$).]{
		\scalebox{0.24}{\includegraphics{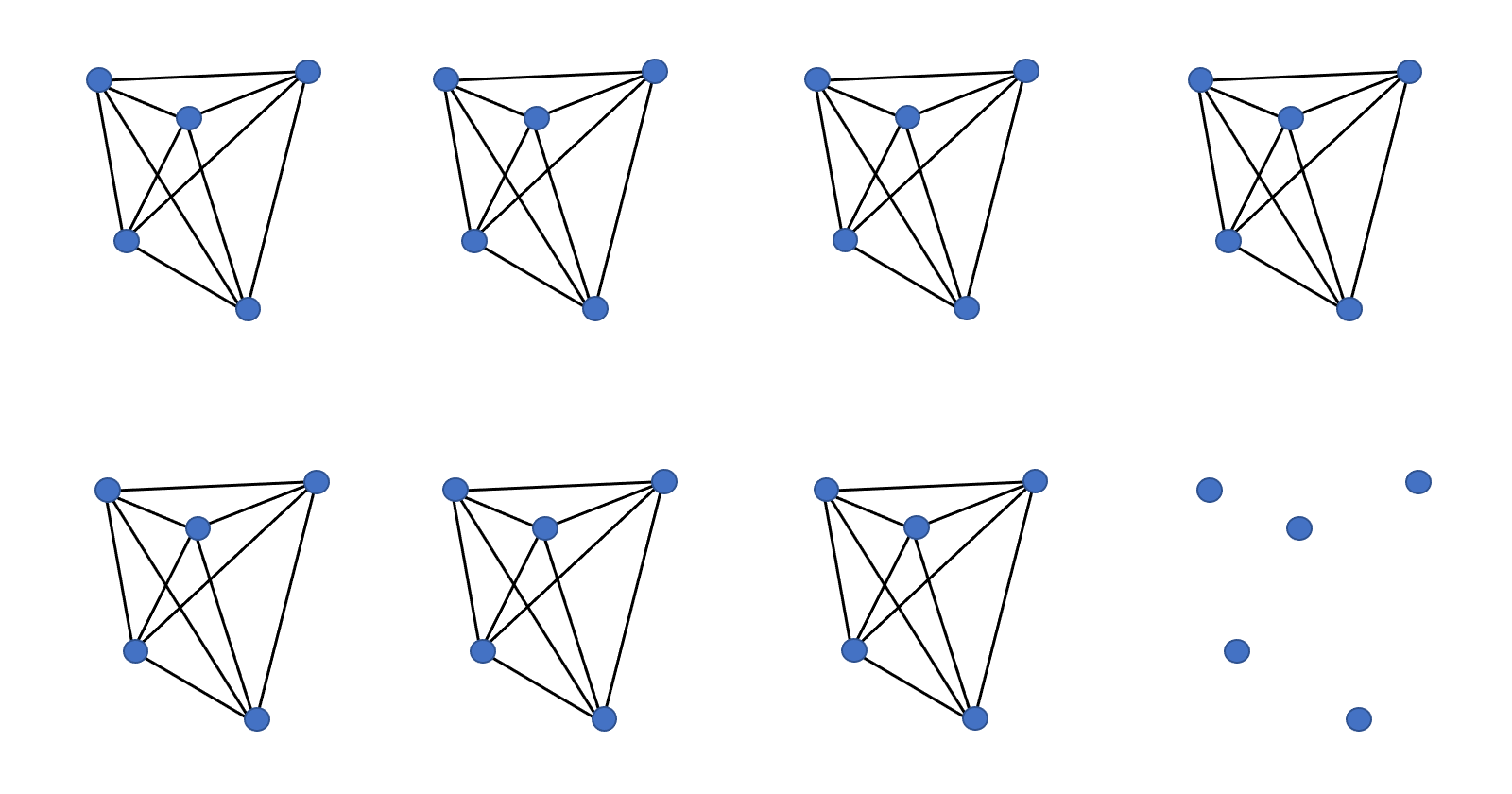}}}
	\caption{An illustration of $w_{\sf Y}$ and $w_{\sf N}$: pairs of vertices with distance $1$ are connected by edges. \label{fig: metrics}}
\end{figure}

We define $G_{\sf Y}$ to be the (unweighted) graph on the set $V$ of vertices in which there is an edge between every pair of vertices that come from the same group, and define $G_{\sf N}$ to be the graph on $V$ in which there is an edge between every pair of vertices that come from the same regular group.
So the problem of determining whether a graph is $G_{\sf Y}$ or $G_{\sf N}$ is essentially determining whether the graph consists of several equal-sized cliques or several equal-sized cliques plus an independent set of the same size.
It is easy to see that, the problem of distinguishing between metrics $w_{\sf Y}$ and $w_{\sf N}$ in the $\emph{\sf MST}^k_{\sf apx}$ problem is equivalent to the following graph-theoretic communication problem.

\paragraph{Clique or independent set problem ($\coi_{N,m}$):} 
This is a communication game between $k$ players, in which each player $P_i$ is given two subsets $R_i,R'_i\subseteq V\times V$ with $R'_i\subseteq R_i$. Intuitively, the subset $R_i$ represents the set of all pairs of vertices whose edge/non-edge information is given to the player $P_i$, and the subset $R'_i$ represents the set of all pairs in $R_i$ that has an edge connecting them. The goal of the players is to decide whether the graph $G=(V,E)$ where $E=\bigcup_{1\le i\le k}R'_i$ is $G_{\sf Y}$ or $G_{\sf N}$, if (i) $\bigcup_{1\le i\le k}R_i=V\times V$; and (ii) the information received by the players is consistent, i.e., there does not exists a pair $v,v'\in V$ and a pair $i,j\in [k]$, such that $(v,v')\in R_i\cap R_j$ and $(v,v')\in R'_i\setminus R'_j$.
If either (i) or (ii) does not hold, then any output of the protocol will be viewed as a correct output.

%It is not hard to see that the $\textsf{YES}$ and $\textsf{NO}$ cases in Figure~\ref{fig:multi-MST-yes} and Figure~\ref{fig:multi-MST-no} are instances of $\coi$. We consider the following black communication model: each player gets a set of pair of vertices, and for each pair, the player get the information that whether there is an edge between them. We prove the following lower bound.

The main result of this section is the following theorem, which immediately implies that the communication complexity of any randomized protocol that returns an $\alpha$-approximation of problem $\emph{\sf MST}^k_{\sf apx}$ has communication complexity at least $\Omega(\sqrt{n/\alpha})$, which in turn implies \Cref{thm: p pass alpha MST lower}.

\begin{theorem} \label{thm:coi}
For any integers $m$ and $N$, any randomized protocol $\pi$ for the $\coi_{N,m}$ problem in the blackboard model with $k=\floor{2\log N}$ players has communication complexity $\Omega(m/\log^3 N)$.
\end{theorem}

The remainder of this section is dedicated to the proof of \Cref{thm:coi}.
As a warm up, we first consider the special case where $N=2$, namely the $\coi_{2,m}$ problem in the 2-player blackboard model (note that in this case $k=2=2\log N$). In this case, graph $G_{\sf Y}$ is a perfect matching and graph $G_{\sf N}$ is a matching of size $m-1$. We show in \Cref{COI: special case} that any randomized protocol for this problem has communication complexity $\Omega(m)$. %At a high level, we prove the information cost lower bound. We first prove that the information cost of $\coi_{2,2}$ is $\Omega(1)$, and then uses a directed sum type argument to prove that the information cost and communication complexity of the matching problem is $\Omega(m)$.
We then generalize the ideas used in \Cref{COI: special case} to the general case, and complete the proof of \Cref{thm:coi} in \Cref{COI: general}.

\subsection{Communication Lower Bound for Bipartite Matching}
\label{COI: special case}

In this section, we consider the problem $\coi_{2,2m}$ (we use $2m$ instead of $m$ for convenience) in the 2-player blackboard model, and 
show that any randomized protocol for this problem has communication complexity $\Omega(m)$.
At a high level, we first show that the information cost of any randomized protocol for $\coi_{2,2}$ problem is $\Omega(1)$, and then use a direct sum type argument to prove that the information complexity and hence the communication complexity of the problem $\coi_{2,2m}$ is $\Omega(m)$.

We first construct an input distribution $\nu$ for the problem $\coi_{2,2}$ as follows. 
We denote the vertex set by $U\cup V$ where $U=\set{u_1,u_2}$ and $V=\set{v_1,v_2}$.
Recall that $\coi_{2,2}$ is a 2-player game.
%, in which Alice receives a pair $R_A,R'_A$ of subsets, and Bob receives a pair $R_B,R'_B$ of subsets.
Alice's three possible inputs are:
\begin{itemize}
\item $A_1$:  $R_A=(U\cup V)\times (U\cup V)$, and $R'_A=\set{(u_1,v_2),(u_2,v_1)}$;
\item $A_2$:  $R_A=(U\cup V)\times (U\cup V)$, and $R'_A=\set{(u_1,v_1),(u_2,v_2)}$;
\item $A_3$:  $R_A=(U\cup V)\times (U\cup V)\setminus \set{(u_2,v_2)}$, and $R'_A=\set{(u_1,v_2)}$.
\end{itemize}
Bob's three possible inputs are:
\begin{itemize}
\item $B_1$:  $R_B=\set{(u_1,u_2),(v_1,v_2),(u_2,v_2)}$, and $R'_B=\set{(u_2,v_2)}$;
\item $B_2$:  $R_B=\set{(u_1,u_2),(v_1,v_2)}$, and $R'_B=\emptyset$;
\item $B_3$:  $R_B=\set{(u_1,u_2),(v_1,v_2),(u_2,v_2)}$, and $R'_B=\emptyset$.
\end{itemize}

\begin{figure}[h]
	\centering
	\includegraphics[scale=0.17]{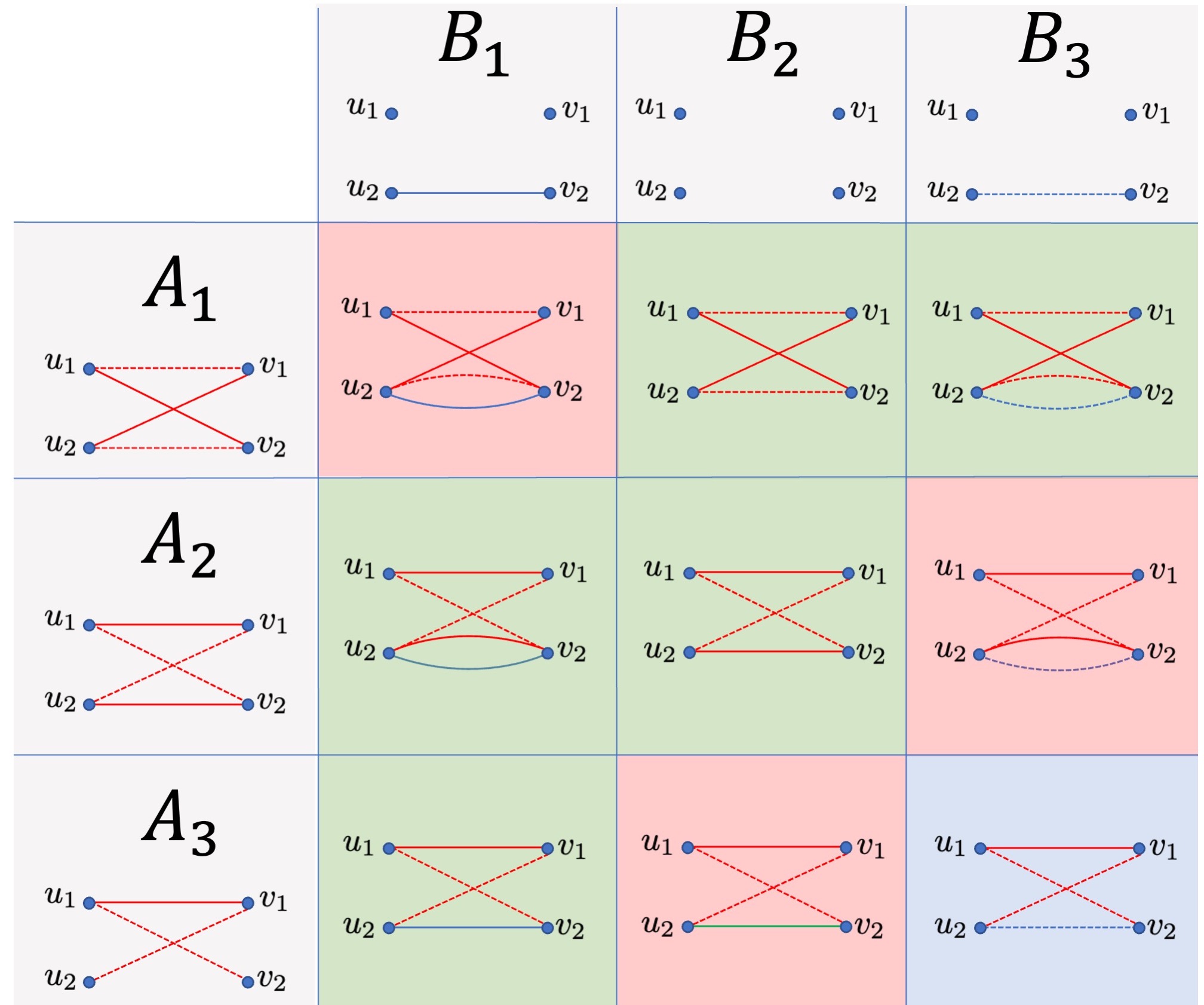}
	\caption{An illustration of the input pairs in distribution $\nu$ for the problem $\coi_{2,2}$. The red and blue lines represent edge/non-edge information in Alice's and Bob's inputs, respectively. Solid lines represent edges (pairs in $R'_A, R'_B$), and dashed lines represent non-edges (pairs in $R_A\setminus R'_A, R_B\setminus R'_B$). Note that the pairs $(u_1,u_2)$ and $(v_1,v_2)$ are given to both Alice and Bob in all inputs as non-edges, and so they are omitted in the figure for simplicity.
	The green line in the cell $(A_2,B_3)$ indicates that the information regarding pair $(u_2,v_2)$ is not given to either Alice or Bob (and thus, this input pair is invalid). %Alice's inputs are the same in each row, and Bob's inputs are the same in each column. 
	The green cells represent valid inputs where the graph $G=G_{\sf Y}$, the blue cell represents a valid input where the graph $G=G_{\sf N}$, and the red cells represent invalid inputs.}\label{fig:matching_input}
\end{figure}

See \Cref{fig:matching_input} for an illustration of the input pairs. %Both Alice and Bob have three possible inputs, but some of them are invalid. 
We define the distribution $\nu$ on the input pairs to be the uniform distribution on the set $\set{(A_1,B_2),(A_1,B_3),(A_2,B_1),(A_2,B_2),(A_3,B_1)}$. We prove the following claim.
%Define a distribution of input $\nu$ as these five valid inputs with a perfect matching with probability $0.2$ each. 

\begin{claim}
\label{clm: Omega(1) CC}
For any randomized protocol $\pi'$ for the problem $\coi_{2,2}$, $\ic_{\nu}(\pi')=\Omega(1)$.
\end{claim}
\begin{proof}
%Next, we prove that if $\pi'$ has success probability $0.999$, then $\ic_{\nu}(\pi')=\Omega(1)$. 
%For any $1\le i,j \le 3$, we say $I^*_A = i$ and $I^*_B=j$ if Alice and Bob's input is the $i^{th}$ row and $j^{th}$ input in Figure~\ref{fig:matching_input}.
For each $i\in \set{1,2,3}$, we let $X^A=i$ iff Alice's input is $A_i$. We define the variable $X^B$ similarly.
For any pair $i,j\in \set{1,2,3}$, we denote by $\Pi'_{ij}$ the random variable representing the transcript of $\pi'$ when Alice's input is $A_i$ and Bob's input is $B_j$ (namely $X^A=i$ and $X^B=j$).
We will show that $\ic_{\nu}(\pi')\ge 10^{-3}$. Assume for contradiction that this does not hold. Then %(note that although some input is not valid, we can still use the protocal to get the transcript). Suppose $\cc_{\nu}(\pi') < 0.001$, then
\begin{align*}
10^{-3} &> \ic_{\nu}(\pi') \\
&= \II_{\nu}(\Pi';X^A|X^B) + \II_{\nu}(\Pi';X^B|X^A) \\
&= \frac{2}{5}\cdot\bigg(\II_{\nu}(\Pi';X^A|X^B=1) + \II_{\nu}(\Pi';X^A|X^B=2) + \II_{\nu}(\Pi';X^B|X^A=1) + \II_{\nu}(\Pi';X^B|X^A=2) \bigg),
\end{align*}
where the last equality follows from the definition of distribution $\nu$ (note that $\Pr[X^B=3]=1/5$ and  $\Pr[X^B=1]=\Pr[X^B=2]=2/5$, and if $X^B=3$, then $X^A=1$ with probability $1$).

Note that, from \Cref{fac:hel-dkl}, 
$$
\II_{\nu}(\Pi'; X^A |X^B=1) = \frac{1}{2}\cdot
D_{\textnormal{\textsf{KL}}}\bigg(\Pi_{21}\text{ } \bigg|\bigg|\text{ }\frac{\Pi_{21}+\Pi_{31}}{2}\bigg)
+
\frac{1}{2}\cdot
D_{\textnormal{\textsf{KL}}}\bigg(\Pi_{31}\text{ } \bigg|\bigg|\text{ }\frac{\Pi_{21}+\Pi_{31}}{2}\bigg)
\ge \hel^2(\Pi_{21},\Pi_{31}). 
$$
Therefore, $\hel^2(\Pi_{21},\Pi_{31}) < (5/2)\cdot 10^{-3}$, and so $\hel(\Pi_{21},\Pi_{31}) < 1/20$. Similarly, we get that $\hel(\Pi_{21},\Pi_{22}) < 1/20$. Then from \Cref{clm:triangle}, $\hel(\Pi_{31},\Pi_{22}) \le \hel(\Pi_{31},\Pi_{21})+\hel(\Pi_{21},\Pi_{22}) < 0.1$, and then from \Cref{fac:rect}, we get that $\hel(\Pi_{21},\Pi_{32}) < 0.1$. From \Cref{clm:triangle}, $\hel(\Pi_{22},\Pi_{32})\le\hel(\Pi_{22},\Pi_{21})+\hel(\Pi_{21},\Pi_{32})  < 0.15$. Similarly, $\hel(\Pi_{22},\Pi_{23}) < 0.15$, so $\hel(\Pi_{23},\Pi_{32}) < 0.3$. Again by \Cref{fac:rect}, $\hel(\Pi_{22},\Pi_{33})<0.3$. Lastly, from \Cref{fac:hel-tvd}, we get that $\tvd{\Pi_{22}}{\Pi_{33}} < \frac{3\sqrt{2}}{10} < 0.425$.

However, this means that, with probability at least $2\cdot(0.5-0.425)=0.15$, the transcript of $\pi'$ on the input pair $(A_2,B_2)$ and the transcript of $\pi'$ on the input pair $(A_3,B_3)$ are identical, and thus the output in these two cases are identical. However, when the input pair is $(A_2,B_2)$, the graph $G$ is $G_{\sf Y}$, while if the input pair is $(A_3,B_3)$, the graph $G$ is $G_{\sf N}$. 
This contradicts the fact that $\pi'$ is a randomized protocol for the problem $\coi_{2,2}$ (as it needs to succeed with probability at least $99/100$).
\end{proof}

We now show a communication complexity lower bound for the problem $\coi_{2,2m}$ via a reduction from problem $\coi_{2,2}$ using a standard direct sum type argument. We first construct an input distribution for $\coi_{2,2m}$ as follows.

We denote by $U\cup V$ the vertex set, where $|U|=|V|=2m$. The set $U$ is further partitioned into $m$ subsets: $U=\bigcup_{1\le i\le m}U^i$, where for each $1\le i\le m$, $U^i=\set{u^i_1,u^i_2}$. Similarly, the set $V$ is further partitioned into $m$ subsets: $V=\bigcup_{1\le i\le m}V^i$, where for each $1\le i\le m$, $V^i=\set{v^i_1,v^i_2}$. From the setting of the problem $\coi_{2,2m}$, the graph $G_{\sf Y}$ is a perfect matching between $U$ and $V$, and the graph $G_{\sf N}$ is a size-$(2m-1)$ matching between $U$ and $V$. 

%We denote $u^i_j$ and $v^i_j$ where $1 \le i \le m$ and $1 \le j \le 2$. 

The non-edge information for every pair in $(U\times U)\cup (V\times V)$ is always given to both Alice and Bob. 
For each pair $i,i'\in [m]$ with $i\ne i'$ and for every pair $(u,v)$ of vertices with $u\in U^i$ and $v\in V^{i'}$, the non-edge information is also given to both Alice and Bob. 
The distribution of the input edge/non-edge information between each $(U^i,V^i)$-pair is identical to $\nu$ (with $u^i_1$ viewed as $u_1$, $u^i_2$ viewed as $u_2$, etc), and the inputs between different pairs $(U^i,V^i)$, $(U^{i'},V^{i'})$ are independent. This completes the definition of the input distribution, and we denote this distribution by $\mu$.

%Let $\textsf{M}_i$ be the matching problem ($\coi_{2,2}$) that defines on the subgraph induced by $\{u^i_1,u^i_2,v^i_1,v^i_2\}$, then the original problem $\textsf{M}(G) = \bigwedge_{i=1}^m \textsf{M}_i$.

%Let $\mu$ be the distribution of input on $4n$ vertices such that each $\textsf{M}_i$ is sampled by $\nu$ independently.

We prove the following claim.

\begin{claim}
\label{COL 2m}
For any randomized protocol $\pi$ for the problem $\coi_{2,m}$, $\cc_{\mu}(\pi)=\Omega(m)$.
\end{claim}

%In $\coi_{2,2m}$, there is a bipartite graph $G = (U,V,E)$ where $\card{U} = \card{V} = 2m$. Let $f(u,v) = 1$ if there is an edge between $(u,v)$, otherwise $f(u,v)=0$. Alice and Bob both receive a set of pairs of vertices $S_A$ and $S_B$ respectively such that $S_A \cup S_B = U \times V$. For any pair of vertices $(u,v)$ in $S_A$ (or $S_B$), Alice (or Bob) receive $f(u,v)$. In the $\textsf{YES}$ case, $G$ contains a perfect matching and in the $\textsf{NO}$ case, the graph is a $m-1$ matching. We prove that any protocol that success with probability at least $0.999$ requires $\Omega(m)$ bits of communication.

\begin{proof}
%We use a direct sum type argument to prove that $\ic_{\mu}(\pi) = \Omega(m)$ for any $\pi$ that is correct on every input with probabilty at least $0.999$. We give a protocol that solves the problem on 4 vertices $\pi'$ with successful probability at least $0.999$ basing on $\pi$, prove that $\ic_{\nu}(\pi') = \frac{1}{n} \ic_{\mu}(\pi)$, and then prove that $\ic_{\nu}(\pi') = \Omega(1)$.
Let $\pi$ be any randomized protocol for the problem $\coi_{2,m}$.
We use $\pi$ to construct a randomized protocol $\pi'$ for problem $\coi_{2,2}$ as follows.

\begin{tbox}
    \textbf{Protocol} $\pi'$: A protocol for $\coi_{2,2}$ using a protocol $\pi$ for $\coi_{2,m}$.

    \smallskip

    \textbf{Input:} An instance $\set{(R_A,R'_A), (R_B,R'_B)}$ of problem $\coi_{2,2}$. \\
    \textbf{Output:} A label $\sf Y$/$\sf N$ indicating whether $G$ is the graph $G_{\sf Y}$ in problem $\coi_{2,2}$ (a perfect matching on $4$ vertices) or $G$ is the graph $G_{\sf N}$ in problem $\coi_{2,2}$ ($G$ only contains a single edge).

    \begin{enumerate}
        \item Alice and Bob create an instance $\set{(\hat R_A, \hat R'_A), (\hat R_B, \hat R'_B)}$ of problem $\coi_{2,m}$ as follows:
            \begin{enumerate}
                %\item Let $L=(L_1,L_2,\dots,L_m)$ and $R=(R_1,R_2,\dots,R_m)$, where each $L_i$ and $R_i$ contains $2$ vertices.  
                \item Alice and Bob both add all pairs $(u,u')$ with $u,u'\in U$ and all pairs $(v,v')$ with $v,v'\in V$ to $\hat R_A$ and $\hat R_B$ but not to $\hat R'_A$ or $\hat R'_B$. Similarly, for every pair $u,v$ with $u\in U^i$ and $v\in V^{i'}$ such that $i\ne i'$, add the pair $(u,v)$ to $\hat R_A$ and $\hat R_B$ but not to $\hat R'_A$ or $\hat R'_B$.
                \item Alice and Bob use \underline{public coins} to sample an index $i^* \in [m]$.
                \item For each $1 \le i < i^*$, Alice and Bob generate an instance $\set{(R_A(i),R'_A(i)),(R_B(i),R'_B(i))}$ of $\coi_{2,2}$ as follows: they first use \underline{public coins} to generate Alice's input $(R_A(i),R'_A(i))$ according to $\nu$, and then Bob uses \underline{private coins} to generate his own input $(R_B(i),R'_B(i))$ according to $\nu$ conditioned on Alice's input.
                \item For each $i^* < i \le m$, Alice and Bob generate an instance $\set{(R_A(i),R'_A(i)),(R_B(i),R'_B(i))}$ of $\coi_{2,2}$ as follows: they first use \underline{public coins} to generate Bob's input $(R_B(i),R'_B(i))$ according to $\nu$, and then Alice uses \underline{private coins} to generate her own input $(R_A(i),R'_A(i))$ according to $\nu$ conditioned on Bob's input.
                \item  Alice sets $(R_A(i^*),  R'_A(i^*))=( R_A,  R'_A)$, and Bob sets $(R_B(i^*),  R'_B(i^*))=( R_B,  R'_B)$.
                \item  Alice adds all pairs in $\bigcup_{1\le i\le m} R_A(i)$ to $\hat R_A$, adds all pairs in $\bigcup_{1\le i\le m} R'_A(i)$ to $\hat R'_A$; and Bob adds all pairs in $\bigcup_{1\le i\le m} R_B(i)$ to $\hat R_B$, adds all pairs in $\bigcup_{1\le i\le m} R'_B(i)$ to $\hat R'_B$.
            \end{enumerate} 
        \item Alice and Bob then run protocol $\pi$ on the instance $\set{(\hat R_A, \hat R'_A), (\hat R_B, \hat R'_B)}$ of problem $\coi_{2,m}$, and return the output of $\pi$.
    \end{enumerate}
\end{tbox}

Since for each $i \neq i^*$, the information pair $(R_A(i),R'_A(i)),(R_B(i),R'_B(i))$ is generated according to $\nu$. 
If we denote by $G$ the graph for the input instance $(R_A, R'_A),(R_B, R'_B)$ of problem $\coi_{2,2}$, and denote by $\hat G$ the graph for the constructed instance $(\hat R_A, \hat R'_A),(\hat R_B, \hat R'_B)$ of problem $\coi_{2,m}$, then $G=G_{\sf Y}$ iff graph $\hat G$ is a perfect matching. Therefore, the output of $\pi'$ is correct iff the output of $\pi$ is correct. On the other hand, it is easy to see that the distribution of the constructed instance $(\hat R_A, \hat R'_A),(\hat R_B, \hat R'_B)$ is identical to $\mu$.

%Now we prove that $\ic_{\nu}(\pi') = \frac{1}{n} \ic_{\mu}(\pi)$. 
For each $j\in \set{1,2,3}$, we let $X^A=j$ iff Alice's input $(R_A, R'_A)$ for the $\coi_{2,2}$ problem is $A_j$, and we define the variable $X^B$ similarly.
For each $1\le i\le m$ and for each $j\in \set{1,2,3}$, we let $\hat X^A_i=j$ iff the pair $(R_A(i), R'_A(i))$ constructed in the protocol is $A_j$, and we define variable $\hat X^B_i$ similarly.
%Let $I_A$ and $I_B$ be the input Alice and Bob get from the input $\{G,S_A,S_B\}$. For any $i$, let $I^i_A$ and $I^i_B$ be the input Alice and Bob get from $\{G_i,S^i_A,S^i_B\}$, and let $I^*_A$ and $I^*_B$ be the input Alice and Bob get from $\{G^*,S^*_A,S^*_B\}$. Then we have 
Lastly, we denote $\hat X_A=(\hat X^A_1,\ldots, \hat X^A_m)$ and $\hat X_B=(\hat X^B_1,\ldots, \hat X^B_m)$.
Then, from \Cref{clm: Omega(1) CC},
\begin{align*}
\Omega(1)=\II_{\nu}(\Pi';X^B|X^A) &\le  \II_{\mu}(\Pi ; X^B\mid i^*, \hat X^A_1,\ldots, \hat X^A_m,\hat X^B_1,\ldots, \hat X^B_{i^*-1}) \\
&= \frac{1}{m} \cdot \sum_{i=1}^m \II_{\mu}(\Pi;\hat X^B_i \mid \hat X^A,\hat X^B_1,\dots,\hat X^B_{i-1}) \\
&= \frac{1}{m} \cdot \II_{\mu}(\Pi;\hat X^B | \hat X^A).
\end{align*}
where the first equality is from \ref{fac:mul-ind} in \Cref{clm: basic properties} the mutual independence between the random variables in $\set{\hat X^A_t\mid 1\le t\le m}$ and $\set{\hat X^B_t\mid 1\le t\le i^*}$, the second equality follows from the fact that $i^*$ is chosen uniformly at random from the set $[m]$, and the last equality is from the chain rule (\ref{fact:chain-rule} in \Cref{clm: basic properties}).

Lastly, from \Cref{fac:multi-cc}, we get that $\cc_{\mu}(\pi)\ge  \II_{\mu}(\Pi;\hat X^B | \hat X^A)=\Omega(m)$. This completes the proof of \Cref{COL 2m}.
\end{proof}

\subsection{Communication Lower Bound for MST Estimation in Metric Streams}
\label{COI: general}

In this subsection, we generalize the ideas and techniques used in \Cref{COI: special case} to prove a lower bound for the communication complexity of the $\coi_{N,m}$ problem, thus establishing \Cref{thm:coi}. 
Let $k=\log N$. We consider the communication game $\coi_{N,m}$ in the blackboard model between $2k$ players, denoted by $A^1,\dots,A^k,B^1,\dots,B^k$. 
Similar to \Cref{COI: special case}, we will first show that the information complexity of any randomized protocol for problem $\coi_{N,2}$ is $\tilde{\Omega}(1)$, and then show that the communication complexity of any randomized protocol for problem $\coi_{N,m}$ is $\tilde{\Omega}(m)$ via a direct sum type argument.

%Each player gets a set of pairs of vertices and gets the information of whether there is an edge between them. Moreover, for any pair of vertices, at least one player gets the information.

The crucial step in this subsection is to construct a hard input distribution for the problem $\coi_{N,2}$, that (i) can be used to show that the information complexity of the problem $\coi_{N,2}$ is $\tilde{\Omega}(1)$; and (ii) can be easily amplified (for example, via a product distribution) into an input distribution for the problem $\coi_{N,m}$ such that the direct sum argument can be applied. 
We first provide some intuition. 
Let $\nu$ be an input distribution for problem $\coi_{N,2}$.
Intuitively, we would like to mimick the construction in \Cref{COI: special case} for showing that the information complexity of the problem $\coi_{2,2}$ is $\Omega(1)$ (see \Cref{fig:matching_input}).
On the one hand, to achieve property (ii), $\nu$ should be supported mostly, if not always, on Yes-instances (where the underlying graph is $G_{\sf Y}$ in $\coi_{2,2}$), since otherwise, if with some probability $p=\omega(1/m)$ the distribution is supported on No-instances, then the probability that the product input distribution $\nu^{m/2}$ for problem $\coi_{N,2}$ is valid is at most $(1-p)^{m/2}+m\cdot p(1-p)^{m/2-1}$, which is too small and therefore insufficient for an $\tilde \Omega(1)$ information complexity lower bound of problem $\coi_{N,2}$.
On the other hand, to make a valid input distribution for problem $\coi_{2,2}$, and also to ensure that there is always some player who gets partial information and is therefore uncertain about the output, mimicking the construction in \Cref{COI: special case}, we decompose a clique of size $N=2^k$ into the union of $k$ bi-cliques, and then pair the players $A_i,B_i$ to give them partial inputs similar to the ones in \Cref{fig:matching_input}.

%To prove the lower bound, we consider the problem $\coi_{N,2}$. 
We now provide the construction of such an input distribution $\nu$ for problem $\coi_{N,2}$ in more details.
We denote the vertex set by $P\cup Q$, where $|P|=|Q|=N=2^k$.
Each vertex $p\in P$ are labeled by a string in $\set{0,1}^k$. Abusing the notation, we also denote by $p=(p_1,\ldots, p_k)$ the string that corresponds to the vertex $p$. For each $1\le i\le k$, we define
$U^P_{i}=\set{p\in P\mid p_i=0}$, namely the set of all vertices in $P$ that corresponds to a string whose $i$-th bit is $0$, and we define $V^P_{i}=\set{p\in P\mid p_i=1}$. 
Clearly, for each $i\in [k]$, $(U^P_{i},V^P_{i})$ is a partition (or a cut) of $P$ (that we denote as the \emph{$i$-th partition/cut of $P$}), the partitions (or cuts) in $\set{(U^P_{i},V^P_{i})\mid i\in [k]}$ are distinct, and each edge of $E(P)$ participates in at least one of these cuts.
Similarly, we label each vertex $q\in Q$ by a string in $\set{0,1}^k$, that we denote by $q$ as well, abusing the notation, and we define sets $U^Q_{i},V^Q_{i}$ the \emph{$i$-th partition/cut} $(U^Q_{i},V^Q_{i})$ of $Q$ similarly, for each $1\le i\le k$.
The edge/non-edge information given to the players will respect these partitions (and not more fine-grained).

We now describe the possible inputs to the players. Let $i$ be an index in $[k]$. The $(k+2)$ possible inputs for player $A^i$ are described below:
\begin{itemize}
	\item for each $j\in [k], j\ne i$, a general input $A^i_{j}$:
	\\ $R_{A^i}\setminus R'_{A^i} =(U^P_j \cup V^Q_j) \times (U^Q_j \cup V^P_j)$, and \\
	$R'_{A^i}=\bigg((U^P_j \cup V^Q_j)\times (U^P_j \cup V^Q_j)\bigg) \cup \bigg((U^Q_j \cup V^P_j)\times (U^Q_j \cup V^P_j)\bigg)$;\\
	(note that this input is with respect to the $j$-th partitions but not the $i$-th ones, also note that this input contains complete edge/non-edge information, i.e., in this input $A^{i}$ is given the whole underlying graph, which is $G_{\sf Y}$, consisting of a clique on $(U^P_j \cup V^Q_j)$ and a clique on $(U^Q_j \cup V^P_j)$);
	\item special input $A^i_{\row(1)}$ (corresponding to the first row in \Cref{fig:matching_input}):
	\\ $R_{A^i}\setminus R'_{A^i}=(U^P_i \times V^P_i) \cup (U^Q_i \times V^Q_i)$, and $R'_{A^i}=(U^P_i \times V^Q_i) \cup (U^Q_i \times V^P_i)$;\\
	(if we view set $U^P_i$ as $u_1$, set $U^Q_i$ as $u_2$, set $V^P_i$ as $v_1$, and set $V^Q_i$ as $v_2$, then this is exactly the first row in \Cref{fig:matching_input}, and similarly for the next two special inputs);
	\item special input $A^i_{\row(2)}$ (corresponding to the second row in \Cref{fig:matching_input}):\\ 
	 $R_{A^i}\setminus R'_{A^i}=(U^P_i \times V^Q_i) \cup (U^Q_i \times V^P_i)$, and $R'_{A^i}=(U^P_i \times V^P_i) \cup (U^Q_i \times V^Q_i)$;
	\item special input $A^i_{\row(3)}$ (corresponding to the third row in \Cref{fig:matching_input}):
	\\  $R_{A^i}\setminus R'_{A^i}= (U^P_i\times V^Q_i)\cup (U^Q_i\times V^P_i)$, and $R'_{A^i}=U^P_i\times V^P_i$.
\end{itemize}

The $(k+2)$ possible inputs for player  $B^i$  are:
\begin{itemize}
	\item for each $j\in [k], j\ne i$, a general input $B^i_{j}$:
	\\ $R_{B^i}\setminus R'_{B^i} = (U^P_j \cup V^Q_j) \times (U^Q_j \cup V^P_j)$, and \\
	$R'_{B^i}= \bigg((U^P_j \cup V^Q_j)\times (U^P_j \cup V^Q_j)\bigg) \cup \bigg((U^Q_j \cup V^P_j)\times (U^Q_j \cup V^P_j)\bigg)$;\\
	(note that this input is with respect to the $j$-th partitions but not the $i$-th ones, also note that this input contains complete edge/non-edge information, i.e., in this input $B^{i}$ is given the whole underlying graph, which is $G_{\sf Y}$, consisting of a clique on $(U^P_j \cup V^Q_j)$ and a clique on $(U^Q_j \cup V^P_j)$);
	\item special input $B^i_{\col(1)}$ (corresponding to the first column in \Cref{fig:matching_input}):\\ 
	$R_{B^i}\setminus R'_{B^i}=\emptyset$, and $R'_{B^i}=(U^Q_i \times V^Q_i)$;\\
	(if we view set $U^P_i$ as $u_1$, set $U^Q_i$ as $u_2$, set $V^P_i$ as $v_1$, and set $V^Q_i$ as $v_2$, then this is exactly the first column in \Cref{fig:matching_input}, and similarly for the next two special inputs)
	\item special input $B^i_{\col(2)}$ (corresponding to the second column in \Cref{fig:matching_input}):
	\\  $R_{B^i}\setminus R'_{B^i}= \emptyset$, and $R'_{B^i}=\emptyset$.
	\item special input $B^i_{\col(3)}$ (corresponding to the third column in \Cref{fig:matching_input}):
	\\  $R_{B^i}\setminus R'_{B^i}= U^Q_i\times V^Q_i$, and $R'_{B^i}=\emptyset$.
\end{itemize}

We now describe the distribution $\nu$. It is a uniform distribution on the following set of $(2k+3)$ input combinations:
\begin{itemize}
	\item for each $j\in [k]$, a general input combination $\Phi^j_{1,2}$:\\
	$(A^1_{j},\ldots,A^{j-1}_{j}, A^j_{\row(1)},A^{j+1}_{j},A^k_{j},B^1_{j},\ldots,B^{j-1}_{j}, B^j_{\col(2)},B^{j+1}_{j},B^k_{j})$;\\
	(intuitively, this is a combination in which all players except $A^j$ and $B^j$ know the whole graph, which is the union of a clique on $(U^P_j \cup V^Q_j)$ and a clique on $(U^Q_j \cup V^P_j)$, while players $A^j$ and $B^j$ are in the first-row-second-column scenario of \Cref{fig:matching_input})
	\item for each $j\in [k]$, a general input combination $\Phi^j_{1,3}$:\\
	$(A^1_{j},\ldots,A^{j-1}_{j}, A^j_{\row(1)},A^{j+1}_{j},A^k_{j},B^1_{j},\ldots,B^{j-1}_{j}, B^j_{\col(3)},B^{j+1}_{j},B^k_{j})$;\\
	(similarly, this is a combination in which all players except $A^j$ and $B^j$ know the whole graph, which is the union of a clique on $(U^P_j \cup V^Q_j)$ and a clique on $(U^Q_j \cup V^P_j)$, while players $A^j$ and $B^j$ are in the first-row-third-column scenario of \Cref{fig:matching_input})
	\item special input combination $\Phi_{2,1}$: $(A^1_{\row(2)},\ldots,A^k_{\row(2)},B^1_{\col(1)},\ldots,B^k_{\col(1)})$;\\
	(the underlying graph of this input combination is the union of a clique on $P$ and a clique on $Q$; where the edge-information are distributed into $k$ bi-cliques to the players, and similarly for the next two combinations)
	\item special input combination $\Phi_{2,2}$: $(A^1_{\row(2)},\ldots,A^k_{\row(2)},B^1_{\col(2)},\ldots,B^k_{\col(2)})$;
	\item special input combination $\Phi_{3,1}$: $(A^1_{\row(3)},\ldots,A^k_{\row(3)},B^1_{\col(1)},\ldots,B^k_{\col(1)})$.
\end{itemize}

We denote $\Omega=\set{\Phi^j_{1,2}\mid j\in [k]}\cup \set{\Phi^j_{1,3}\mid j\in [k]}\cup \set{\Phi_{2,1},\Phi_{2,2},\Phi_{3,1}}$.
It is not hard to verify that, for all input combinations in $\Omega$, each pair of vertices appears simultaneously in at least one set of $\set{R_{A^j}}_{j\in [k]}\cup \set{R_{B^j}}_{j\in [k]}$, and the information received by the players is consistent. Moreover, in all these combinations, the underlying graph is $G_{\sf Y}$ (a union of two cliques). In combinations $\Phi_{2,1}, \Phi_{2,2}$ and $\Phi_{3,1}$, the cliques are on vertex sets $P$ and $Q$. For any $1\le j \le k$, in combination $\Phi^j_{1,2}$ and $\Phi^j_{1,3}$, the cliques are on vertex sets $U^P_j \cup V^Q_j$ and $U^Q_j \cup V^P_j$.

To give some intuition why the information complexity of any randomized protocol is $\tilde\Omega(1)$ over the above distribution. Observe that, in each general input combination $\Phi^j_{1,2}$ or $\Phi^j_{1,3}$, there are two players $A^j, B^j$ who do not know whether or not the underlying graph is $G_{\sf Y}$ or $G_{\sf N}$, so they need to be notified by $\tilde \Omega(1)$ bit exchange. On the other hand, in each special input combination, no player knows whether or not the underlying graph is $G_{\sf Y}$ or $G_{\sf N}$.

We now provide a formal proof of an information complexity lower bound for any protocol for problem $\coi_{N,2}$.
Fix an index $1\le i\le k$. We define a random variable $X_i$ as follows. For each $j\in [k], j\ne i$, we let $X_i=j$ iff the input of player $A^i$ is $A^i_j$; for each $t\in [3]$, $X_i=\row(t)$ iff the input of player $A^i$ is $A^i_{\row(t)}$, so $X_i$ takes value from $([k]\setminus\set{j})\cup  \set{\row(1),\row(2),\row(3)}$. Similarly, we define a random variable $Y_i$ according to the input of player $B^i$, that takes a value from $([k]\setminus\set{j})\cup  \set{\col(1),\col(2),\col(3)}$. We denote $X=(X_1,\ldots,X_k)$, and for each $1\le i\le k$, we denote $X_{-i}=(X_1,\ldots,X_{i-1},X_{i+1},\ldots,X_k)$, and we define $Y, Y_{-i}$ similarly.

We prove the following claim.

\begin{claim} \label{cla:multi-small}
For any randomized protocol $\pi'$ for problem $\coi_{N,2}$, there exists an index $i\in [k]$, such that either $\II_{\nu}(\Pi';X_{-i}|X_{i}) = \Omega(1/k^3)$ or $\II_{\nu}(\Pi';Y_{-i}|Y_{i}) = \Omega(1/k^3)$ holds.
\end{claim}
\begin{proof}
Assume for contradiction that for each $1\le i\le k$, $\II_{\nu}(\Pi';X_{-i}|X_{i}) < 1/(10^5 k^3)$ and $\II_{\nu}(\Pi';Y_{-i}|Y_{i}) < 1/(10^5 k^3)$, then for each $1\le i \le k$, from \Cref{fac:hel-dkl},
	\begin{align*}
	\frac{1}{10^5\cdot k^3} & > \II_{\nu}(\Pi';X_{-i}|X_{i}) \\
	&\ge \Pr_{\nu}[X_i=\row(1)] \cdot \II_{\nu}(\Pi';X_{-i}|X_{i}=\row(1)) \\
	&= \frac{2}{2k+3} \cdot 	\frac{1}{2}\cdot\bigg(
	D_{\textnormal{\textsf{KL}}}\bigg(\Pi'(\Phi_{1,2}^i)\text{ } \bigg|\bigg|\text{ }\frac{\Pi'(\Phi_{1,2}^i)+\Pi'(\Phi_{1,3}^i)}{2}\bigg)
	+
	D_{\textnormal{\textsf{KL}}}\bigg(\Pi'(\Phi_{1,3}^i)\text{ } \bigg|\bigg|\text{ }\frac{\Pi'(\Phi_{1,2}^i)+\Pi'(\Phi_{1,3}^i)}{2}\bigg)
	\bigg) \\
	&\ge \frac{2}{2k+3} \cdot \hel(\Pi'(\Phi_{1,2}^i),\Pi'(\Phi_{1,3}^i))^2.
	\end{align*}
So $\hel(\Pi'(\Phi_{1,2}^i),\Pi'(\Phi_{1,3}^i)) < 1/40k$. Similarly, for each index $1\le i\le k$, $\hel(\Pi'(\Phi_{1,2}^i),\Pi'(\Phi_{2,2})) < 1/40k$, 
$\hel(\Pi'(\Phi_{2,1}),\Pi'(\Phi_{2,2})) < 1/40k$,
and
$\hel(\Pi'(\Phi_{2,1}),\Pi'(\Phi_{3,1})) < 1/40k$.
We define the following new input combinations:
\begin{itemize}
\item $\Phi_{2,2}=\big(A^1_{\row(2)},\ldots,A^k_{\row(2)},B^1_{\col(3)},\ldots,B^k_{\col(3)}\big)$;
\item $\Phi_{3,2}=\big(A^1_{\row(3)},\ldots,A^k_{\row(3)},B^1_{\col(2)},\ldots,B^k_{\col(2)}\big)$;
\item $\Phi_{3,3}=\big(A^1_{\row(3)},\ldots,A^k_{\row(3)},B^1_{\col(3)},\ldots,B^k_{\col(3)}\big)$.
\end{itemize}
 We prove the following claim.
\begin{observation} \label{cla:i23}
$\hel(\Pi'(\Phi_{2,2}),\Pi'(\Phi_{2,3})) < 1/10$.
\end{observation}
\begin{proof}
For each subset $S \subseteq [k]$, we define a new input combination 
$$\Phi_{2,3}^{S}=
\bigg(
\big(A^i_{\row(2)}\big)_{i\in [k]},\big(B^i_{\col(3)}\big)_{i\in S},\big(B^i_{\col(2)}\big)_{i\notin S}\bigg).$$
Note that $\Phi_{2,3}^{[k]}=\Phi_{2,3}$. We prove by inducetion on $|S|$ that, for every $S\subseteq [k]$, $\hel(\Pi'(\Phi_{2,2}),\Pi'(\Phi_{2,3}^S)) \le \card{S}/10k$. Note that \Cref{cla:i23} follows immediately if this is true.

The base case is when $|S|=1$, namely $S=\set{i}$ for some $i\in [k]$. Note that $\hel(\Pi'(\Phi_{1,2}^i),\Pi'(\Phi_{1,3}^i)) < 1/40k$ and $\hel(\Pi'(\Phi_{1,2}^i),\Pi'(\Phi_{2,2})) < 1/40k$, from \Cref{clm:triangle}, we get that $\hel(\Pi'(\Phi_{1,3}^i),\Pi'(\Phi_{2,2})) < 1/20k$. Then from \Cref{fac:multi-rect}, $\hel(\Pi'(\Phi_{1,2}^i), \Pi'(\Phi_{2,3}^{\set{i}})) < 1/20k$, and so $\hel(\Pi'(\Phi_{2,2}), \Pi'(\Phi_{2,3}^{\set{i}})) < 1/10k$, and therefore the claim is true when $|S|=1$.
%This means $\hel(\pi'(I_{22}),\pi'(I_{23}^{S})) \le \frac{\card{S}}{10k}$ when $\card{S}=1$.
		
Assume now the claim is true for all subsets $S\subseteq[k]$ with $|S|\le t$. Consider now any subset $S$ with $|S|=t+1$.
Pick an arbitrary element $i \in S$. From the induction hypothesis, $\hel(\Pi'(\Phi_{2,2}), \Pi'(\Phi_{2,3}^{S\setminus \set{i}})) \le (\card{S}-1)/10k$. On the other hand, since $\hel(\Pi'(\Phi_{2,2}), \Pi'(\Phi_{2,3}^{\set{i}})) < 1/10k$, we have $\hel(\Pi'(\Phi_{2,3}^{\set{i}}), \Pi'(\Phi_{2,3}^{S\setminus \set{i}})) < \card{S}/10k$. Then from \Cref{fac:multi-rect}, $\hel(\Pi'(\Phi_{2,2}), \Pi'(\Phi_{2,3}^{S})) < |S|/10k$. 
Therefore, the claim is true. This completes the proof of \Cref{cla:i23}.
\end{proof}
	
On the other hand, since $\hel(\Pi'(\Phi_{2,1}),\Pi'(\Phi_{2,2})) < 1/40k$,
and
$\hel(\Pi'(\Phi_{2,1}),\Pi'(\Phi_{3,1})) < 1/40k$, we get that $\hel(\Pi'(\Phi_{2,2}),\Pi'(\Phi_{3,1})) < 1/20k$. Then from \Cref{fac:multi-rect}, $\hel(\Pi'(\Phi_{2,1}),\Pi'(\Phi_{3,2})) < 1/20k$, and so $\hel(\Pi'(\Phi_{2,2}),\Pi'(\Phi_{3,2})) < 1/10k \le 1/10$. Combined with \Cref{cla:i23}, $\hel(\Pi'(\Phi_{2,3}),\Pi'(\Phi_{3,2})) < 1/5$. 
Next from \Cref{fac:multi-rect}, $\hel(\Pi'(\Phi_{2,2}),\Pi'(\Phi_{3,3})) < 1/5$. From \Cref{fac:hel-tvd},  $\tvd{\Pi'(\Phi_{2,2})}{\Pi'(\Phi_{3,3})} < 0.3$, which means, with probability at least $2\cdot(0.5-0.3)=0.4$, protocol $\pi'$ will have identical output on input combinations  $\Phi_{2,2}$ and $\Phi_{3,3}$. However, $\Phi_{3,3}$ is a valid input combination where the underlying graph is $G_{\sf N}$, while $\Phi_{2,2}$ is a valid input combination where the underlying graph is $G_{\sf Y}$. This contradicts the assumption that $\pi'$ is an $1/10$-error randomized protocol for problem $\coi_{N,2}$.
\end{proof}

We now proceed to show the communication complexity lower bound for $\coi_{N,m}$ via a reduction from $\coi_{N,2}$.
We first construct an input distribution $\mu$ for $\coi_{N,m}$ as follows. 
The vertex set is partitioned into $m/2$ groups of size $2N$ each. For every pair of vertices that come from different group, the non-edge information of this pair is given to all players. 
The distribution on the edge/non-edge information within each group is $\nu$, and the distribution for different groups are independent.

%We then use a direct sum type argument to prove that $\cc_{\mu}(\pi) = \Omega(m / k^2)$ for any randomized protocol $\pi$ for problem $\coi_{N,m}$. We give a protocol that solves the problem $\coi_{N,2}$ vertex $\pi'$ with a successful probability of at least 0.9 basing on $\pi$. 

Let $\pi$ be a randomized protocol $\pi$ for problem $\coi_{N,m}$. We use $\pi$ to construct a procotol $\pi'$ for problem $\coi_{N,2}$ as follows.

\begin{tbox}
    \textbf{Protocol} $\pi'$: A protocol for $\coi_{N,2}$ using the protocol $\pi$ for $\coi_{N,m}$.

    \smallskip

    \textbf{Input:} An instance $\Phi^*$ of $\coi_{N,2}$. \\
    \textbf{Output:} A label $\sf Y$ or $\sf N$ indicating whether $G=G_{\sf Y}$ or $G=G_{\sf N}$ in problem $\coi_{N,2}$.

    \begin{enumerate}
        \item Player $A_1$ choose an index $i$ from $[m/2]$ uniformly at sample.
        \item Player $A_1$ independently sample $(m/2-1)$ instances of $\coi_{N,2}$ problem from distribution $\nu$ on $\Omega$, and denote them by $\Phi_1,\dots,\Phi_{i-1},\Phi_{i+1},\dots,\Phi_{m/2}$.
        \item Player $A_1$ sends the number $i$ and instance $\Phi_1,\dots,\Phi_{i-1},\Phi_{i+1},\dots,\Phi_{m/2}$ to all other players.
        \item The players run the protocol $\pi$ on the instance $\Phi' = (\Phi_1,\dots,\Phi_{i-1},\Phi^*,\Phi_{i+1},\dots,\Phi_{m/2})$ of the $\coi_{N,m}$ problem, where the instance $\Phi_j$ is on the $j$-th group of vertices and the instance $\Phi^*$ is on the $i$-th group of vertices, and then return the output of $\pi$.
    \end{enumerate}
\end{tbox}

We now proceed to prove a lower bound on the information complexity of $\pi'$. Let $\Phi$ be an instance of either the problem $\coi_{N,m}$ or the problem $\coi_{N,2}$. For each $1\le i\le k$, we denote by $X_i(\Phi)$ the input of player $A^i$ in the instance $\Phi$. We denote $X(\Phi)=(X_j(\Phi))_{j\in [k]}$ and 
$X_{-i}(\Phi)=(X_j(\Phi))_{j\in [k],j\ne i}$.
Similarly, we define $Y_i(\Phi)$, $Y(\Phi)$ and $Y_{-i}(\Phi)$ for inputs of players $B^1,\ldots,B^k$.

Observe that, for each $j \neq i$, instance $\Phi_{j}$ is sampled from $\nu$, so instance $\Phi'$ of problem $\coi_{N,m}$ is a $\textsf{YES}$ instance iff instance $\Phi^*$ of problem $\coi_{N,2}$ is a $\textsf{YES}$ instance. Moreover, since instance $\Phi^*$ is also sampled from $\nu$, the distribution of instance $\Phi'$ is exactly $\mu$. 

\begin{claim} \label{cla:ds-mst}
For each $1\le j\le k$,
$\II_{\nu}(\Pi';X_{-j}(\Phi^*)|X_j(\Phi^*)) \le \frac{2}{m}\cdot \II_{\mu}(\Pi;X_{-j}(\Phi')|X_j(\Phi'))$, and\\ $\II_{\nu}(\Pi';Y_{-j}(\Phi^*)|Y_j(\Phi^*)) \le \frac{2}{m} \cdot \II_{\mu}(\Pi;Y_{-j}(\Phi')|Y_j(\Phi'))$.
\end{claim}

\begin{proof}
%For any $1 \le j \le k$, for any $\ell \neq i$, the index $i$, the input $\Phi_{\ell}$ and $X_j(\Phi_{\ell})$ are independent with $X_{-j}(\Phi^*)$, given $X_j(\Phi^*)$. Let $\Phi_i = \Phi^*$. Let $\Pi$ and $\Pi'$ be the transcripts when running protocal $\Pi$ and $\Pi'$ in the algorithm. We have 
Denote $\Phi_i = \Phi^*$. From the construction of protocol $\pi'$, $\Pi' = (i, \Phi_{1},\dots,\Phi_{i-1},\Phi_{i+1},\dots,\Phi_{m/2},\Pi)$.
From \ref{fac:mul-ind} in \Cref{clm: basic properties},
\begin{align*}
\II_{\nu}(\Pi';X_{-j}(\Phi^*)|X_j(\Phi^*)) &= \II_{\mu}(\Pi;X_{-j}(\Phi^*)\mid X_j(\Phi^*))\\
        &\le \II_{\mu}(\Pi;X_{-j}(\Phi^*) \mid  i,\Phi_1 \dots,\Phi_{i-1},X_j(\Phi_{i+1}),\dots,X_j(\Phi_{m/2}),X_j(\Phi^*)) \\
        & = \II_{\mu}(\Pi;X_{-j}(\Phi_i) \mid  i, X_j(\Phi_1) \dots, X_j(\Phi_{m/2}) , X_{-j}(\Phi_1),\dots,X_{-j}(\Phi_{i-1})) \\
        & = \II_{\mu}(\Pi;X_{-j}(\Phi_i) \mid i, X_j(\Phi') , X_{-j}(\Phi_1),\dots,X_{-j}(\Phi_{i-1})) \\
        & = \frac{2}{m}\cdot \sum_{i=1}^{m/2} \II_{\mu}(\Pi;X_{-j}(\Phi_i) \mid X_j(\Phi') , X_{-j}(\Phi_1),\dots,X_{-j}(\Phi_{i-1})) \\
        & = \frac{2}{m}\cdot \II_{\mu}(\Pi, X_{-j}(\Phi') | X_j(\Phi'))
    \end{align*}
where the last equation is from the chain rule (\ref{fact:chain-rule} in \Cref{clm: basic properties}). Similarly, we can prove that
    $$
    \II_{\nu}(\Pi';Y_{-j}(\Phi^*) | Y_j(\Phi^*)) \le \frac{2}{m}\cdot \II_{\mu}(\Pi, Y_{-j}(\Phi') | Y_j(\Phi')).
    $$
\end{proof}

From \Cref{cla:ds-mst} and \Cref{fac:multi-cc}, for any randomized protocol $\pi$ for the $\coi_{N,m}$ problem, $\cc_{\mu}(\pi) \ge \Omega(m/k^3)=\Omega(m/\log^3 N)$. This completes the proof of \Cref{thm:coi}.

%Now we prove the lower bound for $\coi_{N,m}$. By Claim~\ref{cla:ds-mst}, Claim~\ref{cla:multi-small} and Fact~\ref{fac:multi-cc}, $\cc_{\mu}(\pi) \ge \frac{m}{200k^2}$ for any $\pi$ with success probability at least $0.9$. On the other hand, it is easy to see that any $p$-pass streaming algorithm with $S$ space can be simulated by $k$ party broadcast communication protocal with communication complexity $pkS$. This finishes the proof of \Cref{thm:coi}. 

\iffalse
By replacing each non-edge in $\coi_{N,m}$ problem with distance $M$ between them, we get the lower bound for $\mst$ problem.

\begin{theorem}
    In metric stream model, any $p$ pass streaming algorithm that approximates the size of $\mst$ within a factor of $\alpha$ with probabiity at least $0.9$ have space complexity $\tilde{\Omega}(\frac{\sqrt{n/\alpha}}{p})$.
\end{theorem}
\fi

%\input{special_metric}

\section{Space Lower Bounds for MST and TSP Estimation in Graph Streams}

In this section we provide our lower bound results on graph streams.
We will prove a space complexity lower bound for multi-pass streaming algorithms for MST cost estimation in \Cref{Proof of p pass alpha graph MST lower}, and a space complexity lower bound for one-pass streaming algorithms for TSP-cost estimation in \Cref{Proof of 1 pass TSP lower}.

\subsection{Space Lower Bound for Multi-Pass Streaming Algorithms for MST Estimation}
\label{Proof of p pass alpha graph MST lower}

In this section we show that any randomized $p$-pass streaming algorithm that outputs any approximation for MST estimation has to use at least $\Omega(n/p)$, thus establishing \Cref{thm: p pass alpha graph MST lower}.
%we consider the MST problem in graph stream model. We prove a $\Omega(n)$ lower bound for multi-pass algorithm with any approximation rate.
Similar to \Cref{subsec: 1 pass MST lower proof}, we prove the space lower bound of streaming algorithms by analyzing a 2-player communication game defined as follows.

\paragraph{Two-player graph MST estimation problem ($\emph{\sf MST}^{\sf g}_{\sf apx}$):} This is a communication problem, in which Alice is given a weighted graph $G_A$, and Bob is given a weighted graph $G_B$, with the promise that $V(G_A)=V(G_B)$, $E(G_A)\cap E(G_B)=\emptyset$, and $G_A\cup G_B$ is a connected graph.
Alice and Bob are allowed to communicate in the blackboard model. The goal is to compute an estimate $Y$ of $\mst(G_A\cup G_B)$. 
For any constant $\alpha>0$, we say that $Y$ is an $\alpha$-approximation of $\mst(G_A\cup G_B)$ iff $Y\le \mst(G_A\cup G_B)\le \alpha\cdot Y$.

%Let $\dset$ be a distribution on the input pairs $(G_A,G_B)$ of weighted graphs. We say that a protocol $\pi$ of Alice and Bob in the blackboard model $(\alpha,\delta)$-approximates the problem $\emph{\sf MST}^{\sf g}_{\sf apx}$ over the distribution $\dset$ for some real numbers $\alpha>1$ and $0<\delta<1$, iff, with probability at least $1-\delta$, the estimate $Y$ returned by protocol $\pi$ satisfies that $Y\le \mst(G_A\cup G_B)\le \alpha\cdot Y$, when the input pairs are sampled from $\dset$.

The main result of this section is the following theorem, which immediately implies \Cref{thm: p pass alpha graph MST lower}.

\begin{theorem}
\label{thm: MST graph stream lower bound}
For any parameter $\alpha>1$, for any randomized blackboard protocol $\pi$ that computes an $\alpha$-approximation for the $\emph{\sf MST}^{\sf g}_{\sf apx}$ problem, $\cc(\pi)\ge \Omega(n)$.
\end{theorem}

The remainder of this section is dedicated to the proof of \Cref{thm: MST graph stream lower bound}. 
We will prove \Cref{thm: MST graph stream lower bound} via a reduction from the well-known 2-player communication problem called \emph{Disjointness} ($\textsf{Disj}$). In the $\textsf{Disj}$ problem, Alice is given a string $X^A\in \set{0,1}^n$ and Bob is given a string $X^B\in \set{0,1}^n$. They are allowed to communicate in the blackboard model. The goal is to check if there is an index $i\in [n]$ such that $X^A_i = X^B_i = 1$. 
We use the following result on the communication complexity of $\textsf{Disj}$ problem.

\begin{lemma}[\!\!\!\cite{Bar-YossefJKS04}\!]\label{lem:disj}
For any randomized protocol $\pi$ for the $\textnormal{\textsf{Disj}}$ problem, $\cc(\pi)=\Omega(n)$. 
\end{lemma}

We now show a reduction from $\textnormal{\textsf{Disj}}$ problem to the $\emph{\sf MST}^{\sf g}_{\sf apx}$ problem. Let $(X^A,X^B)$ be an input pair of the $\textnormal{\textsf{Disj}}$ problem. We construct a pair $G_A, G_B$ of weighted graphs as follows. We use the parameter $M = 2 \alpha n$.
The common vertex set is $V(G_A)=V(G_B)=\set{u,u'}\cup \set{v_i\mid 1\le i\le n}$. 
The edge set $E(G_A)$ contains: (i) for each $1\le i\le n$, an edge $(u,v_i)$, that has weight $1$ if $X^A_i=0$ and has weight $M$ if $X^A_i=1$; and (ii) an edge $(u,u')$ with weight $1$.
The edge set $E(G_B)$ contains, for each $1\le i\le n$, an edge $(u',v_i)$, that has weight $1$ if $X^B_i=0$ and has weight $M$ if $X^B_i=1$. 
Clearly, $E(G_A)\cap E(G_B)=\emptyset$ and $G_A\cup G_B$ is a connected graph.

Observe that, if there is an index $i\in [n]$ such that $x_i = y_i = 1$, then the two incident edges of $v_i$, $(u,v_i)$ and $(u',v_i)$, both have weight $M$, and since any spanning tree must contain one of these two edges, we get that $\mst(G_A\cup G_B)\ge M$. If for all indices $i\in [n]$, either $x_i=0$ or $y_i=0$ holds, then for each $i\in [n]$, there is at least one incident edge of $v_i$ with weight $1$. Taking the union of these edges with the edge $(u,u')$, we obtain a spanning tree of $G_A\cup G_B$ with weight $n+1$, so in this case $\mst(G_A\cup G_B)=n+1$. 
Since $M = 2 \alpha n$, any $\alpha$-approximation protocol for $\mst(G_A\cup G_B)$ can indeed determine the value of $\mst(G_A\cup G_B)$. Therefore, any randomized protocol for $\alpha$-approximating the $\emph{\sf MST}^{\sf g}_{\sf apx}$ problem can be turned into a randomized protocol for the $\textsf{Disj}$ problem, and therefore, from \Cref{lem:disj}, it has to have communication complexity $\Omega(n)$.

\iffalse
\subsection{MST: $1$-Pass, $\alpha$-Approximation, $n^{1-\alpha/n}$-Space}

\znote{To Complete}

\znote{Cite ``Dynamic Graph Stream Algorithms in o(n) Space''}

If the input is a graph stream (instead of a graph-metric stream), then the simple greedy algorithm uses space $\tilde O(n)$ and outputs a minimum spanning tree and its weight accurately. However, to obtain any approximation in one pass, it takes $\Omega(n)$-space, since $\Omega(n)$-space is needed even to detect whether or not the input graph is connected.

If the input is a graph-metric stream, then the algorithm simply recording the largest weight in the stream (and therefore only uses $\tilde O(1)$-space) is an $\Omega(n)$-approximation algorithm. In this section, we will show that, if we hope for an $o(n)$-approximation, then we have to use $n^{1-o(1)}$-space.

($k$-cycles vs $2k$-cycles)
\fi

\subsection{Space Lower Bound for One-Pass Streaming Algorithms for TSP Estimation}
\label{Proof of 1 pass TSP lower}

In this subsection we show that, for any constant $\eps>0$, any randomized one-pass streaming algorithm that takes as input a graph stream and estimates the graphic TSP-cost of a the input graph to within a factor of $2-\eps$, must use $\Omega(\eps^2n^2)$ space, thus establishing \Cref{thm: 1 pass TSP lower}.
At a high level, our proof considers the 2-player communication game, in which Alice is given a graph $G_A$ and Bob is given a graph $G_B$, and they are asked to estimate the TSP-cost of metric induced by the shortest-path distance metric of graph $G_A\cup G_B$ to within a factor of $(2-\eps)$. It is well-known that the space complexity of any one-pass $(2-\eps)$-approximation streaming algorithm is lower bounded by the one-way communication complexity of the game. We then show that the one-way communication complexity of the game is $\Omega(\eps^2 n^2)$ by constructing a hard distribution and bounding its information complexity from below by reducing the {\sf Index} problem to it.

\paragraph{Two-player graphic TSP cost approximation problem ($\emph{\sf TSP}^{\sf g}_{\sf apx}$):} This is a 2-player communication problem, in which Alice is given a weighted graph $G_A$ and Bob is given a weighted graph $G_B$, with the promise that $V(G_A)=V(G_B)$, $E(G_A)\cap E(G_B)=\emptyset$ and all weights of the weighted graph $G^*=G_A\cup G_B$ form a partial metric. 
Alice is allowed to send a message to Bob, and then Bob, upon receiving this message, needs to compute an estimate $Y$ of $\tsp(G^*)$. 

Let $0<\delta<1$ be a constant and let $\dset$ be an input distribution.
We say that a protocol $\pi$ $(2-\eps,\delta)$-approximates the problem $\emph{\sf TSP}^{\sf g}_{\sf apx}$ over the distribution $\dset$, if and only if, with probability at least $1-\delta$, the estimate $Y$ returned by protocol $\pi$ satisfies that $Y\le \tsp(G^*)\le (2-\eps)\cdot Y$.

The main result of this subsection is the following theorem, which immediately implies \Cref{thm: 1 pass TSP lower}. 
\begin{theorem}
\label{thm: main_1-pass tsp lower bound}
Let $\eps>0$ be a constant. For any real number $0<\delta< 1/10$, there exists a distribution $\dset$ on input of the problem $\emph{\sf TSP}^{\sf g}_{\sf apx}$, such that, for any one-way protocol $\pi$ that $(2-\eps,\delta)$-approximates $\emph{\sf TSP}^{\sf g}_{\sf apx}$ over $\dset$, $\emph{\cc}^{\textnormal{1-way}}_{\dset}(\pi)=\Omega(\eps^2n^2)$.
\end{theorem}

The remainder of this section is dedicated to the proof of \Cref{thm: main_1-pass tsp lower bound}. 

We now proceed to construct the hard distribution.
Let $r$ be some parameter to be fixed later, and let $p=\floor{n/2r}$.
We first define an unweighted graph $G$ as follows. Its vertex set is $V(G)=\set{u_0,u'_{0}}\cup U \cup U'$, where $U=\set{u_{i,t}\mid i\in [p], t\in [r]}$ and $U'=\set{u'_{j,t}\mid j\in [p], t\in [r]}$.
For each $1\le i\le p$, we denote $U_i=\set{u_{i,t}\mid t\in [r]}$ and similarly for each $1\le j\le p$, we denote
$U'_j=\set{u'_{j,t}\mid t\in [r]}$. Therefore, 
$V=\set{u_0,u'_{0}}\cup \bigg(\bigcup_{i\in [p]}U_i\bigg) \cup \bigg(\bigcup_{j\in [p]}U'_j\bigg)$.
The edge set of $G$ contains (i) the edge $(u_0,v_0)$; (ii) for each $u\in U$, the edge $(u_0, u)$; and (iii) for each $u'\in U'$, the edge $(u'_0,u')$.
For each $1\le i\le p$, we denote $E_{i}=\set{(u_0,u_{i,t})\mid t\in [r]}$ and similarly for each $1\le j\le p$, we denote $E'_j=\set{(u'_0,u'_{j,t})\mid t\in [r]}$. Therefore, 
$E(G)=\set{(u_0,u'_{0})}\cup \bigg(\bigcup_{i\in [p]}E_i\bigg) \cup \bigg(\bigcup_{j\in [p]}E'_j\bigg)$.
See \Cref{fig:baseG} for an illustration.

Let $X$ be an $p\times p$ matrix, where each entry $X_{i,j}$ is either $0$ or $1$. Let $i^*,j^*$ be two indices of $[p]$ (where possibly $i^*=j^*$). We now define the weighted graph $G_{X,i^*,j^*}$ based on $G$ as follows. Its vertex set is $V(G)$. Its edge set contains (i) all edges of $E(G)$; and (ii) for each pair $i,j\in [p]$ such that $X_{i,j}=1$, the edges of $E_{i,j}=\set{(u_{i,t}, u'_{j,t}),(u_{i,t}, u'_{j,t+1})\mid t\in [r]}$.
The weights on edges of $G_{X,i^*,j^*}$ are defined as follows.
Let $L$ be a parameter to be fixed later.
The edges of $E_{i^*}\cup E'_{j^*}$ have weight $(L+2)$.
The edges of $E(G)\setminus (E_{i^*}\cup E_{j^*})$ have weight $1$.
For each pair $i,j\in [p]$ such that $X_{i,j}=1$, the edges of $M_{i,j}$ have weight $(L+2)$.
See \Cref{fig:labeledG} for an illustration.

\begin{figure}[h]
	\centering
	\subfigure[Graph $G$, where $r=2$, $p=3$ and $n=14$.]{\scalebox{0.16}{\includegraphics{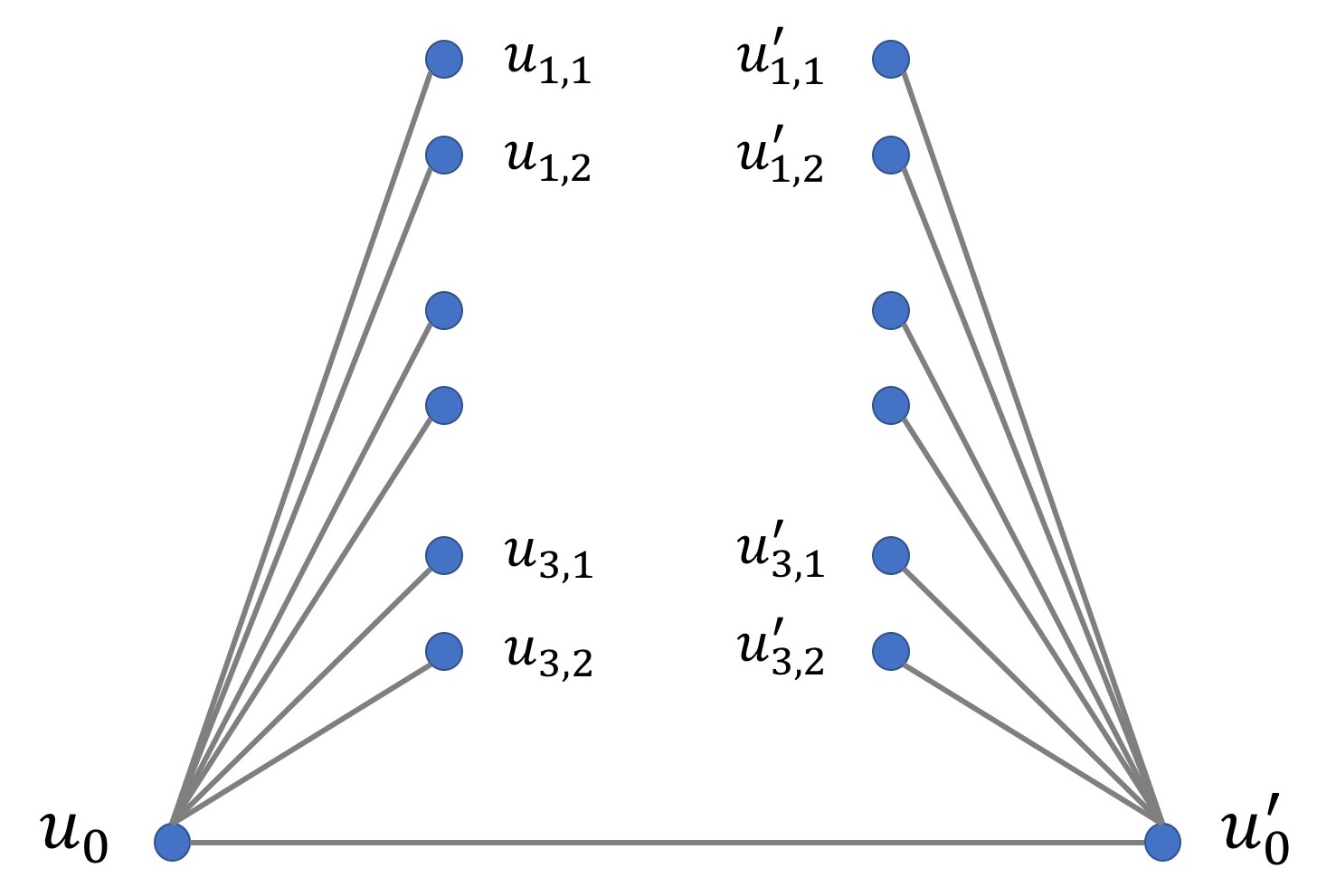}}\label{fig:baseG}}
	\hspace{0.1cm}
	\subfigure[Graph $G_{X,1,1}$, where $1$-entries of $X$ are $X_{1,1},X_{2,3},X_{3,2}$.]{\scalebox{0.16}{\includegraphics{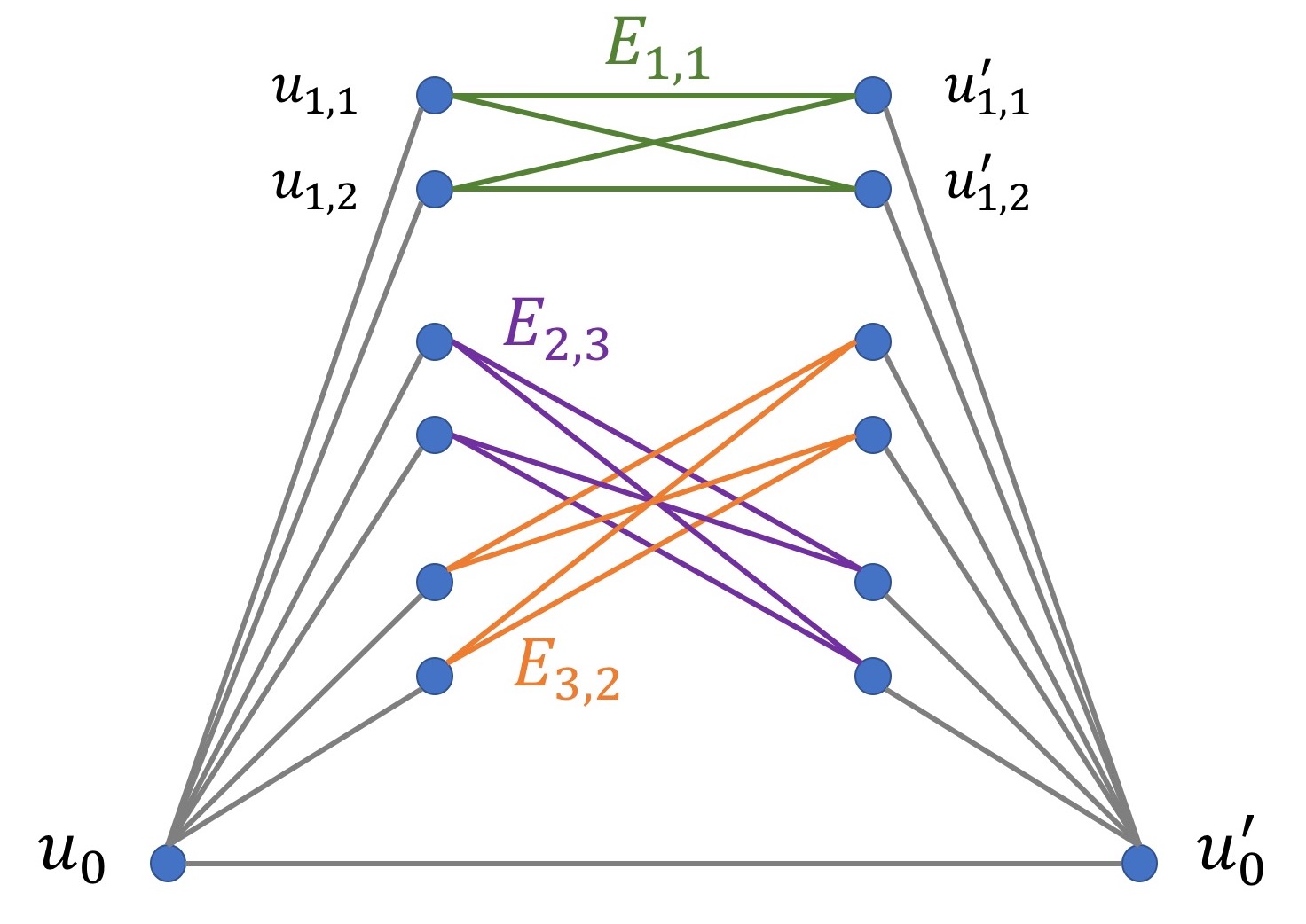}}\label{fig:labeledG}}
	\caption{An illustration of graphs $G$ and $G_{X,i^*,j^*}$.}
\end{figure}

We prove the following observation.
\begin{observation}
\label{obs: yes and no case}
$\mst(G_{X,i^*,j^*})=(n+2r-2)+(2r+1)L$.
If $X_{i^*,j^*}=0$, then $\tsp(G_{X,i^*,j^*})=2\cdot ((n+2r-2)+(2r+1)L)$. If $X_{i^*,j^*}=1$, then $\tsp(G_{X,i^*,j^*})\le 2n-6+(2r+2)L$.
\end{observation}
\begin{proof}
It is easy to see, from the definition of graph $G_{X,i^*,j^*}$, that the edges of $E(G)$ form a minimum spanning tree of $G_{X,i^*,j^*}$, regardless of $X$ and the values of $i^*,j^*$. Therefore, 
\[
\begin{split}
\mst(G_{X,i^*,j^*})= & \text{ }w(u_0,u'_0)
+\sum_{i\in [p]}\sum_{t\in [r]}w(u_0,u_{i,t})
+\sum_{j\in [p]}\sum_{t\in [r]}w(u'_0,u'_{j,t})\\
= & \text{ } L+\bigg(\frac{n}{2}-r-1\bigg)+r\cdot (L+2)+\bigg(\frac{n}{2}-r-1\bigg)+r\cdot (L+2) = (n+2r-2)+(2r+1)L.
\end{split}
\]
Assume that $X_{i^*,j^*}=0$. It is easy to see that the optimal TSP-tour that visits all vertices of $G_{X,i^*,j^*}$ is the Euler-tour of tree $G$, whose cost is $2\cdot \mst(G_{X,i^*,j^*})=(2n+4r-4)+(4r+2)L$.
Assume now that $X_{i^*,j^*}=1$. Now consider the tour $\pi$ that, starting from vertex $u_0$, first sequentially visit vertices $u_{i^*,1}, u'_{j^*,1},u_{i^*,2}, u'_{j^*,2},\ldots,u_{i^*,t}, u'_{j^*,t}, u'_0$, and then visited all other vertices and finally come back to $u_0$ by travelling along edges of $E(G)\setminus (E_{i^*}\cup E'_{j^*})$ at most twice and edge $(u_0,v_0)$ once. See \Cref{fig:good_tour} for an illustration.
It is clear that the cost of such a tour is
\[
w(u_0,u'_0)+w(u_0,u'_{i^*,1})+w(u'_0,u'_{j^*,t})
+\sum_{t\in [r]}w(u_{i^*,t},u'_{j^*,t})
+\sum_{t\in [r-1]}w(u_{i^*,t},u'_{j^*,t+1})
+\sum_{e\in E(G)\setminus (E_{i^*}\cup E'_{j^*})}2\cdot w(e),
\]
which is bounded by $3L+(2r-1)(L+2)+2(n-2r-2)=2n-6+(2r+2)L$.
\begin{figure}[h]
	\centering
	\includegraphics[scale=0.15]{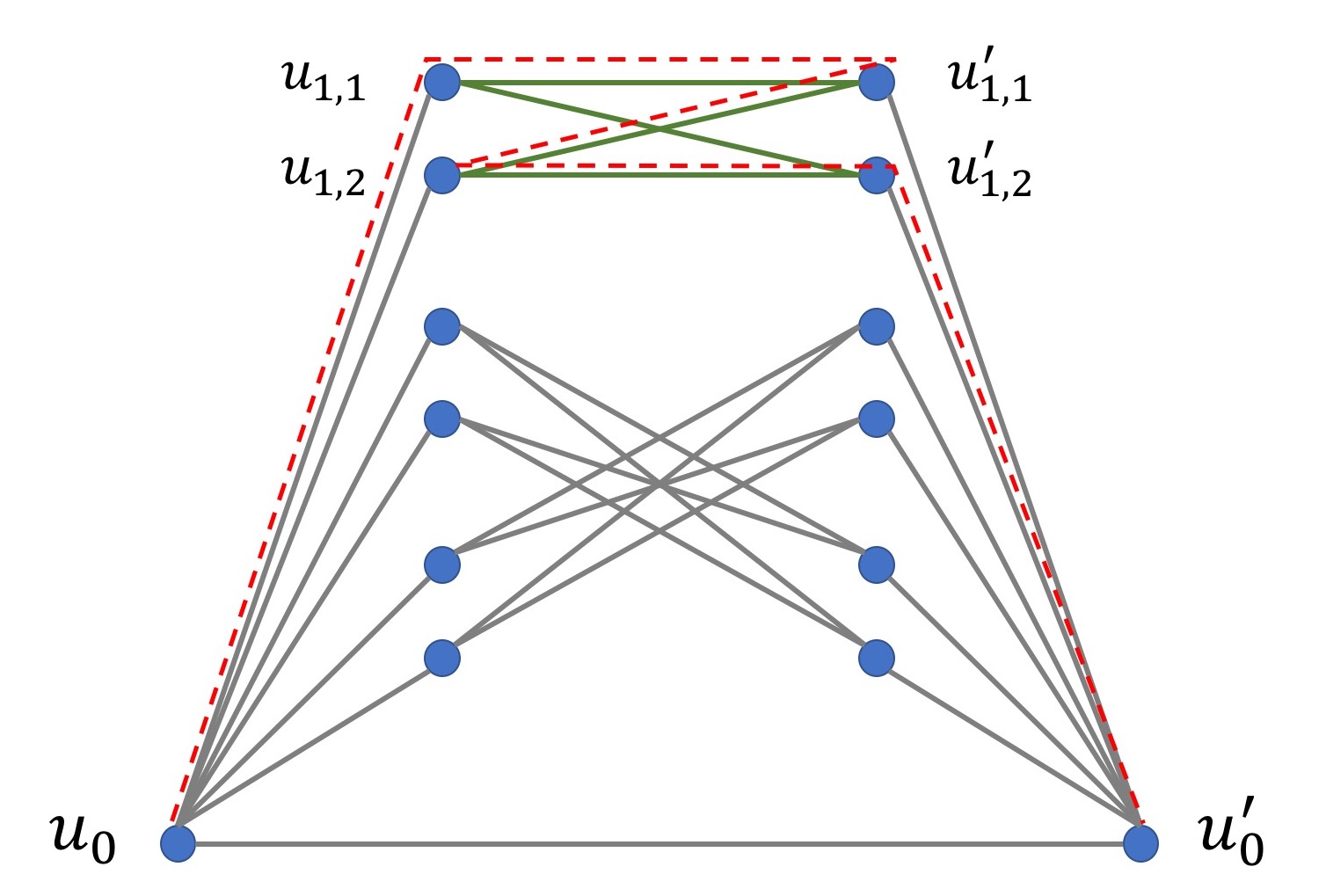}
	\caption{An illustration of part of tour $\pi$ when $X_{i^*,j^*}=1$.}\label{fig:good_tour}
\end{figure}
\end{proof}

\begin{tbox}
\textbf{Distribution} $\dset^{\sf g}_{\sf TSP}(L,r)$: An distribution on pairs $(G_A,G_B)$ of edge-disjoint graphs on $V$: 
	
	\begin{enumerate}
		\item Choose a random matrix $X$, where each entry $X_{i,j}$ is chosen uniformly at random from $\set{0,1}$.
		\item Choose two indices $i^*,j^*$ uniformly at random from $[p]$ (with replacement).
		\item Define graph $G_A$ to be the subgraph of $G_{X,i^*,j^*}$ induced by edges of $\bigcup_{i,j: X_{i,j}=1}E_{i,j}$.
		\item Define graph $G_B$ to be the subgraph of $G_{X,i^*,j^*}$ induced by edges of $E(G_{X,i^*,j^*})\setminus E(G_A)$.
	\end{enumerate}
\end{tbox}

We now show that, if we set $r=1/\eps$ and $L=n^2$, then for any protocol $\pi$ that $(2-\eps,\delta)$-approximates the problem $\emph{\sf TSP}^{\sf g}_{\sf apx}$ over $\dset^{\sf g}_{\sf TSP}(L,r)$, $\emph{\cc}^{\textnormal{1-way}}_{\dset^{\sf g}_{\sf TSP}(L,r)}(\pi)=\Omega(\epsilon^2n^2)$.

We will prove this lower bound by a reduction from the well-known problem $\ind$ to the problem $\emph{\sf TSP}^{\sf g}_{\sf apx}$. For completeness, we first provide a definition of the $\ind$ problem and state the previous result on its communication complexity.

\paragraph{Index Problem ($\ind_k$):} This is a 2-player communication problem, in which Alice is given a vector $X\in \set{0,1}^k$, and Bob is given an index $i^*\in [k]$. Alice is allowed to send a message to Bob, and then Bob, upon receiving this message, needs to outputs the value of $X_{i^*}$, the $i^*$-th coordinate of the input vector $X$ of Alice.

\begin{tbox}
	\textbf{Distribution} $\distIND^{k}$: An distribution on input pairs $(X,i^*)$ of problem $\ind_k$: 
	\begin{enumerate}
		\item Vector $X$ is chosen uniformly at random from $\set{0,1}^k$.
		\item Index $i^*$ is chosen uniformly at random from $[k]$.
	\end{enumerate}
\end{tbox}

The following lower bound on the communication complexity of $\ind_k$ over $\dset^k_{\sf Index}$ is proved in \cite{kremer1999randomized}.

\begin{theorem}[\!\!\!\cite{kremer1999randomized}\!]
Let $0<\delta<1/2$ be any constant. Then for any $\delta$-error one-way protocol $\pi$ for problem $\textnormal{\ind}$ over the distribution $\dset^k_{\textnormal{\sf Index}}$, $\cc^{\textnormal{1-way}}_{\dset^k_{\textnormal{\sf Index}}}(\pi)=\Omega(k)$.
\end{theorem}

Let $\protsp^{\sf g}$ be a protocol that $(2-\eps,\delta)$-approximates the problem $\emph{\sf TSP}^{\sf g}_{\sf apx}$ over the distribution $\dset^{\sf g}_{\sf TSP}(L,r)$.
We first construct the following protocol $\proind$ for solving the problem $\ind$ using the protocol $\protsp^{\sf g}$. 
Note that, if $X$ is a matrix chosen uniformly at random from $\set{0,1}^{p\times p}$, and $(i^*,j^*)$ is a pair chosen uniformly at random from $[p]\times [p]$, then this distribution is equivalent to $\distIND^{p^2}$.

\begin{tbox}
	\textbf{Protocol} $\proind$: A protocol for $\ind_{p^2}$ using a protocol $\protsp^{\sf g}$ for $\tspest^{\sf g}$. 
	
	\smallskip
	
	\textbf{Input:} An instance $(X,(i^*,j^*)) \sim \distIND^{p^2}$. \\
	\textbf{Output:} A value from $\set{0,1}$ as the answer to $\ind_{p^2}$.
	
	\algline
	
	\begin{enumerate}
		\item \textbf{Defining the instance.} Alice and Bob create an instance $(G_A,G_B)$ of $\tspest^{\sf g}$ as follows.
		\begin{enumerate}
		\item Construct the graph $G^*=G_{X,i^*,j^*}$.
		\item Define graph $G_A$ to be the subgraph of $G_{X,i^*,j^*}$ induced by edges of $\bigcup_{i,j: x_{i,j=1}}E_{i,j}$.
		\item Define graph $G_B$ to be the subgraph of $G_{X,i^*,j^*}$ induced by edges of $E(G^*)\setminus E(G_A)$.
		\end{enumerate}
		\item \textbf{Computing the answer.} Alice and Bob run the protocol $\protsp^{\sf g}$ on $(G_A,G_B)$ to compute an estimate $Y$. If $Y< 2\cdot ((n+2r-2)+(2r+1)L)$, then they return $1$; otherwise they return $0$.  
	\end{enumerate}
\end{tbox}

\begin{claim}
\label{clm: protocol graph stream mst to index}
If the one-way protocol $\pi^{\sf g}_{\sf TSP}$ $(2-\eps,\delta)$-approximates the problem $\emph{\sf TSP}^{\sf g}_{\sf apx}$ over the distribution $\dset^{\sf g}_{\sf TSP}(n^2,1/\eps)$, then the protocol $\pi_{\sf Index}$ is a $\delta$-error one-way protocol for $\emph{\ind}$ problem over the distribution $\dset_{\sf Index}$.
\end{claim}
\begin{proof}
	It is easy to see that the distribution of instances $(G_A,G_B)$ for $\tspest^{\sf g}$ created in the reduction by the choice of pairs $(X,i^*,j^*)\sim \distIND$, is exactly the same as the distribution $\dset_{\textnormal{\sf TSP}}^{\sf g}(L,r)$, where $L=n^2$ and $r=1/\eps$. Moreover, from \Cref{obs: yes and no case}, if $X_{i^*,j^*}=1$, then $\tsp(G_{X,i^*,j^*})\le 2n-6+(2r+2)L$; and if $X_{i^*,j^*}=0$, then $\tsp(G_{X,i^*,j^*})=2\cdot ((n+2r-2)+(2r+1)L)$.
	 Therefore, if we set $r=1/\eps$ and $L=n^2$, then the ratio between two values of $\tsp(G^*)$ when $X_{i^*,j^*}=0$ and when $X_{i^*,j^*}=1$ is
	\[
	\frac{2\cdot ((n+2r-2)+(2r+1)L)}{2n-6+(2r+2)L}
	=
	\frac{4rL+(2n+4r+2L-4)}{2rL+(2n+2L-6)}
	> 2-\eps,
	\]
	and it follows that any $(2-\eps)$-approximation of $\tsp(G^*)$ can in fact determine the value of $X_{i^*,j^*}$.
\end{proof}

Combine \Cref{clm: protocol graph stream mst to index} and the communication complexity lower bound for problem $\ind$, we get that, if $\pi^{\sf g}_{\sf TSP}$ is an one-way protocol that $(2-\eps,\delta)$-approximates the problem $\emph{\sf TSP}^{\sf g}_{\sf apx}$ over the distribution $\dset^{\sf g}_{\sf TSP}(n^2,1/\eps)$, then $$\emph{\cc}^{\textnormal{1-way}}_{\dset^{\sf g}_{\sf TSP}(L,r)}(\pi^{\sf g}_{\sf TSP})=\Omega(p^2)=\Omega((n/r)^2)=\Omega(\eps^2n^2).$$ This completes the proof of \Cref{thm: main_1-pass tsp lower bound}.

\section{A Two-Pass Algorithm for TSP Estimation in Graph Streams}

In this section, we present a deterministic $2$-pass $1.96$-approximation algorithm for TSP estimation in graph streams, which uses $\tilde O(n)$ space, thus proving \Cref{thm: 2 pass TSP upper}. Our algorithm will utilize the notion of cover advantage, introduced in \Cref{subsec: cover adv}. This result is in a sharp contrast to \Cref{thm: 1 pass TSP lower} which showed that any single-pass algorithm requires $\Omega(n^2)$ space to obtain a better than $2$-approximation.

\paragraph{Algorithm.}
Let $\alpha,\beta \in (0,1)$ be two constants whose values will be set later.
In the first pass, we simply compute a minimum spanning tree $T$ and its cost $\mst=\sum_{e\in E(T)}w(e)$.
Throughout the second pass, we maintain a subset $E_{\temp}$ of edges, that is initialized to be $\emptyset$, and will only grow over the course of the algorithm.
Upon the arrival of each edge $e$, we compare $w(e)$ with $w(\cov(e)\setminus \cov(E_{\temp}))=\sum_{f\in \cov(e)\setminus \cov(E_{\temp})}w(f)$.
We add the edge $e$ to set $E_{\temp}$ iff $w(e)\le \alpha \cdot w(\cov(e)\setminus \cov(E_{\temp}))$.
Let $E^*$ be the set $E_{\temp}$ at the end of the algorithm. We then compute $w(\cov(E^*))=\sum_{e\in \cov(E^*)}w(e)$. 
If $w(\cov(E^*))\ge \beta \cdot \mst$, then we output $(2-\frac{(1-\alpha)\cdot\beta}{2})\cdot\mst$ as an estimate of $\tsp$; otherwise we output $2\cdot\mst$.
We use the parameters $\alpha=0.715$ and $\beta=0.285$, so $(2-\frac{(1-\alpha)\cdot\beta}{2})\approx\frac{2}{2\alpha(1-\beta)}\approx 1.96$.

\paragraph{Proof of Correctness.} The correctness of the algorithm is guaranteed by the following two claims.

\begin{claim}
\label{clm: 2-pass tsp upper bound}
If $w(\cov(E^*))\ge \beta \cdot \mst$, then $\tsp\le (2-\frac{(1-\alpha)\cdot\beta}{2})\cdot\mst$.
\end{claim}
\begin{proof}
Let $E'$ be the random subset of $E^*$ that includes each edge of $E^*$ independently with probability $1/2$. We will show that the expected total weight of all edges in graph $E(H_{T,E'})$ is at most $(2-\frac{(1-\alpha)\cdot\beta}{2}) \cdot \mst$, namely $\expect[w(E(H_{T,E'}))]\le (2-\frac{(1-\alpha)\cdot\beta}{2})\cdot  \mst$. Note that this implies that there exists a subset $E^{**}$ of $E^*$, such that $w(E(H_{T,E^{**}}))\le (2-\frac{(1-\alpha)\cdot\beta}{2})\cdot \mst$. Combined with \Cref{obs: displace_graph Eulerian}, this implies that there is an Eulerian tour of the same cost (using only edges of graph $H_{T,E^{**}}$). Therefore, there is a TSP-tour of at most the same cost, completing the proof of \Cref{clm: 2-pass tsp upper bound}.

We now show that $\expect[w(E(H_{T,E'}))]\le (2-\frac{(1-\alpha)\cdot\beta}{2})\cdot  \mst$. From the definition of graph $H_{T,E'}$, $E(H_{T,E'})=E'\cup E_{[T,E']}$. On one hand, from the definition of the random subset $E'$, $\expect[w(E')]= w(E^*)/2$.
On the other hand, for each edge $f\in \cov(E^*)$, with probability $1/2$ graph $H_{T,E'}$ contains $1$ copy of it, and with probability $1/2$ graph $H_{T,E'}$ contains $2$ copies of it. Therefore, $\expect[w(E_{[T,E']})]= 2\cdot w(E(T))-w(\cov(E^*))/2=2\cdot \mst-w(\cov(E^*))/2$. 
Altogether, $\expect[w(E(H_{T,E'}))]=2\cdot \mst-(w(\cov(E^*))-w(E^*))/2$.
We use the following observation.
\begin{observation}
\label{obs: marginal contribution}
$w(E^*)\le \alpha \cdot w(\cov(E^*))$.
\end{observation}
\begin{proof}
From the algorithm of constructing the set $E^*$, an edge $e$ is added to the set $E_{\temp}$ we are maintaining, iff $w(\cov(E_{\temp}\cup \set{e}))- w(\cov(E_{\temp})) \ge w(e)/\alpha$.
Therefore, $w(\cov(E^*))\ge w(E^*)/\alpha$.
\end{proof}
%\znote{change 1.99}
From \Cref{obs: marginal contribution},
\[\expect[w(E(H_{T,E'}))]
%=2\cdot \mst-\frac{w(\cov(E^*))-w(E^*)}2
\le 2\cdot \mst-\frac{(1-\alpha)\cdot w(\cov(E^*))}2  \le \bigg(2-\frac{(1-\alpha)\cdot\beta}{2}\bigg)\cdot  \mst.\]
This concludes the proof of \Cref{clm: 2-pass tsp upper bound}.
\end{proof}

We next show that, if we do not find a sufficiently large cover, i.e., the value of $w(\cov(E^*))$ is not sufficiently large compared with $\mst$, then $\tsp$ must be bounded away from $\mst$.

\begin{claim}
\label{clm: 2-pass tsp lower bound}
If $w(\cov(E^*)) < \beta \cdot \mst$, then $\tsp\ge 2\alpha(1-\beta)\cdot\mst$.
\end{claim}
\begin{proof}
Recall that set $V_1(T)$ contains all vertices with odd degree in $T$. Let $M$ be a minimum-cost perfect matching on $V_{1}(T)$. 
From \Cref{lem: tour edge into two cover}, $\tsp\ge 2\cdot w(M)$.

We use the following observations.
%
\iffalse
\begin{observation}
\label{obs: tsp lower bounded by odd-matching}
$\tsp\ge 2\cdot w(M)$.
\end{observation}
\begin{proof}
Let $\pi$ be an optimal TSP-tour. 
Denote $V_1(T)=\set{v_1,\ldots,v_k}$ (where $k$ is even), where the vertices are indexed according to the order in which they are traversed by the tour $\pi$. Let $\pi_1$ be the tour $(v_1,v_2,\ldots,v_k,v_1)$ obtained by shortcutting $\pi$, so $\cost(\pi_1)\le \cost(\pi)$. On the other hand, denote $M_{o}=\set{(v_i,v_{i+1})\mid 1\le i\le k, i\text{ odd}}$ and $M_{e}=\set{(v_i,v_{i+1})\mid 1\le i\le k, i\text{ even}}$, then it is clear that both $M_{o}$ and $M_{e}$ are perfect matchings on vertices of $V_1(T)$, and $\cost(\pi_1)=\cost(M_{o})+\cost(M_{e})$. Therefore, $\cost(\pi_1)=\cost(M_{o})+\cost(M_{e})\ge 2\cdot w(M)$, and it follows that $\cost(\pi) \ge 2\cdot w(M)$.
\end{proof}
\fi
%
%
%From \Cref{obs: tsp lower bounded by odd-matching}, if $w(M)\ge 0.505\cdot\mst$, then $\tsp\le 1.01\cdot\mst$. Therefore, we assume from now on that $w(M)< 0.505\cdot\mst$. %We will show that, in this case, $w(\cov(E^*))\ge \beta\cdot\mst$, a contradiction to the assumption that $w(\cov(E^*))<\beta \cdot\mst$, thus completing the proof of \Cref{clm: 2-pass tsp lower bound}.
%We first prove the following observations.
%
\begin{observation}
\label{obs: odd matching covers the whole tree}
$\cov(M)=E(T)$.
\end{observation}
\begin{proof}
%We will show that $\cov(M)=E(T)$. Since $\cov(e)\subseteq \cov(E(Q_e))$ holds for each edge $e\in M$, this implies that $\cov(\bigcup_{e\in M}E(Q_e))=E(T)$.
Let $T'$ be the tree obtained from $T$ by suppressing all degree-$2$ vertices, so $V(T')$ is exactly the set of special vertices of $T$, $M$ can be viewed as a set of edges (not belonging to $T'$) connecting pairs of vertices in $V(T')$, and it suffices to show that each edge of $E(T')$ is covered by some edge of $M$.

Consider an edge $f'\in E(T')$ where $f'=(u,v)$. Let $(S_u,S_v)$ be cut in $T$ that contains a single edge $f'$.
Clearly, both $T[S_v]$ and $T[S_u]$ have an even number of odd-degree vertices. Therefore, the number of vertices of $S_u$ with an odd degree in $T[S_u]\cup f'$ is odd, and the same holds for $S_v$, and it follows that some edge of $M$ connects a vertex of $S_v$ to a vertex of $S_u$, implying that edge $f'$ is covered by $M$.
\end{proof}

%

%Denote $M=\set{e_1,\ldots,e_r}$, where the edges are indexed according to the order in which they arrive in the stream.

\begin{observation}
\label{obs: small marginal contribution}
For each $e\in M$, $w(\cov(e)\setminus \cov(E^*))< w(e)/\alpha$.
\end{observation}
\begin{proof}
%Assume for contradiction that there exists some edge $e\in M$ such that $w(\cov(e)\setminus \cov(E^*))\ge w(e)/\alpha$.
%
We denote by $Q_e$ the shortest-path in $G$ connecting the endpoints of $e$ (where $G$ is the graph underlying the stream). Since $Q_e$ is a subgraph of $G$, all edges of $Q_e$ will appear in the graph stream of $G$.
Note that $w(Q_e)=w(e)$ and $\cov(e)\subseteq \cov(E(Q_e))$.

We will show that, for every edge $e'\in E(Q_e)$, $w(\cov(e')\setminus \cov(E^*))< w(e')/\alpha$. Note that the observation follows from this assertion, as 
$$w(\cov(e)\setminus \cov(E^*))\le w(\cov(E(Q_e))\setminus \cov(E^*))\le \sum_{e'\in E(Q_e)}w(\cov(e')\setminus \cov(E^*))<\frac{w(Q_e)}{\alpha}=\frac{w(e)}{\alpha}.$$

Consider now any edge $e'\in E(Q_e)$, and assume for contradiction that $w(\cov(e')\setminus \cov(E^*))\ge  w(e')/\alpha$.
Note that set $E_{\temp}$ only grows over the course of the algorithm that computes set $E^*$, and so does the set $\cov(E_{\temp})$. Therefore, when $e'$ arrives in the stream, 
$w(\cov(e')\setminus \cov(E_{\temp}))\ge w(e')/\alpha$ must hold. Then according to the algorithm, the edge $e'$ should be added to $E_{\temp}$ right away, which means that edge $e'$ will eventually belong to $E^*$, leading to $\cov(e')\subseteq \cov(E^*)$ and $w(\cov(e')\setminus \cov(E^*))=0$, a contradiction to the assumption that $w(\cov(e')\setminus \cov(E^*))\ge w(e')/\alpha$.
\end{proof}

From \Cref{obs: odd matching covers the whole tree} and \Cref{obs: small marginal contribution}, we get that
\[
(1-\beta)\cdot\mst\le \mst-w(\cov(E^*))=w(\cov(M)\setminus \cov(E^*)) \le \sum_{e\in M}w(\cov(e)\setminus \cov(E^*))< w(M)/\alpha.
\]
Therefore, $w(M)\ge \alpha(1-\beta)\cdot\mst$. Since $\tsp\ge 2\cdot w(M)$, we conclude that
$\tsp\ge 2\alpha(1-\beta)\cdot\mst$.
\end{proof}

\section{Space Lower Bound for One-Pass Exact TSP Estimation in Metric Streams}
\label{sec: one-pass exact TSP}

In this section we show that any one-pass streaming algorithm that given a metric stream for some metric $w$ on a set of $n$ vertices and computes the value of $\tsp(w)$ exactly has to use $\Omega(n^2)$ space. At a high level, our proof considers the corresponding $2$-player communication game, and show that the one-way communication complexity of the game is $\Omega(n^2)$ by a reduction from the problem $\ind$.

\paragraph{Two-Player TSP Computing Problem ($\tsp_2$):} This is a 2-player one-way communication problem, in which Alice is given partial metric $\bar w_A$ and Bob is given a partial metric $\bar w_B$, both on a set $V$ of vertices that is known to both players, with the promise that $\bar w_A$ and $\bar w_B$ are complimentary. 
The goal of the players is to compute the value of $\tsp(w)$ where $w=\bar w_A\cup \bar w_B$.

The main result of this section is the following theorem.

\begin{theorem}
\label{thm: tsp computing}
Any one-way protocol $\pi$ for the $\tsp_2$ problem has $\cc(\pi)=\Omega(n^2)$.
\end{theorem}

The remainder of this section is dedicated to the proof of \Cref{thm: tsp computing}.

We now show a reduction from the $\ind$ problem to the $\tsp_2$ problem. Assume that $n=2k+1$ is an odd integer. 
We denote the vertex set by $V=U\cup U'$, where $|U|=k+1$ and $|U'|=k$, and in particular, we denote 
$U=\set{u_1,\ldots,u_{k+1}}$ and $U'=\set{u'_1,\ldots,u'_{k}}$.

Let $X$ be a $(k+1)\times (k+1)$ $0/1$-matrix, such that $X_{i,i}=0$ for all $i\in [k+1]$, and $X_{i,j}=X_{j,i}$ for all pairs $i,j\in [k+1]$.
For such a matrix $X$, we define a partial metric $\bar w_A$ as follows. For each pair $u_i,u_j$ of vertices in $U$, $\bar w(u_i,u_j)=X_{i,j}+1$; for each pair $u'_i,u'_j$ of vertices in $U'$, $\bar w(u_i,u_j)=*$ and for each pair $u\in U, u'\in U'$, $\bar w(u,u')=*$. Let $i^*, j^*$ be a pair of distinct indices in $[k+1]$. We define a partial metric $\bar w_B$ as follows. Let $P_{i^*, j^*}$ be any path that has endpoints $u_{i^*}, u_{j^*}$ and alternates between vertices in sets $U$ and vertices in set $U'$, i.e., $P_{i^*, j^*}=(u_{q_1}=u_{i^*},u'_{p_1},u_{q_2},u'_{p_2},\ldots,u_{q_k},u'_{p_k},u_{q_{k+1}}=u_{j^*})$, where $(p_1,p_2,\ldots, p_k)$ is a permutation of $[k]$ and $(q_1,p_2,\ldots, q_{k+1})$ is a permutation of $[k+1]$. Now for each pair $u'_i,u'_j\in U$, we set $\bar w_B(u'_i,u'_j)=2$; for each pair $u_i,u_j\in U$, we set $\bar w_B(u_i,u_j)=*$, and for each pair $u_i\in U, u'_j\in U'$, we set $\bar w_B(u_i,u_j)=1$ iff edge $(u_i, u'_j)$ belongs to path $P_{i^*,j^*}$, otherwise we set we set $\bar w_B(u_i,u_j)=2$.
It is easy to verify that for any matrix $X$ and for any index pair $(i^*,j^*)$, $\bar w_A, \bar w_B$ are valid partial metrics and are complementary.

We now describe the reduction. 
Let $\pi$ be a one-way protocol for the $\tsp_2$ problem. We construct a one-way protocol $\pi_{\ind}$ for the $\ind$ problem as follows.
Assume Alice and Bob are given an instance $(X^A, X^B)$ for the $\ind$ problem, in which $X^A$ is a $0/1$ vector of length $\binom{k+1}{2}$ and $X^B$ is an index in $[\binom{k+1}{2}]$. Note that it is equivalent Alice is given a symmetric matrix $X$ with all diagonal entries $0$ and Bob is given a pair $(i^*,j^*)$ of distinct indices of $[k+1]$. We then let Alice construct the partial metric $\bar w_A$ using matrix $X$ and Bob construct the partial metric $\bar w_B$ using the index pair $(i^*, j^*)$ as described above, and then let them run the protocol $\pi$ on the pair $\bar w_A, \bar w_B$ of partial metrics. If the output of $\pi$ is $2k+1$, then we return $0$ as the output of $\pi_{\ind}$; if the output of $\pi$ is $2k+2$, then we return $1$ as the output of $\pi_{\ind}$.

We now show that, if $\pi$ is a one-way protocol for the $\tsp_2$ problem, then $\pi_{\ind}$ is a one-way protocol for the $\ind$ problem as follows. Note that, from the construction of partial metrics $\bar w_A, \bar w_B$, in the metric $w=\bar w_A\cup \bar w_B$, there is a path $P_{i^*,j^*}$ of weight-$1$ edges, so if $w(u_{i^*},u_{j^*})=1$, then $w$ contains a Hamiltonian cycle of weight-$1$ edges. and therefore $\tsp(w)=2k+1$. On the other hand, since $w(u_{i^*},u_{j^*})\le 2$, $\tsp(w)\le 2k+2$, and it is easy to verify that $\tsp(w)= 2k+1$ implies $w(u_{i^*},u_{j^*})=1$.
Therefore, $\tsp(w)= 2k+1$ iff $X_{i^*,j^*}=0$, and $\tsp(w)= 2k+1$ iff $X_{i^*,j^*}=1$, so $\pi_{\ind}$ is a one-way protocol for the $\ind$ problem. Now \Cref{thm: tsp computing} follows from the classic communication complexity lower bound on one-way protocols for the $\ind$ problem.

\newpage
\part{Query Algorithm I}
%G1 connected part

In this part we present our algorithm in the query model, that, given a distance oracle to a metric $w$ with the promise that the complete weighted graph with edge weights given by $w$ (denoted by $G_w$) contains a minimum spanning tree consisting of only weight-$1$ edges, estimates the value of $\tsp(w)$ to within a factor of $(2-\eps_0)$, for some universal constant $\eps_0>0$, by performing $\tilde O(n^{1.5})$ queries to the oracle, thus establishing \Cref{thm: main G_1 connected}.
We start with some additional preliminaries used in this part, then introduce some crucial subroutines, and finally present the query algorithm and its analysis.

Throughout this part, we denote by $G_1$ the subgraph of $G_w$ induced by all weight-$1$ edges.

\section{Preliminaries}
\label{sec:prelim}

%\znote{To be organized later.}

%Let $G$ be a graph. We say that a subgraph $S$ of $G$ is a \emph{block}, iff $|V(S)|\ge 3$, $S$ is $2$-vertex-connected, and there is no other $2$-vertex-connected subgraph $S'\subseteq G$, such that $S\subsetneq S'$.

Let $G$ be an unweighted graph.
For a pair $v,v'$ of vertices in $G$, we denote by $\dist_G(v,v')$ the length (the number of edges) of the shortest path in $G$ connecting $v$ to $v'$.
%We denote by $N_G(v,r)$ the set of vertices $v'$ in $G$, such that $\dist_G(v,v')\le r$. 

Let $T$ be a tree.
We say that a subgraph $P$ of $T$ is an \emph{induced subpath} of $T$, iff $P$ is a path, and all internal vertices of $P$ have degree $2$ in $T$. We say that $P$ is a \emph{maximal induced subpath} if it is not a proper subpath of any other induced subpath of $T$. Therefore, for any pair $P,P'$ of maximal induced subpaths  of $T$, $P$ and $P'$ are internally vertex-disjoint. The following observation is immediate.

\begin{observation}
	\label{obs:leaves_induced_paths}
	The number of maximal induced subpaths of a tree is at most twice the number of leaves in the tree.
\end{observation}

%We say that $v,v'$ are at tree-distance $d$ in $T$, iff $|E(P^T_{v,v'})|=d$.

%The vertices of $T$ can be partitioned into two sets: set $V_0(T)$ contains all vertices with even degree in $T$, and set $V_1(T)$ contains all vertices with odd degree in $T$.

\paragraph{E-blocks and e-block tree.}
We say that an edge $e$ of a connected graph $G$ is a \emph{bridge}, iff $G\setminus e$ is not connected.
A graph is \emph{$2$-edge-connected} if it does not contain any bridges. Let $G$ be a graph and let $H$ be an induced subgraph of $G$, we say that $H$ is a \emph{$2$-edge-connected block} (or an \emph{e-block} for brevity), iff $H$ is a maximal $2$-edge-connected induced subgraph of $G$ (namely, there is no other $2$-edge-connected induced subgraph $H'\subseteq G$, such that $H\subsetneq H'$).
We denote by $\bset(G)$ the set of e-blocks of $G$. Clearly, e-blocks in $\bset(G)$ are vertex-disjoint. 

The \emph{e-block tree} $\tset_G$ of a connected graph $G$ is defined to be the graph obtained from $G$ by contracting each e-block $B$ of  $\bset(G)$ into a vertex $v_B$.
It is easy to see that, if $G$ is connected, then $\tset_G$ is a tree.
We call such a vertex $v_B$ that corresponds to an e-block in $G$ an \emph{e-block vertex}.
%For each vertex $v$ in $\tset_G$, we naturally define its \emph{weight} $w(v)$ as follows. If $v$ is also a vertex in $G$, then $w(v)=1$;  if $v=v_B$ is an e-block vertex that corresponds to the e-block $B\in \bset(G)$, then $w(v)=|V(B)|$. For a subset $S\subseteq V(\tset_G)$ of vertices or a subgraph $S\subseteq \tset_G$, its weight $w(S)$ is defined to be the sum of weights of all vertices in $S$.

\iffalse{maybe not needed}
\znote{boundary vertex of a cluster}
\paragraph{Simple Tour on a Path.}
\znote{To Complete.}

\paragraph{Euler Tour on a Tree.}
\znote{To Complete.}

\paragraph{Concatenation of Sub-tours.}
\znote{To Complete.}
\fi

We use the following results in previous work and easy propositions.

\begin{lemma}[Theorem 8 in~\cite{chen2020sublinear}]
\label{thm: n^{1.5}-query-2-approx-of-matching-size}
There is a randomized algorithm, that, given any simple graph $G$ on $n$ vertices and any parameter $\epsilon>0$, with probability $1-e^{-\Omega(n)}$, estimates the size of some maximal matching of $G$ to within an additive error of $\epsilon n$, by performing $\tilde O(n^{1.5}/\epsilon^2)$ pair queries.  
\end{lemma}

\begin{lemma} [Corollary of Lemma~2.11 in \cite{chen2020sublinear}] 
\label{lem:matching-TSP}
Let $H$ be a connected graph on $n$ vertices. If $H$ contains a matching $M$ of size $\alpha n$ such that at most $\beta n$ edges of $M$ are bridges of $H$, then the graphic-TSP cost of $H$ is at most  $(2-\frac{2}{3}(\alpha-\beta))n$.
\end{lemma}

\begin{proposition}\label{lem: naive bounds}
If $G_1$ is connected, then $n\le \emph{\opttsp}(w)\le 2n-2$. 
\end{proposition}
\begin{proof}
On one hand, let $(v_1,\ldots,v_{n},v_1)$ be the optimal TSP-tour that traverses all vertices of $V$. Since $w(v_i,v_{i+1})\ge 1$ for all $1\le i\le n$, the cost of the tour is at least $n$. On the other hand, let $T$ be a spanning tree of $G_1$ and consider the Euler tour that traverses each edge of $T$ exactly twice. It is easy to see that such a tour can be easily shortcut into a TSP-tour that traverses all vertices of $V$, with at most the same cost $2n-2$.
\end{proof}

\begin{proposition}
\label{lem: G_1_matching_TSP_lower_bound}
Let $m$ be the size of a maximum matching in $G_1$, then $ \emph{\opttsp}(w)\ge 2n-2m$.
\end{proposition}
\begin{proof}
Let $(v_1,\ldots,v_{n},v_1)$ be an optimal TSP-tour that traverses all vertices of $V$. Observe that the edges in set $M_o=\set{(v_i,v_{i+1})\mid 1\le i\le n, i\text{ odd}}$ form a matching in $G_1$, and so do the edges in set $M_e=\set{(v_i,v_{i+1})\mid 1\le i\le n, i\text{ even}}$. Since the size of a maximum matching in $G_1$ is at most $m$, the set $M_o$ contains at most $m$ edges of weight $1$, and so does the set $M_e$. Therefore, 
at most $2m$ values of $\set{w(v_i,v_{i+1})}_{1\le i\le n}$ are $1$, and it follows that $\sum_{1\le i\le n}w(v_i,v_{i+1})\ge 2(n-2m)+2m=2n-2m$.
\end{proof}

\begin{proposition}
\label{clm: degree_1 vertices implies lower bound}
If $G_1$ contains $L$ degree-$1$ vertices, then $\emph{\opttsp}(w)\ge n+L/2$.
\end{proposition}
\begin{proof}
Let $(v_1,\ldots,v_{n},v_1)$ be an optimal TSP-tour. If $v_i$ is a degree-$1$ vertex in $G_1$, then at most one of $v_{i-1},v_{i+1}$ is its neighbor in $G_1$, so at least one of the pairs $(v_{i-1},v_i), (v_{i},v_{i+1})$ has distance at least $2$. We say that this pair is \emph{in charge of} vertex $v_i$. Clearly, each pair is in charge of at most two degree-$1$ vertices in $G$. Therefore, at least $L/2$ pairs in $\set{(v_i,v_{i+1})\mid 1\le i\le n}$ has distance at least $2$, and it follows that $\sum_{1\le i\le n}w(v_i,v_{i+1})\ge n+L/2$.
\end{proof}

\begin{proposition}
\label{lem: matching_in_tree}
Let $T$ be a tree on $n$ vertices with $L$ leaves. Then $T$ contains a matching of size at least $(n - L)/2$.
\end{proposition}
\begin{proof}
We root $T$ at an arbitrary vertex, and then iteratively construct a matching of $ T$ as follows.
Throughout, we maintain a subset $ M$ of edges of $T$, that is initialized to be $\emptyset$, and each vertex of $ T$ is either marked matched or unmatched.
The algorithm proceeds in iterations. While there is an unmatched vertex, such that all its ancestors are matched and some of its descendants remain unmatched, we match this vertex with any of its unmatched descendants (by adding the edge connecting them into $ M$), and mark both of them matched.
It is easy to see that, at the end of this algorithm, all vertices that are left unmatched are leaves of $T$. Therefore, the matching that we constructed has size at least $ (n - L)/2$.
\end{proof}

\begin{lemma}\label{clm:number_of_degree_1}
There is an efficient algorithm, that, given any constant $\epsilon>0$, with probability at least $1-O(n^{-10})$, either (correctly) claims that the number of degree-1 vertices of $G_1$ is greater than $\epsilon n/2$, or (correctly) claims that the number of degree-1 vertices of $G_1$ is at most $2\epsilon n$. Moreover, the algorithm performs $O(n/\epsilon)$ distance queries.
\end{lemma}
\begin{proof}
We sample $L=200\log n/\epsilon$ vertices from $V$ uniformly at random, and for each sampled vertex $v$, we check whether or not $v$ is a degree-$1$ vertex in $G_1$ by performing pair queries $(v,v')$ for all $v'\in V$. Let $L'$ be the number of degree-$1$ vertices vertices in our sample set. If $L'\ge \epsilon L$, we claim that the number of degree-$1$ vertices is at least $\epsilon n/2$ (a positive report), otherwise we claim that the number of degree-$1$ vertices is at most $2\epsilon n$ (a negative report).
It now remains to show the correctness of the algorithm. First assume that the number of degree-$1$ vertices in $G_1$ is at least $\epsilon n/2$. In this case, by Chernoff bound,
\[
\Pr[L'\ge 2\epsilon L]\le \exp\left(-10\log n\right)=O(n^{-10});\mbox{ and }
\Pr[L'\le \epsilon L/2]\le \exp\left(-10\log n\right)=O(n^{-10}).
\]
We then assume that the number of degree-$1$ vertices in $G_1$ is at most $\epsilon n/2$. Similarly, by Chernoff bound, the probability that the algorithm makes an incorrect claim is $\Pr[L'\ge 2\epsilon L]\le \exp\left(-10\log n\right)=O(n^{-10})$.
Therefore, with high probability, the algorithm's claim is correct.
\end{proof}

\section{Subroutines}

In this section we provide four crucial subroutines in our query algorithm of this part.
The first and the second focus on exploring the neighborhood of a vertex up to a certain size, and then intuitively compute the difference between TSP cost and MST cost within this neighborhood.
The third and the fourth focus on efficiently reconstruct long induced paths of an MST and properly connect them into a TSP tour.

\subsection{Subroutine $\local(\cdot,\cdot)$}

\begin{definition}[Light Vertices]
We say that a vertex $v$ is \emph{$\ell$-light} for some integer $\ell>0$, iff there is an edge $e\in \delta_{G_1}(v)$, such that
(i) $e$ is a bridge in $G_1$; and (ii) the connected component of graph $G_1\setminus e$ that contains $v$ has at most $\ell$ vertices. We call such an edge a \emph{witness} of $v$. 
\end{definition}

It is easy to verify that, if $\ell < n/2$, then each $\ell$-light vertex has exactly one witness.
We will assume from now on that $\ell<n/2$.
We denote by $e_v$ the unique witness edge of an $\ell$-light vertex $v$, and denote by $S_v$ the connected component of $G_1 \setminus e_v$ that contains $v$. 

\begin{definition}[Maximal Light Vertices]
We say that $v$ is \emph{maximal $\ell$-light}, iff $v$ is $\ell$-light, and there is no other $\ell$-light vertex $u$ such that the subgraph $S_u$ contains $v$. In this case we also say that $S_v$ is a \emph{maximal $\ell$-light subgraph} of $G_1$.
\end{definition} 
We denote by $L_{\ell}$ the \emph{set of all maximal $\ell$-light vertices}. 

We now describe a subroutine that, given a vertex $v\in V$ and an integer $s$, explores the neighborhood of $v$ in graph $G_1$ up to size $s$.
This subroutine can help us find the maximal $s$-light subgraph of $G_1$ that $v$ lies in, if it indeed lies in any such subgraph, and will be a crucial building block for our TSP-cost estimation algorithm.
%As we will see, the subroutine is not tailored for the case where $G_1$ is a spanning tree. It can also be applied in the case where $G_1$ is any connected graph.

The subroutine is called $\local(\cdot,\cdot)$, and is defined as follows.
The input to the subroutine consists of: (i) a vertex $v\in V$; and (ii) an integer $1\le s \le n$. The output of the subroutine $\local(v,s)$ is a signal \textsf{Success}/\textsf{Fail} indicating the termination status of the subroutine, and if the status is \textsf{Success}, the subroutine also returns a connected subgraph of $G_1$ that contains the vertex $v$.

The subroutine $\local(v,s)$ proceeds in two phases.
%The first phase is called \emph{exploring phase}, and the second phase is called \emph{checking phase}.
We now describe the first phase, called the \emph{exploring phase}.
Intuitively, in this step we perform a breadth-first search in $G_1$ starting at vertex $v$ and maintain a tree $T$ rooted at $v$ throughout. Initially, $T$ contains a single vertex $v$ as the root. 
The exploring phase consists of several \emph{exploring steps}. An exploring step at a vertex $\hat v$ includes: (i) performing distance queries $w(\hat v,v')$ for all vertices $v'\notin V(T)$; (ii) for all vertices $v'$ with $w(\hat v,v')=1$, adding $v'$ and the edge $(v',\hat v)$ to $T$ (so $v'$ is a child vertex of $\hat v$ in $T$). We say that a vertex is \emph{explored} iff we have performed an exploring step at it.
In the exploring phase, we start by performing an exploring step at $v$. Let $v_1,\ldots,v_k$ be the new vertices that are added to $T$ in this phase, we then perform sequentially, for each $1\le i\le k$, an exploring step at $v_i$. 
%\znote{To Modify until the end of the subroutine in a more BFS-way.}
%For a new active vertex after these steps, we say that it is \emph{$\hat v$-originated} for some vertex $\hat v$, if either it is $\hat v$ itself or it is marked active in the exploring step at $\hat v$. If later on in the exploring phase, some vertex is marked active in the exploring step at some $\hat v$-originated vertex, then we also say that this vertex is $\hat v$-originated. Clearly, all active or explored vertices are $v$-originated.
We then iteratively, for each $1\le i\le k$, take a descendant of $v_i$ that is not yet explored, and perform an exploring step at it. 
%We continue performing exploring steps in a BFS manner.
It is clear that, after each exploring step, the number of explored vertices in $T$ increases by $1$. The exploring phase is terminated once the number of explored vertices reaches $2s$.

In the second phase, which is called the \emph{checking phase}, we find the least common ancestor in $T$ of all vertices in $T$ that are not yet explored.
Clearly, there is at least one such vertex, since otherwise $G_1$ is not connected.
Let $\hat v$ be the vertex we get. If $\hat v = v$, then we return \textsf{Fail}.
Assume now that $\hat v \ne v$, and assume that $\hat v\in V(T_{v_i})$ for some first-level vertex $v_i$ in $T$ (recall that $T_{v_i}$ is the subtree of $T$ rooted at $v_i$).
If $|V(T\setminus T_{v_i})|> s$, then we return \textsf{Fail}.
If $|V(T\setminus T_{v_i})|\le s$, then return \textsf{Success} and consider the tree path in $T$ connecting $v_i$ to $\hat v$.
%Let $\hat V$ be the set of explored vertices $\hat v$ such that all active vertices are $\hat v$-originated. Clearly, $v\in \hat V$, and vertices in $\hat V$ can be arranged in a chain $v=\hat v_0\prec\hat v_1\prec\cdots\prec\hat v_k$, such that for each $0\le j<j'\le k$, vertex $\hat v_{j'}$ is $\hat v_{j}$-originated.
%For each $\hat v\in \hat V$, let $S_{\hat v}$ contain all explored vertices that are not $\hat v$-originated.
Let $v^*$ be the vertex on this path with lowest level, such that $|V(T\setminus T_{v^*})|\le s$. We denote $S=V(T\setminus T_{v^*})$.
We further perform queries between every pair of vertices in $S$, and output $G_1[S]$, the subgraph of $G_1$ induced by vertices in $S$. 
%Assume now that not all active vertices are $\hat v$-originated for any $\hat v\ne v$. In this case the subroutine returns \textsf{Fail}.
This completes the description of the subroutine.

The following observations are immediate.
\begin{observation}
\label{obs: number_of_queries_local}
For every vertex $v\in V$ and every integer $1\le s \le n$, subroutine $\emph{\local}(v,s)$ performs at most $O(ns)$ distance queries.
\end{observation}

\begin{observation}
	If the vertex $v$ is $\ell$-light, then the subroutine $\emph{\local}(v,\ell)$ returns \textsf{Success} and outputs the maximal $\ell$-light subgraph of $G_1$ that contains $v$;
	otherwise the subroutine $\emph{\local}(v,\ell)$ returns \textsf{Fail}.
\end{observation}

\subsection{Subgraph Reconfiguration}

Fix a parameter $0<\ell<n/2$ and a vertex $v\in L_{\ell}$ (the set of all maximal $\ell$-light vertices in $G_1$), and consider the maximal $\ell$-light subgraph $S_v$ of $G_1$. For each vertex $u\in V(S_v)$, we define its \emph{out-reach distance}, denoted as $\ord(u)$, to be the minimum distance between $u$ and a vertex of $V\setminus V(S_v)$, namely $\ord(u)=\min\set{w(u,u')\mid u'\notin V(S_v)}$. 

\begin{definition}[Subgraph Reconfiguration]
A \emph{reconfiguration} of a subgraph $S_v$ is defined to be a multi-set of unordered pairs of vertices in $V(S_v)$ that induce a connected graph on $V(S_v)$. 
Equivalently, a reconfiguration of $S_v$ can be viewed as a connected multi-graph on vertices of $S_v$, whose edges may or may not belong to $S_v$.
\end{definition}

Let $R$ be a reconfiguration of $S_v$.
We will sometimes refer to $R$ as the (multi)-graph induced by edges of $R$. The vertices of this multigraph can be partitioned into two sets: the set $V_0(R)$ contains all vertices with even degree, and the set $V_1(R)$ contains all vertices with odd degree.
The \emph{cost} of $R$ is defined as
$$\cost(R) = \sum_{(u,u') \in R} w(u,u') + \frac{1}{2}\sum_{u \in V_1(R)} \ord(u) - (\card{V(S_v)} - 1).$$
%Intuitively, if $R$ is part of the TSP tour, the cost of $R$ is the extra cost we spend in $T_v$ other than the size of $T_v$. 
The \emph{optimal reconfiguration cost} of $S_v$, that we denote by $\rc(S_v)$, is defined to be the minimum cost of a reconfiguration of subgraph $S_v$.

\begin{lemma}
\label{clm:opt_subgraph_reconf}
There is an algorithm, that, given any maximal $\ell$-light subgraph $S$, computes a reconfiguration $R$ that achieves optimal (minimum) reconfiguration cost, by performing $O(n\ell)$ queries. 
\end{lemma}
\begin{proof}
We first query the distances between all pairs $(u,u')$ of vertices with $u\in V(S)$ and $u'\in V$. We then enumerate all possible reconfigurations and compute the cost of each of them, from the distance information we acquired. Finally, we return the configuration with minimum cost.  
\end{proof}

\subsection{Subroutine $\pbfs(\cdot,\cdot,\cdot,\cdot)$} 
We now describe a subroutine that, given a vertex $v$ and several parameters, explores the neighborhood of $v$ in graph $G_1$ up to a certain distance, by performing a certain number of distance queries.
%This subroutine is a critical building block for the algorithm in this section.
As we will see, the subroutine \pbfs\text{ }is different from the subroutine \local, in the sense that the subroutine \pbfs\text{ }focuses on checking whether or not the neighborhood of a vertex $v$ in $G$ has a path-like structure, instead of checking whether or not the vertex belongs to a maximal light subgraph.
We will show that this checking can be done in a distinctly more efficient manner than the subroutine \local.

The subroutine is called $\pbfs(\cdot,\cdot,\cdot,\cdot)$, and is defined as follows.
The input to the subroutine consists of: (i) a vertex $v\in V$; (ii) an integer $1\le h \le n$ representing the depth of the BFS-tree that the subroutine is aiming for; (iii) an integer $1\le q\le n^2$ representing the number of queries that the subroutine is allowed to perform; and (iv) a real number $\alpha>1$ for determining the termination status of ths subroutine. The output of the subroutine $\pbfs(v,h,q,\alpha)$ consists of: (i) a signal \textsf{Success}/\textsf{Fail} indicating the termination status of the subroutine; and (ii) if the status is \textsf{Success}, a path $P_v$ of length $2h$ and a subgraph $H_v$ of $G_1$, both containing the vertex $v$.

The subroutine $\pbfs(v,h,q,\alpha)$ proceeds in $h$ stages. For each $1\le i\le h$, in the $i$th stage, we explore all vertices that are at distance exactly $i$ from $v$ in $G_1$ by performing distance queries.
Throughout the subroutine, each vertex is either marked as a \emph{level-$i$} vertex for some $0\le i\le h$, or is marked as \emph{unexplored}. Initially, the input vertex $v$ is marked as a level-$0$ vertex and all other vertices are marked as unexplored.
We will also maintain a tree $\hat T^h_v$, that initially contains a single vertex $v$ and no edges.
Along with exploring vertices of $V$, we also keep track of the queries that we have already performed so far, and utilize them for efficiently determining which vertices to explore in the next stage. 
Specifically, throughout the subroutine, we maintain, for each vertex $\hat v\in V$, a set
$$\qset(\hat v)=\set{v'\mid \text{the distance between vertices }\hat v\text{ and }v'\text{ has been queried}},$$ that records all distance information we have obtained about vertex $\hat v$.
%We also say that $v'\in \qset(\hat v)$ if set $\qset(\hat v)$ contains the pair $(v',w(\hat v,v'))$.
Initially, all sets $\set{\qset(\hat v)}_{\hat v\in V}$ are $\emptyset$, and the sets evolve as the subroutine proceeds.

In the first stage, we query distances $w(v,v')$ for all vertices $v'\in V$. We then mark all vertices $v'$ with $w(v,v')=1$ as \emph{level-1} vertices, and add the edge $(v,v')$ to $\hat T^h_v$ for each level-$1$ vertex $v'$.
%and for each vertex $\hat v\in V$, we add the pair $(v,w(\hat v,v))$ to the set $\qset(\hat v)$.

We now fix some $1<i\le h$ and describe the $i$th stage of the subroutine.
Let $v_1,\ldots,v_p$ be all level-$(i\!-\!1)$ vertices, (indexed arbitrarily). We sequentially process vertices $v_1,\ldots,v_p$ as follows. Consider a vertex $v_j$ with $1\le j\le p$.
We would like to discover all unexplored vertices that are at distance $1$ from $v_j$ by performing as few queries as possible.
Specifically, for each unexplored vertex $\hat v$, we query all distances in $\set{w(v',v_j)\mid v'\in \qset(\hat v)}$ and then compute
$X(v_j,\hat v)=\max\set{w(v',\hat v)-w(v',v_j)\mid v'\in \qset(\hat v)}$. 
We then query distance $w(v_j,\hat v)$ for all unexplored vertices $\hat v$ with $X(v_j,\hat v)\le 1$.
We then mark, among the above vertices, those vertices $\hat v$ with $w(v_j,\hat v)=1$ as level-$i$ vertices, and then add the edge $(v_j,\hat v)$ to tree $\hat T^h_v$. 
%Finally, we add, for each such vertex $\hat v$, the pair $(v_j,w(\hat v,v_j))$ to the set $\qset(\hat v)$.
This finishes the description of processing a level-$(i\!-\!1)$ vertex $v_j$, and also finishes the description of the $i$th stage.

The subroutine also maintains a counter that records the number of distance queries performed so far. If the counter reaches $q$ before all $h$ stages are finished, then we abort the subroutine and return \textsf{Fail}. Otherwise, let $T^h_v$ be the resulting tree that we obtain. The output of the subroutine distinguishes between the following cases.

\textbf{Case 1.} If $|V(T^h_v)|\le \alpha h$, and the tree $T^h_v$ contains exactly two level-$h$ vertices $v^*,v^{**}$, and the $v$-$v^*$ path and the $v$-$v^{**}$ path in $T^h_v$ are internally vertex-disjoint, then we return \textsf{Success}. 
We denote by $P_v$ the $v^*$-$v^{**}$ path in tree $T^h_v$, and call it the \emph{support path} of $v$.
We proceed to query all distances between vertices of $V(T^h_v)$. Note that this takes at most $\alpha^2h^2$ additional queries. 
%Let $H_v$ be the graph whose vertex set is $V(T^h_v)$ and edge set contains, for each pair $v'_1,v'_2$ of vertices with $w(v'_1,v'_2)=1$, an edge $(v'_1,v'_2)$. 
From these queries, we are able to recover the subgraph of $G_1$ induced by vertices of $V(T^h_v)$, which we denote by $H_v$.
%Let $\bset(H_v)$ be the set of e-blocks in $H_v$. We then compute a maximum matching in $\bigcup_{B\in\bset(H_v)}B$. 
We then return the path $P_v$ and the subgraph $H_v$.
We say that vertices $v^*,v^{**}$ are \emph{interface vertices} of $P_v$ and $H_v$.

\textbf{Case 2.} If Case 1 does not happen, then we say that Case 2 happens. Equivalently, Case 2 happens iff either $|V(T^h_v)|> \alpha h$ holds, or tree $T^h_v$ does not contain exactly two level-$h$ vertices, or tree $T^h_v$ does contain exactly two level-$h$ vertices $v^*,v^{**}$ but the $v$-$v^*$ path and the the $v$-$v^{**}$ path in $T^h_v$ share a vertex other than $v$. In this case we return \textsf{Fail}.

This completes the description of the subroutine $\pbfs(v,h,q,\alpha)$.
We use the following observations.

\begin{observation}
\label{obs: PBFS queries}
The subroutine $\pbfs(v,h,q,\alpha)$ performs at most $q+\alpha^2h^2$ distance queries.
\end{observation}
\begin{proof}
If the number of performed queries reaches $q$ before all $h$ stages are finished, then the subroutine is aborted and has performed in total at most $q$ queries. If the subroutine ends in Case 1, then at most $\alpha^2h^2$ additional queries are performed. If the subroutine ends in Case 2, then no more queries are performed. \Cref{obs: PBFS queries} now follows.
\end{proof}

We will use following observation in analyzing our TSP-cost estimation algorithm in \Cref{sec: G_1 connected}.

\begin{observation}
\label{obs: Euler tour distance}
For every vertex $u\notin V(T^h_v)$ and for every pair $x,x'$ of vertices such that the distances $w(u,x)$ and $w(u,x')$ are queried by the subroutine, $w(x,x')\ge \min_{u'\in V(T^h_v)}\set{w(u,u')}-1$.
\end{observation}
\begin{proof}
We define $d_u=\min_{u'\in V(T^h_v)}\set{w(u,u')}$ for every vertex $u\notin V(T^h_v)$.
Assume the contrast that there exists some $u\notin V(T^h_v)$ and some pair $x,x'\in \qset(u)$ such that $w(x,x')\le d_u-2$.
Note that, since $u\notin V(T^h_v)$, the set $\qset(u)$ only contains vertices of $T^h_v$ (so $x,x'\in V(T^h_v)$), and the distance $w(x,u)$ ($w(x',u)$, resp.) may only be queried in the iteration of processing $x$ ($x'$, resp.) to find its children in constructing the tree $T^h_v$.
Assume without loss of generality that vertex $x$ was processed before $x'$, so the distance $w(x,u)$ was queried before the distance $w(x',u)$. Consider now the iteration of processing $x'$.
Recall that we only query distances $w(x',\hat v)$ for all unexplored vertices $\hat v$ with $X(x',\hat v)\le 1$, where $X(x',\hat v)=\max\set{w(z,\hat v)-w(z,x')\mid z\in \qset(\hat v)}$. Therefore, if the distance $w(x',u)$ is queried in this iteration, then $\max\set{w(z,u)-w(z,x)\mid z\in \qset(u)}\le 1$. However, since at this moment $x\in \qset(u)$, so $\max\set{w(z,u)-w(z,x')\mid z\in \qset(u)}\ge w(x,u)-w(x,x')\ge d_u-(d_u-2)=2$, a contradiction.
\end{proof}

\subsection{Computing Optimal Proper Tours}

Let $\qset$ be a set of vertex-disjoint paths in $G_1$, and we denote $V'=\bigcup_{Q\in \qset}V(Q)$.
We say that a tour is a \emph{$V'$-tour} iff it traverses all vertices of $V'$ exactly once.
We say that a $V'$-tour $(v_1,\ldots,v_{r},v_1)$ is a \emph{$\qset$-proper tour}, iff for each $1\le i\le r$, either $(v_i,v_{i+1})$ is an edge on some path of $\qset$, or $v_i$ is an endpoint of some path $Q\in \qset$, and $v_{i+1}$ is an endpoint of some other path $Q'\in \qset$. 

\begin{claim}
\label{clm: proper tour cost}
The minimum cost of a $\qset$-proper tour can be computed by performing $O(n+|\qset|^2)$ queries.
\end{claim}
\begin{proof}
We first query, for each path $Q\in \qset$, the distance between all pairs of consecutive vertices on $Q$. 
This takes in total at most $n$ queries.
Let set $\hat V$ contain, for each path $Q\in \qset$, both endpoints of path $Q$. We then query the distances between all pairs of vertices in $\hat V$.
This takes in total $O(|\qset|^2)$ queries.
%Since $k=O(\sqrt{n})$, the total number of queries that we performed is $O(n)$.
We then enumerate all $\qset$-proper tours that traverses all vertices of $V'$ and compute the cost of each of them, from the distance information we acquired. Finally, we return the minimum cost of a $\qset$-proper tour.
\end{proof}

\begin{claim}
\label{clm: long induced path concatenation}
If  $\opttsp(w)\le (1+{\epsilon})n$, $|E(\qset)|\ge (1-\eps')n$, and for all pairs $v,v'$ of vertices belonging to distinct paths of $\qset$, $w(v,v')=1$ implies that $v$ and $v'$ are endpoints of the paths in $\qset$ that they belong to, respectively, then the minimum cost of a $\qset$-proper tour is at most $(1+3\eps+4{\epsilon}')n+2|\qset|$.
\end{claim}
\begin{proof}
Let $\pi^*$ be the optimal TSP-tour that visits all vertices of $V$. We will construct a $\qset$-proper tour that visits all vertices of $V'$ with cost at most $(1+3\eps+4{\epsilon}')n+2|\qset|$ as follows.
	
Let $\pi_1$ be the $V'$-tour obtained from $\pi^*$ by deleting all vertices of $V\setminus V'$ from $\pi^*$, and keeping all vertices of $V'$ in the same order as in $\pi^*$. 
Clearly, $\pi_1$ is a $V'$-tour, and $\cost(\pi_1)\le \cost(\pi^*)\le (1+\eps)n$ due to the triangle inequality. 
We denote by $G'$ the graph on $V'$ that contains every consecutive pair of vertices in $\pi_1$ as an edge. Put in other words, if $\pi_1=(v_{1},\ldots,v_{r},v_{1})$, then $E(G')=\set{(v_{i},v_{i+1})\mid 1\le i\le r}$.
Since $\cost(\pi_1)\le (1+\eps)n$, and $|V'|\ge |E(\qset)|\ge (1-\eps')n$, at most $(\eps+\eps')$ edges of $G'$ have weight more than $1$.
Since for all pairs $v,v'$ of vertices lying on distinct paths of $\qset$, $w(v,v')=1$ implies that $v$ and $v'$ are endpoints of the paths in $\qset$ that they belong to, at most $|\qset|$ edges in $G'$ have weight $1$ but do not belong to $E(\qset)$. Altogether,
at most $|\qset|+(\epsilon+\eps')n$ edges of graph $G'$ do not belong to $E(\qset)$.	
%Moreover, since $|V'|\ge 1-{\epsilon}n$, at most $2{\epsilon}n$ edges of $T$ do not lie in $G'$.
	
We now iteratively modify the graph $G'$ to obtain a graph that contains all edges of $E(\qset)$, as follows.
Throughout, we will maintain a graph $\hat G$ on $V'$, that is initialized to be $G'$. We will ensure that, at any time, $\hat G$ is the union of a set of vertex-disjoint cycles. Since $\pi_1$ is a $V'$-tour, from the definition of $G'$, this condition holds at the beginning of the algorithm.
	The algorithm continues to be executed as long as $E(\qset)\not\subseteq E(\hat G)$.
	We now describe an iteration.
	Let $e=(u,u')$ be an edge in $E(\qset)\setminus E(\hat G)$. Assume that edge $e$ belongs to a path $Q$ of $\qset$. 
	Let $u_1,u_2$ be the neighbors of $u$ in $\hat G$.
	Since $u_1,u_2\ne u$, at least one of the edges $(u_1,u),(u_2,u)$ does not belong to $Q$. Assume without loss of generality that $(u_1,u)$ does not belong to $Q$. We define $u'_1,u'_2$ similarly for $u'$ and assume without loss of generality that $(u'_1,u')$ does not belong to $Q$.
	We then remove the edges $(u_1,u),(u'_1,u')$ from $\hat G$, and add the new edges $(u_1,u'_1),(u,u')$ to $\hat G$. This completes the description of an iteration. Clearly, after this iteration, the degree for every vertex in $\hat G$ is still $2$, so the invariant that $\hat G$ is the union of disjoint cycles still holds.
	Note that after every iteration, the size of $E(\qset)\setminus E(\hat G)$ decreases by at least one. Therefore, the algorithm will eventually terminate in at most $|\qset|+(\epsilon+\eps')n$ iterations.
	%\znote{consider the hanging light subtree structure}
	
	Let $G''$ be the graph that we obtain at the end of the algorithm.
	Let $G''$ be the union of disjoint cycles $C_1,\ldots,C_t$. 
	Clearly, $G''$ contains all edges of $E(\qset)$, each path $Q\in \qset$ is entirely contained in some cycle $C_j$.
	Note that, after each iteration, the cost of $\hat G$, which is defined as $\cost(\hat G)=\sum_{(u,u')\in E(\hat G)}w(u,u')$, increases by at most $2$. This is because, from the triangle inequality, 
	$w(u_1,u'_1)+w(u,u')-w(u_1,u)-w(u'_1,u')\le  2D(u,u')=2$.
	Therefore, $\cost(G'')\le \cost(G')+2(|\qset|+(\epsilon+\eps')n)\le (1+3\eps+2{\epsilon}')n+2|\qset|$.

	We now use cycles $C_1,\ldots,C_t$ and edges of $G_1$ to construct a $\qset$-proper tour, as follows.
	Let $\hat G_1$ be the graph obtained from $G_1$ by contracting, for each $1\le j\le t$, all vertices of $C_j$ into a single node, that we denote by $v_{C_j}$.
	Note that, since $G_1$ is connected and $|E(\qset)|\ge (1-\eps')n$, $\hat G_1$ is also connected and contains at most ${\epsilon}'n$ vertices. Let $\hat E$ be the set of edges of a Steiner tree in $\hat G_1$ that spans all nodes $v_{C_1},\ldots,v_{C_t}$ with minimum cost, so $|\hat E|\le {\epsilon}'n$. Note that each edge of $\hat E$ is also an edge in $G_1$. Moreover, since for all pairs $v,v'$ of vertices lying on distinct paths of $\qset$, $w(v,v)=1$ implies that $v$ and $v'$ are endpoints of the paths in $\qset$ that they belong to, it is easy to see that each edge of $\hat E$ has to connect an endpoint of some path in $\qset$ to another endpoint of some other path in $\qset$. Consider the graph $G''\cup \hat E$. It is easy to see that, we can construct a tour $\pi$ that visits all vertices of $V'$ and endpoints of all edges in $\hat E$, by using each edge of $G''$ at most once and each edge of $\hat E$ at most twice, such that for each path $Q\in\qset'$, all vertices of $Q$ appear consecutively in $\pi_2$ in the same order as they appear on $Q$.
	Moreover, such a tour can be shortcut into a proper $V'$-tour $\pi$ with at most the same cost. See Figure~\ref{fig:proper_tour} for an illustration.
	Therefore, $\cost(\pi)\le \cost(\pi_2)\le \cost(G'')+2|\hat E|\le (1+3\eps+4{\epsilon}')n+2|\qset|$.
	\begin{figure}[!h]
		\centering
		\subfigure[Before: Paths of $\qset$ are shown in green, cycles are shown in red, and edges of $\hat{E}$ are shown in blue.]{\scalebox{0.47}{\includegraphics{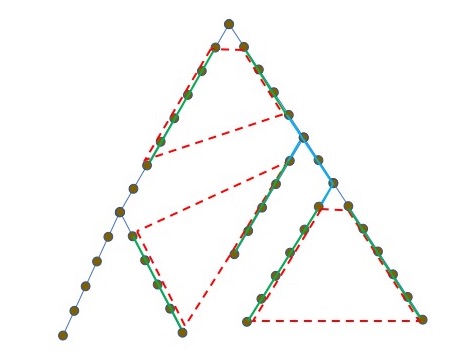}}}
		\hspace{10pt}
		\subfigure[After: The tour $\pi$ that visits all vertices of $V'$ and endpoints of all edges in $\hat E$ is shown in pink. $\pi$ can be further shortcut into a $\qset$-proper tour in a standard way.]{\scalebox{0.47}{\includegraphics{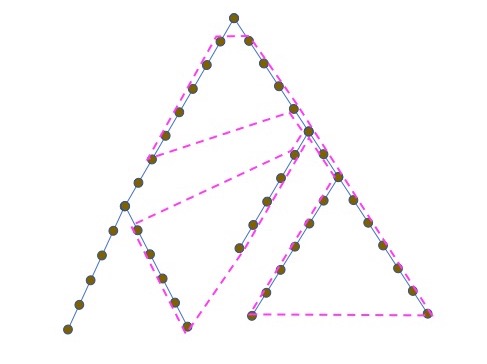}}}
		\caption{An illustration of constructing a $\qset$-proper tour. \label{fig:proper_tour}}
	\end{figure}
\end{proof}

\section{The Algorithm and its Analysis}
\label{sec: G_1 connected}

In this section we present our algorithm and its analysis, completing the proof of \Cref{thm: main G_1 connected}.
At a high level, our algorithm considers the ``bottom part'' (light subtrees) and the ``top part'' (the subtree obtained by peeling off all light subtrees) of an MST separately. If the bottom part is heavy, then we use the first two subroutines to explore random neighborhoods of the bottom parts, and then derive locally either $\tsp$ is bounded away from $2\cdot \mst$ or bounded away from $\mst$. If the ``bottom part'' is light, then almost all cost of the MST belongs to its top part, and then we use the third and the fourth subroutines to reconstruct the MST and then estimate $\tsp$ accordingly.

Parameters: $\epsilon=2^{-100}, \hat\eps=2^{-40}$, $q=n/\hat{\eps}^2$, $\ell=100\sqrt{n}$, $h=\hat \eps \ell/200$ and $\alpha=10/\hat{\eps}$.

\begin{enumerate}
\item \label{G_1_connected_step_1}
We first use the algorithm from~\Cref{clm:number_of_degree_1} to compare the number of degree-$1$ vertices of $G_1$ with $\epsilon n/40$.
If the algorithm reports that the number of degree-$1$ vertices of $G_1$ is greater than $\epsilon n/80$, then we return $2n$ as the estimate of $\opttsp(w)$. 
If the algorithm reports that the number of degree-$1$ vertices of $G_1$ is at most $\epsilon n/20$, go to~\ref{G_1_connected_step_2}.
\item \label{G_1_connected_step_2}
We sample $1/ \epsilon^2$ vertices (with replacement) from $V$ independently and uniformly at random, and perform the subroutine $\local(v,\ell)$ for each sampled vertex $v$.
If at least $\epsilon$-fraction of the subroutines $\local(\cdot,\ell)$ return \textsf{Success}. 
Let $v_1,\ldots,v_{1/ \epsilon^2}$ be the sampled vertices.
For each $1\le i\le 1/ \epsilon^2$, if the subroutine $\local(v_i,\ell)$ returns \textsf{Fail}, then we let $\gamma_i=0$. Otherwise, the subroutine $\local(v_i,\ell)$ returns \textsf{Success} and a maximal $\ell$-light subgraph $S_i$ that contains $v_i$. We apply the algorithm in \Cref{clm:opt_subgraph_reconf} to $S_i$ to compute the optimal reconfiguration cost $\rc(S_i)$ of $S_i$, and then let $\gamma_i=\rc(S_i)/|V(S_i)|$.
We now compare $\sum_{1\le i\le 1/\epsilon^2}\gamma_i$ with $1/80\epsilon$.
If $\sum_{1\le i\le 1/\epsilon^2}\gamma_i\le 1/80\epsilon$, then we return $(2-{\epsilon}/20)n$ as the estimate of $\opttsp(w)$; otherwise we return $2n$ as the estimate of $\opttsp(w)$. 
If less than $\epsilon$-fraction of the subroutines $\local(\cdot,\ell)$ return \textsf{Success}, go to~\ref{G_1_connected_step_3}.
\item \label{G_1_connected_step_3}
We sample $(100\cdot n\log n)/h$ vertices (with replacement) from $V$ independently and uniformly at random, and perform the subroutine $\pbfs(v,h,q,\alpha)$ for each sampled vertex $v$.
Let $\tilde V$ be the set of sampled vertices that the subroutine $\pbfs(\cdot,h,q,\alpha)$ returns \textsf{Success} upon. 
Recall that the subroutine \pbfs\text{ }also returns a path $P_v$ and a graph $H_v$. 
We define graphs $Z=(V,\bigcup_{v\in \tilde V}E(P_v))$ and $H=(V,\bigcup_{v\in \tilde V}E(H_v))$.
%We now distinguish between the following two cases, depending on whether or not the number of isolated vertices of $H$ is at least $10{\epsilon}n$.
We then define vertex sets $V_{\textsf{chunk}}=\bigcup_{v\in \tilde V}V(H_v\setminus P_v)$ and $V_{\textsf{isolated}}=V\setminus (V_{\textsf{chunk}}\cup V(Z))$.
If $|V_{\textsf{isolated}}|\ge 10\hat{\epsilon}n$, go to~\ref{G_1_connected_step_4}.
If $|V_{\textsf{chunk}}|\ge 10\hat{\epsilon}n$, go to~\ref{G_1_connected_step_4'}.
Otherwise,
%($|V_{\textsf{isolated}}|,|V_{\textsf{chunk}}|< 10\hat{\epsilon}n$)
go to~\ref{G_1_connected_step_5}. 
\item \label{G_1_connected_step_4}
%Denote by $\hat V$ the set of isolated vertices of $Z$.
We use the algorithm from \Cref{thm: n^{1.5}-query-2-approx-of-matching-size} with parameter $\hat\epsilon/100$ to obtain an estimate of the size of some maximal matching of $G_1[V_{\textsf{isolated}}]$.
If the estimate is at most $\hat{\epsilon}n$, then we return $2n$ as the estimate of $\opttsp(w)$.
Otherwise, we return $(2-\hat{\epsilon}/200)n$ as the estimate of $\opttsp(w)$.
\item \label{G_1_connected_step_4'}
%Denote by $\hat V$ the set of isolated vertices of $Z$.
We compute the maximum size of a matching consisting of only edges that belong to some e-block of $H$.
If the size is at most $2\hat{\epsilon}n$, then we return $2n$ as the estimate of $\opttsp(w)$.
Otherwise, we return $(2-\hat{\epsilon})n$ as the estimate of $\opttsp(w)$.
\item \label{G_1_connected_step_5}
We use the following claim.
\begin{claim}
\label{clm: vertex disjoint induced paths}
We can efficiently find a set $\qset$ of vertex-disjoint induced subpaths of $Z$, such that each path of $\qset$ has length at least $\Omega(\hat\epsilon^2 h/\log^2 n)$, and $\sum_{Q\in \qset}|E(Q)|\ge (1-40\hat{\epsilon})n$.
\end{claim}
\begin{proof}
Recall that graph $Z$ is obtained from taking the union of paths in $\pset=\set{P_v\mid v\in \tilde V}$, where $P_v$ is the support path returned by the subroutine $\pbfs(v,h,q,\alpha)$. Also recall that $|V(Z)|= |V|-|V_{\textsf{isolated}}|-|V_{\textsf{chunk}}|\ge (1-20\hat{\eps})n$.
We say that a path $P\in \pset$ is \emph{bad} iff the number of vertices $u$ in $P$ with $\deg_R(u)\ne 2$ is more than $20\log n/\hat{\eps}$; otherwise we say that $P$ is \emph{good}.
Clearly, set $\pset$ contains at most $\hat{\eps}/(20\log n)\cdot|\pset|$ bad paths.
Therefore, the total length of bad paths in $\pset$ is at most $\hat{\eps}/(20\log n)\cdot|\pset|\cdot 2h\le \hat{\eps}/(20\log n)\cdot(100n\log n/h)\cdot 2h\le 10\hat \eps n$. 
Consider now a good path $P\in \pset$, and let $P_1,\ldots,P_t$ be the subpaths of $P$ obtained by deleting all vertices $u$ of $P$ with $\deg_R(u)\ne 2$.
We denote $\sset(P)=\set{P_1,\ldots,P_t}$, and say that a path in $\sset(P)$ is \emph{short} iff its length is at most $\hat{\eps}^2 h/(20\log n)^2$; otherwise we say it is \emph{long}.
Clearly, paths $P_1,\ldots,P_t$ are induced subpaths of $Z$. And it is easy to show that the total length of all long paths in $\sset(P)$ is at least $(1-\hat{\eps}/(20\log n))|E(P)|$, namely $$\sum_{1\le i\le t}\mathbf{1}[|E(P_i)|\ge \hat{\eps}^2 h/(20\log n)^2]\cdot |E(P_i)|\ge (1-\hat{\eps}/(20\log n))\cdot |E(P)|.$$
Therefore, the total length of short subpaths in set $\bigcup_{P\in\pset: P\text{ is good}}\sset(P)$ is at most $10\hat{\eps} n$.
Let $\qset$ be the set that contains, for each good path $P\in \pset$, all long subpaths in $\sset(P)$. Then clearly the paths in $\qset$ are vertex-disjoint, and each path of $\qset$ has length at least $\hat \eps^2 h/(20\log n)^2$. Moreover, the total length of all paths in $\qset$ is at least $|V(Z)|-10\hat{\eps}n-10\hat{\eps}n\ge (1-40\hat{\eps})n$.
\end{proof}

We use the algorithm from \Cref{clm: vertex disjoint induced paths} to find a set $\qset$ of vertex-disjoint induced subpaths in $Z$, each of length at least $\Omega(\hat\epsilon^2 h/\log^2 n)$ (so $|\qset|\le  O(n\log^2n/\ell\hat{\eps}^2)$).
We denote $V'=\bigcup_{Q\in \qset}V(Q)$, so $|V'|\ge (1-40\hat{\epsilon})n$.
We then use the algorithm from~\Cref{clm: proper tour cost} to compute the minimum cost of a $\qset$-proper tour. If the cost is at most $(2-100\hat{\epsilon})n$, then we return $(2-\hat{\epsilon})n$ as the estimate of $\opttsp(w)$; otherwise we return $2n$ as the estimate of $\opttsp(w)$.
\end{enumerate}

\paragraph{Proof of Correctness.}
We first count the total number of queries performed by the algorithm. 
The algorithm in \Cref{clm:number_of_degree_1} performs $O(n/\epsilon)=O(n)$ queries.
From \Cref{obs: number_of_queries_local}, the $1/\epsilon^2$ executions of the subroutine $\local(\cdot,\ell)$ perform in total $O(n\ell/\epsilon^2)=O(n^{1.5})$ queries. 
From \Cref{obs: PBFS queries}, the $(100\cdot n\log n)/h$ executions of the subroutine $\pbfs(\cdot,h,q,\alpha)$ perform in total $O((q+\alpha^2h^2)\cdot\sqrt{n}\log n)=\tilde O(n^{1.5})$ queries. 
The algorithm from \Cref{thm: n^{1.5}-query-2-approx-of-matching-size} performs at most $O(n^{1.5}/\hat\eps^2)=O(n^{1.5})$ queries. 
%The algorithm in \Cref{clm: vertex disjoint induced paths} performs $\tilde O(n^{1.5})$ queries. 
The algorithm in \Cref{clm: proper tour cost} performs $\tilde O(n)$ queries.
Therefore, the whole algorithm performs $\tilde O(n^{1.5})$ queries in total.

We now show that the output of the algorithm is an $(2-\Omega({\epsilon}))$-approximation of $\opttsp(w)$. We start by defining several bad events and analyzing their probabilities.

\paragraph{Bad Event $\xi_1$.} 
We say that the bad event $\xi_1$ happens, iff 
\begin{itemize}
	\item either the number of $\ell$-light vertices in $G_1$ is at least $2\epsilon n$, namely $\sum_{v\in L_{\ell}}|V(S_v)|\ge  2\epsilon n$, while less than $\epsilon$-fraction of executions of $\local(\cdot,\ell)$ return \textsf{Success};
	\item or the number of $\ell$-light vertices in $G_1$ is at most $\epsilon n/2$, namely $\sum_{v\in L_{\ell}}|V(S_v)|\le \epsilon n/2$, while at least $\epsilon$-fraction executions of $\local(\cdot,\ell)$ return \textsf{Success}.
\end{itemize}
From Chernoff's bound, 
$\Pr[\xi_1]\le e^{-\Omega(n)}$.

\paragraph{Bad Event $\xi_2$.} 
We say that the bad event $\xi_2$ happens iff either the algorithm in \Cref{clm:number_of_degree_1} reports incorrectly, or the algorithm from \Cref{thm: n^{1.5}-query-2-approx-of-matching-size} reports incorrectly. 
From \Cref{clm:number_of_degree_1} and  \Cref{thm: n^{1.5}-query-2-approx-of-matching-size}, $\Pr[\xi_2]\le O(n^{-10})$.

\paragraph{Bad Event $\xi_3$.} 
We say that the bad event $\xi_3$ happens, iff 
\begin{itemize}
	\item either $\sum_{v\in L_{\ell}}\rc(S_v)\ge {\epsilon}n/40$, while $\sum_{1\le i\le 1/{\epsilon}^2}(\rc(S_i)/|V(S_i)|)\le 1/80\epsilon$;
	\item or $\sum_{v\in L_{\ell}}\rc(S_v)\le {\epsilon}n/160$, while $\sum_{1\le i\le 1/{\epsilon}^2}(\rc(S_i)/|V(S_i)|)\ge 1/80\epsilon$.
\end{itemize}
We define the function $\gamma: V\to \mathbb{R}$ as follows. If $v$ is $\ell$-light, then $\gamma(v)=\rc(S_v)/|V(S_v)|$, otherwise $\gamma(v)=0$. It is clear that $0\le \gamma(v)\le 1$ for all vertices $v\in V$. Therefore, if vertex $v$ is chosen uniformly at random from $V$, $\gamma(v)$ is a random variable in $[0,1]$, and moreover,
\[\mathbb{E}_{v\sim V}[\gamma(v)]=\sum_{v'\in L_{\ell}}\frac{|S_{v'}|}{n}\cdot \frac{\rc(S_{v'})}{|S_{v'}|}=\frac{\sum_{v'\in L_{\ell}}\rc(S_{v'})}{n}.\]
Note that the values $\gamma_1,\ldots,\gamma_{1/{\epsilon}^2}$ that we obtained in the algorithm are $1/{\epsilon}^2$ independent samples of the random variable $\gamma(v)$.
Therefore, from Chernoff's bound, $\Pr[\xi_3]\le e^{-\Omega(n)}$.

%We now start proving the following claims, and then use them to complete the proof of \Cref{thm:G1_connected}.

%We now provide the proof of \Cref{thm:G1_connected}.
We denote by $X$ the output of the algorithm.
We will show that, if none of the bad events $\xi_1, \xi_2,\xi_3$ happens, then $\opttsp(w)\le X\le (2-\Omega({\epsilon}))\cdot\opttsp(w)$.
Recall that, from \Cref{lem: naive bounds}, $n\le \opttsp(w)\le 2n-2$. 

First of all, if the algorithm from \Cref{clm:number_of_degree_1} reports that $G_1$ contains at least ${\epsilon}n/80$ degree-1 vertices.
Since event $\xi_2$ does not happen, and from \Cref{clm: degree_1 vertices implies lower bound}, $\opttsp(w)\ge (1+{\epsilon}/160)n$. Therefore, our estimate $X=2n$ in this case is a $(2-\Omega({\epsilon}))$-approximation of $\opttsp(w)$. Assume from now on that algorithm in \Cref{clm:number_of_degree_1} reports that $G_1$ contains at most ${\epsilon}n/20$ degree-1 vertices.
%Since event $\xi_2$ does not happen, $G_1$ does contain at most $2{\epsilon}n$ degree-1 vertices.
We now distinguish between the following two cases, depending on the number of executions of the subroutine 
$\local(\cdot,\ell)$ that return \textsf{Success} in Step~\ref{G_1_connected_step_2}.

\subsection{Case 1. At least $\eps$-fraction of the executions of $\local(\cdot,\ell)$ return \textsf{Success}}

Since event $\xi_1$ does not happen, $\sum_{v\in L_{\ell}}|V(S_v)|\ge {\epsilon}n/2$.  We distinguish between the following subcases.

\paragraph{Case 1.1. $\sum_{1\le i\le 1/{\epsilon}^2}\gamma_i\le 1/80{\epsilon}$.} 
Since event $\xi_3$ does not happen, $\sum_{v\in L_{\ell}}\rc(S_v)\le {\epsilon}n/40$. We show in the following claim that $\opttsp(w)\le (2-{\epsilon}/20)n$, so our estimate $X=(2-{\epsilon}/20)n$ is a $(2-\Omega({\epsilon}))$-approximation of $\opttsp(w)$.

\begin{claim}
\label{clm:tsp_upper}
If $\sum_{v \in L_{\ell}} \card{V(S_v)} \ge {\epsilon} n/2$, $\sum_{v \in L_{\ell}} \rc(S_v) \le {\epsilon}n/40$ and $G_1$ contains at most $\epsilon n/20$ degree-1 vertices, then $\opttsp(w) \le (2- \epsilon/20)n$.
\end{claim}

\begin{proof}
For each vertex $v\in L_{\ell}$, let $R_v$ be a reconfiguration of the subgraph $S_v$ with minimum cost. 
Recall that set $V_1(R_v)$ contains all vertices with odd degree in multi-graph $R_v$.
For each vertex $u\in V_1(R_v)$, let $u^*$ be a vertex of $V\setminus V(S_v)$ with minimum distance from $u$, namely $u^*=\arg\min_{u'\notin V(S_v)}\set{w(u,u^*)}$ and $w(u,u^*)=\ord(u)$. 
Denote $\Gamma_v=\set{(u,u^*)\mid u\in V_1(R_v)}$. Denote $R = \bigcup_{v \in L_{\ell}} R_v$ and $\Gamma = \bigcup_{v \in L_{\ell}} \Gamma_v$.
\iffalse
Therefore, 
$$|\Gamma|\le \cost(\Gamma)=\sum_{(u,u^*)\in \Gamma}w(u,u^*)\le \sum_{v \in L_{\ell}} \rc(S_v) \le {\epsilon}n/20,$$ and 
$$|R|=\sum_{v\in L_{\ell}}|R_v|\le \sum_{v\in L_{\ell}}\rc(S_v)\le \sum_{v \in L_{\ell}} (\card{V(S_v)} - 1) + {\epsilon}n/20.$$
\fi	
Denote $V'=V \setminus (\bigcup_{v\in L_{\ell}}V(S_v))$.
We define a graph $Y$ such that $V(Y)=V$ and $E(Y)=E_{G_1}(V')\cup R \cup \Gamma$, where $E_{G_1}(V')$ is the set of edges in $G_1$ with both endpoints in $V'$.
%We assign each edge of $Y$ with unit weight.
To show that $\opttsp(w)$ is bounded away from $2n$, we will first show that $\gtsp(Y)$, the graphic-TSP cost of $Y$, is bounded away from $2n$, and then we show that the difference between $\opttsp(w)$ and $\gtsp(Y)$ is small. 

\iffalse
We use the following lemma, which is a corollary of Lemma 2.9 in \cite{chen2020sublinear}.
\begin{lemma} \label{lem:matching-TSP}
If a connected graph $G$ contains a matching $M$ of size $\alpha n$, such that no edge of $M$ is a bridge in $G$, then $\gtsp(G)\le (2-\frac{2\alpha}{3})n$.
\end{lemma} 
\fi

We first show that $\gtsp(Y)\le (2- \epsilon /10)n$.	
Note that here we view all edges of $Y$ as weight-$1$ edges.	
We use the following observation.
\begin{observation}
\label{obs: bridges bounded by RC}
For each vertex $v\in L_{\ell}$, $R_v$ contains at most $\rc(S_v)$ edges that are bridges in $Y$. 
\end{observation}
\begin{proof}
Let $H_v$ be the graph obtained by contracting all vertices of $V\setminus V(S_v)$ into a single node $v'$, so $E(H_v)=R_v\cup \Gamma_v$. From the definition of a reconfiguration and set $\Gamma_v$, it is easy to see that graph $H_v$ is Eulerian graph.
Therefore, the only edges of $R_v$ that are possibly bridges in $Y$ are the edges that have more than one copies contained in $R_v$. Note that, if we delete one copy for each of these edges from $R_v$, then the remaining edges in $R_v$ still induce a connected graph on $V(S_v)$. Therefore, the number of such edges is at most $|R_v|-(|V(S_v)|-1)\le \rc(S_v)$.
\end{proof}

Now fix a vertex $v \in L_{\ell}$. %(with $\card{V(S_v)}>1$) 
Let $T_v$ be a spanning tree of the graph $(V(S_v),R_v)$. Note that such a spanning tree exists since $R_v$ is connected.
Also note that, from the definition of sets $R_v$ and $\Gamma_v$, each vertex in $S_v$ is incident to at least two edges of $R_v \cup \Gamma_v$. Therefore, each leaf of $T_v$ is incident to an edge of $(R_v \cup \Gamma_v) \setminus E(T_v)$. Let $\hat L_v$ be the number of leaves in $T_v$, then $\hat L_v \le 2 (|R_v|+|\Gamma_v|-|E(T_v)|) \le 4\cdot \rc(S_v)$. 
From \Cref{lem: matching_in_tree}, $T_v$ contains a matching of size at least $(\card{V(S_v)}-\hat L_v)/2$.
From \Cref{obs: bridges bounded by RC}, we get that $R_v$ contains a matching consisting of only non-bridge edges with size at least 
$$\frac{\card{V(S_v)}-\hat L_v}{2}-\rc(S_v) \ge \frac{\card{V(S_v)}-4\cdot \rc(S_v)}{2}-\rc(S_v) = \frac{\card{V(S_v)}}{2} -3\cdot\rc(S_v).$$ 
Since sets $\set{R_v}_{v\in L_{\ell}}$ are mutually disjoint, $Y$ contains a matching of size 
\[
\sum_{v \in L_{\ell}:\card{V(S_v)}>1} \left(\frac{\card{V(S_v)}}{2}-2\cdot \rc(S_v)\right) \ge \frac{1}{2}\cdot \sum_{v \in L_{\ell}} \card{V(S_v)} -\frac{n}{2\ell} - 3\cdot \sum_{v \in L_{\ell}} \rc(S_v) \ge \left( \frac{{\epsilon}}{4}-\frac{3{\epsilon}}{40}-\frac{1}{2\ell} \right)n, 
\]
consisting of only non-bridge edges.
From \Cref{lem:matching-TSP}, $\gtsp(Y) \le (2- \epsilon/10)n$.

We now show that $\opttsp(w)-\gtsp(Y) \le  \epsilon n/20$. Note that, together with $\gtsp(Y) \le (2- \epsilon/10)n$, this implies that $\opttsp(w) \le (2- \epsilon/20)n$, thus completing the proof of \Cref{clm:tsp_upper}.
Let $\pi=(v_1,v_2,\ldots,v_{n'},v_1)$ be the optimal graphic-TSP-tour in $Y$, where possibly $n'>n$ since for each $1\le i\le n'$, $(v_i,v_{i+1})$ is required to be an edge of $Y$, so $\pi$ may visit some vertices more than once.
%Consider the TSP-tour $\pi'$ obtained from $\pi$ by naturally shortcuting all non-first appearance of all vertices in $V$.
%Note that all edges of $E_{G_1}(V')$ have weight $1$.
Moreover,
\[
\begin{split}
\cost(\pi)-\cost_g(\pi)
&= \sum_{1\le i\le n'} (w(v_i,v_{i+1})-1)\\
&= \sum_{(v_i,v_{i+1})\in E_{G_1}(V')} (w(v_i,v_{i+1})-1)+ \sum_{(v_i,v_{i+1})\in R\cup \Gamma} (w(v_i,v_{i+1})-1)\\
&= \sum_{v\in L_{\ell}}\sum_{(v_i,v_{i+1})\in R_v\cup \Gamma_v} (w(v_i,v_{i+1})-1)\\
&= 2\cdot\sum_{v\in L_{\ell}}\rc(S_v)\le \eps n/20,
\end{split}
\]
where $\cost_g(\pi)$ is the graphic-TSP-cost of the tour $\pi$, namely the cost of $\pi$ if the weight of each edge is $1$. Clearly, this implies that $\opttsp(w)-\gtsp(Y) \le  \epsilon n/20$.
\end{proof}

\paragraph{Case 1.2. $\sum_{1\le i\le 1/{\epsilon}^2}\gamma_i\ge 1/80{\epsilon}$.} 
Since event $\xi_3$ does not happen, $\sum_{v\in L_{\ell}}\rc(T_v)\ge {\epsilon}n/160$. We show in the next claim that $\opttsp(w)\ge (1+{\epsilon}/1120)n$, so our estimate $X=2n$ in this case is a $(2-\Omega({\epsilon}))$-approximation of $\opttsp(w)$.

\begin{claim}
	\label{clm:tsp_lower}
	$\opttsp(w) \ge n + (\sum_{v \in L_{\ell}} \rc(S_v))/7$. 
\end{claim}
\begin{proof}
	Let $\pi$ be the optimal TSP-tour that traverses all vertices of $V$, and we denote $\alpha=\cost(\pi)-n$.
	We will construct, for each maximal $\ell$-light vertex $v\in L_{\ell}$, a reconfiguration $R_v$ for its maximal $\ell$-light subgraph $S_v$, such that $\sum_{v\in L_{\ell}}\cost(R_v)\le 7\alpha$. Note that this implies  
	$\opttsp(w)\ge n+(\sum_{v\in L_{\ell}}\rc(S_v))/7$.
	
	%We define a multi-graph $H$ as follows. The vertex set of $H$ is $V$. The edge set $E(H)$ consists of two parts. We let set $E_1$ contain, for every pair of vertices that appear consecutively in the tour $\pi$, an edge connecting them. In other words, if $\pi=(v_1,v_2,\ldots,v_n,v_1)$, then $E_1=\set{(v_i,v_{i+1})\mid 1\le i\le n}$. We let set $E_2$ contain, for each edge $e$ of tree $T$ that does not appear in $E_1$, two copies of $e$. We then define $E(H)=E_1\cup E_2$. This finishes the description of $H$. It is easy to verify that $H$ is a Eulerian graph, and $E(T)\subseteq E(H)$. Moreover, since $\cost(\pi)=n+\alpha$ and the only pairs $(v,v')$ with $w(v,v')$ are edges of $T$, we get that $|E_2|\le 2\alpha$, and therefore $\sum_{(v,v')\in E(H)}w(v,v')\le n+3\alpha$.
	
	Let $E(\pi)$ be the set that contains, for every pair of vertices that appear consecutively in the tour $\pi$, an edge connecting them. In other words, if $\pi=(v_1,v_2,\ldots,v_n,v_1)$, then $E(\pi)=\set{(v_i,v_{i+1})\mid 1\le i\le n}$.
	We use the following simple observation.
	\begin{observation}
		\label{obs: number of light subgraphs}
		$|L_{\ell}|\le 2\alpha$.
	\end{observation}
	\begin{proof}
		From the definition of maximal $\ell$-light subgraphs, it is easy to see that, for each $v\in L_{\ell}$, there is at least one edge in $E(\pi)$ connecting a vertex of $S_v$ to a vertex of $S_v$. Note that such an edge does not belong to $G_1$, so it has cost at least $2$. Since $G_1$ contains $|L_{\ell}|$ maximal $\ell$-light subgraphs, $E(\pi)$ contains at least $|L_{\ell}|/2$ such edges, so $n+\alpha=\cost(\pi)\ge n+|L_{\ell}|/2$. It follows that $|L_{\ell}|\le 2\alpha$.
	\end{proof}

	Let $v\in L_{\ell}$ be a maximal $\ell$-light vertex, and let $S_v$ be its maximal $\ell$-light subgraph $S_v$. 
	%Note that every edge of $T_v$ is contained in $E(H)$. 
	We construct the reconfiguration $R_v$ as follows.
	Starting with $R_v=\emptyset$, we first add all edges of $E(\pi)$ with both endpoints in $S_v$ to $R_v$. Assume that the graph on $V(S_v)$ induced by these edges contain $k$ connected components. We then select a set $E'_v$ of $(k-1)$ edges in $E(S_v)$, such that, together with edges of $E(\pi)$ whose both endpoints belong to $V(S_v)$, they induce a connected graph on $V(S_v)$. 
	Since $S_v$ is connected, such a set $E'_v$ clearly exists. 
	We then add two copies for each edge of $E'_v$ to $R_v$. This completes the construction of $R_v$.
	Clearly, $R_v$ is a valid reconfiguration.
	Note that, for a vertex $u\in S_v$, $\deg_{R_v}(u)$ is odd iff there is an edge $(u,u')$ in $E(\pi)$ connecting $u$ with some other vertex $u'\notin S_v$.

	%It remains to show that $\sum_{v\in L_{\ell}}\cost(R_v)\le 7\alpha$. 
	We denote $V'=V \setminus (\bigcup_{v\in L_{\ell}}V(S_v))$, and define
	$$\cost_{\pi}(V')=\sum_{(u,u')\in E(\pi):\text{ }u,u'\in V'}w(u,u')\text{ }+\sum_{(u,u')\in E(\pi):\text{ }u\notin V', u'\in V'}\frac{w(u,u')}{2}.$$
	Clearly, $\cost_{\pi}(V')\ge |V'|$.
	We define $\cost_{\pi}(S_v)$ for all vertices $v\in L_{\ell}$ similarly. 
	Denote $E'=\bigcup_{v\in L_{\ell}}E'_v$. 
	From similar arguments in the proof of \Cref{obs: number of light subgraphs}, we can show that $|E'|\le 2\alpha$.
	%From the construction of $\set{R_v}_{v\in L_{\ell}}$, 
	Therefore, $\cost_{\pi}(V')+\sum_{v\in L_{\ell}}\cost_{\pi}(S_v)=\cost(\pi)+2\cdot\sum_{v\in L_{\ell}}|E'_v|= n+\alpha+2|E'|\le n+5\alpha$.
	Altogether,
	%Combined with \Cref{obs: number of light subgraphs}, 
	we have
	\[
	\begin{split}
	\sum_{v\in L_{\ell}}\cost(R_v)
	&=\sum_{v\in L_{\ell}}\left(
	\sum_{(u,u')\in R_v}w(u,u')+\frac{1}{2}\sum_{u\in V_1(R_v)}\ord(u)-(|V(S_v)|-1)\right)\\
	&\le \sum_{v\in L_{\ell}}\left(\cost_{\pi}(S_v)-(|V(S_v)|-1)\right)\\
	&\le \sum_{v\in L_{\ell}}\cost_{\pi}(S_v)-\sum_{v\in L_{\ell}}|V(S_v)|+|L_{\ell}|\\
	&\le (n+5\alpha)-\cost_{\pi}(V')-\sum_{v\in L_{\ell}}|V(S_v)|+2\alpha\\
	&\le (n+5\alpha)-n+2\alpha \le 7\alpha.
	\end{split}
	\]
	This completes the proof of \Cref{clm:tsp_lower}.
\end{proof}

\subsection{Case 2. Less than $\eps$-fraction of the executions of $\local(\cdot,\ell)$ return \textsf{Success}}

Since event $\xi_1$ does not happen, $\sum_{v\in L_{\ell}}|V(S_v)|\le 2{\epsilon}n$.
In this case the algorithm continues to execute Step~\ref{G_1_connected_step_3}.
Recall that in Step~\ref{G_1_connected_step_3} we execute the subroutine $\pbfs$ on $(100\cdot n\log n)/h$ random vertices of $V$.
Recall that $Z=(V,\bigcup_{v\in \tilde V}E(P_v))$ and $H=(V,\bigcup_{v\in \tilde V}E(H_v))$, where $\tilde V$ is the set of sampled vertices that the subroutine $\pbfs(\cdot,h,q,\alpha)$ returns \textsf{Success} upon.  
Also recall that $V_{\textsf{chunk}}=\bigcup_{v\in \tilde V}V(H_v\setminus P_v)$ and $V_{\textsf{isolated}}=V\setminus (V_{\textsf{chunk}}\cup V(Z))$. 
We further distinguish between the following sub-cases, depending on outcomes of subroutines $\pbfs$.

\paragraph{Case 2.1. $|V_{\textsf{chunk}}|\ge 10\hat{\epsilon}n$.}
We first prove the following claim.

\begin{claim}
\label{clm: most chunk vertices are in e-blocks}
The number of vertices in $H$ that do not belong to any e-block of $H$ is at most $\hat{\eps} n$.
\end{claim}
\begin{proof}
Let $H'$ be the graph obtained from $H$ by deleting from it all $\ell$-light vertices in $V$ (recall that the number of $\ell$-light vertices of $V$ is at most $2\eps n$). Note that every degree-$1$ vertex in $H$ corresponds to an degree-$1$ vertex of $G'_1$. Since the number of degree-$1$ vertices of $G'_1$ is at most $n/\ell$, and since the diameter of each graph in $\set{H_v}_{v\in \tilde V}$ is at most $2h$, the number of vertices that do not belong to an e-block of $H'$ is at most $(n/\ell)\cdot (2h)=\hat{\eps}n/100$. It follows that the number of vertices in $H$ that do not belong to any e-block of $H$ is at most $\hat{\eps} n$.
\end{proof}

Assume first that the size of a maximum matching in $H$ consisting of only edges in e-blocks of $H$ is at least $2\hat{\eps}n$.
Since $H$ is a subgraph of $G_1$, the size of a maximum matching in $G_1$ consisting of only edges in e-blocks of $G_1$ is also at least $2\hat{\eps}n$. From \Cref{lem:matching-TSP}, we get that $\opttsp(w)\le (2-\hat\eps)n$, so our estimate $X=(2-\hat\eps)n$ in this case is a $(2-\Omega(\hat\eps))$-approximation of $\opttsp(w)$.

Assume now that the size of a maximum matching in $H$ consisting of only edges in e-blocks of $H$ is at most $2\hat{\eps}n$. We will show that the size of a maximum matching in $G_1$ is at most $(1/2-\hat{\eps})n$. Then from \Cref{lem: G_1_matching_TSP_lower_bound}, $\opttsp(w)\ge (1+2\hat{\eps})n$, so our estimate $X=2n$ in this case is a $(2-\Omega(\hat\eps))$-approximation of $\opttsp(w)$.
Let $M$ be a maximum matching in $H$ consisting of only edges in e-blocks of $H$. The edges of $M$ can be partitioned into three subsets: set $M_1$ contains edges with both endpoints lying in $V(H)$; set $M_2$ contains edges with both endpoints lying in $V\setminus V(H)$; and set $M_3$ contains all other edges. 
Clearly, $|M_2|\le (n-|V(H)|)/2$. From~\Cref{clm: most chunk vertices are in e-blocks}, $|M_1|\le \hat{\eps}n+2\hat\eps n=3\hat{\eps n}$, since every edge in $M_1$ either contains a vertex that does not belong to any e-block of $H$, or is entirely contained in some e-block of $H$.
Moreover, since the only vertices of $V(H)$ that are connected by edges to vertices of $V\setminus V(H)$ are interface vertices of graphs in $\set{H_v}_{v\in \tilde V}$, so $|M_3|\le 2|\tilde V|\le O(\sqrt{n}\log n)$.
Altogether, 
\[|M|=|M_1|+|M_2|+|M_3|\le (n-|V(H)|)/2+3\hat{\eps n}+O(\sqrt{n}\log n)\le (1/2-\hat\eps) n.\]

\paragraph{Case 2.2. $|V_{\textsf{isolated}}|,|V_{\textsf{chunk}}|< 10\hat{\epsilon}n$.}
Recall $V(Z)=V\setminus (V_{\textsf{isolated}}\cup V_{\textsf{chunk}})$, so $|V(Z)|\ge (1-20\hat{\eps})n$.
From algorithm in \Cref{clm: vertex disjoint induced paths}, we obtain a set $\qset$ of $\tilde O(n/h)$ paths with $|E(\qset)|\ge (1-40\hat{\eps})n$. It is easy to see, from the definition of the subroutine $\bfs$, that the condition of \Cref{clm: long induced path concatenation} is satisfied by the set $\qset$ of paths.
We then use the algorithm from \Cref{clm: proper tour cost} to compute the minimum cost of a $\qset$-proper tour.
Assume first that the cost is at least $(2-100\hat \epsilon)n$. Then from \Cref{clm: long induced path concatenation}, $\opttsp(w)>(1+\hat \epsilon)n$, and our estimate $X=2n$ in this case is a $(2-\Omega(\hat \eps))$-approximate of $\opttsp(w)$. 
Assume now that the cost is less than $(2-100\hat \epsilon)n$, then since $|E(\qset)|\ge (1-40\hat \epsilon)n$, $\opttsp(w)\le (2-100\hat \epsilon)n+2\cdot40\hat \epsilon n= (2-20\hat \epsilon)n$ (a tour with such cost can be easily constructed from the $\qset$-proper tour and Euler-tour of any spanning tree of $G_1$), and so our estimate $X=(2-\hat \epsilon)n$ in this case is a $(2-\Omega(\hat \eps))$-approximate of $\opttsp(w)$.

\paragraph{Case 2.3. $|V_{\textsf{isolated}}|\ge 10\hat{\epsilon}n$.}
From Step~\ref{G_1_connected_step_4}, we then use the algorithm from \Cref{thm: n^{1.5}-query-2-approx-of-matching-size} with parameter $\hat\epsilon/100$ to obtain an estimate of the size of some maximal matching of $G_1[V_{\textsf{isolated}}]$.
If the estimate is at most $\hat{\epsilon}n$, then we return $2n$ as the estimate of $\opttsp(w)$.
Otherwise, we return $(2-\hat{\epsilon}/200)n$ as the estimate of $\opttsp(w)$.

We first consider the case where the estimate output by the algorithm from \Cref{thm: n^{1.5}-query-2-approx-of-matching-size} is at most $\hat \eps n$.
Since event $\xi_2$ does not happen, the size of some maximal matching in graph $G_1[V_{\textsf{isolated}}]$ is at most $(1+1/100)\hat{\eps}n$, and therefore the size of a maximum matching in  $G_1[V_{\textsf{isolated}}]$ is at most $(2+2/100)\hat{\eps}n$. We use the following claim to show that our estimate $X=2n$ in this case is a $(2-\Omega(\hat\eps))$-approximate of $\opttsp(w)$.

\begin{claim}
	\label{clm: matching_size_lower_bound}	
	If $|V_{\textnormal{\textsf{isolated}}}|\ge 10\hat\epsilon n$, and the size of a maximum matching in graph $G_1[V_{\textsf{isolated}}]$ is at most $(2+2/100)\hat{\eps}n$, then $\opttsp(w)\ge (1+\hat\epsilon/10)n$.
\end{claim}
\begin{proof}
From \Cref{lem: G_1_matching_TSP_lower_bound}, it is sufficient to show that the size of a maximum matching in graph $G_1$ is at most $(1/2-\hat\eps/20)n$. Assume for contradiction that this is not true.
We will show that the size of a maximum matching in $G_1[V_{\textsf{isolated}}]$ is at least $(2+2/10)\hat{\eps}n$, leading to a contradiction.

Let $M$ be a maximum matching in $G_1$, so $|M|\ge (1/2-\hat\epsilon/20)n$. The edges of $M$ can be partitioned into three subsets: set $M_1$ contains edges with both endpoints lying in $V_{\textsf{isolated}}$; set $M_2$ contains edges with both endpoints lying in $V\setminus V_{\textsf{isolated}}$; and set $M_3$ contains all other edges, so $|M|=|M_1|+|M_2|+|M_3|$. 
Since $M_1$ is a matching in $G_1[V_{\textsf{isolated}}]$, it suffices to show that $|M_1|\ge (2+2/10)\hat{\eps}n$.
Clearly, $|M_2|\le (n-|V_{\textsf{isolated}}|)/2$. Since $|V_{\textsf{isolated}}|\ge 10\hat{\eps} n$, it suffices to show that $|M_3|\le 2\hat \eps n$. 
From the definition of $M_3$, each edge of $M_3$ has one endpoint in $V_{\textsf{isolated}}$.
It is easy to verify that, for each vertex $\hat v\in V_{\textsf{isolated}}$, the vertices in $H$ that $\hat v$ may possibly be adjacent to are interface vertices of graphs of $\set{H_v\mid v\in \tilde V}$, and there are at most $|\tilde V|=\tilde O(n^{0.5})$ of them.
Therefore, 
$|M_3|\le \tilde O(n^{0.5}) <2\hat{\eps}n$.
\end{proof}

We now consider the case where the estimate output by the algorithm from \Cref{thm: n^{1.5}-query-2-approx-of-matching-size} is greater than $\hat \eps n$.
Since event $\xi_2$ does not happen, the size of some maximum matching in graph $G_1[V_{\textsf{isolated}}]$ is at most $0.99\hat{\eps}n$. The following crucial claim immediate implies that our estimate $X=(2-\hat\eps/200)n$ in this case is a $(2-\Omega(\hat\eps))$-approximate of $\opttsp(w)$ with high probability, and thus completes the proof of \Cref{thm: main G_1 connected}.
The proof of \Cref{clm: matching_size_upper_bound} is long and technical, and is therefore deferred to \Cref{sec: Proof of matching_size_upper_bound}.
\begin{claim}
\label{clm: matching_size_upper_bound}
If $|V_{\textnormal{\textsf{isolated}}}|\ge 10\hat \epsilon n$, and the size of a maximum matching in $G_1[V_{\textnormal{\textsf{isolated}}}]$ is at least $0.99\hat \epsilon n$, then with probability $1-O(n^{-99})$, $\opttsp(w)\le (2-\hat \epsilon/125)n$.
\end{claim}

%We say that a vertex $v$ of $G_1$ is \emph{smooth}; iff (i) vertex $v$ also appears in the pruned subgraph $G_1'$; (ii) $v$ is at tree-distance at least $2\sqrt{n}$ from all degree-not-$2$ vertices of $G_1'$; and (iii) the total weight of all vertices that are within tree-distance $2\sqrt{n}$ from $v$ in $G_1'$ is at most $10\sqrt{n}$.

\subsection{Proof of \Cref{clm: matching_size_upper_bound}}
\label{sec: Proof of matching_size_upper_bound}

%\znote{To Complete. \fbox{$\epsilon=2^{-100}, \hat\eps=2^{-40}$, $q=1000n/\hat{\eps}$, $\ell=100\sqrt{n}$ and $h=\hat \eps \ell/200$}}
In this section we provide the proof of \Cref{clm: matching_size_upper_bound}.
We start with the some new definitions in subsections \ref{subsec: new defs_1}, \ref{subsec: new defs_2} and then complete the proof in \Cref{subsec: finally the complete proof}.

\subsubsection{Pruned Subgraph and its Augmented E-Block Tree}
\label{subsec: new defs_1}

%\paragraph{The Pruned Subgraph.}
Recall that $L_{\ell}$ is the set of maximal $\ell$-light vertices in $V$.
We construct another vertex-weighted graph $G_1'$, that we call the \emph{pruned subgraph} of $G_1$, by starting from $G_1$ and processing each vertex of $L_{\ell}$ as follows. 	
For each maximal $\ell$-light vertex $v\in L_{\ell}$, recall that $e_v$ is the witness edge of $v$, and we denote by $\hat v$ the other endpoint of $e_v$, which we call the \emph{witness vertex} for $v$. By definition, $\hat v$ is not $\ell$-light.
Let $G_1'$ be the graph obtained from $G_1$ by deleting, for each vertex $v\in L_{\ell}$, all vertices in the subgraph $S_v$ and their incident edges.
It is easy to see that $G_1'$ is a subgraph of $G_1$. We use the following simple observation.

\begin{observation}
	\label{obs: number of deg-1 in pruned subgraph}
	Graph $G_1'$ contains at most $n/\ell$ degree-1 vertices.
\end{observation}
\begin{proof}
	For each vertex that is not the witness vertex for any vertex of $L_{\ell}$, we define its weight to be $1$. 
	For each vertex $u$ such that $u$ is the witness vertex for maximal $\ell$-light vertices $v_1,\ldots,v_r\in L_{\ell}$ (namely $\hat v_1=\cdots =\hat v_r=u$), we define its weight to be $1+\sum_{1\le i\le r}|V(S_{v_i})|$.
	Therefore, the total weight of all vertices in $G'_1$ is $n$.
	From the construction of $G_1'$, each degree-1 vertex $u$ of $G_1'$ must be the witness vertex for some vertex of $L_{\ell}$, and its weight is at least $\ell$.
	\Cref{obs: number of deg-1 in pruned subgraph} then follows.
\end{proof}

We then construct a tree $T$ as follows. We start from the e-block tree $\tset_{G'_1}$ of the pruned subgraph $G'_1$, and process all e-blocks of $G'_1$ whose corresponding e-block vertices in $\tset_{G'_1}$ have degree $2$ one-by-one as follows.
Consider such an e-block $B$ of $G'_1$. 
Note that there are at most two vertices of $B$ that are incident to an edge of $\delta_{G_1}(B)$. We denote them by $s,t$ (where possibly $s=t$). We replace the e-block vertex $v_B$ in $\tset_{G'_1}$ with an arbitrary shortest path in $B$ connecting $s$ to $t$.
This completes the construction of $T$. %We also distribute the total weight of this e-block in $G'_1$ to the remaining vertices. 
We say that $T$ is the \emph{augmented e-block tree} for graph $G_1'$.
We denote by $L$ the total length of all shortest paths that are added to e-block tree $\tset_{G'_1}$ in the construction of tree $T$.

We use the following simple observations.

\begin{observation}
\label{obs: basics about T}
The following statements are true for $T$:
\begin{itemize}
\item $T$ contains at most $n/\ell$ leaves;
\item all bridges of graph $G'_1$ are contained in $T$;
\item every degree-$2$ vertex of $T$ is also a vertex of the original graph $G_1$; and
\item for every pair $v,v'$ of degree-$2$ vertices in $T$ that belong to the same maximal induced subpath of $T$, $\dist_{G_1}(v,v')=\dist_T(v,v')$.
\end{itemize}
\end{observation}

\iffalse
Since $G'_1$ contains at most $n/\ell$ degree-$1$ vertices, It is easy to see that $T$ contains at most $n/\ell$ leaves, and therefore it contains at most $2n/\ell$ maximal induced subpaths. 
Note that all bridges of $G'_1$ are contained in $T$. 
Note that every degree-$2$ vertex of $T$ is also a vertex of the original graph $G_1$, and for any pair $v,v'$ of degree-$2$ vertices in $T$ that belong to the same maximal induced subpath of $T$, $\dist_{G_1}(v,v')=\dist_T(v,v')$.
We use the following observation.
\fi

\begin{observation}
\label{obs: total length of deg-2 block paths}
$G_1$ contains a matching of size at least $L/2$ consisting of only non-bridge edges.
\end{observation}

\subsubsection{Representatives and Smooth Vertices}
\label{subsec: new defs_2}

Recall that every vertex of $V$ is either an $\ell$-light vertex or a vertex of $V(G'_1)$.
We first define, for each vertex $v\in V$, a set $R(v)$ of vertices of $T$, that we call the set of \emph{representatives} of $v$ in $T$, as follows:
\begin{itemize}
	\item if a vertex $v$ belongs to $G_1'$ but does not belong to any e-block of $G_1'$, then we define $R(v)=\set{v}$;
	\item if a vertex $v$ belongs to some e-block $B$ of $G_1'$, whose corresponding e-block vertex $v_B$ is a special vertex in the e-block tree $\tset_{G'_1}$ of $G_1$, then we let $R(v)$ contain only the vertex of $T$ that represents the e-block $B$ of $G_1'$;
	\item if a vertex $v$ belongs to some e-block $B$ of $G_1'$, whose corresponding e-block vertex $v_B$ is a degree-$2$ vertex in $\tset_{G'_1}$, then we let $R(v)$ contain all vertices on the shortest path that are added to $\tset_{G'_1}$ when processing the e-block vertex $v_B$.
	%pick arbitrarily a vertex on the path $Q_B$ in $T$ that corresponds to the e-block $B$, and designate it as $r(v)$;
	\item if $v$ is an $\ell$-light vertex, then we let $v'$ be the unique witness vertex in $G_1$ of the maximal $\ell$-light subgraph that contains $v$, and we define $R(v)$ depending on whether or not $v'$ belongs to some e-block of $G'_1$ similarly.
\end{itemize} 
	
We say that a vertex $v$ in $T$ is \emph{smooth}, iff (i) it is at tree-distance at least $2h$ from all special vertices of $T$; and
(ii) the number of vertices $\hat v$ in $V$ such that $\dist_T(R(\hat v), v)\le 2h$ is at most $10h/\hat{\eps}$, namely 
$|\set{\hat v\mid \dist_T(R(\hat v),v)\le 2h}|\le 10h/\hat{\eps}$, 
where $\dist_T(R(\hat v), v)=\min_{v'\in R(\hat v)}\set{\dist_T(v',v)}$.
	
%We say that a vertex $v$ in $T$ is \emph{smooth}, iff (i) it is at tree-distance at least $2h$ from all special vertices of $T$; and (ii) the number of vertices in $V$ that are within distance $2h$ from $v$ is at most $10h/\hat{\eps}$, namely $|\set{v'\mid w(v,v')\le 2h}|\le 10h/\hat{\eps}$.
%It is clear from the definition that, for each smooth vertex $v$, $|V(H_v)|\le 10h/\hat \eps$.
%If a vertex $v$ in $T$ does not satisfy (i), then we call it a \emph{type-1 non-smooth} vertex.
%If a vertex $v$ in $T$ does not satisfy (ii), then we call it a \emph{type-2 non-smooth} vertex.
We use the following claim.

\begin{claim} \label{obs:badnum}
The number of non-smooth vertices is at most $0.45\hat \eps n+L$.
\end{claim}
	
\begin{proof}
We say that a vertex $v$ in $T$ is a \emph{type-$1$} (\emph{type-$2$}, resp.) non-smooth vertex, iff vertex $v$ does not satisfy property (i) (property (ii), resp.).
%If a vertex $v$ in $T$ does not satisfy (ii), then we call it a \emph{type-2 non-smooth} vertex.	
From \Cref{obs: basics about T}, $T$ contains at most $n/\ell$ leaves, and therefore $T$ contains at most $2n/\ell$ maximal induced subpaths (from \Cref{obs:leaves_induced_paths}). It is easy to see that each maximal induced subpath of $T$ contains at most $4h$ type-1 non-smooth vertices. 
Therefore, the number of type-1 non-smooth vertices is at most $4h\cdot (2n/\ell)=\hat \eps n/25$. 

Let $P$ be a maximal induced subpath in $T$.
Let $\qset_P$ be the set of shortest paths that are added to $\tset_{G'_1}$ in the construction of path $P$ of $T$.
Let $P'$ be the path obtained from $P$ by contracting all paths in $\qset_P$, so $|V(P)|=|V(P')|+\sum_{Q\in \qset_P}|E(Q)|$.
Let $v_1,\ldots,v_k$ be all type-2 non-smooth vertices in $P'$, where the vertices appear on $P'$ in this order. 
Define $S_P=\set{v_1,v_{4h+3},v_{8h+5},\ldots,v_{(4h+2)\cdot\floor{k/(4h+2)}+1}}$.
It is clear that, if we denote $B(v,2h)=\set{\hat v\mid \dist_T(R(\hat v),v)\le 2h}$, then for every pair $v_i,v_j$ of vertices in $S_P$, $B(v_i,2h)\cap B(v_j,2h)=\emptyset$. Moreover, if we denote by $\pset$ the set of all maximal induced subpaths in $T$, then the sets $\set{B(v,2h)}_{v\in S(P), P\in \pset}$ are mutually disjoint. Therefore, $|\bigcup_{P\in \pset}S_P|\le n/(10h/\hat\eps)$.
On the other hand, the number of type-2 non-smooth vertices in $T$ is at most $(4h+2)\cdot |\bigcup_{P\in \pset}S_P|+\sum_{P\in \pset}|E(\qset_P)|$. Therefore, the number of type-2 non-smooth vertices is at most $(4h+2)\cdot n/(10h/\hat\eps)+L\le  0.41\hat \eps n+L$.
Altogether, the number of non-smooth vertices is at most $ 0.45\hat \eps n+L$.
\end{proof}

\iffalse
We use the following claim.

\begin{claim} \label{clm:fail-upper}
If $|V_{\textsf{isolated}}|\ge 10\hat\epsilon n$, and the size of a maximum matching in graph $G_1[V_{\textsf{isolated}}]$ is at least $0.99\hat\epsilon n$, then with probability $1-O(n^{-99})$, there exist a set $\set{(u_i,w_i)}_{1\le i\le k}$ of pairs of vertices in $V$, such that $\sum_{1\le i\le k}w(u_i,w_i)\le \hat\epsilon n/2000$, and
the size of a maximum matching consisting of only edges in e-blocks of the graph 
$G_1\cup \set{(u_i,w_i)}_{1\le i\le k}$ 
%$(V, E(G_1)\cup \set{(u_i,w_i)}_{1\le i\le k})$ 
is at least $\hat\epsilon n/40$.
\end{claim}
	
We will provide the proof of \Cref{clm:fail-upper} later, after we complete the proof of \Cref{clm: matching_size_upper_bound} using it.	
\znote{to complete}

The remainder of this section is dedicated to the proof of \Cref{clm:fail-upper}. 
%Denote $\tilde G_1=(V,E(G_1)\cup \set{(u_i,w_i)}_{1\le i\le k}).$
\fi

\subsubsection{Completing the Proof of \Cref{clm: matching_size_upper_bound}}
\label{subsec: finally the complete proof}

Recall that we have denoted by $L$ the total length of all shortest paths that are added to e-block tree $\tset_{G'_1}$ in the construction of tree $T$.
From \Cref{obs: total length of deg-2 block paths}, if $L\ge \hat \eps n/80$, then $G_1$ contains a matching of size $\hat \eps n/80$ with no bridge edges. From \Cref{lem:matching-TSP}, this implies that $\opttsp(w)\le (2-\hat \epsilon/120)n$, and \Cref{clm: matching_size_upper_bound} then follows.
Therefore, we assume from now on that $L< \hat \eps n/80$.
And then from \Cref{obs:badnum}, the number of non-smooth vertices is at most $\hat{\eps}n/2$.

Let $M$ be a maximum matching in graph $G_1[{V_{\textsf{isolated}}}]$, so $|M|\ge 0.99\hat\epsilon n$. If at least $\hat \eps n/40$ edges of $M$ belong to e-blocks of $G_1[{V_{\textsf{isolated}}}]$, then the claim is true since all such edges also belong to e-blocks of $G_1$. 
Therefore, we only need to consider the case where at least $0.99\hat\epsilon n-\hat \eps n/40 \ge 0.96 \hat \eps n$ edges of $M$ are bridges in $G_1$. 
%Let $M'\subseteq M$ be the set of bridges in $G_1$.
Since the number of $\ell$-light vertices in $V$ is at most $2\eps n$, at least $(0.96\hat \eps -2\eps)n\ge 0.95\hat \eps n$ edges of $M$ also belong to the pruned subgraph $G'_1$ as bridges. 
Recall that all bridges of $G'_1$ are contained in $T$, so at least $0.95\hat\eps n$ edges of $M$ belong to $T$.
We now delete from these edges, all edges with at least one endpoint that is at tree-distance at most $2h$ from some special vertex of $T$. 
Clearly, we have deleted at most $2h\cdot (2n/\ell)\le 0.02\hat{\eps}n$ vertices.
We denote by $M'$ the set of remaining edges of $M$ in $T$, so $|M'|\ge 0.93\hat{\eps} n$. 
We denote by $V'$ the set of endpoints of edges in $M'$, so $\card{V'} \ge 1.86 \hat \eps n$.

\begin{claim}
\label{clm: many failed vertices in T}
If $|V_{\textnormal{\textsf{isolated}}}|\ge 10\hat \epsilon n$, and the size of a maximum matching in $G_1[V_{\textnormal{\textsf{isolated}}}]$ is at least $0.99\hat \epsilon n$, then with probability $1-O(n^{-99})$, $T$ contains at least $0.93 \hat \eps n$ vertices that subroutine $\pbfs(\cdot,h,q,\alpha)$ returns \textnormal{\textsf{Fail}} on.
\end{claim}
\begin{proof}
Recall that in Step~\ref{G_1_connected_step_3}, we have sampled a set of $(100\cdot n\log n)/h$ random vertices in $V$ to execute the subroutine $\pbfs$ on, and a vertex of $V$ is eventually added to $V_{\textsf{isolated}}$ if it does not belong to any graph of $\set{H_v}_{v\in \tilde V}$, where $\tilde V$ is the subset of sampled vertices that the subroutine $\pbfs$ returns \textsf{Success} on.
For any vertex $v\in V$, if the set $B(v,h)=\set{v'\mid w(v,v')\le h}$ contains at least $h$ vertices that the subroutine $\pbfs(\cdot,h,q,\alpha)$ returns \textsf{Success} upon, then with probability $1-(1-h/n)^{100n\log n/h}=1-O(n^{-100})$, at least one of these vertices is sampled in Step~\ref{G_1_connected_step_3} and the corresponding graph $H_{v'}$ returned by the subroutine $\pbfs(v',h,q,\alpha)$ contains $v$, and therefore $v$ will not be added to $V_{\textsf{isolated}}$. 
From the union bound over all vertices of $V$, with probability at least $1-O(n^{-99})$, for each vertex $v$ in $V_{\textsf{isolated}}$, set $B(v,h)$ contains at most $h$ vertices that the subroutine $\pbfs(\cdot,h,q,\alpha)$ returns \textsf{Success} upon. 
We assume from now on that this is the case.
		
For each vertex $v$ in $V'$, let $N_v$ be the set of vertices that are at tree-distance at most $h$ from $v$ in $T$, and let $\hat N_v\subseteq N_v$ be the subset that contains all vertices that the subroutine $\pbfs(\cdot,h,q,\alpha)$ returns \textsf{Fail} on.
For each vertex $v\in V'$, it is immediate that $N_v\subseteq B(v,h)$, so $|N_v\setminus \hat N_v|\le h$, and since $\card{N_v}= 2h+1$, we have $|\hat N_v|\ge h+1$. 
On the other hand, from the definition of $V'$, every vertex in $N_v$ is at tree-distance at least $h$ from all special vertices of $T$. Therefore, each vertex of $T$ appears in at most $(2h+1)$ sets of $\set{N_v}_{v\in V'}$.
Altogether, if we let $V_f$ be the set of vertices in $T$ that the subroutine $\pbfs(\cdot,h,q,\alpha)$ returns \textsf{Fail} upon, then
$|V_f|\cdot (2h+1)\ge \sum_{v\in V'}|N_v|\ge \sum_{v\in V'}|\hat N_v|\ge |V'|\cdot (h+1)$, 
and it follows that $|V_f|\ge \card{V'} (h+1)/(2h+1) > 0.93 \hat \eps n$.
\end{proof}

%From \Cref{obs:badnum}, among all vertices of $T$ that the subroutine $\pbfs(\cdot,h,q,\eps,\alpha)$ returns \textsf{Fail} on, at most $\hat \eps n/2$ of them are non-smooth. 
Let $\hat V$ be the set of all smooth vertices in $T$ that the subroutine $\pbfs(\cdot,h,q,\alpha)$ returns \textsf{Fail} on. From \Cref{obs:badnum} and \Cref{clm: many failed vertices in T}, so $\card{\hat V} \ge 0.93\hat{\eps}n-\hat{\eps}n/2\ge 0.43\hat \eps n$. 

We delete from $\hat V$ all vertices such that $|V(T^h_v)|\ge 10h/\hat{\eps}$. Since the number of $\ell$-light vertices is at most $2\eps n$, at most $(2\eps n)\cdot (2h+1)/(10h/\hat{\eps})\le \eps n$ vertices are deleted. Therefore, at least $0.4\hat \eps n$ vertices remain in $\hat V$. 
%Note that, for each smooth vertex $v$ of $T$, $|B(v,h)|\le |\set{\hat v\mid \dist_T(R(\hat v),v)\le 2h}|\le 10h/\hat{\eps}$.
It is easy to see that, if the subroutine $\pbfs(v,h,q,\alpha)$ returns \textsf{Fail} for some remaining vertex $v$ in $\hat V$, then either tree $T^h_v$ contains at least $3$ level-$h$ vertices, or the subroutine $\pbfs(v,h,q,\alpha)$ has performed $q$ queries before all its $h$ stages are finished. 
We partition the remaining set $\hat V$ into two subsets accordingly: set $\hat V_1$ contains vertices $v$ such that $T^h_v$ contains at least $3$ level-$h$ vertices; and set $\hat V_2$ contains vertices $v$ such that $\pbfs(v,h,q,\alpha)$ performs $q$ queries before all its $h$ stages are finished. We distinguish between the following cases.
		
\paragraph{Case 2.3.1. $\card{\hat V_1}\ge 0.2 \hat \eps n$.}
Consider now a vertex $v\in \hat V_1$. Since $v$ is a smooth vertex, $T$ contains at most $2$ vertices that are at tree-distance $h$ from $v$. 
Since tree $T^h_v$ contains at least $3$ level-$h$ vertices, there exists a vertex $u\notin V(T)$ such that $w(u,v)=h$. 
We call such a vertex a \emph{failure-certificate} for $v$, and designate any such vertex as $u(v)$. 
Clearly, either $u(v)$ is an $\ell$-light vertex, or $u(v)$ belongs to some e-block of $G'_1$. 
It is not hard to see that each $\ell$-light vertex may be a failure-certificate for at most $2$ smooth vertices. 
Since there are at most $2\eps n$ $\ell$-light vertices, set $\hat V_1$ contains at most $4\eps n$ smooth vertices that have an $\ell$-light failure-certificate. 
Therefore, $\hat V_1$ contains at least $0.19\hat \eps n$ vertices $v$ such that the failure-certificate $u(v)$ belongs to some e-block of $G'_1$. We then show using in the following observation that the size of a maximum matching consisting of only edges in e-blocks of $G_1'$ is at least $0.19 \hat \eps n /6 \ge \hat\eps n/40$.
And then together with \Cref{lem:matching-TSP}, this implies that $\opttsp(w)\le (2-\hat\eps/60)n$, thus completing the proof of \Cref{clm: matching_size_upper_bound} in this case.
		
\begin{observation} \label{obs:type-1}
Every e-block $B$ of $G'_1$ contains a matching of size at least $\frac{1}{6} \bigg|\set{v\in \hat V_1 \text{ }\big|\text{ } u(v)\in V(B)}\bigg|$.
\end{observation}
\begin{proof}
Since vertices of $\hat V_1$ are smooth, they are at tree-distance at least $2h$ from all special vertices of $T$. Therefore, if $B$ contains at least one failure-certificate $u(v)$ for some vertex $v\in \hat V_1$, then $B$ must contain exactly two vertices that are incident to edges of $\delta_{G_1}(B)$, which we denote by $s_B,t_B$, and $T$ contains a shortest path in $B$ connecting $s_B$ to $t_B$, which we denote by $Q_B$.

Denote $S=\set{v\in \hat V_1 \text{ }\big|\text{ } u(v)\in V(B)}$ and $x=|S|$.
If $|E(Q_B)|\ge x/3$, then $Q_B$ contains a matching of size at least $x/6$.
Assume now that $|E(Q_B)|< x/3$. 
Since $\hat V_1\subseteq V(T)$, at least $2x/3$ vertices of $S$ do not belong to $B$. 
Therefore, at least one of the two subtrees obtained by deleting all edges of $Q_B$ from $T$ contains at least $x/3$ vertices of $R$.
Assume without loss of generality that the subtree containing vertex $s_B$, that we denote by $T_B$, contains at least $x/3$ vertices of $S$. 
From the definition of $T$, it is easy to see that the distances in $\set{w(v,s_B)\mid v\in V(T_B)\cap \hat V_1}$ are mutually distinct, which implies that the distances in $\set{w(u(v),s_B)\mid v\in V(T_B)\cap \hat V_1}$ are also mutually distinct. Therefore, the diameter of the e-block $B$ is at least $x/3$, and it follows that $B$ contains a matching of size at least $x/6$, completing the proof of \Cref{obs:type-1}.
\end{proof}

\paragraph{Case 2.3.2. $|\hat V_2|\ge 0.2\hat \eps n$.} 
%In this case the proof of \Cref{clm:fail-upper} is similar to the proof of \Cref{clm: smooth vertex fail then upper}. 

For each vertex $v\in \hat V_2$, we designate an arbitrary vertex in the set $R(v)$ as the \emph{representative} of $v$, denoted by $r(v)$.
We use the following claim.
		
\begin{claim} \label{clm:bfs-fail-connect}
For each vertex $v \in \hat V_2$, there exists a pair $(\vin,\vout)$ of vertices in $V$, such that: (i) the tree-distance between vertices $r(\vin)$ and $v$ in $T$ is at most $h$; (ii) the tree-distance between vertices $r(\vin)$ and $r(\vout)$ in $T$ is at least $h$; and (iii) $w(\vin,\vout) \le h/500$.
\end{claim}
\begin{proof}
%Let $N_v$ be the set of vertices that are at tree-distance at most $h$ from $v$ in $T$. 
%Consider the queries being asked during $\bfs(v,\sqrt{n},1000n)$. 
For each $u\notin T^h_v$, we denote by $d_{u}$ the minimum distance between $u$ and a vertex in $T^h_v$, namely $d_u=\min_{u'\in V(T^h_v)}\set{w(u,u')}$. 
Recall that set $\qset(u)$ contains all vertices $v'$, where the distance between $v'$ and $u$ has been queried in the subroutine $\pbfs(u,h,q,\alpha)$.
As a corollary of \Cref{obs: Euler tour distance}, every pair of vertices in $\qset(u)$ are at tree-distance at least $(d_u-1)$ in $T$. Therefore, $\card{\qset(u)}\le 2\card{V(T^h_v)}/(d_u-1)$, since otherwise there exists a pair of vertices in $\qset(u)$ that are at distance at most $(d_u-2)$ in the Euler tour of $T^h_v$, and so they are at distance at most $(d_u-2)$, a contradiction.

Denote $N_v=V(T^h_v)$. From the definition of the subroutine $\pbfs(v,h,q,\alpha)$, if the distance between a pair $v'_1,v'_2$ of vertices are ever queried, then at least one of $v'_1,v'_2$ belongs to set $N_v$.
Let $N'_v$ be the set of vertices $u$ such that the tree-distance between vertices $r(u)$ and $v$ in $T$ is at most $2h$. Since $v$ is a smooth vertex, $\card{N'_v} \le 10h/\hat\eps$. %Now consider all queries being asked during the process, 
Therefore, the number of performed queries on the distances between a pair of vertices in $N'_v$ is at most $(10h/\hat\eps)^2\le 100n$. Since $\pbfs(v,h,q,\alpha)$ returns \textsf{Fail}, there exists a vertex $x\notin N'_v$ with $\card{\qset(x)} \ge (q-100n)/n \ge 2^{20}/\hat{\eps}$. From \Cref{obs: Euler tour distance}, $d_w \le \frac{2\cdot (10h/\hat{\eps})}{2^{20}/\hat{\eps}} +1 \le  h/500$. We designate $x$ as $\vout$, and define $\vin=\arg_{v'\in N_v}\set{w(v',\vout)}$.
It remain to verify the properties (i) and (ii).
On one hand, since $\vin \in N_v$ and $v$ is a smooth vertex, $r(\vin)$ and $v$ are at tree-distance at most $h$ in $T$. On the other hand, since $r(\vout)$ and $v$ are at tree-distance at least $2h$ in $T$, from the triangle inequality, $r(\vout)$ and $r(\vin)$ are at tree-distance at least $h$ in $T$.
\end{proof}

Consider now a vertex $v \in \hat V_2$, let $X'_v$ be the tree-path in $T$ connecting $\vin$ to $\vout$.
Since $v$ is smooth, vertices $\vin$ and $v$ belong to the same maximal induced subpath of $T$. We denote by $X_v$ the intersection between this maximal induced subpath and $X'_v$. Clearly, $X_v$ is a path of length at least $h$.
We use the following claim.

%\iffalse
\begin{claim}
\label{clm: selected shortcuts}
Let $P$ be a maximal induced subpath of $T$ and define $\hat V(P)= \hat V_2 \cap V(P)$.
Then there exists a subset $\bar V(P)\subseteq \hat V(P)$, such that: (i) for any pair $v,v'$ of vertices in $\bar V(P)$, the paths $X_v$ and $X_{v'}$ are disjoint; and (ii) $\sum_{v\in \bar V(P)} |V(X_v)| \ge |\hat V(P)|/5$.
\end{claim}
\begin{proof}
We iteratively construct the set $\bar V(P)$ as follows. 
Throughout, we maintain a set $Z$ of vertices of $\hat V(P)$, initialized $\emptyset$.
The algorithm proceeds in iterations, and will continue to be executed as long as there is a vertex $v\in \hat V(P)$, such that the path $X_v$ is disjoint from all paths in $\set{X_{u}\mid u\in Z}$.
We now describe an iteration. Let $\hat v$ be a vertex in $\hat V(P)$ that, among all vertices $v$ in $\hat V(P)$ such that $X_v$ is disjoint from all paths in $\set{X_{u}\mid u\in Z}$, maximizes $|X_v|$. We add vertex $\hat v$ into $Z$ and continue to the next iteration.
This finishes the description of the algorithm.
Let $\bar V(P)$ be the set $Z$ at the end of this algorithm.
It is clear from the algorithm that, for any pair $v,v'$ of vertices of $\bar V(P)$, the paths $X_v$ and $X_{v'}$ are disjoint. It remains to show that $\sum_{v\in \hat V'(P)} |V(X_v)| \ge |\hat V(P)|/5$.

We claim that, if $v=\arg_{\hat v\in \hat V(P)}\max\set{|X_{\hat v}|}$, then the number of vertices $v'$ such that $X_v\cap X_{v'}\ne \emptyset$ is at most $5|X_v|$. Note that this immediately implies that $\sum_{v\in \bar V(P)} |V(X_v)| \ge |\hat V(P)|/5$.
It remains to prove the claim.
Let $v_1$ be a vertex of $\hat V(P)$ such that $X_{\hat v} \cap X_{v} \neq \emptyset$. Since $v_1$ is a smooth vertex in $P$, $\vin_1$ also belongs to $P$. Note that $\vin_1$ is an endpoint of path $X_{v_1}$, and we denote by $x$ the other endpoint of $X_{v_1}$. 
From the optimality of $v$, $|V(X_{v_1})|\le \card{V(X_{v})}$. Since the tree-distance between $v_1$ and $X_{v_1}$ in $T$ is at most $h$, the tree-distance between $v_1$ and $X_v$ in $T$ is at most $\card{V(X_{v_1})}+h \le 2\card{V(X_{v})}$. Clearly, there are at most $5\card{X_{v}}$ vertices on $P$ that are at tree-distance at most $2\card{X_{v}}$ from $X_{v}$ in $T$. 
\end{proof}

Let $\pset$ be the set of maximal induced subpaths of $T$. 
Denote $\bar V=\bigcup_{P\in\pset} \bar V(P)$.
% where for each $P\in \pset$, the set $\bar V(P)$ is given by \Cref{clm: selected shortcuts}. We then denote $\rset=\set{X_v\mid v\in \bar V}$.
From \Cref{clm: selected shortcuts}, $|E(\rset)| \ge |\hat V|/5$, the paths of $\rset$ are mutually disjoint, and each path of $\rset$ have length at least $h$ each.
Therefore, $|\rset| \le |E(\rset)|/h$, and $\sum_{v\in \bar V}w(\vin,\vout)\le (h/500) \cdot (|E(\rset)|/h) = |E(\rset)|/500$. 
On the other hand, if we add, for each vertex $v\in \bar V$, the edge $(\vin,\vout)$ to $T$, then all edges of $E(\rset)$ are contained in some e-block in the resulting graph, and are therefore also contained in some e-block of $G_1$. 
Additionally, it is easy to see that there is a matching of size at least $|E(\rset)|/3$ that only contains edges of $E(\rset)$. Via similar arguments in the proof of \Cref{clm:tsp_upper}, we can show that 
$$\opttsp(w) \le 2n-\frac{|E(\rset)|}{3}\cdot \frac 2 3+\frac{|E(\rset)|}{500}\le 
2n-\frac{|E(\rset)|}{5}\le (2-\hat\eps/125)n.$$

%\fi

\newpage
\part{Query Algorithm II}
%Given MST part

In this part we present our second algorithm in the query model, that, given a distance oracle to a metric $w$ and an MST of the complete weighted graph with edge weights given by $w$, estimates the value of $\tsp(w)$ to within a factor of $(2-\eps_0)$, for some universal constant $\eps_0>0$, by performing $\tilde O(n^{1.5})$ queries to the oracle, thus establishing \Cref{thm: main with MST}.
We start with some additional preliminaries used in this part, then introduce some crucial subroutines, and finally present the query algorithm and its analysis in \Cref{sec: main_with_MST}.

\section{Preliminaries}

Let $H$ be a subgraph of $G$. 
%The vertices of $H$ can be partitioned into two sets: set $V_0(H)$ contains all vertices with even degree in $H$, and set $V_1(H)$ contains all vertices with odd degree in $H$.
For a set $E'$ of edges, we denote $w(E')=\sum_{e\in E'}w(e)$. For a subgraph $H'$ of $H$, we denote $w(H')=w(E(H'))$, and we denote $w(H)=w(E(H))$.

The following lemma is folklore.
\begin{lemma}
	\label{lem: TSP bounded by Eulerian Multigraph}
	Let $H$ be any connected Eulerian multigraph on $V$. Then $\tsp\le w(E(H))$.
\end{lemma}

We use the following theorem from \cite{behnezhad2021time}.

%\znote{with the reduction in A near-optimal sublinear-time algorithm for approximating the minimum vertex cover size}

\begin{theorem}[Theorem 1.4 in \cite{behnezhad2021time}\!]
\label{thm: unweighted matching estimation}
For any $\hat\eps>0$, there is an algorithm that takes an unweighted graph as input, with probability $1-1/\textnormal{poly}(n)$, computes a $(2,\hat\eps n)$-approximation to the size of maximum matching, by performing $\tilde O(n/\hat\eps^3)$ queries to the adjacency matrix of the input graph.
\end{theorem}

\section{Subroutine I: Reorganization of an Independent Set of $c$-Trees}

In this section we introduce the first crucial subroutine that will be used in the query algorithm.
At a high level, the subroutine in this section handles the problem in which we are given a set of \emph{simple} subtrees of an MST, and we want to estimate the TSP cost of a tour that only visits all vertices of the given subtrees. It turns out that the simplicity (which will be explained later) of the input subtrees can be properly utilized to reduce the number of queries needed in the TSP cost estimation task.

Let tree $T$ be rooted at a vertex $r\in V(T)$. We introduce the following definition.

\begin{definition}
We say that a subtree $F$ of $T$ is a \emph{$c$-\snfl}, for some integer $c>0$, iff tree $F$ contains at most $c$ leaves.
\end{definition}

\begin{definition}
We say that a set $\fset$ of $c$-\snfls in $T$ is \emph{independent}, iff (i) the subtrees in $\fset$ are edge-disjoint; and (ii) for every pair $v,v'$ of vertices in $V(\fset)$ that serve as non-root vertices of different subtrees of $\fset$, $v$ is neither an ancestor nor a descendant of $v'$ in $T$.
\end{definition}

\iffalse
We use the following simple claim.

\begin{claim}
Let $F$ be a $c$-\snfl and let $F'$ be a $c'$-\snfl such that $V(F)\cap V(F')\ne \emptyset$. Then $F\cup F'$ is a $(c+c')$-\snfl. 
\end{claim}
\begin{proof}
On the one hand, since $F, F'$ are subtrees of $T$ and $V(F)\cap V(F')\ne \emptyset$, $F\cup F$ is also a subtree of $T$. On the other hand, since for each vertex $v\in V(F)\cup V(F')$, $\deg_{F\cup F'}(v)=\deg_{F}(v)+\deg_{F'}(v)$, a leaf of subtree $F\cup F'$ has to be either a leaf of subtree $F$ or a leaf of subtree $F'$, so subtree $F\cup F'$ contains at most $c+c'$ leaves.
\end{proof}
\fi

%Let $V'$ be a set of vertices in $T$. We say that a tour is a \emph{$V'$-tour} iff it visits all vertices of $V'$ exactly once, and we denote by $\tsp(V')$ the optimal (minimum) cost of a $V'$-tour.

\paragraph{Skeleton of an independent set $\fset$ of $c$-\snfls.}
Let $\fset$ be an independent set of $c$-\snfls of $T$.
For each $c$-\snfl $F\in \fset$, we denote by $r_F$ the root of $F$ (recall that the root $r_F$ is the vertex of $V(F)$ that is closest to the root $r$ of $T$).
Let tree $T'$ be obtained from $T$ by contracting all edges of $E(T)\setminus E(\fset)$. 
We call tree $T'$ the \emph{skeleton} of $\fset$.
Since the set $\fset$ of $c$-\snfls is independent, it is not hard to see that, equivalently, $T'$ can be viewed as the graph obtained by taking the union of all trees in $\fset$ and then identifying all vertices of $\set{r_F}_{F\in \fset}$.
%Clearly, there is a natural mapping that maps every vertex of $V(\pset)$ to a vertex of $T'$, and it is possible that endpoints of different paths of $\pset$ are mapped to the same vertex of $T'$.
We denote $V'=V(T')$ and denote by $r'$ the vertex of $T'$ that corresponds to vertices of $\set{r_F}_{F\in \fset}$, then every vertex of $V'\setminus \set{r'}$ is also a vertex of $V(\fset)$. 

Based on the given metric $w$ on $V$, 
we define a function $w': V'\times V'\to \mathbb{R}^{\ge 0}$ as follows.
For every pair $v_1,v_2\in V'\setminus\set{r'}$, we define $w'(v_1,v_2)=w(v_1,v_2)$.
For every $v\in V'\setminus\set{r'}$, we define $w'(r',v)=w(v_F,v)$, where $F$ is the $c$-\snfl of $\fset$ that contains $v$. Note that $w'$ may not be a metric on $V'$. But since $w$ is a metric, the restriction of $w'$ on $(V'\setminus \set{r'})\times (V'\setminus \set{r'})$ is a metric. We call $w'$ the \emph{induced metric} of $w$ on $V'$.

We prove the following observation.

\begin{observation}
$T'$ is a minimum spanning tree of $w'$.
\end{observation}
\begin{proof}
Assume for contradiction that $T'$ is not a spanning tree of $w'$. Then there is a pair $v,v'\in V'$ with $(v,v')\notin E(T')$, such that $w'(v,v')<\max\set{w(e)\mid e\in P_{T'}(v,v')}$. Assume first that $v,v'\ne r$, so $w'(v,v')=w(v,v')$. Since both $v,v'$ are vertices of $V$, from the construction of $T'$, $E(P_{T'}(v,v'))\subseteq E(P_{T}(v,v'))$, and it follows that $w(v,v')<\max\set{w(e)\mid e\in P_{T}(v,v')}$, which is a contradiction to the fact that $T$ is a minimum spanning tree of $w$. Assume now that $v=r$. We denote by $P$ the path of $\pset$ that contains $v'$, and denote by $v''$ the endpoint of $P$ that is closer to $r$ in $T$. From the definition of $T'$, $w'(v,v')=w(v',v'')$. Therefore, $w(v',v'')<\max\set{w(e)\mid e\in P_{T'}(v,v')}=\max\set{w(e)\mid e\in P_{T}(v',v'')}$, a contradiction to the fact that $T$ is a minimum spanning tree of $w$.
\end{proof}

\paragraph{Special walks on a skeleton.}
A \emph{walk} on $V'$ is a sequence $\sigma=(v_1,\ldots,v_{k})$ of vertices of $V'$, where vertices in sequence $\sigma$ are allowed to repeat. 
%A walk $\sigma$ is \emph{closed} iff $v_k=v_1$. 
%A walk $\sigma$ is \emph{complete} iff every vertex of $V$ appears at least once in the sequence $\sigma$. 
We say that a walk $\sigma$ on $V'$ is a \emph{special walk}, iff $\sigma$ starts and ends at $r'$, and every vertex of $V\setminus \set{r'}$ is visited exactly once (while the vertex $r'$ can be visited an arbitrary number of times).
Let $w: V\times V\to \mathbb{R}^{\ge 0}$ be any metric on $V$, and let $w': V'\times V'\to \mathbb{R}^{\ge 0}$ be the induced metric of $w$ on $V'$. The \emph{minimum special walk cost} of $w'$, denoted by $\walkcost(w')$, is defined to be the minimum cost of any special walk on $V'$, where the cost of a walk $\sigma$ is defined to be $w'(\sigma)=\sum_{1\le i\le k-1}w'(v_i,v_{i+1})$.

The main result of this section is the following lemma, whose proof is deferred to \Cref{subsec: proof of spider_walk_estimation}.

\begin{lemma}
\label{lem: spider_walk_estimation}
There is an efficient algorithm, that takes as input a constant $0<\eps<10^{-9}$, an integer $0<c<10^3$, a metric $w$ on $V$ and a set $\fset$ of independent $c$-\snfls of $T$, either correctly reports that $\walkcost(w')\le \big(2-\frac{\eps^4}{2\log (1/\eps)}\big)\cdot \mst(w')$ where $w'$ is the induced metric of $w$ on $V'=V(T')$ and $T'$ is the skeleton of $\fset$, or correctly reports that $\walkcost(w')\ge (2-c_0\cdot c\cdot \eps)\cdot\mst(w')$, where $c_0$ is some universal constant, by performing $\tilde O(n)$ queries.
\end{lemma}

We then obtain the following crucial corollary, whose proof is deferred to \Cref{subsec: Proof of independent path reorganization}.

\begin{corollary}
\label{lem: independent path reorganization}
There is an efficient algorithm, that takes as input a constant $0<\eps<10^{-9}$, an integer $0<c<10^3$, a metric $w$ on $V$ and an independent set $\fset$ of $c$-\snfls in $T$ such that $\sum_{F\in \fset}w(F)\ge (\frac{1}{2}+\eps)\cdot\mst(w)$,
computes an $\big(2-O(\frac{\eps^4}{\log (1/\eps)})\big)$-approximation of $\tsp(w)$, by performing $\tilde O(n)$ queries.
\end{corollary}

\subsection{Proof of \Cref{lem: spider_walk_estimation}}
\label{subsec: proof of spider_walk_estimation}

In this section we provide the proof of \Cref{lem: spider_walk_estimation}.
We denote the $c$-\snfls of $\fset$ by $F_1,\ldots,F_k$, and for each $1\le i\le k$, we denote $w'_i=w'(F_i)$, so $\mst(w')=\sum_{1\le i\le k}w'_i$.

Consider a pair $F_i, F_j$ of distinct $c$-\snfls in $\fset$. Let $V^*_i$ be the set of special vertices of $F_i$, and we define set $V^*_j$ similarly. We define $\zeta_{i,j}$ to be the maximum cover advantage of any subset of edges with both endpoints in $V^*_i\cup V^*_j$ on $F_i\cup F_j$. Since $|V^*_i|, |V^*_j|\le 2c$, the value of $\zeta_{i,j}$ can be computed with $O(c^2)$ queries.

We construct a weighted graph $L$ as follows. We use a parameter $0.9<\alpha<1$ whose value will be set later. The vertex set of $L$ is $V(L)=\set{1,\ldots,k}$. For every pair $i,j\in V(L)$, set $E(L)$ contains the edge $(i,j)$ iff $\zeta_{i,j}\ge (1-\alpha)\cdot(w'_i+w'_j)$, and, if such an edge exists in $E(L)$, then it has weight $w'_i+w'_j$. Let $\maxmat(L)$ be the maximum total weight of a matching in $L$. We prove the following claims.

\begin{claim}
\label{clm: max weighted matching estimation}
There is an algorithm that performs $\tilde O(c^2n)$ queries and outputs an estimate $X$ of $\maxmat(L)$, such that 
$$X\le \maxmat(L)\le \frac{16\cdot \log (\frac{1}{1-\alpha})}{(1-\alpha)^2}\cdot X+\frac{2}{(1-\alpha)^3}\cdot\sum_{F\in \fset}\adv^*(F)+\frac{\mst(w')}{\log^{8}n}.$$ 
\end{claim}

%\znote{replace $1+\eps$ with 2}

\begin{proof}
Denote $q= \log_2 (n^2/\eps)$. For each $1\le t\le q$, we define set $V_t=\set{i\mid 2^t\le w'_i< 2^{t+1}}$. Clearly, sets $V_1,\ldots,V_{q}$ partition set $V(L)$.
Define $\ell=\ceil{\log_{2}\big(\frac{1}{(1-\alpha)^2}\big)}$. 
Let graph $L'$ be obtained from $L$ by removing all edges $(i,i')$ such that 
$\zeta_{i,i'}\ge 2\cdot \min\set{w'_{i}, w'_{i'}}+(1-\alpha)^3\cdot\max\set{w'_{i}, w'_{i'}}$.

We first prove the following observations.
\begin{observation}
Let $t,t'$ be any pair such that $t'>t+\ell$, and let $i,i'$ be any two indices such that $i\in V_t$ and $i'\in V_{t'}$. Then the edge $(i,i')$ does not belong to $L'$.
\end{observation}
\begin{proof}
Since $t'>t+\ell$, from the definition of $\ell$, $w'_{i'}>\frac{1}{(1-\alpha)^2}\cdot w'_{i}$. If $\zeta_{i,i'}\ge (1-\alpha)\cdot (w'_{i}+w'_{i'})$, then 
\[
\begin{split}
\zeta_{i,i'}
\ge  (1-\alpha)\cdot w'_{i}
= & \text{ } (1-\alpha)\cdot\bigg((1-\alpha)^2\cdot w'_{i'}+(2\alpha-\alpha^2)\cdot w'_{i'}\bigg)\\
> &\text{ } (1-\alpha)^3\cdot w'_{i'}+ (1-\alpha)\cdot(2\alpha-\alpha^2)\cdot\frac{ w'_{i}}{(1-\alpha)^2}\\
=& \text{ }
(1-\alpha)^3\cdot w'_{i'}+2 \cdot w'_{i}\\
=& \text{ }
(1-\alpha)^3\cdot \max\set{w'_{i}, w'_{i'}}+2 \cdot \min\set{w'_{i}, w'_{i'}}.
\end{split}
\]
From the definition of $L'$, the edge $(i,i')$ does not belong to $L'$.
\end{proof}
\begin{observation}
\label{obs: bad edge can be ignored}
$\maxmat(L\setminus L')\le \frac{2}{(1-\alpha)^3}\cdot\sum_{F\in \fset}\adv^*(F)$.
\end{observation}
\begin{proof}
Let $M''$ be a matching in graph $L\setminus L'$. Consider an edge $(i,j)\in M''$ and assume that $w'_i\ge w'_j$, so $\zeta_{i,j}\ge 2\cdot w'_{j}+(1-\alpha)^3\cdot w'_{i}$.
Let $E'_{i,j}$ be the set of edges with both endpoints in $V^*_i\cup V^*_j$ that achieves the maximum cover advantage on $F_i\cup F_j$.
Consider now the cover advantage of set $E'_{i,j}$ on $F_{i}$. Note that
\[w'(\cov(E'_{i,j},F_i))\ge w'(\cov(E'_{i,j},F_i\cup F_j))-w'(F_j)= \zeta_{i,j}-w'_j> (1-\alpha)^3\cdot w'_{i}.\]
Therefore, since all edges of $E'_{i,j}$ have one endpoint being a special vertex of $F_i$, we get that $\adv^*(F_i)\ge (1-\alpha)^3\cdot w'_{i}\ge \frac{(1-\alpha)^3}{2}\cdot (w'_{i}+w'_{j})$.
Altogether,
\[\sum_{(i,j)\in M''}(w'_i+w'_j)\le \frac{2}{(1-\alpha)^3}\cdot\sum_{(i,j)\in M''}\bigg(\adv^*(F_i)+\adv^*(F_j)\bigg)
\le \frac{2}{(1-\alpha)^3} \cdot\sum_{F\in \fset}\adv^*(F).\]
\end{proof}
%We now show that we can efficiently estimate the value of $\maxmat(L')$ to within a constant factor. 
%
\begin{observation}
\label{obs: pair query can be simulated}
A pair query in graph $L'$ can be simulated by $O(c^2)$ distance queries in $w$.
\end{observation}
\begin{proof}
Since for each pair $i,j\in [k]$, the value of $\zeta_{i,j}$ can be computed using $O(c^2)$ queries, and whether or not the edge $(i,j)$ belongs to graph $L'$ depends solely on the value of $\zeta_{i,j}$, the pair query $(i,j)$ for any pair $i,j\in [k]$ can be obtained by $O(c^2)$ distance queries in $w$.
\end{proof}

For each pair $1\le t,t'\le q$ such that $|t-t'|\le \ell$, we define $L'_{t,t'}$ as the unweighted graph on $V_t\cup V_{t'}$, whose edge set is $E_{L'}(V_t\cup V_{t'})$. We then apply the algorithm from \Cref{thm: unweighted matching estimation} (with $\hat \eps = 1/\log^{9}n$) to the unweighted graph $L'_{t,t'}$. Let $X_{t,t'}$ be the estimate we obtain, so $X_{t,t'}\le \maxmat(L'_{t,t'})\le 2\cdot X_{t,t'}+\hat \eps\cdot |V_t\cup V_{t'}|$. 
Since $q=O(\log n)$ and $\ell=O(1)$, combined with \Cref{obs: pair query can be simulated}, the number of queries performed by all the applications of the algorithm from \Cref{thm: unweighted matching estimation} is $\tilde O(c^2n)$.
For each $0\le r\le \ell$, we compute $Y_r=\sum_{1\le t\le q}X_{t,t+r}\cdot 2^t$. We compute $X=\max_{0\le r\le \ell}\set{Y_r/2}$ and return $X$ as the estimate of $\maxmat(L)$. It remains to show that $X$ satisfies the properties in \Cref{clm: max weighted matching estimation}.

For each pair $1\le t,t'\le q$ such that $|t-t'|\le \ell$, we let $M'_{t,t'}$ be a maximum matching in graph $L'_{t,t'}$, and denote by $M_{t,t'}$ the corresponding matching in $L$. Clearly, for each $0\le r\le \ell$, $M_r=\bigcup_{1\le t\le q}M_{t,t_r}$ is the union of two matchings in $L$. Additionally, according to the definition of the partition $(V_1,\ldots,V_{q})$ and the definition of edges of $L$, we get that $w'(M_r)\ge Y_r$.
Therefore, each $Y_r$ is a lower bound for $2\cdot\maxmat(L)$, and so $X\le \maxmat(L)$. On the other hand, let $M^*$ be a maximum weight matching in $L$. For each $0\le r\le \ell$ and for each $1\le t\le q$, we denote $M^*_{t,t+r}=M^*\cap E(V_t,V_{t+r})$. From \Cref{thm: unweighted matching estimation}, we get that \[w(M^*_{t,t+r})\le 2\cdot 2^{t+\ell}\cdot |M^*_{t,t+r}|\le  2\cdot 2^{t+\ell}\cdot (2\cdot X_{t,t+r}+\hat \eps\cdot |V_t\cup V_{t+r}|).\]
From the definition of sets $V_t$ and $V_{t+r}$,
$|V_t\cup V_{t+r}|\le \mst(w')/2^{t}$.
Therefore, combined with \Cref{obs: bad edge can be ignored}, we get that
\begin{equation*}
\begin{split}
w(M^*) & \le \frac{2}{(1-\alpha)^3}\cdot\sum_{F\in \fset}\adv^*(F)+ \sum_{0\le r\le \ell}\sum_{1\le t\le q}w(M^*_{t,t+r})\\
& \le \frac{2}{(1-\alpha)^3}\cdot\sum_{F\in \fset}\adv^*(F)+\sum_{0\le r\le \ell}\sum_{1\le t\le q}2\cdot 2^{t+\ell}\cdot \bigg(2\cdot X_{t,t+r}+\hat \eps\cdot |V_t\cup V_{t+r}|\bigg)\\
& \le \frac{2}{(1-\alpha)^3}\cdot\sum_{F\in \fset}\adv^*(F)+4\cdot\sum_{0\le r\le \ell}2^{\ell}\cdot \bigg(\bigg(\sum_{1\le t\le q}2^{t}\cdot X_{t,t+r}\bigg)+\hat \eps \cdot \mst(w')\bigg)\\
& = \frac{2}{(1-\alpha)^3}\cdot\sum_{F\in \fset}\adv^*(F)+2^{\ell+2}\cdot\bigg(\hat \eps \cdot \mst(w')+\sum_{0\le r\le \ell} Y_r\bigg)\\
& \le \frac{2}{(1-\alpha)^3}\cdot\sum_{F\in \fset}\adv^*(F)+2^{\ell+2}\cdot(\ell+1) \cdot 2X +2^{\ell+2}\cdot \hat \eps \cdot \mst(w') \\
& \le \frac{2}{(1-\alpha)^3}\cdot\sum_{F\in \fset}\adv^*(F)+\frac{16 \cdot\log \big(\frac{1}{1-\alpha}\big)}{(1-\alpha)^2}\cdot X+\frac{\mst(w')}{\log^{8}n}.
\end{split}
\end{equation*} 
This completes the proof of \Cref{clm: max weighted matching estimation}.
\end{proof}

\begin{claim}
\label{clm: walk cost upper bound by max-matching}
$\walkcost(w')\le 2\cdot \mst(w')-\frac{(1-\alpha)}{2}\cdot \maxmat(L)$.
\end{claim}

\begin{proof}
Let $M$ be the max-weight matching in $L$. We construct a special walk $\sigma$ as follows. Consider first a pair $(i,j)\in M$. From \Cref{lem: cover_advantage}, we know that there is a walk $\sigma_{i,j}$ that starts and ends at $r'$, visits all vertices of $F_i\cup F_j$ exactly once, and has cost $w'(\sigma_{i,j})\le w'(F_i)+w'(F_j)-\zeta_{i,j}/2$.
Consider now an index $i$ that does not participate in any pair in $M$, we define $\sigma_{i}$ to be the walk obtained from the Euler tour of $F_i$ by deleting the repeated appearance of all vertices. So the walk $\sigma_i$ starts and ends at $r'$, visits all vertices of $V(F_i)$ exactly once, and $w'(\sigma_{i})\le 2\cdot w'_i$.
Finally we define the walk $\sigma$ as the concatenation of all walks $\set{\sigma_{i,j}\mid (i,j)\in M}\cup \set{\sigma_{i}\mid i \text{ does not participate in any pair of } M}$.
It is not hard to verify that $\sigma$ is a special walk.
From the above discussion, we get that
\begin{equation*}
\begin{split}
w'(\sigma) & \le \sum_{(i,j)\in M} (w'_i+w'_j-\zeta_{i,j})+ \sum_{i:(i,*)\notin M}2\cdot w'_i\\
& \le 2\cdot \mst(w')-\frac{(1-\alpha)}{2}\cdot  \sum_{(i,j)\in M}(w'_i+w'_j)\le  2\cdot\mst(w')-\frac{(1-\alpha)}{2}\cdot \maxmat(L).
\end{split}
\end{equation*}
This completes the proof of \Cref{clm: walk cost upper bound by max-matching}.
\end{proof}

\begin{claim}
\label{clm: walk cost lower bound by max-matching}
$\walkcost(w')\ge \bigg(2-3c\cdot(1-\alpha)\bigg)\cdot \mst(w')-
3c\cdot
\maxmat(L)
-4c\cdot\sum_{F\in \fset} 
\adv^*(F)$.
\end{claim}

\begin{proof}
Let $\sigma$ be the minimum cost special walk on $V'$, so $w'(\sigma)=\walkcost(w')$.
Define $V^*=\big(\bigcup_{1\le i\le k}V^*_i\big)\setminus \set{r'}$ to be the set of all non-root special vertices of all subtrees in $\fset$.
%Recall that the set of leaves in spider $T'$ is $\set{v^*_1,\ldots,v^*_k}$. 
%We first convert $\sigma$ into a special walk $\sigma^*$ on $V^*\cup \set{r'}$ as follows.
%We define $\sigma^*$ to be the walk on $V^*\cup \set{r'}$ obtained from $\sigma$ by deleting all vertices of $V'\setminus (V^*\cup \set{r'})$. In order to prove \Cref{clm: walk cost lower bound by max-matching}, we will first analyze the difference between $w'(\sigma)$ and $w'(\sigma^*)$, and then prove a lower bound of $w'(\sigma^*)$.
%
Denote $V^*\setminus \set{r'}=\set{v^*_1,\ldots,v^*_s}$.
Since each vertex of $V^*$ is visited exactly once in the walk $\sigma$, up to naming, we can assume without loss of generality that the walk $\sigma$ can be partitioned into subwalks $\sigma=(v^*_1,\sigma_1,v^*_2,\sigma_2,\ldots,v^*_s,\sigma_s,v^*_1)$, where for each $1\le i\le s$, the walk $\sigma_i$ does not contain any vertex of $V^*$.

Fix an index $1\le i\le s$ and consider the subwalk $(v^*_i,\sigma_i,v^*_{i+1})$. Assume first that $r'\notin V(\sigma_i)$. From the definition of $w'$, for every pair $v,v'\in V'\setminus \set{r'}$, $w'(v,v')=w(v,v')$. Since $w$ is a metric, the total cost of the subwalk $(v^*_i,\sigma_i,v^*_{i+1})$ is at least $w'(v^*_i,v^*_{i+1})$, which is, assuming $v^*_{i}\in V^*_j$ and $v^*_{i+1}\in V^*_{j'}$, at least $\cov(e^*_i,F_j\cup F_{j'})-\zeta_{j,j'}$. Assume now that $r'\in V(\sigma_i)$. We further partition the walk $\sigma_i$ at the first and the last appearances of $r'$ as $\sigma_i=(\sigma^{a}_i,r',\sigma^{b}_i, r', \sigma^{c}_i)$ (if $r'$ appears only once in $\sigma_i$, then $\sigma^{b}_i$ is an empty subwalk). 
Assume $v^*_i$ belongs to $V^*_j$, the set of special vertices of tree $F_j$.
Note that, the subwalk $\sigma^a_i$ may or may not contain vertices outside of $F_j$, and it is not hard to see that, in either case, the $w'$-cost of the subwalk $(v^*_i,\sigma^{a}_i,r')$ is at least $w'(v^*_i,r_{F_j})-\adv^*(F_j)$, and similarly, if $v^*_{i+1}$ belongs to $V^*_{j'}$, then the $w'$-cost of the subwalk $(v^*_i,\sigma^{c}_i,r')$ is at least $w'(v^*_{i+1},r_{F_{j'}})-\adv^*(F_{j'})$.
Therefore, if we denote $e^*_i=(v^*_i,v^*_{i+1})$, then 
\[
\begin{split}
w'(v^*_i,\sigma_i,v^*_{i+1})\ge &\text{ } w'(v^*_i,r_{F_j})+w'(v^*_{i+1},r_{F_{j'}})-\adv^*(F_j)-\adv^*(F_{j'})\\
\ge &\text{ } w'(P^{F_j}_{v^*_i,r_{F_j}})+w'(P^{F_{j'}}_{v^*_{i+1},r_{F_{j'}}})-2\cdot\adv^*(F_j)-2\cdot \adv^*(F_{j'})\\
= &\text{ } \cov(e^*_i,F_j\cup F_{j'})-2\cdot\adv^*(F_j)-2\cdot \adv^*(F_{j'})\\
= &\text{ } \cov(e^*_i,T')-2\cdot\adv^*(F_j)-2\cdot \adv^*(F_{j'}).
\end{split}\]

Denote $E^*=\set{(e^*_i,e^*_{i+1})\mid i\in [s]}$. We say that an edge $(e^*_i,e^*_{i+1})$ in $E^*$ is a type-1 edge iff $r'\notin V(\sigma_i)$, otherwise we say that it is a type-2 edge. We denote by $E^*_1$ the set of all type-1 edges, and we define set $E^*_2$ similarly. For each $1\le j\le k$, we denote by $E^*_{1,j}$ the set of all type-1 edges with at least one endpoint in $V^*_j$, and we define the set $E^*_{2,j}$ similarly.

On the one hand, clearly for each $1\le j\le k$, $|E^*_{1,j}|\le |V^*_j|\le 2c$. Therefore, from Vizing's theorem \cite{vizing1964estimate}, set $E^*_1$ can be partitioned into at most $2c+1$ subsets, such that in each of these subsets, no two edges have endpoint in the same set of $\set{V^*_1,\ldots,V^*_k}$.
On the other hand, clearly for each $1\le j\le k$, $|E^*_{1,j}|\le |V^*_j|\le 2c$. Altogether, if we define, for each $1\le i\le s$, $j(i)$ to be the index $j\in [k]$ such that $v^*_i\in V^*_j$, then
\[
\begin{split}
\walkcost(w')\ge &\text{ }
\bigg(2-3c\cdot(1-\alpha)\bigg)\cdot \mst(w')-
3c\cdot
\maxmat(L)
-4c\cdot\sum_{F\in \fset} 
\adv^*(F)\\
\ge &\text{ }
2\cdot \mst(w')-\sum_{i: e^*_i\in E^*_1} 
\zeta_{j(i),j(i+1)}
-
\sum_{i: e^*_i\in E^*_2} 
\bigg(2\cdot\adv^*(F_{j(i)})+2\cdot \adv^*(F_{j(i+1)})
\bigg)\\
\ge &\text{ }
2\cdot \mst(w')-\sum_{i: e^*_i\in E^*_1} 
\zeta_{j(i),j(i+1)}
-2\cdot 2c\cdot \sum_{F\in \fset} 
\adv^*(F)\\
\ge &\text{ }
2\cdot \mst(w')-
(2c+1)\cdot 
\bigg(\maxmat(L)+(1-\alpha)\cdot\mst(w')\bigg)
-4c\cdot\sum_{F\in \fset} 
\adv^*(F)\\
= &\text{ }
\bigg(2-(2c+1)(1-\alpha)\bigg)\cdot \mst(w')-
(2c+1)\cdot 
\maxmat(L)
-4c\cdot\sum_{F\in \fset} 
\adv^*(F)\\
= &\text{ }
\bigg(2-3c\cdot(1-\alpha)\bigg)\cdot \mst(w')-
3c\cdot 
\maxmat(L)
-4c\cdot\sum_{F\in \fset} 
\adv^*(F).
\end{split}
\]
This completes the proof of \Cref{clm: walk cost lower bound by max-matching}.
\end{proof}

\paragraph{Algorithm.}
We now describe the algorithm for \Cref{lem: spider_walk_estimation}.
We use the parameters
$\alpha=1-\eps$ and set $c_0=100$.

%\znote{To reorganize from here}

We first apply the algorithm from \Cref{lem: estimate special adv of a set of segments} to tree $T$ and the independent set $\fset$ of $c$-\snfls with parameter $2\eps^4$. If the algorithm from \Cref{lem: estimate special adv of a set of segments} reports that $\sum_{P\in \pset}\adv^*(P)>2\eps^4 \cdot \mst(w')$, then we report that $\walkcost(w')\le \big(2-\frac{\eps^4}{2\log (1/\eps)}\big)\cdot \mst(w')$. Otherwise, 
%we know that $\sum_{P\in \pset}\adv^*(P)<4\eps^4\cdot \mst(w')$, and 
we then construct the graph $L$ as described before and apply the algorithm from \Cref{clm: max weighted matching estimation} to obtain an estimate $X$ of  $\maxmat(L)$.
If $X\ge \frac{\eps^3}{\log (1/\eps)}\cdot \mst(w')$, we report that $\walkcost(w')\le\big(2-\frac{\eps^4}{2\log (1/\eps)}\big)\cdot \mst(w')$. 
Otherwise, we report that $\walkcost(w')\le (2-c_0\cdot \eps)\cdot \mst(w')$.

\paragraph{Proof of Correctness.}
We now prove the correctness of the algorithm.
Since the algorithm calls the algorithm from \Cref{lem: estimate special adv of a set of segments} once and the algorithm from \Cref{clm: max weighted matching estimation} once, it performs $\tilde O(c^2n)$ queries.
%Let $X$ be the estimate returned by the algorithm. We now show that $\tsp\le X\le (2-\Omega(\eps))\cdot \tsp$.

Assume first that the algorithm from \Cref{lem: estimate special adv of a set of segments} reports that $\sum_{F\in \fset}\adv^*(F)>2\eps^4 \cdot \mst(w')$.
Then from similar arguments in the proof of \Cref{lem: cover_advantage}, it is easy to show that 
$$\walkcost(w')\le 2\cdot\mst(w')-\frac{1}{2}\cdot \sum_{F\in \fset}\adv^*(F)\le (2-\eps^4) \cdot \mst(w'),$$
so our report that $\walkcost(w')\le \big(2-\frac{\eps^4}{2\cdot\log (1/\eps)}\big)\cdot \mst(w')$ is correct in this case. We assume now that the algorithm from \Cref{lem: estimate special adv of a set of segments} reports that $\sum_{P\in \pset}\adv^*(P)< 4\eps^4 \cdot \mst$. In this case, assume first that the estimate $X$ given by algorithm from \Cref{clm: max weighted matching estimation} satisfies that $X\ge \frac{\eps^3}{\log (1/\eps)}\cdot \mst(w')$. Then from \Cref{clm: walk cost upper bound by max-matching} and \Cref{clm: max weighted matching estimation},
\[
\walkcost(w')\le 2\cdot \mst(w')-\frac{(1-\alpha)}{2}\cdot \maxmat(L)\le 2\cdot \mst(w')-\frac{\eps}{2}\cdot X
\le \bigg(2-\frac{\eps^4}{2\cdot\log (1/\eps)}\bigg) \cdot \mst(w').
\]
so our report that $\walkcost(w')\le \big(2-\frac{\eps^4}{2\log (1/\eps)}\big)\cdot \mst(w')$ is correct.
Assume now $X< \frac{\eps^3}{\log (1/\eps)}\cdot \mst(w')$. Then from \Cref{clm: walk cost lower bound by max-matching} and \Cref{clm: max weighted matching estimation}, 
\iffalse
\begin{equation*}
\begin{split}
\walkcost(w',r') &\ge 2\alpha\cdot (\mst(w')-\maxmat(L))- 2\cdot \sum_{P\in \pset}\adv^*(P)\\
& = 2\alpha\cdot \mst(w')- 2\alpha\cdot \maxmat(L)-2\cdot \sum_{P\in \pset}\adv^*(P)\\
& \ge 2\alpha\cdot \mst(w')- 2\alpha\cdot \bigg( \frac{4C^* \cdot \log (\frac{1}{1-\alpha})}{(1-\alpha)^2}\cdot X+\frac{8}{3(1-\alpha)}\cdot\sum_{P\in \pset}\adv^*(P)\bigg)-8\eps^2\cdot \mst(w')\\
& \ge 2\cdot\mst(w')-2\eps\cdot \mst(w')- 2\cdot \bigg(\frac{4C^* \cdot \log (1/\eps)}{\eps^2}\cdot X+\frac{32\eps^2\cdot\mst(w')}{3\eps}\bigg)-8\eps^2\cdot \mst(w')\\
& \ge 2\cdot\mst(w')-\mst(w')\cdot \left(2\eps+ 8C^*\cdot \eps+\frac{32}{3}\cdot \eps+8\eps^2\right)\\
& \ge (2-c_0\cdot \eps)\cdot \mst(w'),
\end{split}
\end{equation*}
\fi
%
\begin{equation*}
\begin{split}
\walkcost(w')\ge &\text{ }
\bigg(2-3c\cdot(1-\alpha)\bigg)\cdot \mst(w')-
3c\cdot
\maxmat(L)
-4c\cdot\sum_{F\in \fset} 
\adv^*(F)\\
\ge &\text{ }
(2-3c\cdot \eps)\cdot \mst(w')-
3c\cdot
\maxmat(L)
-4c\cdot\sum_{F\in \fset} 
\adv^*(F)\\
\ge &\text{ }
(2-3c\cdot \eps)\cdot \mst(w')-
3c\cdot
\maxmat(L)
-4c\cdot 4\eps^4\cdot\mst(w')\\
\ge &\text{ }
(2-4c\cdot\eps)\cdot \mst(w')-
3c\cdot
\maxmat(L)\\
\ge &\text{ }
(2-4c\cdot \eps)\cdot \mst(w')-
3c\cdot
\bigg(\frac{16\cdot \log (\frac{1}{1-\alpha})}{(1-\alpha)^2}\cdot X+\frac{2}{(1-\alpha)^3}\cdot\sum_{F\in \fset}\adv^*(F)+\frac{\mst(w')}{\log^8 n}\bigg)\\
\ge &\text{ }
(2-5c\cdot \eps)\cdot \mst(w')-
3c\cdot
\bigg(\frac{16\cdot \log (1/\eps)}{\eps^2}\cdot X+\frac{2}{\eps^3}\cdot\sum_{F\in \fset}\adv^*(F)\bigg)\\
\ge &\text{ }
(2-5c\cdot \eps)\cdot \mst(w')-
3c\cdot
\frac{16\cdot \log (1/\eps)}{\eps^2}\cdot X
-24\eps\cdot\mst(w')\\
\ge &\text{ }
(2-29c\cdot \eps)\cdot \mst(w')-
3c\cdot
\frac{16\cdot \log (1/\eps)}{\eps^2}\cdot \frac{\eps^3}{\log (1/\eps)}\cdot \mst(w')\\
\ge &\text{ }
\bigg(2-\big(29c+ 48c\big)\cdot \eps\bigg)\cdot \mst(w')\ge (2-c_0\cdot c\cdot \eps)\cdot \mst(w').
\end{split}
\end{equation*}
where the last step is due to the definition of $c_0$, and so our report that $\walkcost(w',r')>(2-c_0\cdot\eps) \cdot \mst(w')$ is correct in this case.
This completes the proof of \Cref{lem: spider_walk_estimation}.

\subsection{Proof of \Cref{lem: independent path reorganization}}
\label{subsec: Proof of independent path reorganization}

In this subsection we provide the proof of \Cref{lem: independent path reorganization}.
%using \Cref{lem: estimate adv of a set of segments}, \Cref{clm: tsp lower bound by walkcost}, \Cref{clm: tsp upper bound by walkcost} and \Cref{lem: spider_walk_estimation}.
%
We first prove the following claims.

\begin{claim}
	\label{clm: tsp lower bound by walkcost}
	$\tsp(w)\ge \walkcost(w')- 2\cdot\sum_{F\in \fset}\adv^*(F)$.
\end{claim}
\begin{proof}
Let $\pi$ be an optimal TSP-tour on $V$, so $w(\pi)=\tsp$. We will convert tour $\pi$ into a special walk on $V'$ as follows. 
%Recall that, for each path $P$ in $\fset$, $v_P$ is the closer-to-$r$ endpoint of $P$, and $v'_P$ is the other endpoint of $P$. Let $T'_0,T'_1,T'_2,\ldots,T'_p$ be the connected components obtained from $T$ by deleting all edges and internal vertices of the paths in $\fset$, where $T'_0$ is the connected component that contains all vertices of $\set{v_P}_{P\in \fset}$. We denote $V'_0=V(T'_0)$, $V^*=V'\setminus \set{r'}$ and $U=\bigcup_{1\le j\le p}V(T'_j)$, so set $V$ is partitioned by sets $V'_0, V^*, U$.
%
Let $\uset$ be the set of connected components of $T\setminus \big(\bigcup_{F\in \fset}F\big)$. It is easy to see that there is a unique connected component $U_0\in \uset$ that contains all vertices of $\set{r_F\mid F\in \fset}$.

Recall that each vertex of $V'\setminus\set{r'}$ is a vertex of $V$.
We denote $V'\setminus\set{r'}=\set{v'_1,\ldots,v'_s}$, where the vertices are indexed according to their appearance in $\pi$.
%Assume without loss of generality that $\pi$ starts and ends at a vertex of $V'\setminus\set{r'}$. 
We first partition the tour $\pi$ at vertices of $V'$ into $\pi=(v'_1,\sigma_1,v'_2,\sigma_2,\ldots,\sigma_{s},v'_1)$, where each $\sigma_i$ is a simple walk in $V(\uset)$.
%So $V'=\set{v'_1,v'_2,\ldots,v'_{k-1}}$.
Assume without loss of generality (up to naming on vertices of $V'\setminus \set{r'}$) that $\sigma_{1}$ contains a vertex of $U_0$. 
We then construct a walk $\sigma$ as $\sigma=(v'_1,r',v'_2,v'_3,\ldots,v'_{s},v'_1)$. Clearly, $\sigma$ is a special walk on $V'$. 
	
We now show that $w'(\sigma)\le \tsp+2\cdot\sum_{F\in \fset}\adv^*(F)$. On the one hand, for each $2\le i\le s$, from the definition of $w'$ and the triangle inequality of metric $w$, $w'(v'_i,v'_{i+1})$ is at most the total $w$-cost of the walk $(v'_i,\sigma_i,v'_{i+1})$. On the other hand, let $v'$ be an arbitrary vertex of $V(U_0)\cap V(\sigma_1)$, then the total $w$-cost of the walk $(v'_1,\sigma_1,v'_{2})$ is at least $w(v'_1,v')+w(v', v'_2)$.
Assume that $v'_1\in F_1$ and $v'_2\in F_2$.
Then from the definition of $w'$, $w'(v'_1,r')+w'(r', v'_2)=w(v'_1,r_{F_1})+w(v'_2,r_{F_2})$.

Now if we consider the edge $e'_1=(v'_1,v')$, from \Cref{lem: single_edge_adv}, we get that
\[\adv^*(F_1)\ge \adv(e'_1,F_1)\ge w(P^{T}_{v'_1,r_{F_1}})-w(v'_1,v')\ge w'(r',v'_1)-w(v'_1,v'),\]
and similarly, if we denote $e'_2=(v'_1,v')$, then
\[\adv^*(F_2)\ge \adv(e'_2,F_2)\ge w(P^{T}_{v'_2,r_{F_2}})-w(v'_2,v')\ge w'(r',v'_2)-w(v'_2,v').\]
%then it is easy to verify that $\adv(E'_1,F_1)=w(v'_1,u_1)-w(v'_1,v')$ and therefore $\adv'(P_1)\ge w(v'_1,u_1)-w(v'_1,v')$ (since the edge $(v'_1,v')$ has an endpoint $v'$, which is a special vertex of the path in $\fset$ that contains $v'_1$). Similarly, we can derive that $\adv'(P_2)\ge w(v'_2,u_2)-w(v'_2,v')$. 
Since $\adv^*(F_1)+\adv^*(F_2)\le 2\cdot\sum_{F\in \fset}\adv^*(F)$, we get that
\[\bigg(w(v'_1,v')+w(v', v'_2)\bigg)-\bigg(w'(v'_1,r')+w'(v'_2,r')\bigg)\le 2\cdot\sum_{F\in \fset}\adv^*(F).\]
Altogether, we get that
\[\walkcost(w')\le w'(\sigma)\le \tsp -  (w(v'_1,v')+w(v', v'_2))+w'(v'_1,r')+w'(v'_2,r')\le \tsp+2\cdot\sum_{F\in \fset}\adv^*(F).\]
This completes the proof of \Cref{clm: tsp lower bound by walkcost}.
\end{proof}

\begin{claim}
	\label{clm: tsp upper bound by walkcost}
	$\tsp(w)\le \walkcost(w')+2\cdot w(E(T)\setminus E(\fset))$.
\end{claim}

\begin{proof}
Let $\sigma$ be the special walk on $V'$ with minimum cost. We first partition $\sigma$ at every appearance of $r'$ into $\sigma=(r',\sigma_1,r',\sigma_2,r'\ldots,r',\sigma_s,r')$. From the definition of a special walk, each $\sigma_i$ is a walk in $V'$ that does not contain $r'$, and the sets $\set{V(\sigma_i)}_{1\le i\le s}$ partitions $V'\setminus \set{r'}$.
Recall that every vertex in $V'\setminus \set{r'}$ is also a vertex of $V'$. 
	
We construct a Eulerian multigraph $H$ on $V$ as follows. 
Similar as the proof of \Cref{clm: tsp lower bound by walkcost}, we let $\uset$ be the set of connected components obtained from $T$ by deleting all subtrees of $\fset$, and let $U_0$ be the unique connected component in $\uset$ that contains all vertices of $\set{r_F\mid F\in \fset}$. So $E(T)=E(\fset)\cup(\bigcup_{U\in \uset}E(U))$. 

For each $U\in \uset\setminus \set{U_0}$, we denote by $H_U$ the multigraph on $V(U)$ that contains two copies of every edge of $E(U)$, so $w(H_U)=2\cdot w(U)$. 
We define multigraph $H_{U_0}$ as follows.
For each $1\le i\le s$, we denote by $v_i$ ($v'_i$, resp.) the first (last, resp.) vertex of $\sigma_i$, and we denote by $u_i$ ($u'_i$, resp.) the closer-to-$r$ endpoint of the subtree in $\fset$ that contains $v_i$ ($v'_i$, resp.). %We denote $V''=\set{u_i,u'_i\mid 1\le i\le s}$, and define set 
We denote $M=\set{(u_i,u'_i)\mid 1\le i\le s}$.
We set $H_{U_0}=(V(U_0),E_{[U_0,M]})$ (see definition in \Cref{subsec: cover adv}), so $w(H_{U_0})\le 2\cdot w(U_0)$.
Then for each $1\le i\le s$, we define the $L_i$ to be the graph induced by edges of $E(\sigma_i)\cup\set{(v_i,u_i),(v'_i,u'_i)}$. So for each $1\le i\le s$, $w(L_i)$ equals the $w'$-cost of the subwalk $(r',\sigma_i,r')$ of $\sigma$, and it follows that $\sum_{1\le i\le s}w(L_i)=\walkcost(w')$.
Finally, we define graph
$H=(\bigcup_{1\le i\le s}L_i)\cup (\bigcup_{U\in \uset}H_U)$.
It is not hard to verify that $H$ is a connected Eulerian graph on $V$. 
From \Cref{lem: TSP bounded by Eulerian Multigraph}, $\tsp\le w(H)$.
	
On the other hand, 
$$w(H)= \sum_{1\le i\le s}w(L_i)+\sum_{U\in \uset}w(H_U)\le 
\walkcost(w')+2\cdot \sum_{U\in \uset}w(U)=\walkcost(w')+2\cdot w(E(T)\setminus E(\fset)).$$
Altogether, we get that $\tsp\le \walkcost(w')+2\cdot w(E(T)\setminus E(\fset))$.
\end{proof}

We now describe the algorithm for \Cref{lem: independent path reorganization}.
We use parameters $\eps_1=\eps/4$ and $\eps_2=\eps_0/(4\cdot c\cdot c_0)$, so $\eps_1,\eps_2=\Omega(\eps)$.

We first apply the algorithm from \Cref{lem: estimate special adv of a set of segments} to tree $T$ and the independent set $\fset$ of $c$-\snfls in $T$, with parameter $\eps_1$. If it reports that $\sum_{F\in \fset}\adv^*(F)>\eps_1 \cdot \mst(w)$, then we return $(2-\eps_1/2) \cdot \mst$ as an estimate of $\tsp$. Otherwise, we continue to apply the algorithm from \Cref{lem: spider_walk_estimation} to metric $w$, the set independent set $\fset$ of $c$-\snfls in $T$ with parameter $\eps_2$. If the algorithm reports that $\walkcost(w')\le (2-\eps^4_2/2\log (1/\eps_2))\cdot\mst(w')$, then we return $(2-\eps^4_2/4\log (1/\eps_2)) \cdot \mst$ as an estimate of $\tsp$. Otherwise, we return $2\cdot \mst$.

We now prove the correctness of the above algorithm.
First, since the algorithm calls the algorithm from \Cref{lem: estimate special adv of a set of segments} once and the algorithm from \Cref{lem: spider_walk_estimation} at most once, it performs $\tilde O(n)$ queries.
Let $X$ be the estimate returned by the algorithm. We now show that $\tsp\le X\le (2-\Omega(\eps^4/\log (1/\eps)))\cdot \tsp$.

Assume first that the algorithm from \Cref{lem: estimate special adv of a set of segments} reports that $\sum_{F\in \fset}\adv^*(F)>\eps_1 \cdot \mst$.
Then from \Cref{lem: cover_advantage} and the fact that $\adv(F)\ge \adv^*(F)$ for every tree $F\in \fset$, 
$$\tsp\le 2\cdot\mst-\frac{1}{2}\cdot \sum_{F\in \fset}\adv(F)\le 2\cdot\mst-\frac{1}{2}\cdot \sum_{F\in \fset}\adv^*(F)\le \bigg(2-\frac{\eps_1}{2}\bigg) \cdot \mst,$$ 
so our estimate $X=(2-\eps_1/2) \cdot \mst$ in this case is a $(2-\Omega(\eps))$-approximation of $\tsp$. Assume now that the algorithm from \Cref{lem: estimate special adv of a set of segments} reports that $\sum_{F\in \fset}\adv^*(F)< 2\eps_1 \cdot \mst$. In this case, assume first that the algorithm  from \Cref{lem: spider_walk_estimation} reports that $\walkcost(w')\le (2-\eps^4_2/2\log (1/\eps_2))\cdot\mst(w')$. Then from \Cref{clm: tsp upper bound by walkcost}, 
\begin{equation*}
\begin{split}
\tsp &\le \walkcost(w',r')+2\cdot w(E(T)\setminus E(\fset))\\
& \le \bigg(2-\frac{\eps^4_2}{2\log (1/\eps_2)}\bigg)\cdot \bigg(\frac{1}{2}+\eps\bigg)\cdot\mst+ 2\cdot \bigg(\frac{1}{2}-\eps\bigg)\cdot\mst \le \bigg(2-\frac{\eps^4_2}{4\log (1/\eps_2)}\bigg)\cdot\mst,
\end{split}
\end{equation*}
so our estimate $X=(2-\eps^4_2/4\log (1/\eps_2))\cdot \mst$ in this case is a $(2-\Omega(\eps^4/\log (1/\eps)))$-approximation of $\tsp$.
Assume now that the algorithm from \Cref{lem: spider_walk_estimation} reports that $\walkcost(w')\ge (2-c_0\cdot c\cdot \eps_2)\cdot\mst(w')$. Then from \Cref{clm: tsp lower bound by walkcost}, 
\begin{equation*}
\begin{split}
\tsp & \ge \walkcost(w')- 2\cdot\sum_{F\in \fset}\adv^*(F)\\
& \ge (2-c_0\cdot c\cdot \eps_2)\cdot\mst(w')- 2\eps_1\cdot \mst\\
& \ge (2-c_0\cdot c\cdot  \eps_2)\cdot\bigg(\frac{1}{2}+\eps\bigg)\cdot \mst- 2\eps_1\cdot \mst\\
& \ge \bigg(1+2\eps-\frac{c_0\cdot c\cdot  \eps_2}{2}-c_0\cdot c\cdot  \eps_2\cdot\eps-2\eps_1\bigg)\cdot\mst \ge (1+\eps)\cdot\mst.
\end{split}
\end{equation*}
and so our estimate $X=2 \cdot \mst$ in this case is a $\big(2-\Omega\big(\eps^4/\log (1/\eps)\big)\big)$-approximation of $\tsp$.
This completes the proof of \Cref{lem: independent path reorganization}.

\section{Subroutine II: Light Subtrees, Segments, and Extensions}
\label{sec: light subtrees}

In this section, we introduce the second main subroutine that will be used in the query algorithm. Intuitively, the subroutine in this section partitions the tree into top and bottom parts, similar to the techniques used in the first query algorithm in Part II. In addition, we also further partition the top part of the MST into segments, that will be used later as basic objects to compute cover advantage on.

\begin{definition}[Maximal Light Vertices]
	We say that a vertex $v$ is \emph{maximal $\ell$-light} for some integer $\ell>0$, iff the subtree $T_v$ contains at most $\ell$ vertices, and the subtree $T_{v'}$ contains more than $\ell$ vertices, where $v'$ is the parent node of $v$ in $T$.
\end{definition}

We denote by $L_{\ell}$ the set of all maximal $\ell$-light vertices.
Let $v$ be a maximal $\ell$-light vertex and let $v'$ be the parent vertex of $v$ in $T$. We denote $T^+_v=T_{v}\cup \set{(v,v')}$, so $T^+_v$ is a subtree of $T$ (but note that $T^+_v$ is not equal to $T_{v'}$ since $v'$ may have other child nodes in $T$). We call $T^+_v$ the \emph{$\ell$-light subtree} of $T$ at $v$. We denote $w_{\lig}=\sum_{v\in L_{\ell}}w(T^+_v)$.

Next, we compute a partition of the tree $T$ into segments (which we will define later), as follows.
Let $T'$ be the tree obtained from $T$ by deleting from it, for each vertex $v\in L_{\ell}$, all edges and vertices (except for the root) of $T^+_v$. We use the following observation.

\begin{observation}
	Tree $T'$ has at most $n/\ell$ leaves.
\end{observation}
\begin{proof}
	Let $v$ be a leaf of $T'$. Then clearly $|V(T_v)|\ge \ell$, since otherwise $v$ is an $\ell$-light vertex and should not belong to $T'$. On the other hand, let $v,v'$ be any pair of leaves in $T'$, then subtrees $T_v, T_{v'}$ are vertex-disjoint. It follows that $T'$ has at most $n/\ell$ leaves.
\end{proof}

For every vertex $v\in V(T')$, we denote by $X_{\ell}(v)$ the set of all maximal $\ell$-light child vertices of $v$ in $T$, and denote by $\tset_{\ell}(v)$ the set of all $\ell$-light subtrees rooted at vertices of $X_{\ell}(v)$. For each $v\in V(T')$, we denote $n_{\ell}(v)=\sum_{u\in X_{\ell}(v)}|V(T_u)|$, and for each path $P$ in $T'$, we denote $n_{\ell}(P)=\sum_{v\in V(P)}n_{\ell}(v)$.
We prove the following lemma.

\begin{lemma}
\label{lem: partitioning into segments}
We can efficiently partition $T'$ into a set $\pset$ of $O(n/\ell)$ vertex-disjoint paths, such that, for each path $P\in \pset$, either $P$ contains a single-vertex of $T$, or satisfies that $n_{\ell}(P)\le 11\ell$.
\end{lemma}
\begin{proof}
We say that a vertex $v$ of $T'$ is \emph{giant} iff $n_{\ell}(v)\ge 10\ell$. Clearly, the number of giant vertices in $T'$ is at most $n/(10\ell)$.
For each vertex $v\in T'$ that is either a special vertex or a giant vertex $v$, we add a path into $\pset$ that contains the single vertex $v$.
We then delete from $T'$ all special vertices and all giant vertices. It is easy to see that the remaining subgraph of $T'$ is the union of at most $O(n/\ell)$ vertex-disjoint subpaths of $T'$, that we denote by $\qset$. We further partition each path of $\qset$ as follows. Let $Q=(v_1,\ldots,v_r)$ be a path of $\qset$, we compute a sequence of indices $i_0=0< i_1 <\ldots< i_k<r$, such that, for each $1\le j\le k$, $\ell\le \sum_{i_{j-1}+1\le t\le i_j}n_{\ell}(v_t)\le 11\ell$, so $k\le n_{\ell}(Q)/\ell$. It is easy to see that this can be done since $Q$ does not contain giant vertices. We then define, for each $1\le j\le k$, path $P_j=(v_{i_{j-1}+1},\ldots, v_{i_j})$, and we define $P_{k+1}=(v_{i_{k}+1},\ldots, v_{r})$. Clearly, paths $P_1,\ldots,P_{k+1}$ partition $Q$, and we add them into set $\pset$. Altogether, we have added a total of $$O\bigg(\frac n \ell\bigg)+\sum_{Q\in \qset} \bigg(\frac{n_{\ell}(Q)}{\ell}+1\bigg)\le O\bigg(\frac n \ell\bigg)+|\qset|+\bigg(\frac{\sum_{Q\in \qset}n_{\ell}(Q)}{\ell}\bigg)=O\bigg(\frac n \ell\bigg)$$ 
paths into $\pset$, and each path $P$ of $\pset$ is either a single-vertex path, or satisfies that $n_{\ell}(P)\le 11\ell$.
\end{proof}

\paragraph{Segments of $T$.}
We apply the algorithm from \Cref{lem: partitioning into segments} to tree $T$. Let $\pset$ be the resulting set of vertex-disjoint paths of $T$.
For each path $P\in \pset$, we denote by $T_P$ the subtree of $T$ that contains (i) all vertices of $P$; and (ii) for each vertex $v\in P$, the set of all $\ell$-light subtrees rooted at vertices of $X_{\ell}(v)$. In other words, $T_P$ is the subgraph of $T$ induced by all vertices of $V(P)\cup (\bigcup_{v\in P}V(\tset_{\ell}(v)))$.
We call each subtree $T_P$ a \emph{segment} of tree $T$. For each vertex $v\in T'$ that is either a special vertex or a giant vertex, we also call each $\ell$-light subtree in $\tset_{\ell}(v)$ a \emph{segment} of $T$. Therefore, if we denote by $V'$ the set of all special and giant vertices in $T'$, then the set of segments of $T$ is $\sset=\set{T_P\mid P\in \pset}\cup \big(\bigcup_{v\in V'}\tset_{\ell}(v)\big)$. Clearly, the segments are edge-disjoint subgraphs of $T$, and each segment contains $O(\ell)$ vertices.

\paragraph{Extension of $T'$.} Let $v$ be a vertex of $T'$. Recall that $X_{\ell}(v)$ is the set of all $\ell$-light child vertices of $v$, and $\tset_{\ell}(v)$ is the set of $\ell$-light subtrees rooted at vertices of $X_{\ell}(v)$. We say that a path $Q$ in $T$ is \emph{$v$-nice}, iff one endpoint of $Q$ is $v$ and the other endpoint of $Q$ is a vertex of $V(\tset_{\ell}(v))$.
In other words, $Q$ is $v$-nice iff it is a nice path of $T$ with $v$ as its close-to-$r$ endpoint. 
%
\iffalse{original "constant-branch" definition of extensions}
For a path $P\in \pset$, we say that a path $Q$ in $T$ is \emph{$P$-nice}, iff $Q$ is $v$-nice for some $v\in V(P)$.
We say that a subgraph $F$ of $T$ is a $k$-\emph{extension} of $T$, iff (i) $F=\big(\bigcup_{P\in \pset}F_P\big)\cup \big(\bigcup_{v\in V'}F_v\big)$, where for each $v\in V'$, $F_v$ is a subtree of $T_{v}$ obtained by taking the union of at most $4$ $v$-nice paths, and for each $P\in \pset$, $F_P$ is a forest of $T_P$ obtained by taking the union of at most $4$ $P$-nice paths; and (ii) $F$ contains at most $k$ leaves of $T$.
We say that $F$ is a maximum $k$-extension, iff $F$ is the $k$-extension with maximum total weight $w(F)$. The following observation is immediate.
\fi
%
We say that a subgraph (not necessarily connected) $F$ of $T\setminus T'$ is a $k$-\emph{extension} of $T'$, iff $F=\bigcup_{1\le i\le k}F_i$, where each $F_i$ is a $v$-nice path for some $v\in V(T')$.
We say that $F$ is a maximum $k$-extension of $T'$, iff $F$ is the $k$-extension of $T'$ with maximum total weight $w(F)$. 
We denote by $w^{\sf ext}_k(T')$ the weight of a maximum $k$-extension.
The following observation is immediate.

\begin{observation}
\label{obs: compute k-extension}
Given any integer $k$, we can efficiently compute $w^{\sf ext}_k(T')$ and a $k$-extension $F$ of $T'$ that achieves the weight $w^{\sf ext}_k(T')$.
\end{observation}

\section{The Main Algorithm and its Analysis} 
\label{sec: main_with_MST}
%Let $T'$ be the tree obtained from $T$ by deleting from it, for each vertex $v\in L_{\ell}$, all edges and vertices (except for the root of) $T^+_v$. 

In this section we provide the main algorithm in this part and then complete the proof of \Cref{thm: main with MST}.
At a high level, our algorithm computes the cover advantage of the the top part and the bottom part of the MST separately. On the one hand, the cover advantage of the bottom parts (light subtrees) can be estimated efficiently in an exhaustive way. On the other hand, the top part is a subtree that contains few special vertices, so its special cover advantage can be estimated efficiently. Finally, for the cover advantage achieved by edges with an endpoint between different light subtrees, we utilize the first subroutine to estimate it. We then carefully combine the estimates to compute the final output.

Parameters: $\ell=\sqrt{n}, \epsilon=2^{-100}/c_0$, where $c_0$ is the constant in \Cref{lem: spider_walk_estimation}.

\begin{enumerate}
\item \label{WithMST_step_1}
We first apply the algorithm from \Cref{lem: max special cover advantage} to the subtree $T'$ of $T$ that is defined in \Cref{sec: light subtrees}, and compute the value of $w(T')$. If $\adv^*(T')\ge (\eps/10)\cdot \mst$, then we return $X=(2-\eps/20)\cdot \mst$ as the estimate of $\tsp$.
If $\adv^*(T')\le  (\eps/10)\cdot \mst$ and $w(T')\ge (1/2+\eps)\cdot \mst$, then we return $X=2\cdot \mst$ as the estimate of $\tsp$. Otherwise, go to \ref{WithMST_step_2}.
%We first query all distances between pairs $u,u'$ of vertices such that $u$ is a special vertex of $T'$, and $u'\in V$. Then, using the information acquired, we compute a set $E'$ of edges that, over all sets $\hat E$ satisfying that (i) $E(T')\subseteq \cov(\hat E)$ and (ii) for each edge $e\in \hat E$, at least one endpoint of $e$ is a special vertex of $T'$, minimizes the total cost. If $w(T')\ge c_1\cdot \mst$ and $w(E')\le (1-\eps)\cdot w(T')$, then we return $(2-c_1\eps/2)\cdot \mst$ as the estimate of $\tsp$; otherwise go to \ref{WithMST_step_2}. \znote{or just the min-cost odd matching?}
\item \label{WithMST_step_2}
We compute a maximum weight independent set $\fset$ of $800$-\snfls in $T\setminus T'$ (note that this only uses edge weights of $T$ and does not require performing additinal queries). If $w(\fset)\ge (1/2+\eps)\cdot\mst$, then we apply the algorithm from \Cref{lem: independent path reorganization} to $\fset$, and return the output of \Cref{lem: independent path reorganization} as an estimate of $\tsp$.
If $\eps\cdot\mst<w(\fset)< (1/2+\eps)\cdot\mst$, we then apply the algorithm from \Cref{lem: spider_walk_estimation} to $\fset$.
If it reports that $\walkcost(w')\le \big(2-\frac{\eps^4}{2\log (1/\eps)}\big)\cdot \mst(w')$ (see the definition of $w'$ in \Cref{lem: spider_walk_estimation}), then we return $X=\big(2-\frac{\eps^5}{2\log (1/\eps)}\big)\cdot \mst$ as the estimate of $\tsp$.
Otherwise, go to \ref{WithMST_step_3}.
If $w(\fset)\le \eps\cdot\mst$, then go to \ref{WithMST_step_3}.
\iffalse{previous middle-point step}
\item 
For each maximal $\ell$-light vertex $v\in L_{\ell}$, we construct the middle subtree $T^{\midd}_v$ as described in \Cref{sec: light subtrees} and compute $\sum_{v\in L_{\ell}}w(T^{\midd}_v)$. 
If $\sum_{v\in L_{\ell}}w(T^{\midd}_v)\ge (1/4+\eps/2)\cdot \mst$, then we apply the algorithm from \Cref{obs: construct independent paths from mid-subtrees} and obtain a set $\pset$ of independent paths in $T$, and then apply the algorithm from \Cref{lem: independent path reorganization} to $\pset$ to get an estimate of $\tsp$. Otherwise, go to \ref{WithMST_step_3}.
\fi
%Let $\dset_{\sep}$ be the distribution on the set $\tset_{\sep}$ of separated subtrees, such that each tree $\hat T\in \tset_{\sep}$ has probability $w(\hat T)/w_{\sep}$. We perform the following operations for $100\log n$ times. Sample a separated subtree $\hat T$ from $\tset_{\sep}$ according to the distribution $\dset_{\sep}$, compute $\tsp(\hat T)$ by performing queries to the distances between all pairs of vertices in $V(T)$, and let $\gamma(\hat T)=\tsp(\hat T)/2w(\hat T)$.
%Let $\hat \gamma$ be the mean of all values $\gamma(\hat T)$ for $100\log n$ sampled subtrees $\hat T$. If $w_{\sep}\ge c_2\cdot \mst$ and $\hat\gamma\le 1-\eps$, then we return $(2-c_2\eps)\cdot \mst$ as the estimate of $\tsp$; otherwise go to \ref{WithMST_step_3}.
\item \label{WithMST_step_3} 
We apply the algorithm from \Cref{lem: estimate adv of a set of segments} to the set $\sset$ of segments of $T$ defined in \Cref{sec: light subtrees}, and parameter $\eps/1000$. If the algorithm reports that $\sum_{S\in \sset}\adv(S)\ge (\eps/1000)\cdot \mst$, then we return $X=(2-\eps/2000)\cdot \mst$ as the estimate of $\tsp$. Otherwise, go to \ref{WithMST_step_4}.
\item \label{WithMST_step_4} 
We apply the algorithm from \Cref{obs: compute k-extension} to compute a maximum $k$-extension $F$ of $T'$, and its weight $w^{\sf ext}_k(T')$. Denote $T''=T'\cup F$, and then we apply the algorithm from \Cref{lem: max special cover advantage} to tree $T''$.
If $\adv(T'')\ge (\eps/10)\cdot \mst$, then we return $X=(2-\eps/20)\cdot \mst$ as the estimate of $\tsp$.
If $\adv(T'')<  (\eps/10)\cdot \mst$, and $w(T'')>(1/2+\eps)\cdot \mst$, then we return $X=2\cdot \mst$ as the estimate of $\tsp$. Otherwise, go to \ref{WithMST_step_5}.
\item \label{WithMST_step_5} 
Return $X=2\cdot \mst$ as the estimate of $\tsp$.
\end{enumerate}

%\paragraph{Proof of Correctness.}

We first count the total number of queries performed by the algorithm. First, in Step \ref{WithMST_step_1} we invoke the algorithm from  \Cref{lem: max special cover advantage}, which performs $\tilde O(n^{1.5})$ queries, since the number of special vertices in $T'$ is at most $O(n/\ell)=O(\sqrt{n})$.
Second, in Step \ref{WithMST_step_2} we invoke the algorithm from \Cref{lem: independent path reorganization}, which performs $\tilde O(n)$ queries.
Third, in Step \ref{WithMST_step_3} we invoke the algorithm from \Cref{lem: estimate adv of a set of segments}, which performs in total $\tilde O(n^{1.5})$ queries, since each segment contains $O(\ell)=O(\sqrt{n})$ vertices.
Fourth, in Step \ref{WithMST_step_4} we invoke the algorithm from \Cref{lem: max special cover advantage}, which performs $\tilde O(n^{1.5})$ queries, since the number of special vertices in $T''$ is at most $O(n/\ell+k)=O(\sqrt{n})$.
Altogether, the algorithm performs in total $\tilde O(n^{1.5})$ queries.

We now show that the algorithm indeed returns a $(2-\Omega(\eps^5/\log (1/\eps)))$-approximation of $\tsp$.

\textbf{Case 1.} Assume that, in Step \ref{WithMST_step_1}, the algorithm from  \Cref{lem: max special cover advantage} returns that $\adv(T')\ge (\eps/10)\cdot \mst$. From \Cref{lem: cover_advantage}, we get that $\tsp\le 2\cdot \mst-(\eps/10)\cdot \mst/2=(2-\eps/20)\cdot \mst$, so our estimate $X=(2-\eps/20)\cdot \mst$ in this case is an $(2-\Omega(\eps))$-approximation of $\tsp$.

\textbf{Case 2.} Assume that, in Step \ref{WithMST_step_1}, the algorithm from \Cref{lem: max special cover advantage} returns that $\adv(T')\le (\eps/10)\cdot \mst$, and $w(T')\ge (\frac{1}{2}+\eps)\cdot \mst$. From \Cref{lem: cover_advantage lower bound}, we get that $\tsp\ge 2\cdot w(T')-2\cdot\adv(T')\ge  (1+\frac{9\eps}{5})\cdot \mst$, so our estimate $X=2\cdot \mst$ in this case is an $(2-\Omega(\eps))$-approximation of $\tsp$.

\textbf{Case 3.} Assume that, in Step \ref{WithMST_step_2}, the independent set $\fset$ of $800$-\snfls satisfies that $w(\fset)> (1/2+\eps)\cdot\mst$. Then from \Cref{lem: independent path reorganization}, the output is an $(2-\Omega(\eps^4/\log (1/\eps)))$-approximation of $\tsp$. Assume that $\eps\cdot\mst<w(\fset)< (1/2+\eps)\cdot\mst$ and $\walkcost(w')\le \big(2-\frac{\eps^4}{2\log (1/\eps)}\big)\cdot \mst(w')$, then from similar arguments in \Cref{subsec: Proof of independent path reorganization}, it is easy to show that the estimate $X=\big(2-\frac{\eps^5}{2\log (1/\eps)}\big)\cdot \mst$ in this case is a $\big(2-\frac{\eps^5}{2\log (1/\eps)}\big)$-approximation of $\tsp$.

\textbf{Case 4.} Assume that, in Step \ref{WithMST_step_3}, the algorithm from \Cref{lem: estimate adv of a set of segments} reports that $\sum_{S\in \sset}\adv(S)\ge (\eps/1000)\cdot \mst$. From \Cref{lem: cover_advantage}, we get that $\tsp\le (2-\eps/2000)\cdot \mst$, so our estimate $X=(2-\eps/2000)\cdot \mst$ in this case is a $(2-\Omega(\eps))$-approximation of $\tsp$.

\textbf{Case 5.} Assume that we return an estimate in Step \ref{WithMST_step_4}. Similar to Case 1 and Case 2, it is easy to show that our estimate is a $(2-\Omega(\eps))$-approximation of $\tsp$.

We assume from now on that Cases $1$ to $5$ do not happen.
In other words,
\begin{properties}{P}
\item $\adv(T')< (\eps/10)\cdot \mst$;
\label{prop_1}
\item the maximum weight of an independent set $\fset$ of $800$-\snfls in $T\setminus T'$ is at most $(\frac{1}{2}+\eps)\cdot \mst$;
\label{prop_2}
\item $w(T')< (\frac{1}{2}+\eps)\cdot \mst$; 
\label{prop_3}
\item $\adv(\sset):=\sum_{S\in \sset}\adv(S)\le (\eps/500)\cdot \mst$; and
\label{prop_4}
\item $w(T')+\text{ max weight of an }(100\sqrt n)\text{-extension of }T' < (\frac{1}{2}+\eps)\cdot \mst$.\label{prop_5}
\end{properties}
We will show that the above five properties implies that $\tsp\ge (1+\Omega(1))\cdot\mst$, which implies that our estimate $X=2\cdot \mst$ in Step \ref{WithMST_step_5} is a $(2-\Omega(1))$-approximation of $\tsp$.

We first prove the following lemma.

%To prove the correctness of the algorithm, we will modify $\pi$ several time and prove that the size of $\pi$ would not change much. Before we doing that, we first prove the following claim.
\begin{claim} \label{clm:in-edge}
Let $\hat T$ be a subtree of $T$ and let $\hat E$ be a set of edges that do not belong to $T$. For each edge $f \in E(\hat{T})$, we denote by $\cov_{\hat E}(f)$ the number of edges in $\hat E$ that covers $f$. Then for any integer $k \ge 2$ and for any real number $0<\delta<1/k$, 
\[w(\hat E) \ge \bigg(1-(k-1)\delta\bigg)\cdot \bigg(\sum_{f\in E(\hat T)}w(f)\cdot\min\set{k,\cov_{\hat E}(f)}\bigg) - \frac{\adv(\hat{T})}{2\delta}.\]
\end{claim}

\begin{proof}
Let $E'$ be a random subset of $\hat E$ that includes every edge of $\hat E$ independently with probability $2\delta$. 
Consider an edge $f\in E(\hat T)$. If $\cov_{\hat E}(f)=c$, then the probability that $f\in \cov(E')$ is $1-(1-2\delta)^c \ge 2\delta c - 2\delta^2c(c-1)$, and if $c\ge k$, then $\Pr[f\in \cov(E')]\ge 2\delta k - 2\delta^2k(k-1)$.
Therefore, from linearity of expectation, we get that
$\expect[w(E')]=2\delta\cdot w(\hat E)$ and 
\[
\begin{split}
\expect[w(\cov(E',\hat T))]\ge &\text{ } \sum_{f\in E(\hat T)}w(f)\cdot\bigg(2\delta \cdot\min\set{k,\cov_{\hat E}(f)} \cdot\bigg(1-\delta\cdot\big(\min\set{k,\cov_{\hat E}(f)}-1\big)\bigg)\bigg)\\
\ge &\text{ }\sum_{f\in E(\hat T)}w(f)\cdot\bigg(2\delta \cdot\min\set{k,\cov_{\hat E}(f)}\cdot (1-\delta(k-1))\bigg)\\
= &\text{ } 2\delta \cdot \big(1-\delta(k-1)\big)\cdot\bigg(\sum_{f\in E(\hat T)}w(f)\cdot \min\set{k,\cov_{\hat E}(f)}\bigg).
\end{split}\] 
From the definition of $\adv(\hat T)$,
$\adv(\hat T)\ge \expect[w(E')]-\expect[w(\cov(E',\hat T))]$. Now \Cref{clm:in-edge} follows from the above inequalities.
%
\iffalse
    It is sufficient to prove the case when each $e'$ appears in $E'$ exactly $\max\{k,\cov_E(e')\}$ times. Consider the random set $\tilde{E}$ such that each edge $e \in E$ is sampled with probability $2\delta$, and let $\hat{E}'$ be the set of edges in $\bar{T}$ that covered by $\hat{E}$. For any edge $e'$, if $\cov_E(e')=i \le k$, then $e' \in \hat{E}'$ with probability $1-(1-2\delta)^i \ge 2\delta i - 2\delta^2i(i-1)$, the inequality is due to the assumption that $\delta<\frac{1}{k}$. If $\cov_E(e')\ge k$, then $e' \in \hat{E}'$ with probabilty at least $1-(1-2\delta)^k = 2\delta k -2\delta^2k(k-1)$. We have $\expect[w(\hat{E})]=2\delta w(E)$ and $\expect[w(\hat{E})']\ge 2\delta(1-(k-1)\delta)w(E')$, so there exist a subset $\hat{E} \subseteq E$ such that $w(\hat{E}') - w(\hat{E}) \ge 2\delta((1-(k-1)\delta)w(E')-w(E))$ where $\hat{E}'$ is the set of edges $e' \in \bar{T}$ that are covered by $\hat{E}$. On the other hand, by definition of cover advantage, $\adv(\bar{T}) \ge w(\hat{E}')-w(\hat{E})$. So $\adv(\bar{T}) \ge 2\delta((1-(k-1)\delta)w(E')-w(E))$, which means $w(E) \ge (1-(k-1)\delta)w(E) - \frac{1}{2\delta}\adv(S)$.\fi
\end{proof}

%\subsection{Useful Tools}

%We partition set $E_{\pi'}$ into three subsets as follows: set $E_1$ contains all edges of $E_{\pi'}$ connecting a pair of vertices of $T\setminus T'$, set $E_2$ contains all edges of $E_{\pi'}$ connecting a vertex of $T'$ to a vertex of $T\setminus T'$, and set $E_3$ contains all edges of $E_{\pi'}$ connecting a pair of vertices of $T'$, so $E_{\pi'}=E_1\cup E_2\cup E_3$.
%
%We further partition set $E_1$ as follows: set $E'_1$ contains all edges of $E_1$ connecting a pair of vertices from the same segment of $\sset$, and set $E''_1$ contains all edges of $E_1$ connecting a pair of vertices from different segments of $\sset$, so $E_1=E'_1\cup E''_1$.

Let $\pi$ be an optimal TSP-tour that visits all vertices of $V$, so $\tsp=w(\pi)$.
We construct a set $E'$ of edges as follows. We start with the set $E(\pi)$ of the edges traversed in tour $\pi$. Then for each edge $e\in E(\pi)$ such that either both endpoints of $e$ belong to $V(T')$ or both endpoints of $e$ belong to the same segment of $\sset$, we replace $e$ with edges in $E(P^T_e)$, the tree path in $T$ connecting the endpoints of $e$. Denote by $E'_{\pi}$ the resulting set after this step. Note that $E'_{\pi}$ may contain many copies of the same edge. 
Finally, for each edge $e$ that has more than $2$ copies contained in $E'_{\pi}$, if $E'_{\pi}$ contains an odd number of copies of $e$, then we delete all but one copies from $E'_{\pi}$; if $E'_{\pi}$ contains an even number of copies of $e$, then we delete all but two copies from $E'_{\pi}$. Denote by $E'$ the resulting set of edges, so each edge has at most $2$ copies contained in $E'$. 
It is easy to verify that the graph induced by edges of $E'$ (with multiplicity) is connected and Eulerian.
Additionally, we prove the following observations.

\begin{observation}
\label{obs: new set weight bound}
$w(E')-w(E_{\pi})\le \sqrt{\eps}\cdot\mst$.
\end{observation}
\begin{proof}
Let $S$ be a segment in $\sset$. We denote by $E_{\pi}(S)$ the subset of $E_{\pi}$ that contains all edges with both endpoints in $S$.
We apply \Cref{clm:in-edge} with $\delta=\sqrt{\eps}/5$ and $k=2$ (so $\delta<1/k$) to set $E_{\pi}(S)$ of edges and segment $S$, and get
\[
w(E_{\pi}(S))\ge \bigg(1-\frac{\sqrt{\eps}}{5}\bigg)\cdot \bigg(\sum_{f\in E(S)}w(f)\cdot\min\set{2,\cov_{E_{\pi}(S)}(f)}\bigg) - \frac{\adv(S)}{2\sqrt\eps/5}.
\]
On the other hand, we denote by $E'(S)$ the subset of the resulting set $E'$ that contains all edges with both endpoints in $S$. Then from the construction of $E'$, it is easy to see that, for each edge $f\in E(S)$, the number of its copies contained in $E'(S)$ is at most $\min\set{2,\cov_{E_{\pi}(S)}(f)}$. Therefore,
\[
w(E_{\pi}(S))\ge \bigg(1-\frac{\sqrt{\eps}}{5}\bigg)\cdot 
w(E'(S)) - \frac{\adv(S)}{2\sqrt\eps/5}.
\]
Similarly, for all other segments in $\sset$ and for tree $T'$ we can get similar inequalities. Altogether, if we denote by $E^*_{\pi}$ the subset of $E_{\pi}$ that contains all edges such that neither both endpoints belong to some segment in $\sset$ nor both endpoints belong to $V(T')$, then clearly $E'=E^*_{\pi}\cup \big(\bigcup_{S\in \sset}E'(S)\big)$, and so
\[
\begin{split}
w(E_{\pi})= & \text{ }
w(E^*_{\pi})+w(E_{\pi}(T'))+\sum_{S\in \sset}w(E_{\pi}(S))\\
\ge & \text{ } w(E^*_{\pi})+ \bigg(1-\frac{\sqrt{\eps}}{5}\bigg)\cdot
w(E'(T')) - \frac{\adv(T')}{2\sqrt\eps/5}+\sum_{S\in \sset}\bigg(\bigg(1-\frac{\sqrt{\eps}}{5}\bigg)\cdot
w(E'(S)) - \frac{\adv(S)}{2\sqrt\eps/5}\bigg)\\
= & \text{ } w(E')- \frac{\sqrt{\eps}}{5}\cdot
\bigg(w(E'(T')) + \sum_{S\in \sset}
w(E'(S))\bigg) 
-\frac{\adv(T')+\sum_{S\in \sset}\adv(S)}{2\sqrt\eps/5}\\
\ge & \text{ } w(E')- \frac{\sqrt{\eps}}{5}\cdot
\mst 
-\frac{(\eps/10)\cdot\mst +(\eps/500)\cdot\mst}{2\sqrt\eps/5}> w(E')- \sqrt{\eps}\cdot
\mst,
\end{split}
\]
where the last but one inequality utilizes Properties \ref{prop_1} and \ref{prop_4}.
\end{proof}

\begin{observation} \label{clm:upper}
	There is a subset $E^*$ of $E'$, such that (i) $E(T')\subseteq \cov(E^*, T)$; and (ii) $E^*$ contains $O(n/\ell)$ edges with at least one endpoint in $V(T\setminus T')$. 
\end{observation}

\begin{proof}
Recall that in \Cref{lem: partitioning into segments} we partitioned tree $T'$ into a set $\pset$ of vertex-disjoint paths, such that, for each $P\in \pset$, either $\pset$ contains a single vertex (that is either a special vertex or a giant vertex) of $T'$, or contains no special vertices of $T'$. We construct the set $E^*$ of edges by processing paths of $\pset$ one-by-one, as follows.

Initially, we set $E^*=\emptyset$.
Consider a path $P\in \pset$. Assume first that path $P$ consists of a single vertex $v$ of $T'$. Then for each edge $f$ of $T'$ that is incident to $v$, we arbitrarily pick an edge $e\in E'$ such that $f\in \cov(e, T')$, and add $e$ to $E^*$. Clearly, such an edge exists since the graph induced by edges of $E'$ is connected and Eulerian on $V$.
Assume now that path $P$ contains no special vertices of $T'$. We denote by $f_1(P),f_2(P)$ the two edges of $T'$ that are incident to the endpoints of $P$. We then let $e_1(P)$ be an edge in $E'$ that, among all edges $e$ in $E'$ that covers $f_1(P)$, maximizes $|\cov(e,T')\cap E(P)|$, and we define edge $e_2(P)$ similarly for $f_2(P)$. 
If $E(P)\subseteq \cov(\set{e_1(P),e_2(P)}, T')$, then we add edges $e_1(P),e_2(P)$ to set $E^*$.
If $E(P)\not\subseteq \cov(\set{e_1(P),e_2(P)}, T')$, then let $P'$ be the subpath of $P$ that is not covered by either $e_1(P)$ or $e_2(P)$. From the construction of $E'$, it is easy to verify that all edges of $P'$ belong to $E'$ (since otherwise they are covered by an non-tree-edge with both endpoints lying in $T_P$, the segment corresponding to path $P$, a contradiction). We then add edges $e_1(P),e_2(P)$ and all edges of $E(P')$ to set $E^*$. This completes the construction of set $E^*$.

It is easy to verify that $E(T')\subseteq E^*$ from the construction. We now show that the number of edges in $E^*$ with an endpoint in $V(T\setminus T')$ is $O(n/\ell)$. In fact, for a path $P\in \pset$ that does not contain a special vertex of $T'$, we added to $E^*$ at most two edges  with an endpoint in $V(T\setminus T')$; and for all single-special-vertex paths of $\pset$, since $T'$ contains $O(n/\ell)$ leaves, we added a total of $O(n/\ell)$ edges to $E^*$. Therefore, the number of edges in $E^*$ with an endpoint in $V(T\setminus T')$ is $O(n/\ell)$. This completes the proof of \Cref{clm:upper}.
\iffalse
	We cover every segment one by one. By definition of a segment, it either contains only contians a special vertex in the $T'$, or a path in $T'$ that does not contain a special vertex. If a segment contains a special vertex $v$ in $T'$, then for each edge $e$ in $T'$ that incident on $v$, we select one edge in $\pi'$ that covers $e$. If a segment contains a path in $T'$ that does not contain a special vertex, let $e_1$, $e_2$ be two edges in $T'$ that connects to the path (but not in the path). We select an edge $e'_1$ that covers $e_1$ and also covers the largest number of edges in the path, and an edge $e'_2$ that covers $e_2$ and also covers the largest number of edges in the path. If $e'_1$ and $e'_2$ together do not cover the whole path, then for any edge that is not covered by these two edges, it is only covered by an edge that both endpoint is inside the segment. However, since $\pi'$ does not contain edges such that both endpoints in the same segment but not a tree path. So these edges are covered by themselves in $T'$, and we also add these edges into the set.
	
	Now we consider the number of edges that contains an endpoint not in $T'$. These edges either cover an endpoint of a maximal induced subpath in $T'$ or cover an edges that inciedent to a segment. Since the number of maximal induced subpath in $T'$ and the number of segments in $T'$ are both $O(n/\ell)$, the number of such edges is also $O(n/\ell)$.\fi
\end{proof}

%We modify $\pi$ as follows: Let $E'$ be the set of edges in $\pi$ that both endpoints are in the same segment or both endpoints are in $t'$. For each edge $e \in E'$, we replace $E'$ with the edges in $T$ that are covered by $e$. Let $\hat{E}$ be the resulting multi-set. Furthermore, for each edge $e$ in $T$ that appears more than twice in $\hat{E}$, we keep deleting 2 copies of $e$ until it appears one or two times in $\hat{E}$. By Claim~\ref{clm:in-edge}, for any $0<\delta<0.5$, $w(E') \ge (1-\delta)w(\hat{E}) - \frac{1}{2\delta}(\sum_{S \in \sset} \adv(S) + \adv(T'))$. Let $\delta=\sqrt{\eps}/5$, since each edge in $T$ appears at most twice in $\hat{E}$, $\adv(T') \le \eps /10$ and $\adv(\sset) \le (\eps/500) \mst$, we have $w(\hat{E})-w(E') \le 2\delta \mst + \frac{1}{2\delta}(\eps/5)\mst \le \sqrt{\eps}\mst$. On the other hand, let $\pi'=(\pi \setminus E' ) \cup \hat{E}$, $\pi'$ is still a tour that visits each vertex at least once and we have $w(\pi')-w(\pi) \le \sqrt{\eps}\mst$. 

%There are four types of edges in $\pi'$, the first type is the edges that both endpoints are in $T'$, let $E_4$ be the set of such edges. The second type is the edges that exactly one endpoint is in $T'$, let $E_1$ be the set of such edges. The third type is the edges that both endpoints are not in $T'$, and they belongs to different segments, let $E_2$ be the set of such edges. The last type is the edges that both endpoints are in the same segment, by definition of $\pi'$, these edges are all tree edges. Let $E_3$ be the set of last type of edges.

Let $E''$ be the subset of $E'$ that contains all  edges $e$ such that (i) both endpoints of $e$ belong to the same segment of $\sset$; and (ii) there is no edge $e'\in E', e'\ne e$, such that $e\in \cov(e',T)$. In other words, edges of $E''$ are the bridges in the graph induced by edges of $E'$ (ignoring multiplicities), and it is easy to verify that each edge has two copies contained in $E'$. We prove the following simple observation.

\begin{observation} \label{clm:tree-edge-cover}
$w(E') \ge \mst + w(E'')$.
\end{observation}

\begin{proof}
Let $T_1,\ldots,T_r$ be the connected components of graph $T\setminus E''$. From the definition of set $E''$, it is easy to see that each edge of $E'\setminus E''$ has both endpoints belong to the same connected component $T_i$, and moreover, if we denote by $E'_i$ the subset of $E'$ that contains all edges of $E'\setminus E''$ with both endpoints in $T_i$, then the graph induced by edges of $E'_i$ (with multiplicity) is connected and Eulerian on $V(T_i)$. Therefore, $w(E'_i)\ge w(T_i)$ since $T_i$ is the MST on $V(T_i)$. Altogether,
\[w(E')\ge 2\cdot w(E'')+\sum_{1\le i\le r}w(E'_i)\ge 2\cdot w(E'')+\sum_{1\le i\le r}w(T_i)\ge  w(E'')+\mst.\]
%If an edge is covered by an edges that is not in $E_3$, then all of its ancesters in the segment is covered by that edge. So if we remove edges in $E_3$ from $T$, the resulting graph is still a tree, we call it $T'_3$. Let $\pi'_3$ be the tour that remove the edges in $E_3$. $\pi'_3$ is a tour of $T'_3$ that visits each vertex at least once. So $w(\pi'_3)\ge w(T'_3)$. On the other hand, since each edge in $T$ is covered by $\pi'$ at least twice, we have $w(\pi') - w(\pi'_3) \ge 2w_3$, so $w(\pi') \ge w(T'_3)+2w_3 = \mst +w_3$.
\end{proof}

From \Cref{clm:tree-edge-cover}, if $w(E'')>2\sqrt{\eps}\cdot \mst$, then $w(E')>(1+2\sqrt{\eps})\cdot\mst$ and then from \Cref{obs: new set weight bound}, $w(\pi)>(1+2\sqrt{\eps})\cdot\mst-\sqrt{\eps}\cdot \mst \ge (1+\sqrt\eps)\cdot\mst$. Therefore, our estimation $X=2\cdot\mst$ in this case is a $(2-\Omega(\sqrt \eps))$-approximation of $\tsp$.
We assume from now on that $w(E'') \le 2\sqrt{\eps}\cdot \mst$.

We use the following lemma, whose proof is deferred to \Cref{Proof of lem:buttom}.

\begin{lemma} \label{lem:buttom}
Assume $w(E'')\le 2\sqrt{\eps}\cdot \mst$.
Let $T''$ be a subtree of $T$ such that $T'\subseteq T''$ and $w(T''\setminus T') \le \mst/20$.
Let $\hat E_{T''}$ be the subset of $E'$ that contains all edges with at least one endpoint in $V(T \setminus T'')$. Then 
\[w(\hat E_{T''}) \ge \bigg(1.9- 5\cdot\frac{w(T''\setminus T')}{\mst}\bigg)\cdot w(T\setminus T'') - 10\sqrt{\eps}\cdot\mst.\]
\end{lemma}

%Denote $w(T')=w(T')$ and $(w(T)-w(T'))=w(T)-w(T')$.
Assume that $w(T')\le 0.47\cdot\mst$, so $w(T\setminus T')\ge 0.53\cdot\mst$.
Applying \Cref{lem:buttom} to subtree $T''=T'$ (so $w(T''\setminus T')=0$ and $w(T\setminus T'')\ge 0.53\cdot\mst$),
we get that $w(E')>0.53\cdot \mst\cdot 1.9-10\sqrt{\eps}\cdot\mst$, and then from \Cref{obs: new set weight bound}, $w(\pi)> (1.007 - 11\sqrt{\eps})\cdot\mst$,
and so our estimation $X=2\cdot\mst$ in this case is a $(2-\Omega(1))$-approximation of $\tsp$. 
We assume from now on that $w(T') > 0.47\cdot\mst$.

Let $E^*$ be the set of edges in \Cref{clm:upper}. 
%From the second property in \Cref{clm:upper}, $E^*\cap E''=\emptyset$.
Since $E(T')\subseteq \cov(E^*, T)$, from the definition of $\adv(T')$ and Property \ref{prop_1}, we get that $w(E^*) \ge w(T')-\adv(T') \ge (0.47 - \eps/10) \cdot\mst$. On the other hand, let $F$ be the subgraph of $T$ induced by edges of $\cov(E^*,T)$, so $T'\subseteq F$. Since $E^*$ only contains $O(n/\ell)$ edges with an endpoint in $V(T\setminus T')$, it is easy to verify that $F\setminus T'$ is an $O(n/\ell)$-extension of $T'$. From Property \ref{prop_3}, $w(F) \le (1/2+\eps)\cdot \mst$, so $w(F\setminus T') \le (0.03+\eps)\cdot\mst$. 

Last, applying \Cref{lem:buttom} to subtree $F$ (note that $w(F\setminus T') \le \mst/20$), and noticing that $\hat E_F\cap E^*=\emptyset$ (since from the definition of $F$, every edge of $E^*$ has both endpoints in $F$), we get that
\[
\begin{split}
w(E')\ge w(E^*)+w(\hat E_F) \ge & \text{ } \bigg(0.47-\frac{\eps}{10}\bigg)\cdot\mst+(1.9-0.2)\cdot(w(T)-w(F)) - 10\sqrt{\eps}\cdot\mst\\
 \ge & \text{ } 
\bigg(0.47-\frac{\eps}{10}\bigg)\cdot\mst+ 1.7 \cdot \bigg(\frac{1}{2}-\eps\bigg)\cdot\mst - 10\sqrt{\eps}\cdot\mst\\
 \ge & \text{ }   1.32\cdot\mst - 12\sqrt{\eps}\cdot\mst.
\end{split}
\] 
Then from \Cref{obs: new set weight bound}, $w(\pi) \ge (1.32-13\sqrt{\eps})\cdot\mst>1.3\cdot\mst$, so our estimation $X=2\cdot\mst$ in this case is a $(2-\Omega(1))$-approximation of $\tsp$.
This completes the proof of \Cref{thm: main with MST}.

\subsection{Proof of Lemma~\ref{lem:buttom}}
\label{Proof of lem:buttom}

Recall that we are given a subtree $T''$ of $T$ such that $T'\subseteq T''$.
For convenience, we denote $\hat E=\hat E_{T''}$ to be the subset of $E'$ that contains all edges with at least one endpoint in $V(T\setminus T'')$.
Throughout the subsection, we set constant $c=400$.

Denote $F=T\setminus T''$. Clearly, $F$ is a subgraph of the graph obtained by taking the union of all segments, namely $F\subseteq \bigcup_{S\in \sset} S$. 
We partition set $\hat E$ into four subsets as follows: set $\hat E_1$ contains all edges of $\hat E$ with exactly one endpoint in $F$; and set $\hat E_2$ contains all edges of $\hat E$ with its endpoints lying in different segments; set $\hat E_3$ contains all edges $e$ of $\hat E$ such that both endpoints of $e$ belong to the same segment and $e$ as a tree edge is covered by some edge of $\hat E_1\cup \hat E_2$; and set $\hat E_4$ contains all edges $e$ of $\hat E$ such that both endpoints of $e$ belong to the same segment and $e$ as a tree edge is not covered by any edge of $\hat E_1\cup \hat E_2$. Clearly, $\hat E=\hat E_1\cup \hat E_2\cup \hat E_3 \cup \hat E_4$, and from the construction of $\hat E$ and $E'$, sets $\hat E_3$ and $\hat E_4$ are subsets of $E(F)$. Moreover, $\hat E_4\subseteq E''$, so $w(\hat E_4)\le 2\sqrt{\eps}\cdot \mst$, and from the previous discussion on set $E''$, each edge of $\hat E_4$ has at least two copies contained in $\hat E_4$. In other words, if we denote by $F_4$ the set of edges without multiplicity in $\hat E_4$, then $w(\hat E_4)\ge 2\cdot w(F_4)$.

Let $E^*$ be the subset of $E(F)$ that contains all edges that are covered by either at least $c$ edges of $\hat E_1$ or at least $c$ edges of $\hat E_2$.
We distinguish between the following two cases, depending on whether or not the total weight of edges in $E^*$ is large enough.

\subsubsection*{Case 1. $w(E^*)> w(F\setminus  F_4)\cdot (2/c)$}

We construct a new set $\hat E'_1$ of edges from set $\hat E_1$, similarly to the construction of set $E'$ from set $E(\pi)$, as follows.
We start with the set $\hat E_1$. For each edge $e\in \hat E_1$ such that both endpoints of $e$ belong to the same segment, we replace $e$ with edges in $\cov(e,F)$, the set of edges in $F$ that are covered by $e$. 
%Denote by $E^*_1$ the resulting set after this step. Note that $E'_{\pi}$ may contain many copies of the same edge. 
Finally, for each edge $e$ that has more than $c$ copies contained in the current edge set, if the current set contains an odd number of copies of $e$, then we delete from it all but one copies of $e$; if the current set contains an even number of copies of $e$, then we delete from it all but two copies of $e$. Denote by $\hat E'_1$ the resulting set of edges, so each edge has at most $c$ copies contained in $\hat E'_1$. 
From \Cref{clm:in-edge} (by setting $\delta=\sqrt{\eps}/5c$ and $k=c$),  and observing that, for each edge $f\in E(F)$, the number of its copies contained in $\hat E'_1$ is at most $\min\set{c, \cov_{\hat E_1}(f)}$,
we get that
\[
\begin{split}
w(\hat E_1)\ge & \text{ } \bigg(1-(c-1)\cdot\frac{\sqrt{\eps}}{5c}\bigg)\cdot\bigg(\sum_{f\in E(F)}w(f)\cdot\min\set{c,\cov_{\hat E_1}(f)}\bigg) - \frac{\sum_{S\in \sset}\adv(S)}{2\sqrt{\eps}/5c}\\
\ge & \text{ }
\bigg(1-\frac{\sqrt{\eps}}{5}\bigg)\cdot w(\hat E'_1) - \frac{(\eps/500)\cdot \mst}{2\sqrt{\eps}/5c}
\ge w(\hat E'_1)-3\sqrt{\eps}\cdot\mst,
\end{split}
\]
where the second inequality utilizes Property \ref{prop_4}.

We then construct a new set $\hat E'_2$ from $\hat E_2$ in the same way as the construction of $\hat E'_1$ of edges from $\hat E_1$. From similar arguments, and noticing that each edge of $\hat E_2$ covers the edges of at most two connected component of $F$, we get that 
$2\cdot w(\hat E_2)\ge  w(\hat E'_2)-3\sqrt{\eps}\cdot\mst$.

Note that, from the construction of new sets $\hat E'_1, \hat E'_2$, it is easy to verify that every edge in $E^*$ has at least $(c-1)$ of its copies contained in set $\hat E'_1\cup \hat E'_2$. Also, every edge of $F\setminus F_4$ has at least two of its copies contained in $\hat E'_1\cup \hat E'_2\cup \hat E_3$. Therefore,
\[
\begin{split}
w(\hat E) = & \text{ } w(\hat E_1)+w(\hat E_2)+w(\hat E_3)+w(\hat E_4)\\
\ge & \text{ } w(\hat E'_1)-3\sqrt{\eps}\cdot\mst + \frac{w(\hat E'_2)-3\sqrt{\eps}\cdot\mst}{2}+w(\hat E_3)+2\cdot w(F_4)\\
\ge & \text{ }\frac{w(\hat E'_1\cup \hat E'_2\cup \hat E_3)}{2}-5\sqrt{\eps}\cdot\mst+2\cdot w(F_4)\\
\ge & \text{ }\bigg(\frac{c-1}{2}\bigg)\cdot w(E^*)+w(F\setminus F_4)-5\sqrt{\eps}\cdot\mst+2\cdot w(F_4)\\
\ge & \text{ }\bigg(\frac{c-1}{2}\bigg)\cdot w(F\setminus F_4)\cdot \frac{2}{c}+w(F\setminus F_4)-5\sqrt{\eps}\cdot\mst+2\cdot w(F_4)\\
\ge & \text{ }\bigg(\frac{2c-1}{c}\bigg)\cdot w(F)-5\sqrt{\eps}\cdot\mst,
\end{split}
\]
which implies that the required inequality of \Cref{lem:buttom} holds.

\subsubsection*{Case 2. $w(E^*)\le w(F\setminus F_4)\cdot (2/c)$}

Define forest $F'=F\setminus (F_4\cup E^*)$ and denote $\eta=w(F)/\mst$.
Since $\eta< 0.05$, and $w(T')\le (\frac{1}{2}+\eps)\cdot \mst$, $w(F')\ge 0.44\cdot\mst$.
We claim that 
\begin{equation}
\label{eqn: only one}
w(\hat E'_1)+w(\hat E_2)+w(\hat E_3)\ge (1.91-5\eta)\cdot w(F').
\end{equation}
Note that, if the inequality \ref{eqn: only one} is true, then since we have assumed that $w(E^*)\le w(F\setminus F_4)\cdot (2/c)$,
\[
\begin{split}
w(\hat E) = & \text{ } w(\hat E_1)+w(\hat E_2)+w(\hat E_3)+w(\hat E_4)\\
\ge & \text{ } w(\hat E'_1)-3\sqrt{\eps}\cdot \mst +w(\hat E_2)+w(\hat E_3)+2 \cdot w(F_4)\\
\ge & \text{ } (1.91-5\eta)\cdot w(F')+2\cdot w(F_4)-3\sqrt{\eps}\cdot \mst\\
\ge & \text{ } (1.9-5\eta)\cdot w(F\setminus F_4)+2 \cdot w(F_4)-3\sqrt{\eps}\cdot \mst\\
\ge & \text{ } (1.9-5\eta)\cdot w(F)-3\sqrt{\eps}\cdot \mst,
\end{split}
\]
which implies that the required inequality of \Cref{lem:buttom} holds.

The remainder of this subsection is dedicated to the proof of inequality \ref{eqn: only one}.
We first prove the following observation.

\begin{observation}
$F'=F\setminus (F_4\cup E^*)$ is an independent set of $2c$-\snfls.
\end{observation}
\begin{proof}
On the one hand, note that each edge of $F\setminus E^*$ is covered at most $2c$ times by edges of $\hat E_1\cup \hat E_2$. Also note that, from the definition of sets $\hat E_1$ and $\hat E_2$, if an edge of $F\setminus E^*$ is covered by some edge of $\hat E_1\cup \hat E_2$, then all its ancestors are also covered by the same edge. Therefore, $F\setminus E^*$ is an independent set of $2c$-\snfls. On the other hand, for the same reasons, it is easy to verify that edges of $F_4$ are at the bottom of each connected component of $F$. Therefore, $F'=(F\setminus E^*)\setminus F_4$ is also an independent set of $2c$-\snfls.
\end{proof}

Let $T_{F'}$ be the skeleton of $F'$. From the construction of sets $\hat E'_1, \hat E_2$, it is easy to verify that the graph induced by edges of $\hat E'_1\cup \hat E_2 \cup \hat E_3$ is connected and Eulerian on $V(T_{F'})$, so the minimum special walk cost on $T_{F'}$ is at most $w(\hat E'_1\cup \hat E_2 \cup \hat E_3)$.

Since $w(\hat E_4)\le 2\sqrt{\eps}\cdot \mst$, 
\[
\begin{split}
w(F')\ge & \text{ } \bigg(1-\frac{2}{c}\bigg)\cdot w(F\setminus F_4)\ge \bigg(1-\frac{2}{c}\bigg)\cdot\bigg(w(T\setminus T')-\eta\cdot\mst -2\sqrt{\eps}\cdot \mst\bigg)\\
\ge & \text{ } w(T\setminus T')-(\eta+2\sqrt{\eps}+2/c)\cdot\mst.
\end{split}
\]

Let $F^*$ be the $(2c)$-\snfl computed in Step \ref{WithMST_step_2} of the algorithm. From the maximality of $F^*$, $w(F^*)\ge w(F')$. Let $T_{F^*}$ be the skeleton of $F^*$. Consider the special walk on $T_{F^*}$ constructed as follows. We start with the tour induced by all edges of $\hat E'_1\cup \hat E_2 \cup \hat E_3\cup \hat E_4$, and delete all vertices that do not belong to $F^*$. 
Assume for contradiction that $w(\hat E'_1\cup \hat E_2 \cup \hat E_3)<(1.91-5\eta)\cdot w(F')$. Then the cost of such a walk is at most
\[
\begin{split}
w(\hat E'_1\cup \hat E_2 \cup \hat E_3\cup \hat E_4)+2\cdot (\eta+2\sqrt{\eps}+4/c)\cdot\mst 
< & \text{ } (1.91-5\eta)\cdot w(F')+ (2\eta+4\sqrt{\eps}+4/c)\cdot\mst\\
< & \text{ } 1.91\cdot w(F')+ 0.02\cdot \mst\\
< & \text{ } 1.99\cdot w(F')\\
\le & \text{ } 1.99\cdot w(F^*),
\end{split}
\]
where the last but one step uses the property that $w(F')>0.44\cdot\mst$ (and so $5\eta\cdot w(F')>2\eta\cdot \mst$).
This is a contradiction to the fact that the minimum special walk cost on $F^*$ is at least $(2-10^3\cdot c_0\cdot\eps)\cdot w(F^*)$. Therefore, inequality \ref{eqn: only one} holds. This completes the proof of the correctness of the algorithm.

\section{Future Directions}

%In this work, we studied the problems of MST and TSP cost estimation in the streaming and query settings, via a novel notion called \emph{cover advantage}.
In this work, we studied the problems of MST and TSP cost estimation in the streaming
and query settings. For TSP cost estimation, we introduced and utilized a novel notion called \emph{cover advantage} that may prove useful for solving this problem in other computational models also.
In the streaming setting, an interesting open problem is to obtain a one-pass $o(n^2)$-space $(2-\eps)$-approximate estimation of TSP cost in the metric stream.
In the query model, we believe a major open problem is to obtain an $o(n^2)$-query $(2-\eps)$-approximate estimation of TSP-cost in general metrics.

\section*{Acknowledgements}

We thank Santhoshini Velusamy for pointing to us a mistake in a previous version of the paper.

%\newpage
\bibliography{REF}

\begin{thebibliography}{ORRR12}

\bibitem[Beh21]{behnezhad2021time}
Soheil Behnezhad.
\newblock Time-optimal sublinear algorithms for matching and vertex cover.
\newblock {\em arXiv preprint arXiv:2106.02942}, 2021.

\bibitem[BJKS04]{Bar-YossefJKS04}
Ziv Bar{-}Yossef, T.~S. Jayram, Ravi Kumar, and D.~Sivakumar.
\newblock An information statistics approach to data stream and communication
  complexity.
\newblock {\em J. Comput. Syst. Sci.}, 68(4):702--732, 2004.

\bibitem[BR14]{BravermanR14}
Mark Braverman and Anup Rao.
\newblock Information equals amortized communication.
\newblock {\em {IEEE} Trans. Inf. Theory}, 60(10):6058--6069, 2014.

\bibitem[BRRS23]{behnezhad2023sublinear}
Soheil Behnezhad, Mohammad Roghani, Aviad Rubinstein, and Amin Saberi.
\newblock Sublinear algorithms for tsp via path covers.
\newblock {\em arXiv preprint arXiv:2301.05350}, 2023.

\bibitem[Chr76]{christofides1976worst}
Nicos Christofides.
\newblock Worst-case analysis of a new heuristic for the travelling salesman
  problem.
\newblock Technical report, Carnegie-Mellon Univ Pittsburgh Pa Management
  Sciences Research Group, 1976.

\bibitem[CKK20]{chen2020sublinear}
Yu~Chen, Sampath Kannan, and Sanjeev Khanna.
\newblock Sublinear algorithms and lower bounds for metric tsp cost estimation.
\newblock {\em arXiv preprint arXiv:2006.05490}, 2020.

\bibitem[CQ17]{CQ17}
Chandra Chekuri and Kent Quanrud.
\newblock Approximating the held-karp bound for metric {TSP} in nearly-linear
  time.
\newblock In {\em 58th {IEEE} Annual Symposium on Foundations of Computer
  Science, {FOCS} 2017, Berkeley, CA, USA, October 15-17, 2017}, pages
  789--800, 2017.

\bibitem[CQ18]{CQ18}
Chandra Chekuri and Kent Quanrud.
\newblock Fast approximations for metric-{TSP} via linear programming.
\newblock {\em CoRR}, abs/1802.01242, 2018.

\bibitem[CRT05]{ChazelleRT05}
Bernard Chazelle, Ronitt Rubinfeld, and Luca Trevisan.
\newblock Approximating the minimum spanning tree weight in sublinear time.
\newblock {\em {SIAM} J. Comput.}, 34(6):1370--1379, 2005.

\bibitem[CS09]{czumaj2009estimating}
Artur Czumaj and Christian Sohler.
\newblock Estimating the weight of metric minimum spanning trees in sublinear
  time.
\newblock {\em SIAM Journal on Computing}, 39(3):904--922, 2009.

\bibitem[Gao18]{Gao18}
Zhihan Gao.
\newblock On the metric s-t path traveling salesman problem.
\newblock {\em {SIAM} Review}, 60(2):409--426, 2018.

\bibitem[Har11]{prahladhlecture}
Prahladh Harsha.
\newblock Lecture notes on communication complexity.
\newblock 2011.

\bibitem[KKG21]{karlin2021slightly}
Anna~R Karlin, Nathan Klein, and Shayan~Oveis Gharan.
\newblock A (slightly) improved approximation algorithm for metric tsp.
\newblock In {\em Proceedings of the 53rd Annual ACM SIGACT Symposium on Theory
  of Computing}, pages 32--45, 2021.

\bibitem[KLS15]{karpinski2015new}
Marek Karpinski, Michael Lampis, and Richard Schmied.
\newblock New inapproximability bounds for tsp.
\newblock {\em Journal of Computer and System Sciences}, 81(8):1665--1677,
  2015.

\bibitem[KNR99]{kremer1999randomized}
Ilan Kremer, Noam Nisan, and Dana Ron.
\newblock On randomized one-round communication complexity.
\newblock {\em Computational Complexity}, 8(1):21--49, 1999.

\bibitem[Lin91]{Lin91}
Jianhua Lin.
\newblock Divergence measures based on the shannon entropy.
\newblock {\em {IEEE} Trans. Inf. Theory}, 37(1):145--151, 1991.

\bibitem[MM18]{MnichM18}
Matthias Mnich and Tobias M{\"{o}}mke.
\newblock Improved integrality gap upper bounds for traveling salesperson
  problems with distances one and two.
\newblock {\em European Journal of Operational Research}, 266(2):436--457,
  2018.

\bibitem[ORRR12]{onak2012near}
Krzysztof Onak, Dana Ron, Michal Rosen, and Ronitt Rubinfeld.
\newblock A near-optimal sublinear-time algorithm for approximating the minimum
  vertex cover size.
\newblock In {\em Proceedings of the twenty-third annual ACM-SIAM symposium on
  Discrete Algorithms}, pages 1123--1131. Society for Industrial and Applied
  Mathematics, 2012.

\bibitem[Viz64]{vizing1964estimate}
Vadim~G Vizing.
\newblock On an estimate of the chromatic class of a p-graph.
\newblock {\em Discret Analiz}, 3:25--30, 1964.

\bibitem[YYI12]{yoshida2012improved}
Yuichi Yoshida, Masaki Yamamoto, and Hiro Ito.
\newblock Improved constant-time approximation algorithms for maximum matchings
  and other optimization problems.
\newblock {\em SIAM Journal on Computing}, 41(4):1074--1093, 2012.

\end{thebibliography}

\end{document}